\documentclass[11pt]{article}
\usepackage[margin=1.35in]{geometry}

\usepackage{graphicx}
\usepackage{float}
\usepackage{multirow}
\usepackage{booktabs}
\usepackage{bbm}
\usepackage{amsmath,amssymb,amsthm}
\usepackage[nameinlink,capitalize]{cleveref}
\usepackage[authoryear,round]{natbib}
\setcitestyle{aysep={,},yysep={;}}
\numberwithin{equation}{section}

\theoremstyle{plain}
\newtheorem{theorem}{Theorem}[section]      
\newtheorem{lemma}[theorem]{Lemma}          
\newtheorem{proposition}[theorem]{Proposition}
\newtheorem{corollary}[theorem]{Corollary}

\theoremstyle{definition}

\newtheorem{assumption}[theorem]{Assumption}

\theoremstyle{remark}
\newtheorem{remark}[theorem]{Remark}
\newtheorem{example}[theorem]{Example}

\newcommand{\indep}{\perp\!\!\!\perp}

\crefname{theorem}{theorem}{theorems}
\Crefname{theorem}{Theorem}{Theorems}
\crefname{lemma}{lemma}{lemmas}
\Crefname{lemma}{Lemma}{Lemmas}
\crefname{proposition}{proposition}{propositions}
\Crefname{proposition}{Proposition}{Propositions}
\crefname{corollary}{corollary}{corollaries}
\Crefname{corollary}{Corollary}{Corollaries}
\crefname{definition}{definition}{definitions}
\Crefname{definition}{Definition}{Definitions}
\crefname{assumption}{assumption}{assumptions}
\Crefname{assumption}{Assumption}{Assumptions}
\crefname{remark}{remark}{remarks}
\Crefname{remark}{Remark}{Remarks}
\crefname{example}{example}{examples}
\Crefname{example}{Example}{Examples}

\title{A General Exposure-Mapping-Agnostic Framework for Causal Inference under Interference}
\author{Yihui He, Eric J. Tchetgen Tchetgen}
\date{\today}

\begin{document}
\maketitle
\begin{abstract}
We develop a general framework for design-based causal inference under interference in cluster experiments conducted via two-stage randomization on a network of interconnected units, without relying on exposure mapping assumptions, exclusion of cross-cluster interference, or Bernoulli treatment assignments. Within this framework, we establish a complete characterization of linear weighted estimators (LW) as defined by \cite{Godambe1955}, with weights depending on the local neighborhood assignments, that achieve unbiased identification of various network causal effects under interference. This general class includes several new estimators with improved theoretical guarantees and superior finite-sample performance relative to existing approaches such as standard inverse-probability-of-treatment weighting. For most estimators in this class, we establish central limit theorems and conservative variance estimators, which allows us to describe the distinct asymptotic behavior exhibited by different weighting schemes potentially of interest. In particular, we study how randomization at the cluster-level affects the asymptotic behavior of various estimators, and we identify a subclass of cluster-agnostic LW estimators whose convergence rates are independent of the number of clusters and attain the parametric root-N rate, where N denotes the total number of units. Notably, for complete randomization we develop new techniques that may be of independent interest, both to establish a central limit theorem for sums of general dependent statistics and to construct conservative and bias-corrected variance estimators. We complement our theoretical results with extensive simulation studies that offer practical guidance on the choice of weighting method and experimental design under a wide range of interference structures.

\end{abstract}

\section{Introduction}
Two-stage randomized trials \cite{Hudgens2008} have emerged as a widely used experimental design for studying causal effects of intervention policies with spillover effects, with applications spanning economics, epidemiology, and related fields. In this design, a population of $N$ units is partitioned into $n$ clusters, and treatments are assigned at both the cluster and unit levels: cluster-level assignment determines the distribution of unit-level treatments within each cluster, while unit-level assignments directly affect outcomes. This design generalizes the traditional cluster randomized trials where treatment is assigned only at the cluster level, and enables the analysis of policies in which treatment is partially saturated or heterogeneously implemented. Despite its broad applicability, existing analytical frameworks for two-stage randomized trials have largely been restricted to settings with only within-cluster interference, implicitly assuming that interactions across clusters are negligible. Recently, \cite{Leung2025} and \cite{Lu2026} challenge this assumption and propose distinct methodological solutions.  By way of motivation, we illustrate these challenges  with the following example from \cite{Makofane2023} of a cluster randomized trial, where standard cluster randomized trial methods will generally fail to appropriately account for these challenges. 

\begin{example}[Cluster randomized trial on health intervention uptake]
\cite{Makofane2023} conducted a cluster randomized trial in South Africa aimed at increasing HIV testing uptake. Communities were randomly assigned to treatment, where all individuals were offered a financial micro-incentive to use an HIV home-test kit, or to control, where no incentive to use the test was provided. While comparing average outcomes between treated and control communities yielded a standard estimate of the overall effect, the study found that uptake depended on individual-level family networks: individuals were more likely to test when their family members were also treated. Since families often spanned multiple communities, and in some cases geographically distant ones, interference could extend beyond cluster boundaries, making inferences based on simple comparisons at the cluster level no longer valid and difficult to interpret. Moreover, spillover effects from different family members could be heterogeneous and difficult to model due to latent patterns of influence in family decision-making.
\end{example}

To address these challenges, we begin by defining new estimands that compare potential outcomes under counterfactual distributions in which all clusters are uniformly assigned to treatment or control, rather than under distributions in which treatment varies across clusters. These estimands remain aligned with the causal effect framework of \cite{Halloran1991}, which compares each unit's potential outcomes in populations with and without a policy intervention. In traditional frameworks, treated and control clusters are used to approximate such populations. In contrast, we directly define these policy environments through counterfactual treatment laws imposed on the entire population. By construction, our definition does not rely on assumptions about the outcome structure, such as the absence of cross-cluster interference. Furthermore, our framework allows users to interpret the target policy more explicitly as an intervention applied across all clusters, under which the unit-level treatment mechanisms are uniformly shifted from those under control clusters to those under treated clusters. Building on this formulation, our estimand is applicable to the whole class of experimental designs within the two-stage randomization framework of \cite{Hudgens2008}. In other words, our framework accommodates arbitrary randomization schemes at both the cluster and unit levels beyond independent Bernoulli designs (e.g., complete randomization). In this way, one preserves the full flexibility of the two-stage randomization framework for studying the effects of complex interventions.

To formalize valid inference about the proposed estimands with some degree of robustness afforded by the design-based causal inference framework, we establish the unbiasedness of a large class of \emph{linear weighted} (LW) estimators, as defined in \cite{Godambe1955}. We highlight that this unbiasedness property requires correct specification only of the interference network, obviating any reliance on well-specified so-called \textit{exposure mapping} function introduced by \cite{Aronow2017}. An exposure mapping function is often assumed to be a known summary statistic through which a unit’s neighbors’ treatments affect its outcome; however, this assumption may be overly restrictive and will often be violated in practice. Such violation may arise from misspecification of the interference neighborhood, unmodeled heterogeneity within neighborhoods, or complex interactions among spillover effects induced by different neighbors. Under our identification strategy which does not involve the specification of an exposure mapping, we avoid bias due to unmodeled heterogeneity and interaction effects, and improve robustness to network misspecification through conservative specifications of the interference network (i.e., allowing more connections than necessary).

A major contribution of this work is the careful study of convergence rates, asymptotic normality, and the development of conservative variance estimators for Hájek-type variants of the proposed class of unbiased LW estimators. Specifically, we show that general estimators in this class achieve asymptotic normality at a rate of $O(n^{-1/2})$, which depends on the number of clusters, while a subclass of LW estimators with cluster-agnostic weights attains a rate of $O(N^{-1/2})$, which instead depends on the total sample size. We further characterize the distinct asymptotic properties of these cluster-agnostic estimators, including faster convergence rates, potential impact of excessive finite-sample variance of the weights, and the need for alternative variance estimators. Under asymptotic regimes in which $n$ grows with $N$ at different rates, our results provide guidance for selecting estimators within this class to balance convergence rate and weight-stability, thereby improving overall finite-sample performance.

We establish asymptotic results under a broad class of randomization schemes and interference networks. Within the whole class of two-stage randomized designs, we impose restrictions on only one level of randomization—either at the cluster level or at the unit level, depending on the estimator—and these restrictions still permit a wide range of schemes at that level, including i.i.d. Bernoulli randomization, standard complete randomization, and stratified complete randomization. For the interference network, asymptotic normality at the standard convergence rate holds as long as the maximum number of neighboring units and neighboring clusters in the underlying network are suitably controlled.

In terms of technical contributions, as we demonstrate it is not trivial to establish formal asymptotic theory for complete randomization, as existing techniques for analyzing sums of weakly dependent sequences (e.g., \cite{Shi2025,Li2017,Fang2015,Kojevnikov2021,Ogburn2024,Liu2014}) do not yield central limit theorems under complete randomization for a general summand without invariance to permutations. To address this challenge, we develop a novel approach based on Stein’s method, which controls the Stein discrepancy $\mathbb{E}[f(S) - S f'(S)]$ using a coupling decomposition that appears to be novel, and thus substantially different from coupling decompositions in the literature. This technique may be of independent interest, as it provides a pathway to establish central limit theorems for general, non-symmetric summands under complete randomization. The same dense dependence also makes variance estimation nonstandard; we show that the standard Heteroskedasticity and Autocorrelation Consistent (HAC) variance estimators remain conservative under complete randomization and, by characterizing in closed form the additional dependence term they omit, further construct a bias-corrected estimator that removes the resulting excess conservativeness.

Finally, we evaluate our methods in extensive empirical experiments similar to those in \cite{Leung2025}. In simulation studies evaluating overall causal effects, we demonstrate the consistently superior accuracy of the marginal Radon–Nikodym derivative estimators within our class relative to alternative methods, including those in \cite{Leung2025} and the naive difference-in-means estimator comparing treated and control clusters. Furthermore, our estimators maintain accurate point estimation and valid confidence interval coverage under complete randomization, while achieving lower mean squared error than under two-stage randomized designs with independent Bernoulli assignment. 

\section{Related works and contributions}
To our knowledge, \cite{Leung2025} and \cite{Lu2026} are the two primary methodological papers addressing cross-cluster interference. We provide a comprehensive comparison between our work and theirs. Similar to our setting, \cite{Leung2025} studies the causal effects of population-level intervention policies in two-stage randomized trials without an exposure mapping assumption. Their estimand and estimator can be viewed as special cases within our estimand class and unbiased LW estimator class, respectively. In particular, our framework extends beyond their approach by considering a broader class of estimators and systematically comparing their estimator's performance within this family. In terms of network structure, \cite{Leung2025} assumes that interference and clusters are both determined by geographic proximity, with spillover effects that decay with distance but remain non-zero across all clusters. In contrast, we allow for a more general class of interference networks that need not align with clustering and geographic proximity, and focus on settings with well-defined interference neighborhoods to avoid the complexities associated with globally decaying dependence. Turning to \cite{Lu2026}, they focus on causal estimands defined through contrasts across different exposure mapping levels, both within and across clusters, rather than contrasts induced by a population-level intervention policy. Thus, they address a fundamentally different research question. Moreover, our framework helps reconcile the difference in convergence rates between the two papers. Specifically, \cite{Leung2025} establishes a convergence rate of $O(n^{-1/2})$, whereas \cite{Lu2026} obtains a rate of $O(N^{-1/2})$. We show that both rates arise within our class of estimators. In particular, the $O(n^{-1/2})$ rate applies to our general estimators, whereas the $O(N^{-1/2})$ rate relies on stronger independence assumptions and, in general, is attained only by our cluster-agnostic estimators. Finally, compared with both \cite{Leung2025} and \cite{Lu2026}, our framework accommodates more flexible treatment assignment mechanisms.

Beyond the cross-cluster interference literature discussed above, another related strand of work studies unbiased identification of causal effects in general experiments with a known interference network. A foundation for much of this literature is \cite{Aronow2017}, which introduced a framework based on a known exposure mapping from population-level treatment assignments to unit-level exposures. Several recent studies in this literature relax this requirement: \cite{Svje2023}, \cite{Gao2025}, \cite{Leung2025}, and \cite{Harshaw2025} study identification with a known interference network but an unknown or misspecified exposure mapping.
The first two papers focus on estimands defined via exposure mapping functions, but establish identification without requiring correct specification of the mapping. In contrast, the latter two avoid exposure mappings in both the estimand definition and identification. In particular, \cite{Harshaw2025} develop an approach to identify general causal estimands as linear functionals of potential outcome functions. Using the Riesz representation theorem and the fact that the Riesz representer does not depend on the functional form of the potential outcomes, they construct unbiased LW estimators using only information about the interference neighborhood. Such LW estimators extend the identification strategies in the literature on experimental analysis, which typically rely on inverse probability weighting. In two-stage randomization, we show that this approach still works, however it is fundamentally different from some other valid identification strategies based on inverse-probability-of-treatment weighting (e.g., \cite{Leung2025}). The non-uniqueness in unbiased identification arises because additional structure is imposed on the potential outcomes beyond the interference neighborhood of the unit-level treatment variable. Cluster-level treatments that do not directly affect outcomes constitute one such structure, while additional assumptions on interference (e.g., exposure mappings) may induce others. In such settings, the space of weights for unbiased LW estimators becomes strictly larger than the space implied by the restricted potential outcome functions and the Riesz representation is not the unique identifying functional.

It is also informative to compare our new framework with existing approaches for establishing central limit theorems under interference. To derive asymptotic normality in this setting, most existing work either assumes a sparse dependency graph, under which independence holds for most unit pairs (e.g., \cite{Ogburn2024,Leung2025}), or follows \cite{Kojevnikov2021} in requiring that dependence between units decays sufficiently fast with respect to a well-defined notion of network distance. While both approaches apply to settings with independent Bernoulli randomization, neither accommodates complete randomization. Under complete randomization, dependence is typically induced across all units and does not decay sufficiently fast to satisfy the conditions in \cite{Kojevnikov2021}. The literature that directly addresses complete randomization (e.g., \cite{Liu2014,Li2017,Shi2025}) is limited in a different respect: it focuses on asymptotic results for sums with a restricted class of symmetric summands that are invariant under certain permutations of treatment assignments. These approaches do not readily extend to our framework, where LW estimators generally do not admit such a symmetric summand representation.

Prior work exists on asymptotic analysis of weakly dependent sequences under dense dependency graphs providing useful intuition for analyzing complete randomization, however these methods differ from ours in key aspects and do not directly deliver central limit theorems for our setting. For example, our characterization of weak dependence under complete randomization is related to the mixing coefficients in \cite{Merlevde2009}, which quantify the gap between unconditional and conditional distributions of summands. However, our framework differs in the choice of conditioning sigma-fields, and we allow for substantially larger dependence coefficients. Moreover, our decomposition of the Stein discrepancy is related to the idea of Stein coupling in \cite{Chen2010} and \cite{Fang2015}. To capture the dependence structure induced by complete randomization, however, we construct a more elaborate coupling than those previously used in the standard Stein coupling framework.

To summarize, the major contributions of our work lie in the following aspects:

\begin{enumerate}
\item[(i)] \textit{Generalized causal effect estimands in imperfect cluster randomized experiments.} 
We generalize the definitions of direct, indirect, total, and overall effects in \cite{Hudgens2008} to arbitrary cluster randomized experiments by comparing potential outcomes under policies in which all clusters receive the active treatment versus receiving the control treatment. These estimands maintain a tractable interpretation in the presence of cross-cluster interference and complex within-cluster treatment assignments. Compared with estimands defined as contrasts across exposure mapping levels, our estimands place less emphasis on how individual characteristics shape treatment effects, but provides a more robust measure of overall policy impact.

\item[(ii)] \textit{Complete characterization of unbiased estimators without exposure mappings.} 
Generalizing \cite{Harshaw2025}, we identify the complete class of unbiased linear weighted estimators for our estimands without relying on exposure mapping assumptions, where the weights may depend on the neighborhood assignments $(\mathbf W_{\mathcal N_{ij}},\mathbf C_{\mathcal N_{ij}})$ and unbiasedness is required uniformly over all outcomes $Y_{ij}\in\mathcal L(\mathbf W_{\mathcal N_{ij}})$, yielding a fully nonparametric and robust identification strategy. For the Hájek-type variants of representative estimators in this class, we conduct asymptotic analysis and simulation studies to better understand their relative advantages. These results provide both theoretical insights and practical guidance for selecting suitable estimators across different research settings. The estimators we propose demonstrate substantially improved accuracy relative to existing approaches.

\item[(iii)] \textit{Estimation and inference in flexible settings.} 
Our framework provides feasible and valid estimation and inference strategies under a broad class of randomization schemes, potential outcome models, interference network specifications, and asymptotic regimes governing the growth of the number of clusters $n$ relative to the sample size $N$. We also develop the \texttt{general\_cluster\_exp} package to enable efficient computation for estimation and inference of causal effects under standard i.i.d.\ Bernoulli and standard complete randomization schemes.

\item[(iv)] \textit{New techniques for inference under complete randomization.}
We introduce a novel technique, based on Stein's method and coupling, that establishes central limit theorems for sums of general, non-symmetric statistics under complete randomization, together with new variance estimation techniques for this setting. These techniques may be of independent interest beyond cluster experiments.
\end{enumerate}

\section{Setup and identification}
Suppose there are $n$ clusters indexed by $i=1,\dots,n$, where cluster $i$ contains $N_i$ individuals indexed by $j=1,\dots,N_i$. Let $N=\sum_{i=1}^n N_i$ denote the total number of units. Let $C_i \in \{0,1\}$ denote the cluster-level assignment, $W_{ij} \in \{0,1\}$ the individual-level treatment, and $Y_{ij}$ the observed outcome. Let
    \[
    \mathbf W
    =
    (W_{11},\dots,W_{1N_1},W_{21},\dots,W_{nN_n})^\top
    \]
    denote the stacked vector of individual-level treatment assignments  for the entire population,
    ordered first by cluster index $i=1,\dots,n$ and then by individual index $j=1,\dots,N_i$ within each cluster.
    Similarly, let
    \[
    \mathbf C
    =
    (C_1,\dots,C_n)^\top
    \]
    denote the vector of cluster-level treatment assignments.
    This ordering is fixed throughout and is used implicitly whenever vector or matrix operations
    involving $\mathbf W$ or $\mathbf C$ are considered. $\mathbf{w}$ and $\mathbf{c}$ will appear later as the realization of $\mathbf W$ and $\mathbf C$. To formally describe the treatment mechanism, we have the following assumption.

\begin{assumption}[Treatment mechanism]
\label{ass:cluster_assignment}
The cluster assignment $\mathbf{C}=(C_1,\ldots,C_n)$ affects treatment assignments only through the mechanisms specified below.
\begin{enumerate}
\item[(i)] For each cluster $i=1,\ldots,n$, conditional on the realized cluster assignments $\mathbf{C}$, the cluster-level vector of individual-level treatment assignments
\[
\mathbf{W}_i=(W_{i1},\ldots,W_{iN_i})^T
\]
is generated according to a cluster-specific assignment law,
\[
\mathbf{W}_i \mid \mathbf{C}
\;\sim\;
\phi_{C_i},
\]
where $\phi_0$ and $\phi_1$ are predetermined probability laws indexed by the cluster type $C_i$.
The assignment laws $\phi_0$ and $\phi_1$ do not depend on potential outcomes and do not depend on the realized values of $\mathbf{C}$ beyond the index $C_i$, but it can vary with different cluster size. Within a cluster, the assignment law $\phi_{c}$ may correspond to either an independent treatment regime (e.g., Bernoulli assignment) or a dependent treatment regime (e.g., complete randomization).
\item[(ii)] Across clusters, the assignment vectors
\(
\{\mathbf{W}_i\}_{i=1}^n
\)
are independent conditional on $\mathbf{C}$.
\end{enumerate}
\end{assumption}

Under the treatment mechanism described above, we furthermore assume that the cluster-level treatment $\mathbf{C}$ influences the outcomes only through
its effect on the treatment assignment vector $\mathbf{W}$. Formally, we have the following assumption on potential outcomes.
\begin{assumption}[Potential outcome]
\label{ass:potential_outcome}
For each unit $(i,j)$, there exists a potential outcome function
\[
Y_{ij}(\mathbf{w}), \qquad \mathbf{w}\in\{0,1\}^{\sum_{i}N_i},
\]
corresponding to the unit's outcome under hypothetical treatment assignment vector $\mathbf{w}$; such that under a network consistency assumption, the observed outcome satisfies
\[
Y_{ij} = Y_{ij}(\mathbf{W}) \quad \text{almost surely}.
\]
\end{assumption}

To simplify the potential outcome structure, we assume neighborhood interference through a known undirected network defined over individual units. Notably, we allow interference neighborhoods to extend across the clusters defined above. Let $\mathcal{N}_{ij}$ denote the known neighborhood of unit $(i,j)$ in this network, consisting of individual units (rather than clusters) and containing both the neighboring units and the unit itself, $(i,j)\in\mathcal{N}_{ij}$.
With this notation in place, we impose the following assumption.
\begin{assumption}[Neighborhood interference]\label{ass:interference}
For any unit $(i,j)$, $Y_{ij}(\mathbf{w}) = Y_{ij}(\mathbf{w}')$ whenever
$\mathbf{w}_{\mathcal{N}_{ij}} = \mathbf{w}'_{\mathcal{N}_{ij}}$, where $\mathcal{N}_{ij}$
denotes the interference neighborhood of unit $(i,j)$.
Under this assumption, the potential outcome depends on $\mathbf{w}$ only through
$\mathbf{w}_{\mathcal{N}_{ij}}$, and we can therefore also write it as $Y_{ij}(\mathbf{w}_{\mathcal{N}_{ij}})$.
\end{assumption}

To furthermore record how neighborhoods extend across clusters, we project the individual-level neighborhood onto the cluster level and define the \emph{cluster-level interference neighborhood} of unit $(i,j)$ as
\[
\mathcal{N}_{ij}^{(cl)} = \{\, i' : \exists\, j' \text{ such that } (i',j') \in \mathcal{N}_{ij} \,\},
\]
namely, the set of clusters containing at least one unit in $\mathcal{N}_{ij}$.
Then, let
\[
\mathbf{W}_{\mathcal{N}_{ij}}
=
\{W_{i'j'} : (i',j') \in \mathcal{N}_{ij}\}
\]
denote the collection of individual-level treatment assignments in the neighborhood, and let
\[
\mathbf{C}_{\mathcal{N}_{ij}}
:=
\mathbf{C}_{\mathcal{N}_{ij}^{(cl)}}
=
\{C_{i'} : i' \in \mathcal{N}_{ij}^{(cl)}\}
\]
denote the corresponding cluster-level assignments. By Assumption~\ref{ass:interference}, the potential outcome of unit $(i,j)$ depends on $\mathbf W$ only through $\mathbf{W}_{\mathcal{N}_{ij}}$, 
and $\mathbf{W}_{\mathcal{N}_{ij}}$ depends on $\mathbf C$ only through $\mathbf{C}_{\mathcal{N}_{ij}}$.

To make meaningful contrasts between $\phi_1$ and $\phi_0$, we use $\phi_c$ to define a counterfactual assignment law $P_{\phi_c}$ that represents a hypothetical experiment in which this same individual-level treatment rule is applied uniformly across all clusters. By definition, $P_{\phi_c}$ preserves the cluster-level randomization scheme of the original experiment, so that
\[
P_{\phi_c}(\mathbf{C}) = P(\mathbf{C}).
\]
However, under $P_{\phi_c}$ and conditional on $\mathbf{C}$, the within-cluster treatment vector
\[
\mathbf{W}_i|\mathbf{C}\sim \phi_c
\]
holds independently for all clusters. The counterfactual laws $P_{\phi_0}$ and $P_{\phi_1}$ therefore provide uniform reference designs that isolate the effect of switching the individual-level treatment rule across all clusters, and serve as the basis for defining regime-specific potential outcomes and causal contrasts. 

To ensure that information from the observed assignment mechanism can be transported to the counterfactual regimes of interest, we impose the following local transportability assumption. Intuitively, this assumption requires the support of the realized assignment mechanism to be sufficiently rich to cover that of the counterfactual regimes. The assumption provides two alternative absolute continuity conditions that differ in strength and correspond to different identification results developed later. For our main analysis, only the first condition is required, while the second, stronger condition enables theoretically sharper identification.
\begin{assumption}[Local measure transportability]
\label{ass:measure}
For each treatment regime $\phi \in \{\phi_0, \phi_1\}$ and each unit $(i,j)$, we impose one of the following absolute continuity conditions of the counterfactual assignment measure $P_{\phi}$ with respect to the observed assignment measure $P$, defined on the measurable space induced by the interference neighborhood of unit $(i,j)$:
\begin{enumerate}
    \item[(i)] The marginal measure $P_{\phi}(\mathbf{W}_{\mathcal{N}_{ij}})$ is absolutely continuous with respect to $P(\mathbf{W}_{\mathcal{N}_{ij}})$, i.e.,
    \[
    P_{\phi}\!\left(\mathbf{W}_{\mathcal{N}_{ij}}\right)
    \ll
    P\!\left(\mathbf{W}_{\mathcal{N}_{ij}}\right).
    \]
    Equivalently, for any realization $\mathbf{w}_{\mathcal{N}_{ij}}$,
    \[
    P\!\left(
    \mathbf{W}_{\mathcal{N}_{ij}}=\mathbf{w}_{\mathcal{N}_{ij}}
    \right)=0
    \;\;\Rightarrow\;\;
    P_{\phi}\!\left(
    \mathbf{W}_{\mathcal{N}_{ij}}=\mathbf{w}_{\mathcal{N}_{ij}}
    \right)=0.
    \]

    \item[(ii)] The joint measure $P_{\phi}(\mathbf{W}_{\mathcal{N}_{ij}}, \mathbf{C}_{\mathcal{N}_{ij}})$ is absolutely continuous with respect to $P(\mathbf{W}_{\mathcal{N}_{ij}}, \mathbf{C}_{\mathcal{N}_{ij}})$, i.e.,
    \[
    P_{\phi}\!\left(\mathbf{W}_{\mathcal{N}_{ij}}, \mathbf{C}_{\mathcal{N}_{ij}}\right)
    \ll
    P\!\left(\mathbf{W}_{\mathcal{N}_{ij}}, \mathbf{C}_{\mathcal{N}_{ij}}\right).
    \]
    Equivalently, for any realizations $(\mathbf{w}_{\mathcal{N}_{ij}}, \mathbf{c}_{\mathcal{N}_{ij}})$,
    \[
    P\!\left(
    \mathbf{W}_{\mathcal{N}_{ij}}=\mathbf{w}_{\mathcal{N}_{ij}},
    \mathbf{C}_{\mathcal{N}_{ij}}=\mathbf{c}_{\mathcal{N}_{ij}}
    \right)=0
    \;\;\Rightarrow\;\;
    P_{\phi}\!\left(
    \mathbf{W}_{\mathcal{N}_{ij}}=\mathbf{w}_{\mathcal{N}_{ij}},
    \mathbf{C}_{\mathcal{N}_{ij}}=\mathbf{c}_{\mathcal{N}_{ij}}
    \right)=0.
    \]
\end{enumerate}
\end{assumption}
\begin{remark}[Sufficient conditions for Assumption \ref{ass:measure}]
Assumption \ref{ass:measure}(i) holds under either of the following conditions:
\begin{itemize}
    \item For any unit $(i,j)$ and any $\mathbf{c}_{\mathcal{N}_{ij}} \in \{0,1\}^{|\mathcal{N}_{ij}^{(cl)}|}$,
    \begin{equation}\label{eq:transportability}
    0 < P\!\left(\mathbf{C}_{\mathcal{N}_{ij}} = \mathbf{c}_{\mathcal{N}_{ij}}\right) < 1.
    \end{equation}
    \item For any cluster $i$, any feasible value $\mathbf{w}_i\in \{0,1\}^{N_i}$, and any cluster-level treatment regime $\phi$,
    \[
    0 < P_{\phi}\!\left(\mathbf{W}_i = \mathbf{w}_i\right) < 1.
    \]
\end{itemize}
In contrast, Assumption \ref{ass:measure}(ii) is guaranteed only by the second condition. Under the first condition alone, this guarantee does not generally hold; however, it can be recovered if the second condition holds for at least one of the treatment regimes $\phi \in \{\phi_0, \phi_1\}$. 
As an illustration of the failure of Assumption \ref{ass:measure}(ii) when only the first condition above holds, consider a traditional cluster randomized trial design without within-cluster randomization, where for every cluster $i$,
\[
P_{\phi_0}(\mathbf{W}_i = \mathbf{0}) = 1
\quad \text{and} \quad
P_{\phi_1}(\mathbf{W}_i = \mathbf{1}) = 1.
\]
Here, we write $\mathbf{0}$ and $\mathbf{1}$ denote vectors of zeros and ones, respectively. 
In this setting, Assumption \ref{ass:measure}(ii) fails because the counterfactual distribution assigns positive probability to configurations such as $P_{\phi_1}(W_{ij}=1, C_i=0) > 0$, whereas such events have zero probability under the observed assignment mechanism, i.e., $P(W_{ij}=1, C_i=0) = 0$. 
\end{remark}

For any treatment regime $\phi\in\{\phi_0,\phi_1\}$ and individual-level treatment $w\in\{0,1\}$,
we follow \citet{Hudgens2008} and define the \emph{marginal individual average potential outcome} under regime $\phi$ as
\[
\bar{Y}_{ij}(\phi)
:=
\mathbb{E}_{\phi}\!\left[ Y_{ij} \right]
=
\mathbb{E}_{\mathbf C}\!\left[
\mathbb{E}_{\mathbf W\sim\phi}\!\left[
Y_{ij}(\mathbf W)
\right]
\right].
\]
Similarly, for unit $(i,j)$ under individual treatment $w$ and regime $\phi$, suppose $P_{\phi}(W_{ij}=w)>0$, and the corresponding \emph{individual average potential outcome} is
\[
\bar{Y}_{ij}(w,\phi)
:=
\mathbb{E}_{\phi}\!\left[ Y_{ij}\mid W_{ij}=w \right]
=
\mathbb{E}_{\mathbf C}\!\left[
\mathbb{E}_{\mathbf W\sim\phi}\!\left[
Y_{ij}(\mathbf W)\mid W_{ij}=w
\right]
\right].
\]
Averaging over individuals within cluster $i$, we define the \emph{cluster-level average potential outcomes}
\[
\bar{Y}_i(w,\phi)
=
\frac{1}{N_i}\sum_{j=1}^{N_i}\bar{Y}_{ij}(w,\phi),
\qquad
\bar{Y}_i(\phi)
=
\frac{1}{N_i}\sum_{j=1}^{N_i}\bar{Y}_{ij}(\phi).
\]
To aggregate across clusters, we allow for a general collection of nonnegative known cluster weights $\{g_i\}_{i=1}^n$ satisfying $\sum_{i=1}^n g_i = 1$. 
Following the terminology in \cite{Hudgens2008}, the \emph{population average potential outcome} and its marginal counterpart are then defined as
\[
\bar{Y}(w,\phi)
=
\sum_{i=1}^n g_i\,\bar{Y}_i(w,\phi),
\qquad
\bar{Y}(\phi)
=
\sum_{i=1}^n g_i\,\bar{Y}_i(\phi).
\]
Notably, the “population” considered in this definition consists solely of the observed units in the randomized trial. We do not treat the observed data as a random sample drawn from a larger superpopulation.

Population-level causal effects can then be defined as contrasts of the \emph{population average potential outcomes} or \emph{marginal population average potential outcomes} under different individual treatment statuses and treatment regimes. We focus exclusively on contrasts of population-level potential outcomes, rather than those at the cluster or individual level, since all potential outcomes are defined with respect to population-level intervention policies, rendering their aggregate effects more directly interpretable. Specifically, the \emph{population-level direct causal effect} under treatment regime $\phi$
is defined as
\[
\mathrm{CE}^{D}(\phi)
:=
\bar{Y}(1,\phi)-\bar{Y}(0,\phi),
\]
capturing the effect of changing an individual’s own treatment status
while holding fixed the stochastic treatment rule governing all other units.
The \emph{population-level indirect causal effect} under individual treatment $w$
is defined as
\[
\mathrm{CE}^{I}(w;\phi_1,\phi_0)
:=
\bar{Y}(w,\phi_1)-\bar{Y}(w,\phi_0),
\]
which isolates the effect of switching the treatment regime applied to others
while holding an individual’s own treatment fixed.
The \emph{population-level total causal effect} is defined as
\[
\mathrm{CE}^{T}(\phi_1,\phi_0)
:=
\bar{Y}(1,\phi_1)-\bar{Y}(0,\phi_0),
\]
combining both direct and indirect components by simultaneously changing
the individual’s treatment status and the treatment regime applied to others.
By construction, the total causal effect satisfies
\[
\mathrm{CE}^{T}(\phi_1,\phi_0)
=
\mathrm{CE}^{D}(\phi_1)
+
\mathrm{CE}^{I}(0;\phi_1,\phi_0),
\]
mirroring the decomposition established in \citet{Hudgens2008}.
Finally, the \emph{population-level overall causal effect} summarizes the
average effect of switching the treatment regime from $\phi_0$ to $\phi_1$
without conditioning on individual treatment status, and is defined as
\[
\mathrm{CE}^{O}(\phi_1,\phi_0)
:=
\bar{Y}(\phi_1)-\bar{Y}(\phi_0).
\]

In this paper, we estimate the \emph{population average potential outcomes} and the \emph{marginal population average potential outcomes} using a class of linear weighted (LW) estimators as defined in \cite{Godambe1955}; the corresponding causal effect estimates are obtained as linear combinations of these estimators. The following theorem characterizes the full class of weights that achieve unbiased identification of the average potential outcome estimands without relying on a specified exposure mapping. For any random variable $R$, let $\mathcal{L}(R)$ denote the class of Lebesgue-integrable functions of $R$. For the marginal and non-marginal average potential outcomes of each unit $(i,j)$, we consider estimators of the form \[Y_{ij}\beta_{ij}(\mathbf{W}_{\mathcal{N}_{ij}}, \mathbf{C}_{\mathcal{N}_{ij}}),\] and characterize the class of weight functions $\beta_{ij}$ that ensures unbiased identification for all $Y_{ij} \in \mathcal{L}(\mathbf{W}_{\mathcal{N}_{ij}})$. Since $\mathbf{W}_{\mathcal{N}_{ij}}$ is discrete, this class includes all well-defined functions on its support. Aggregating these unbiased individual-level estimators then yields unbiased estimators of the corresponding population-level marginal and non-marginal average potential outcomes.

\begin{theorem}[Unbiased identification weights]\label{thm:all_weights}
Fix a treatment regime $\phi \in \{\phi_0,\phi_1\}$ and define the marginal Radon--Nikodym derivative
\[
\alpha_{ij}(\phi)
=
\frac{
dP_{\phi}\big(\mathbf{W}_{\mathcal{N}_{ij}}\big)
}{
dP\big(\mathbf{W}_{\mathcal{N}_{ij}}\big)
},
\]
which reweights the observed assignment distribution to the counterfactual law $P_{\phi}$ at the level of the neighborhood $\mathcal{N}_{ij}$.

Under Assumptions \ref{ass:potential_outcome}--\ref{ass:interference} and \ref{ass:measure}(i), for any $h \in \mathcal{L}(\mathbf{W}_{\mathcal{N}_{ij}})$, a weight $\beta_{ij} \in \mathcal{L}(\mathbf{W}_{\mathcal{N}_{ij}}, \mathbf{C}_{\mathcal{N}_{ij}})$ satisfies
\[
\mathbb{E}\!\left[ Y_{ij}\,\beta_{ij} \right]
=
\mathbb{E}_{\phi}\!\left[ Y_{ij}\,h(\mathbf{W}_{\mathcal{N}_{ij}}) \right]
\quad \text{for all } Y_{ij} \in \mathcal{L}(\mathbf{W}_{\mathcal{N}_{ij}})
\]
if and only if
\[
\mathbb{E}\!\left[\beta_{ij} \mid \mathbf{W}_{\mathcal{N}_{ij}}\right] = \alpha_{ij}(\phi)\,h(\mathbf{W}_{\mathcal{N}_{ij}}).
\]
In particular:
\begin{enumerate}
    \item[(a)] Taking $h \equiv 1$: a weight $\beta_{ij}(\phi) \in \mathcal{L}(\mathbf{W}_{\mathcal{N}_{ij}}, \mathbf{C}_{\mathcal{N}_{ij}})$ satisfies
    \[
    \mathbb{E}\!\left[ Y_{ij}\,\beta_{ij}(\phi) \right] = \bar{Y}_{ij}(\phi)
    \quad \text{for all } Y_{ij} \in \mathcal{L}(\mathbf{W}_{\mathcal{N}_{ij}})
    \]
    if and only if $\mathbb{E}\!\left[\beta_{ij}(\phi)\mid\mathbf{W}_{\mathcal{N}_{ij}}\right]=\alpha_{ij}(\phi)$.

    \item[(b)] For $w \in \{0,1\}$ with $P_{\phi}(W_{ij}=w) > 0$, taking $h(\mathbf{w}) = {1(w_{ij}=w)}/{P_{\phi}(W_{ij}=w)}$: a weight $\beta_{ij}(w,\phi) \in \mathcal{L}(\mathbf{W}_{\mathcal{N}_{ij}}, \mathbf{C}_{\mathcal{N}_{ij}})$ satisfies
    \[
    \mathbb{E}\!\left[ Y_{ij}\,\beta_{ij}(w,\phi) \right] = \bar{Y}_{ij}(w,\phi)
    \quad \text{for all } Y_{ij} \in \mathcal{L}(\mathbf{W}_{\mathcal{N}_{ij}})
    \]
    if and only if $\mathbb{E}\!\left[\beta_{ij}(w,\phi)\mid\mathbf{W}_{\mathcal{N}_{ij}}\right]=\alpha_{ij}(\phi)\frac{1(W_{ij}=w)}{P_{\phi}(W_{ij}=w)}$.
\end{enumerate}
\end{theorem}

Following \cite{Harshaw2025}, the marginal Radon--Nikodym derivative weights, $\alpha_{ij}(\phi)$ and $\alpha_{ij}(\phi)\frac{1(W_{ij}=w)}{P_{\phi}(W_{ij}=w)}$, can be interpreted as the Riesz representers of the corresponding estimands. Specifically, these estimands are viewed as linear functionals of $Y_{ij}(\mathbf{W}_{\mathcal{N}_{ij}})$ on the Hilbert space $\mathcal{L}(\mathbf{W}_{\mathcal{N}_{ij}})$ equipped with the inner product induced by product expectations. However, since our weights can be constructed in the larger space $\mathcal{L}(\mathbf{W}_{\mathcal{N}_{ij}}, \mathbf{C}_{\mathcal{N}_{ij}})$, it is possible to construct alternative unbiased weights by exploiting additional structure in $\mathbf{C}$. 
    
Among all weights that achieve unbiasedness for all $Y_{ij}(\mathbf{W}_{\mathcal{N}_{ij}})\in\mathcal{L}(\mathbf{W}_{\mathcal{N}_{ij}})$, the marginal Radon--Nikodym derivative weights attain the minimum variance.
From Theorem~\ref{thm:all_weights}, we have
\[
\mathbb{E}[\beta_{ij}(\phi)\mid \mathbf W_{\mathcal N_{ij}}]
=
\mathbb{E}[\alpha_{ij}(\phi)\mid \mathbf W_{\mathcal N_{ij}}]
=
\alpha_{ij}(\phi)
\]
for any unbiased weight $\beta_{ij}(\phi)$. Hence,
\[
\mathrm{Var}\!\bigl(\beta_{ij}(\phi)\bigr)
\;\geq\;
\mathrm{Var}\!\bigl(\alpha_{ij}(\phi)\bigr).
\]
The same argument applies to the weights identifying non-marginal potential outcomes, $\beta_{ij}(w,\phi)$ and $\alpha_{ij}(\phi)\frac{1(W_{ij}=w)}{P_{\phi}(W_{ij}=w)}$. 

Although the marginal Radon--Nikodym derivative weights are variance-optimal at the unit level, incorporating additional structure from $\mathbf{C}$ into the weights can help reduce dependence across units. Ideally, one would prefer weights that depend only on nearby units in the interference network. However, cluster-level treatments induces dependence among all units within a cluster, and this dependence does not vanish with increasing sample size or cluster size. Consequently, it can offset the benefits of designs that aim to achieve (approximate) independence under within-cluster randomization schemes, $\phi_0$ and $\phi_1$.
To address this issue, we consider weights that satisfy unbiasedness conditional on any realization of the cluster-level treatment vector $\mathbf{C}$:
\[
\mathbb{E}\!\left[ Y_{ij}\,\beta_{ij}(w,\phi)\mid \mathbf{C} \right] = \bar{Y}_{ij}(w,\phi)
\text{ and }
\mathbb{E}\!\left[ Y_{ij}\,\beta_{ij}(\phi)\mid \mathbf{C} \right] = \bar{Y}_{ij}(\phi) \quad \text{for all } Y_{ij} \in \mathcal{L}(\mathbf{W}_{\mathcal{N}_{ij}}).
\]
Under this formulation, the influence of cluster-level treatments on the estimators is mitigated. Moreover, when $W_{ij}$ is approximately independent across $j$ conditional on $C_i$, for two units $(i,j)$ and $(i,j')$ within the same cluster that are far away in the interference network,
\begin{equation}\label{eq:cluster_agnostic}
\operatorname{Cov}(Y_{ij}\beta_{ij},\, Y_{ij'}\beta_{ij'})
=
\mathbb{E}\!\bigl[\operatorname{Cov}(Y_{ij}\beta_{ij},\, Y_{ij'}\beta_{ij'} \mid \mathbf{C})\bigr]
= 0,
\end{equation}
implying weaker correlation between summands in the estimator.
We refer to such weights as \emph{cluster-agnostic unbiased weights}. While they may not minimize unit-level variance, they could reduce the variance of aggregated estimators in settings where cross-unit dependence is a primary concern. In the following theorem, we will characterize such weights.

\begin{theorem}[Cluster-agnostic identification weights]\label{thm:cluster_agnostic_weights}
Fix a treatment regime $\phi \in \{\phi_0,\phi_1\}$. Suppose Assumptions
\ref{ass:cluster_assignment}--\ref{ass:interference} hold.

\begin{enumerate}
    \item[(i)] Suppose Assumption \ref{ass:measure}(ii) holds. Define the complete Radon--Nikodym derivative
    \[
    \alpha_{ij}^{(\mathrm{comp})}(\phi)
    =
    \frac{
    dP_{\phi}\big(\mathbf{W}_{\mathcal{N}_{ij}}, \mathbf{C}_{\mathcal{N}_{ij}}\big)
    }{
    dP\big(\mathbf{W}_{\mathcal{N}_{ij}}, \mathbf{C}_{\mathcal{N}_{ij}}\big)
    },
    \]
    which reweights the observed assignment distribution to the counterfactual law $P_{\phi}$ at the level of the neighborhood $\mathcal{N}_{ij}$.
    For any $h \in \mathcal{L}(\mathbf{W}_{\mathcal{N}_{ij}})$, a weight $\beta_{ij} \in \mathcal{L}(\mathbf{W}_{\mathcal{N}_{ij}}, \mathbf{C}_{\mathcal{N}_{ij}})$ satisfies
    \[
    \mathbb{E}\!\left[
    Y_{ij}\,\beta_{ij}
    \;\middle|\;
    \mathbf{C}
    \right]
    =
    \mathbb{E}_{\phi}\!\left[
    Y_{ij}\,h(\mathbf{W}_{\mathcal{N}_{ij}})
    \right]
    \quad \text{for all } Y_{ij} \in \mathcal{L}(\mathbf{W}_{\mathcal{N}_{ij}})
    \]
    if and only if
    \[
    \beta_{ij} = \alpha_{ij}^{(\mathrm{comp})}(\phi)\,h(\mathbf{W}_{\mathcal{N}_{ij}}).
    \]
    In particular:
    \begin{enumerate}
        \item[(a)] Taking $h \equiv 1$: a weight $\beta_{ij}(\phi) \in \mathcal{L}(\mathbf{W}_{\mathcal{N}_{ij}}, \mathbf{C}_{\mathcal{N}_{ij}})$ satisfies
        \[
        \mathbb{E}\!\left[ Y_{ij}\,\beta_{ij}(\phi) \mid \mathbf{C} \right] = \bar{Y}_{ij}(\phi)
        \quad \text{for all } Y_{ij} \in \mathcal{L}(\mathbf{W}_{\mathcal{N}_{ij}})
        \]
        if and only if $\beta_{ij}(\phi) = \alpha_{ij}^{(\mathrm{comp})}(\phi)$.

        \item[(b)] For $w \in \{0,1\}$ with $P_{\phi}(W_{ij}=w) > 0$, taking $h(\mathbf{w}) = {1(w_{ij}=w)}/{P_{\phi}(W_{ij}=w)}$: a weight $\beta_{ij}(w,\phi) \in \mathcal{L}(\mathbf{W}_{\mathcal{N}_{ij}}, \mathbf{C}_{\mathcal{N}_{ij}})$ satisfies
        \[
        \mathbb{E}\!\left[ Y_{ij}\,\beta_{ij}(w,\phi) \mid \mathbf{C} \right] = \bar{Y}_{ij}(w,\phi)
        \quad \text{for all } Y_{ij} \in \mathcal{L}(\mathbf{W}_{\mathcal{N}_{ij}})
        \]
        if and only if $\beta_{ij}(w,\phi) = \alpha_{ij}^{(\mathrm{comp})}(\phi)\frac{1(W_{ij}=w)}{P_{\phi}(W_{ij}=w)}$.
    \end{enumerate}

    \item[(ii)] Suppose Assumption \ref{ass:measure}(ii) does not hold.
    \begin{enumerate}
        \item[(a)] There exists no weight $\beta_{ij}(\phi) \in \mathcal{L}(\mathbf{W}_{\mathcal{N}_{ij}}, \mathbf{C}_{\mathcal{N}_{ij}})$ such that
        \[
        \mathbb{E}\!\left[ Y_{ij}\,\beta_{ij}(\phi) \mid \mathbf{C} \right] = \bar{Y}_{ij}(\phi)
        \quad \text{for all } Y_{ij} \in \mathcal{L}(\mathbf{W}_{\mathcal{N}_{ij}}).
        \]

        \item[(b)] There exists at least one \(w\in\{0,1\}\) with
\(P_{\phi}(W_{ij}=w)>0\) such that no weight
\[
\beta_{ij}(w,\phi)
\in
\mathcal{L}(\mathbf{W}_{\mathcal{N}_{ij}},\mathbf{C}_{\mathcal{N}_{ij}})
\]
satisfies
\[
\mathbb{E}\!\left[
Y_{ij}\,\beta_{ij}(w,\phi)
\mid \mathbf{C}
\right]
=
\bar{Y}_{ij}(w,\phi)
\quad
\text{for all }
Y_{ij}\in\mathcal{L}(\mathbf{W}_{\mathcal{N}_{ij}}).
\]
    \end{enumerate}
\end{enumerate}
\end{theorem}

\begin{remark}
The theorem establishes a stronger form of unbiasedness than is required for estimation. While standard unbiasedness only guarantees
\[
\mathbb{E}\!\left[
Y_{ij}\,
\beta_{ij}(\phi)
\right]
=
\bar{Y}_{ij}(\phi)\text{ and }\mathbb{E}\!\left[
Y_{ij}\,
\beta_{ij}(w,\phi)
\right]
=
\bar{Y}_{ij}(w,\phi),
\]
the result here guarantees that unbiasedness holds conditionally on any value of $\mathbf C$. We refer to this property as \emph{cluster-agnostic unbiasedness}. In the next section, we show that this property leads to distinct asymptotic behavior.
\end{remark}

\begin{proposition}[Examples of unbiased identification weights]\label{prop:examples_of_weights}
Under Assumptions \ref{ass:cluster_assignment}--\ref{ass:interference}, the following weight functions yield unbiased identification of individual average potential outcomes under their respective conditions.

\begin{enumerate}
    \item[(i)] \textbf{Marginal Radon--Nikodym derivative weights.}
    Suppose Assumption~\ref{ass:measure}(i) holds. The weights
    \[
    \beta_{ij}(\phi_c) = \alpha_{ij}(\phi_c),
    \]
    and, for $w\in\{0,1\}$ with $P_{\phi_c}(W_{ij}=w)>0$,
    \[
    \beta_{ij}(w,\phi_c) = \alpha_{ij}(\phi_c)\frac{1(W_{ij}=w)}{P_{\phi_c}(W_{ij}=w)},
    \]
    achieve unbiased identification for all $Y_{ij} \in \mathcal{L}(\mathbf{W}_{\mathcal{N}_{ij}})$.

    \item[(ii)] \textbf{Complete Radon--Nikodym derivative / cluster-agnostic weights.}
    Suppose Assumption~\ref{ass:measure}(ii) holds. The weights
    \[
    \beta_{ij}(\phi_c) = \alpha_{ij}^{(\mathrm{comp})}(\phi_c),
    \qquad
    \beta_{ij}(w,\phi_c) = \alpha_{ij}^{(\mathrm{comp})}(\phi_c)\frac{1(W_{ij}=w)}{P_{\phi_c}(W_{ij}=w)},
    \]
    achieve cluster-agnostic unbiased identification for all $Y_{ij} \in \mathcal{L}(\mathbf{W}_{\mathcal{N}_{ij}})$.

    \item[(iii)] \textbf{Exposure-mapping-level Radon--Nikodym derivative weights.}
    Suppose Assumption~\ref{ass:measure}(i) holds. Let
    $\mathbf e(\mathbf W_{\mathcal N_{ij}})$ be an arbitrary exposure mapping that
    summarizes the channels through which $\mathbf W_{\mathcal N_{ij}}$ affects
    the outcome $Y_{ij}$, with $W_{ij}$ included as one of its components. Define
    \[
    \beta_{ij}(\phi_c)
    =
    \mathbb{E}\!\left[\alpha_{ij}(\phi_c)\mid \mathbf{e}(\mathbf{W}_{\mathcal{N}_{ij}})\right]
    =
    \frac{dP_{\phi_c}\big(\mathbf{e}(\mathbf{W}_{\mathcal{N}_{ij}})\big)}
    {dP\big(\mathbf{e}(\mathbf{W}_{\mathcal{N}_{ij}})\big)},
    \]
    and, for $w\in\{0,1\}$ with $P_{\phi_c}(W_{ij}=w)>0$,
    \[
    \beta_{ij}(w,\phi_c)
    =
    \frac{dP_{\phi_c}\big(\mathbf{e}(\mathbf{W}_{\mathcal{N}_{ij}})\big)}
    {dP\big(\mathbf{e}(\mathbf{W}_{\mathcal{N}_{ij}})\big)}
    \frac{1(W_{ij}=w)}{P_{\phi_c}(W_{ij}=w)}.
    \]
    These weights achieve unbiased identification for all $Y_{ij} \in \mathcal{L}(\mathbf{e}(\mathbf{W}_{\mathcal{N}_{ij}}))$, that is, over the potentially restricted class of potential outcomes induced by $\mathbf{e}(\mathbf{W}_{\mathcal{N}_{ij}})$, rather than the full class $\mathcal{L}(\mathbf{W}_{\mathcal{N}_{ij}})$.

    Moreover, the Radon--Nikodym derivative admits the decomposition
    \[
    \frac{dP_{\phi_c}\big(\mathbf{e}(\mathbf{W}_{\mathcal{N}_{ij}})\big)}
    {dP\big(\mathbf{e}(\mathbf{W}_{\mathcal{N}_{ij}})\big)}
    =
    \sum_{\mathbf{s} \in \mathrm{supp}(\mathbf{e}(\mathbf{W}_{\mathcal{N}_{ij}}))}
    \frac{
    1\big(\mathbf{e}(\mathbf{W}_{\mathcal{N}_{ij}})=\mathbf{s}\big)\,
    P_{\phi_c}\big(\mathbf{e}(\mathbf{W}_{\mathcal{N}_{ij}})=\mathbf{s}\big)
    }{
    P\big(\mathbf{e}(\mathbf{W}_{\mathcal{N}_{ij}})=\mathbf{s}\big)
    },
    \]
    which resembles standard inverse probability weighting.

    \item[(iv)] \textbf{Inverse-probability-of-treatment (IPT) weights in \cite{Leung2025}.}
    Suppose $P(\mathbf{C}_{\mathcal{N}_{ij}}=c\mathbf{1})>0$, where $\mathbf{C}_{\mathcal{N}_{ij}}=c\mathbf{1}$ denotes the event that every cluster in $\mathcal{N}_{ij}$ receives treatment $c$. Define
    \[
    \beta_{ij}(\phi_c) = \frac{1(\mathbf{C}_{\mathcal{N}_{ij}}=c\mathbf{1})}{P(\mathbf{C}_{\mathcal{N}_{ij}}=c\mathbf{1})},
    \]
    and, for $w\in\{0,1\}$ with $P(\mathbf{C}_{\mathcal{N}_{ij}}=c\mathbf{1},\,W_{ij}=w)>0$,
    \[
    \beta_{ij}(w,\phi_c) = \frac{1(\mathbf{C}_{\mathcal{N}_{ij}}=c\mathbf{1},\,W_{ij}=w)}{P(\mathbf{C}_{\mathcal{N}_{ij}}=c\mathbf{1},\,W_{ij}=w)}.
    \]
    These weights achieve unbiased identification for all $Y_{ij}\in\mathcal{L}(\mathbf{W}_{\mathcal{N}_{ij}})$.
\end{enumerate}
\end{proposition}
\begin{remark}
  The function $\mathbf{e}(\mathbf{W}_{\mathcal{N}{ij}})$ in (iii) is the \emph{exposure mapping function} of \cite{Aronow2017}, and the weights here adapt their Horvitz–Thompson estimator to two-stage estimands. The corresponding weights generally do not belong to the unbiased estimator class characterized in Theorem \ref{thm:all_weights}, as their unbiasedness holds only over the restricted class of potential outcomes for which the exposure mapping is correctly specified.
These weights can be interpreted as projections of the marginal Radon--Nikodym derivative weights in (i) onto the smaller space of $\mathcal{L}(\mathbf{e}(\mathbf{W}_{\mathcal{N}_{ij}}))$. As a result, they typically exhibit lower variance and can improve efficiency when the exposure mapping assumption is well justified. 
Beyond the direct use of exposure-mapping-based estimators for efficiency gains, the coexistence of exposure-mapping-agnostic and exposure-mapping-based estimators also creates an opportunity to assess whether the exposure mapping is correctly specified. Under correct specification, one may be able to characterize the discrepancy between these two classes of estimators and thereby develop tests in the spirit of a Hausman specification test. Interested readers may refer to \cite{Gao2026} for related ideas.
\end{remark}

\begin{remark}
When the unit-level treatment assignment is complex in its distribution but the cluster-level assignment is simple, the weights in~(iv) are often the most computationally efficient choice for unbiased estimation. Alternatively, to approximate the marginal Radon--Nikodym derivative weights in~(i) with reduced variance, one can regress
\(
\frac{\mathbf{1}\!\left(\mathbf{C}_{\mathcal{N}_{ij}}=c\mathbf{1}\right)}{P\!\left(\mathbf{C}_{\mathcal{N}_{ij}}=c\mathbf{1}\right)}
\)
on the covariate vector $\mathbf{W}_{\mathcal{N}_{ij}}$, motivated by the identity
\[
\mathbb{E}\!\left[
\frac{\mathbf{1}\!\left(\mathbf{C}_{\mathcal{N}_{ij}}=c\mathbf{1}\right)}{P\!\left(\mathbf{C}_{\mathcal{N}_{ij}}=c\mathbf{1}\right)}
\,\Bigg|\,
\mathbf{W}_{\mathcal{N}_{ij}}
\right]
=\alpha_{ij}(\phi_c)
\]

\end{remark}
\begin{remark}
Under the classical cluster randomized trial design, where all units in treated clusters receive treatment and all units in control clusters receive control, we may consider the weights in (i), (ii), and (iv). As noted in the remark following Assumption \ref{ass:measure}, Assumption \ref{ass:measure}(ii) is violated in this setting. Consequently, the complete Radon--Nikodym derivative weights in (ii) are not well-defined, and cluster-agnostic unbiased identification is not attainable.

For the weights identifying marginal average potential outcomes, note that under the counterfactual regime $\phi_c$,
\[
P_{\phi_c}\!\left(\mathbf{W}_{\mathcal{N}_{ij}} = c\mathbf{1}\right) = 1,
\]
so the support of $\mathbf{W}_{\mathcal{N}_{ij}}$ under $P_{\phi_c}$ is degenerate at $c\mathbf{1}$, which is a vector of either all ones or all zeros. The marginal Radon--Nikodym derivative then reduces to
\[
\alpha_{ij}(\phi_c)
=
\frac{1(\mathbf{W}_{\mathcal{N}_{ij}} = c\mathbf{1})}
{P(\mathbf{W}_{\mathcal{N}_{ij}} = c\mathbf{1})}
=
\frac{1(\mathbf{C}_{\mathcal{N}_{ij}} = c\mathbf{1})}
{P(\mathbf{C}_{\mathcal{N}_{ij}} = c\mathbf{1})}.
\]
Thus, the weights in (i) identifying marginal average potential outcomes coincide with the IPT weights in (iv).

For non-marginal average potential outcomes, only $\bar{Y}_{ij}(c,\phi_c)$ is well-defined, as the unit-level treatment is deterministic under $\phi_c$. In this case,
\[
\alpha_{ij}(\phi_c)\frac{1(W_{ij}=c)}{P_{\phi_c}(W_{ij}=c)}
=
\alpha_{ij}(\phi_c)
=
\frac{1(\mathbf{C}_{\mathcal{N}_{ij}} = c\mathbf{1})}
{P(\mathbf{C}_{\mathcal{N}_{ij}} = c\mathbf{1})}
=
\frac{1(\mathbf{C}_{\mathcal{N}_{ij}} = c\mathbf{1},\, W_{ij}=c)}
{P(\mathbf{C}_{\mathcal{N}_{ij}} = c\mathbf{1},\, W_{ij}=c)},
\]
again showing that the weights in (i) and (iv) coincide in this setting.
\end{remark}
\section{Estimation and asymptotic properties}
In subsequent analysis, we focus on the weights $\beta_{ij}(\phi)$ and $\beta_{ij}(w,\phi)$ that can achieve unbiased estimation of \emph{marginal population average potential outcome} and \emph{population average potential outcome}. These weights are characterized by the following assumption.
\begin{assumption}[Unbiased identification weights]\label{ass:weight}
In our LW estimator, we impose the following conditions on the weights $\beta_{ij}(\phi)$ and $\beta_{ij}(w,\phi)$ to ensure unbiased identification.
\begin{enumerate}
    \item[(i)] \textbf{(Weights identifying marginal estimands)} For each $\phi\in\{\phi_0,\phi_1\}$
    and all $(i,j)$, we have well-defined marginal weight $\beta_{ij}(\phi)\in\mathcal{L}(\mathbf{W}_{\mathcal{N}_{ij}},\mathbf{C}_{\mathcal{N}_{ij}})$ satisfying
    \[
        \mathbb{E}[\beta_{ij}(\phi)Y_{ij}]=\bar{Y}_{ij}(\phi)\quad\text{and}\quad \mathbb{E}[\beta_{ij}(\phi)]=1.
    \]

    \item[(ii)] \textbf{(Weights identifying non-marginal estimands)} For each pair
    $(w,\phi)\in\{0,1\}\times\{\phi_0,\phi_1\}$, suppose $P_\phi(W_{ij}=w)$ is either strictly
    positive for all $(i,j)$ or equal to zero for all $(i,j)$. For the pairs with $P_\phi(W_{ij}=w)>0$
    for all $(i,j)$, we have well-defined non-marginal weight $\beta_{ij}(w,\phi)\in\mathcal{L}(\mathbf{W}_{\mathcal{N}_{ij}},\mathbf{C}_{\mathcal{N}_{ij}})$ satisfies
    \[
        \mathbb{E}[\beta_{ij}(w,\phi)Y_{ij}]=\bar{Y}_{ij}(w,\phi)\quad\text{and}\quad \mathbb{E}[\beta_{ij}(w,\phi)]=1.
    \]
\end{enumerate}
\end{assumption}
\begin{remark}
This assumption does not require the weights to deliver unbiased identification over the entire space $\mathcal L(\mathbf W_{\mathcal N_{ij}})$ of
potential outcome functions. Rather, it requires the identification moment
condition in Theorem~\ref{thm:all_weights} to hold only for the actual
potential outcome function $Y_{ij}(\cdot)$ and for the constant function in
that space. This allows for the use of weights that exploit exposure-mapping structures. Nevertheless, we recommend adopting an exposure-mapping-agnostic approach when constructing weights in practice. 

Moreover, since we impose conditions directly on the validity of the weights, we do not explicitly assume measure transportability in the subsequent analysis, as those conditions are only needed to guarantee the existence of unbiased identification weights.
\end{remark}
Among the weight class in Assumption~\ref{ass:weight}, we use the following assumption to characterize a smaller weight class that achieves the cluster-agnostic unbiasedness property discussed above.
\begin{assumption}[Cluster-agnostic unbiased identification weights]\label{ass:weight_cluster_agnostic}
In our LW estimator, we impose the following conditions on the weights $\beta_{ij}(\phi)$ and $\beta_{ij}(w,\phi)$ to ensure cluster-agnostic unbiased identification.
\begin{enumerate}
    \item[(i)] \textbf{(Weights identifying marginal estimands)} For each $\phi\in\{\phi_0,\phi_1\}$
    and all $(i,j)$, we have well-defined marginal weight $\beta_{ij}(\phi)\in\mathcal{L}(\mathbf{W}_{\mathcal{N}_{ij}},\mathbf{C}_{\mathcal{N}_{ij}})$ satisfying
    \[
        \mathbb{E}[\beta_{ij}(\phi)Y_{ij}|\mathbf{C}]=\bar{Y}_{ij}(\phi)\quad \text{and}\quad \mathbb{E}[\beta_{ij}(\phi)|\mathbf{C}]=1.
    \]

    \item[(ii)] \textbf{(Weights identifying non-marginal estimands)} For each pair
    $(w,\phi)\in\{0,1\}\times\{\phi_0,\phi_1\}$, suppose $P_\phi(W_{ij}=w)$ is either strictly
    positive for all $(i,j)$ or equal to zero for all $(i,j)$. For the pairs with $P_\phi(W_{ij}=w)>0$
    for all $(i,j)$, we have well-defined non-marginal weight $\beta_{ij}(w,\phi)\in\mathcal{L}(\mathbf{W}_{\mathcal{N}_{ij}},\mathbf{C}_{\mathcal{N}_{ij}})$ satisfies
    \[
        \mathbb{E}[\beta_{ij}(w,\phi)Y_{ij}|\mathbf{C}]=\bar{Y}_{ij}(w,\phi)\quad \text{and}\quad \mathbb{E}[\beta_{ij}(w,\phi)|\mathbf{C}]=1.
    \]
\end{enumerate}
\end{assumption}
\begin{remark}
Similar as in Assumption \ref{ass:weight}, this assumption requires the identification moment
condition in Theorem~\ref{thm:cluster_agnostic_weights} to hold only for the actual
potential outcome function $Y_{ij}(\cdot)$ and for the constant function within the space $\mathcal L(\mathbf W_{\mathcal N_{ij}})$. Also, we do not need to make the measure transportability assumption if we already make this assumption.
\end{remark}
Given the weights satisfying Assumption \ref{ass:weight}, we can construct unbiased LW estimator for $\bar{Y}(\phi)$ as
\[\hat{\tau}_{naive}(\phi)
=\displaystyle
\sum_{i=1}^n \frac{g_i}{N_i}\sum_{j=1}^{N_i} Y_{ij}\,\beta_{ij}(\phi),\]
and unbiased LW estimator for $\bar{Y}(w,\phi)$ as
\[\hat{\tau}_{naive}(w,\phi)
=\displaystyle
\sum_{i=1}^n \frac{g_i}{N_i}\sum_{j=1}^{N_i} Y_{ij}\,\beta_{ij}(w,\phi).\]
Following standard practice in the study of weighting estimators, we focus on the \emph{Hájek-type} variants of the naive estimators introduced above, which typically exhibit improved finite-population accuracy. They are defined as
\[
\hat{\tau}(\phi)
=
\frac{
\displaystyle
\sum_{i=1}^n \frac{g_i}{N_i}\sum_{j=1}^{N_i} Y_{ij}\,\beta_{ij}(\phi)
}{
\displaystyle
\sum_{i=1}^n \frac{g_i}{N_i}\sum_{j=1}^{N_i} \beta_{ij}(\phi)
},\text{ and } \hat{\tau}(w,\phi)
=
\frac{
\displaystyle
\sum_{i=1}^n \frac{g_i}{N_i}\sum_{j=1}^{N_i} Y_{ij}\,\beta_{ij}(w,\phi)
}{
\displaystyle
\sum_{i=1}^n \frac{g_i}{N_i}\sum_{j=1}^{N_i} \beta_{ij}(w,\phi)
}.
\]

In the remainder of this section, we analyze the asymptotic properties of the class of Hájek estimators. We establish results for designs in which either the cluster-level randomization scheme or the unit-level randomization scheme belongs to the independent Bernoulli class or the complete randomization class.

Before formal analysis, we introduce several notational conventions used throughout the subsequent analysis to facilitate the presentation of our results. Let
\[
\boldsymbol{\beta}_{\mathrm{marg}}(\phi)
=
\bigl(\beta_{11}(\phi),\dots,\beta_{1N_1}(\phi),\dots,\beta_{nN_n}(\phi)\bigr)
\]
denote the row vector of marginal weights under regime $\phi$. Stacking these vectors yields the $2\times N$ marginal weight matrix
\[
\mathbf{B}_{\mathrm{marg}}
=
\begin{pmatrix}
\boldsymbol{\beta}_{\mathrm{marg}}(\phi_1) \\
\boldsymbol{\beta}_{\mathrm{marg}}(\phi_0)
\end{pmatrix}.
\]
Similarly, let
\[
\boldsymbol{\beta}(w,\phi)
=
\bigl(\beta_{11}(w,\phi),\dots,\beta_{nN_n}(w,\phi)\bigr)
\]
denote the row vector of non-marginal weights. Stacking all well-defined rows in the canonical order $(1,\phi_1),(0,\phi_1),(1,\phi_0),(0,\phi_0)$ yields the non-marginal weight matrix $\mathbf{B}$, which has $N$ columns and at most four rows.

Also, we introduce vector notation for the estimands and their corresponding estimators. Define the vector of \emph{marginal population average potential outcomes} as
\[
\boldsymbol{\tau}_{\mathrm{marg}}
=
\bigl(
\bar Y(\phi_1),
\bar Y(\phi_0)
\bigr)^{\top},
\]
and the corresponding vector of estimators as
\[
\hat{\boldsymbol{\tau}}_{\mathrm{marg}}
=
\bigl(
\hat{\tau}(\phi_1),
\hat{\tau}(\phi_0)
\bigr)^{\top}.
\]
Similarly, define the vector of \emph{marginal individual average potential outcomes} as
\[
\boldsymbol{\tau}_{ij,\mathrm{marg}}
=
\bigl(
\bar Y_{ij}(\phi_1),
\bar Y_{ij}(\phi_0)
\bigr)^{\top}.
\]
In parallel, let $\boldsymbol{\tau}$ denote the vector of \emph{population average potential outcomes}, defined as the collection of all well-defined $\bar Y(w,\phi)$, and let $\hat{\boldsymbol{\tau}}$ denote the corresponding vector of estimators, defined as the collection of all well-defined $\hat{\tau}(w,\phi)$. Similarly, let $\boldsymbol{\tau}_{ij}$ denote the vector of \emph{individual average potential outcomes}, defined as the collection of all well-defined $\bar Y_{ij}(w,\phi)$. These vectors are all arranged in the canonical order $(1,\phi_1),(0,\phi_1),(1,\phi_0),(0,\phi_0)$.

Using the weight matrices $\mathbf{B}$ and $\mathbf{B}_{\mathrm{marg}}$, the Hájek estimators admit the compact vectorized representations
\begin{align*}
\hat{\boldsymbol\tau}
&=
\left(\sum_{(i,j)}\frac{g_i}{N_i}\mathbf B_{(*,ij)}\right)^{\!\circ -1}
\circ
\sum_{(i,j)}\frac{g_i}{N_i}Y_{ij}\mathbf B_{(*,ij)},
\\
\hat{\boldsymbol\tau}_{\mathrm{marg}}
&=
\left(\sum_{(i,j)}\frac{g_i}{N_i}(\mathbf B_{\mathrm{marg}})_{(*,ij)}\right)^{\!\circ -1}
\circ
\sum_{(i,j)}\frac{g_i}{N_i}Y_{ij}(\mathbf B_{\mathrm{marg}})_{(*,ij)},
\end{align*}
where $\mathbf B_{(*,ij)}$ and $(\mathbf B_{\mathrm{marg}})_{(*,ij)}$ denote the columns of $\mathbf B$ and $\mathbf B_{\mathrm{marg}}$ corresponding to unit $(i,j)$, $\circ$ denotes the Hadamard (element-wise) product, and $(\cdot)^{\circ -1}$ its element-wise inverse.

Throughout the subsequent analysis, \(I\) denotes the identity matrix, while \(\mathbf{1}\) and \(\mathbf{0}\) denote a vector of ones and a vector of zeros, respectively. 
To discuss asymptotic orders, we use the conventional asymptotic notations \(O(\cdot)\) and \(o(\cdot)\). Building on these notations, we furthermore define \(f_n=\Omega(g_n)\) to mean \(g_n=O(f_n)\), \(f_n=\omega(g_n)\) to mean \(g_n=o(f_n)\), and \(f_n=\Theta(g_n)\) to mean that both \(f_n=O(g_n)\) and \(g_n=O(f_n)\) hold.
We also use the stochastic analogues \(O_P(\cdot)\) and \(o_P(\cdot)\). Specifically, \(f_n=O_P(g_n)\) means
\(
\sup_n P\left(\left|f_n/g_n\right|>M\right)\to 0
\quad \text{as } M\to\infty,
\)
and \(f_n=o_P(g_n)\) means that, for every \(M>0\),
\(
P\left(\left|f_n/g_n\right|>M\right)\to 0
\quad \text{as } n\to\infty.
\)

\subsection{Theoretical results}
In this subsection, we first introduce a set of regularity assumptions for the asymptotic analysis. Among them, assumption (i) is relatively strong and is only used for the analysis of cluster-agnostic estimators. Assumption (ii) is a weaker version of (i) and is employed in the analysis of general estimators. The remaining conditions are required for the analysis of both types of estimators.
\begin{assumption}[Regularity for asymptotic analysis]\label{ass:CLT_regularity}
\begin{enumerate}
    \item[(i)] The interference neighborhood sizes are uniformly bounded: there exists a global constant $C_{\mathcal{N}}<\infty$, independent of sample size, such that $|\mathcal{N}_{ij}|\leq C_{\mathcal{N}}$ for all $(i,j)$.  
    \item[(ii)] The interference neighborhoods exhibit uniformly bounded cluster-level complexity: there exists a global constant $C_{\mathcal{N}^{(cl)}}<\infty$, independent of sample size, such that
    \[
    |\mathcal{N}_{ij}^{(cl)}|\leq C_{\mathcal{N}^{(cl)}}
    \]
    for any unit $(i,j)$, and
    \[
    \bigl|\cup_{j=1}^{N_i}\mathcal{N}_{ij}\bigr|\leq C_{\mathcal{N}^{(cl)}}N_i
    \]
    for any cluster $i$.
    \item[(iii)] The potential outcome functions are uniformly bounded: there exists a finite constant $C_Y$, independent of sample size, such that
    \[
    \max_{i,j,\mathbf{w}}|Y_{ij}(\mathbf{w})|\leq C_Y.
    \]
    \item[(iv)] The weight functions are uniformly bounded: there exists a finite constant $C_{\beta}$, independent of sample size, such that
    \[
    \max_{i,j,w,\phi}|\beta_{ij}(w,\phi)|\leq C_{\beta} \text{ and }\max_{i,j,\phi}|\beta_{ij}(\phi)|\leq C_{\beta}.
    \]
    \item[(v)] The cluster sizes are balanced: there exists a global constant $C_N<\infty$, independent of sample size, such that
    \[
    \frac{\max_i N_i}{\min_i N_i} \leq C_N.
    \]
    \item[(vi)] The weights assigned to clusters are balanced: $\max_i\{g_i\} = O\!\left(\frac{1}{n}\right)$.
\end{enumerate}
\end{assumption}
\begin{remark}
Assumptions (i) and (ii) are introduced to characterize the complexity of interference, namely how one unit’s treatment affects another unit’s outcome, rather than the dependence induced by the treatment assignment mechanism.

Assumption (ii) serves as our primary assumption for controlling the complexity of the interference network, and is analogous to standard network complexity conditions in the two-stage randomization literature. 
Its first inequality imposes a pointwise bound on the number of neighboring clusters for each unit $(i,j)$, namely $|\mathcal{N}_{ij}^{(cl)}|\leq C_{\mathcal{N}^{(cl)}}$. 
Its second inequality requires that, within each cluster, the union of all
unit-level interference neighborhoods contains at most $O(N_i)$ units. In
other words, the subpopulation affected by the treatments in each cluster
must be proportional to the cluster size. 
Heuristically, the two inequalities require each unit's interference
neighborhood to be concentrated within a relatively small number of clusters,
rather than dispersed across too many clusters. Notably, condition (ii)
permits the size of each unit's interference neighborhood to diverge with the
sample size: the first inequality allows up to $O(N/n)$ neighboring units for each unit,
while the second requires sufficient overlap among such neighborhoods for units in the same cluster.

In practice, Assumption~(ii) accommodates two distinct interference patterns.
The first is geographic spillover. A unit's geographic neighborhood may contain
a large and growing number of nearby units, but units in the same cluster tend
to have highly overlapping geographic neighborhoods. Consequently, the union of all
interference neighborhoods within a cluster typically remains comparable to
the original cluster in both geographic extent and sample size, as required by
the second inequality.
The second is social spillover transmitted through friendship or family ties.
Such spillovers are typically non-geographic and sparse, with each unit linked
to only $O(1)$ other units. These links may generate disproportionately large
effects, but their bounded number ensures that both inequalities remain
satisfied. 
In \cite{Leung2025}, interference is modeled as geographic spillover only, with each
unit primarily affected by finitely many nearby clusters. Their
interference-complexity assumption is therefore comparable to assumption (ii) here.

In contrast, assumption (i) is a strictly stronger condition and is used only for cluster-
agnostic estimators to obtain faster convergence rates. This is intuitive: cluster-agnostic
estimators can only simplify the dependence induced by the treatment assignment mechanism, but not that arising from interference between units. When assumption (i) imposes a weaker level of interference, it leaves more room for carefully-designed estimators to improve their asymptotic performance. This is related to an assumption made by Lu et al. (2026), which
similarly restricts each unit to being affected by only finitely many others.

\end{remark}
\begin{remark}
The global boundedness assumption in (iv) requires both the positivity of treatment assignment mechanism and the stability of the weight function itself.
For the example weight functions in Proposition~\ref{prop:examples_of_weights}, we discuss some sufficient conditions for this assumption.

For the IPT weights in \cite{Leung2025}, the global boundedness is satisfied as long as the corresponding denominators are bounded away from zero.
The denominators in the weights identifying \emph{marginal population average potential outcomes} are given by sums of joint probabilities of the form
\[
P\bigl(\mathbf{C}_{\mathcal{N}_{ij}}=\mathbf{c}_{\mathcal{N}_{ij}}\bigr)
\]
over all admissible realizations of $\mathbf{c}_{\mathcal{N}_{ij}}$. 
For these weights, assumption (iv) holds when assumption (ii) is satisfied together with the following positivity conditions:
\[
P\bigl(C_i=c \mid C_{i_k}=c_k,\ 1\leq k\leq C_{\mathcal{N}}\bigr)\in \{0\}\cup[p,1],
\]
for some constant $p>0$, uniformly over every cluster index $i$, every cluster sequence $\{i_k\}_{k=1}^{C_{\mathcal{N}}}$, every value of $c$, and every value sequence of $\{c_k\}_{k=1}^{C_{\mathcal{N}}}$. 
These positivity conditions are satisfied by a broad class of randomization schemes in which conditioning on finitely many units does not induce sequences of strictly positive probabilities that converge to zero. At the same time, they allow for realizations with exactly zero joint probability, which is possible in complete randomization.
For the weights identifying \emph{population average potential outcomes}, the denominators are typically given by those of the marginal weights multiplied by a nonzero $P_{\phi}(W_{ij}=w)$. For these weights, assumption (iv) holds when the above conditions for marginal weights are satisfied and, in addition,
\[
P_{\phi}(W_{ij}=w)\in \{0\}\cup[q,1]
\]
for some constant $q>0$ and all $(\phi,w)$ and $(i,j)$.

For the marginal and exposure-mapping-level Radon--Nikodym derivative weights, assumption (iv) holds whenever it holds for the IPT weights, so the sufficient conditions above also apply here. For the marginal weights, we have the identities
\[
\begin{aligned}
\mathbb{E}\!\left[
\frac{1(\mathbf{C}_{\mathcal{N}_{ij}}=c\mathbf{1})}{P(\mathbf{C}_{\mathcal{N}_{ij}}=c\mathbf{1})}
\,\Big|\,
\mathbf{W}_{\mathcal{N}_{ij}}
\right]
&=\alpha_{ij}(\phi_c),\\
\mathbb{E}\!\left[
\frac{1(\mathbf{C}_{\mathcal{N}_{ij}}=c\mathbf{1})}{P(\mathbf{C}_{\mathcal{N}_{ij}}=c\mathbf{1})}
\frac{1(W_{ij}=w)}{P_{\phi_c}(W_{ij}=w)}
\,\Big|\,
\mathbf{W}_{\mathcal{N}_{ij}}
\right]
&=\alpha_{ij}(\phi_c)\frac{1(W_{ij}=w)}{P_{\phi_c}(W_{ij}=w)}.
\end{aligned}
\]
Since a conditional expectation never exceeds the essential supremum of its argument, the uniform bound $C_{\beta}$ on the IPT weights carries over to the marginal weights:
\[
\max_{i,j,c}\frac{1(\mathbf{C}_{\mathcal{N}_{ij}}=c\mathbf{1})}{P(\mathbf{C}_{\mathcal{N}_{ij}}=c\mathbf{1})}\leq C_{\beta}
\;\Rightarrow\;
\max_{i,j,c}\alpha_{ij}(\phi_c)\leq C_{\beta},
\]
\[
\begin{aligned}
&\max_{i,j,w,c}\frac{1(\mathbf{C}_{\mathcal{N}_{ij}}=c\mathbf{1})}{P(\mathbf{C}_{\mathcal{N}_{ij}}=c\mathbf{1})}\frac{1(W_{ij}=w)}{P_{\phi_c}(W_{ij}=w)}\leq C_{\beta}\Rightarrow\;
\max_{i,j,w,c}\alpha_{ij}(\phi_c)\frac{1(W_{ij}=w)}{P_{\phi_c}(W_{ij}=w)}\leq C_{\beta}.
\end{aligned}
\]
The exposure-mapping-level weights are handled analogously, being conditional expectations of the IPT weights given $\mathbf{e}(\mathbf{W}_{\mathcal{N}_{ij}})$ rather than $\mathbf{W}_{\mathcal{N}_{ij}}$. In both cases, then, the positivity conditions that ensure assumption (iv) for the IPT weights also ensure it for these Radon--Nikodym derivative weights.

For the complete Radon--Nikodym derivative weights, however, we need to bound away from zero the denominators given by sums of joint probabilities of the form
\[
P\bigl(\mathbf{W}_{\mathcal{N}_{ij}}=\mathbf{w}_{\mathcal{N}_{ij}},\mathbf{C}_{\mathcal{N}_{ij}}=\mathbf{c}_{\mathcal{N}_{ij}}\bigr)
\]
over all admissible realizations of $\mathbf{w}_{\mathcal{N}_{ij}}$ and $\mathbf{c}_{\mathcal{N}_{ij}}$. 
When $\mathbf{W}_{\mathcal{N}_{ij}}$ contains finitely many units, there could be strictly positive probabilities that converge to zero as $N\to\infty$. In this case, the global boundedness may be violated. Therefore, we may need the bounded degree assumption in (i) to obtain the global boundedness of the complete Radon--Nikodym derivative weights. For weights estimating marginal potential outcomes, assumption (iv) holds when assumption (i) is satisfied together with the following positivity conditions:
\[
P\bigl(C_i=c \mid C_{i_k}=c_k,\ 1\leq k\leq C_{\mathcal{N}}\bigr)\in \{0\}\cup[p,1],
\]
\[
P\bigl(W_{ij}=w \mid W_{i_kj_k}=w_k,\ C_{i_k}=c_k,\ 1\leq k\leq C_{\mathcal{N}}\bigr)\in \{0\}\cup[p,1],
\]
for some constant $p>0$, uniformly over every unit $(i,j)$, every unit sequence $\{(i_k,j_k)\}_{k=1}^{C_{\mathcal{N}}}$, every value of $(w,c)$, and every value sequence of $\{(w_k,c_k)\}_{k=1}^{C_{\mathcal{N}}}$. 
For weights identifying non-marginal potential outcomes, we similarly require the additional condition
\[
P_{\phi}(W_{ij}=w)\in \{0\}\cup[q,1]
\]
for some constant $q>0$ and all $(\phi,w)$ and $(i,j)$.
\end{remark}
\begin{remark}
    In (vii), we assume that the weight assigned to each cluster is not substantially larger than the inverse of the number of clusters. This condition is satisfied, for example, when the ratio between the largest and smallest nonzero weights is uniformly bounded, or when an $O(1)$ proportion of clusters have zero weight and the remaining weights are similarly controlled. When only an $o(1)$ proportion of clusters are assigned nonzero weights (e.g., when focusing on the causal effect for a specific cluster), the problem can be reformulated by restricting attention to those clusters and their associated units. Our asymptotic results continue to hold under this reduced data setup. The only case excluded by (vii) is when a small number of clusters receive disproportionately large weights while many others receive small, but nonzero, weights, a situation that is typically not of practical interest.
\end{remark}

Next, we introduce two assumptions that characterize the classes of randomization schemes of interest. These schemes apply to both cluster-level and unit-level randomization within a two-stage framework. It is typically sufficient for only one stage of randomization to belong to these classes in order to establish the asymptotic theory for each estimator. This flexibility allows us to accommodate a wide range of experimental designs.
\begin{assumption}[Independent Bernoulli randomization schemes]\label{ass:independent_rand}

\begin{enumerate}

    \item[(i)] The cluster-level treatment assignments $\{C_i\}_{i=1}^n$ are mutually independent.
    
    \item[(ii)] For each cluster $i$, the unit-level treatment assignments $\{W_{ij}\}_{j=1}^{N_i}$ are mutually independent conditional on $C_i$.
\end{enumerate}
\end{assumption}
\begin{remark}
The randomization schemes considered here allow for i.i.d.\ treatment assignment. They also allow assignment probabilities to vary across units or clusters according to covariates or network structure. Since we are in an experimental setting, we can treat the network structure and covariates as non-random, which does not affect the asymptotic analysis.
\end{remark}
\begin{remark}
We do not impose explicit positivity assumptions on the treatment probabilities in the cluster-level and unit-level randomization schemes. Instead, such conditions are implicitly enforced through Assumption~\ref{ass:measure}, Assumption~\ref{ass:CLT_regularity}(iii), and the requirement that $P_{\phi}(W_{ij}=w)>0$ to ensure that \emph{population average potential outcomes} are well defined.
Assumption~\ref{ass:measure} requires that the realized treatment mechanism has sufficiently rich support, imposing a weak condition on nonzero probabilities as discussed in the remark following it. Assumption~\ref{ass:CLT_regularity}(iv) further requires that conditioning on finitely many units does not induce sequences of strictly positive probabilities that converge to zero, and is not restrictive for commonly used experimental designs. 
The same implicit requirements carry over to Assumption~\ref{ass:complete_rand}; the remark following that assumption therefore focuses only on the conditions specific to complete randomization.
\end{remark}

\begin{assumption}[Complete randomization schemes]\label{ass:complete_rand}
\begin{enumerate}
    \item[(i)] The cluster index set $\mathcal{I} = \{1,\ldots,n\}$ is partitioned into $K$ disjoint groups $\{\mathcal{I}_k\}_{k=1}^K$. The treatment assignment vector $\mathbf{C}$ is generated independently across groups, and for each $k$, the subvector $\mathbf{C}_{\mathcal{I}_k}$ is uniformly distributed over all assignment vectors satisfying
\[
\sum_{i \in \mathcal{I}_k} C_i = I_k,
\]
for some fixed integer $I_k \leq |\mathcal{I}_k|$. In addition, for some constant $p>0$,
\[
\min_{k}|\mathcal{I}_k|=\Omega(n),
\qquad
\frac{I_k}{|\mathcal{I}_k|}\in\{0,1\}\cup[p,1-p]
\quad\text{for all }k.
\]

    \item[(ii)] For each cluster $i$, the unit index set $\mathcal{J}_{i} = \{1,\ldots,N_i\}$ is partitioned into $K_i$ disjoint groups $\{\mathcal{J}_{ik}\}_{k=1}^{K_i}$. The treatment assignment vector $\mathbf{W}_i$ is generated independently across groups, and for each $k$, the subvector $\mathbf{W}_{i\mathcal{J}_{ik}}$ is uniformly distributed over all assignment vectors satisfying
\[
\sum_{j \in \mathcal{J}_{ik}} W_{ij} = J_{ik},
\]
for some fixed integer $J_{ik} \leq |\mathcal{J}_{ik}|$. In addition, for some constant $p>0$,
\[
\min_{k}|\mathcal{J}_{ik}|=\Omega(N_i)
\quad\text{for all }i,
\qquad
\frac{J_{ik}}{|\mathcal{J}_{ik}|}\in\{0,1\}\cup[p,1-p]
\quad\text{for all }i,k.
\]
\end{enumerate}
\end{assumption}
\begin{remark}
Assumption \ref{ass:complete_rand} accommodates stratified complete randomization schemes at both the cluster and unit levels. In each case, the population is partitioned into arbitrary strata, within which standard complete randomization is implemented independently. The formation of strata, as well as the within-stratum assignment distributions, may depend on network structure and covariate information. However, within each cluster we allow no heterogeneity in treatment probability.
\end{remark}
\begin{remark}
The conditions on treatment fractions and stratum sizes,
\[
\min_{k}|\mathcal{I}_k|=\Omega(n),
\qquad
\frac{I_k}{|\mathcal{I}_k|}\in\{0,1\}\cup[p,1-p]
\quad\text{for all }k,
\]
and
\[
\min_{k}|\mathcal{J}_{ik}|=\Omega(N_i)
\quad\text{for all }i,
\qquad
\frac{J_{ik}}{|\mathcal{J}_{ik}|}\in\{0,1\}\cup[p,1-p]
\quad\text{for all }i,k,
\]
serve two roles. First, they bound the cross-treatment dependence induced by complete randomization. Second, they satisfy the positivity conditions implicitly required by Assumption~\ref{ass:CLT_regularity}(iv): for the example weight functions in Proposition~\ref{prop:examples_of_weights}, they suffice to verify this assumption under Assumption~\ref{ass:CLT_regularity}(ii) for most estimators, and under the stronger Assumption~\ref{ass:CLT_regularity}(i) for the complete Radon--Nikodym derivative estimator.
\end{remark}
Then, with either the cluster-level or unit-level treatment following the class of independent Bernoulli randomization schemes or complete randomization schemes defined above, we have the following upper bound on estimator variance.
\begin{theorem}[Order of the variance matrices]\label{thm:variance_order}
Suppose Assumptions~\ref{ass:cluster_assignment}--\ref{ass:interference}, \ref{ass:weight}, 
and \ref{ass:CLT_regularity}(ii)-(vi) hold. Define the asymptotic covariance matrices of 
$\hat{\boldsymbol{\tau}}_{\mathrm{marg}}$ and $\hat{\boldsymbol{\tau}}$ as
\begin{align*}
\mathbf{\Sigma}_{\mathrm{marg}}
&=
\mathrm{Var}\!\left(
\sum_{i=1}^n \frac{g_i}{N_i}
\sum_{j=1}^{N_i}
(\mathbf B_{\mathrm{marg}})_{(*,ij)}\circ(Y_{ij}\mathbf{1}-\boldsymbol{\tau}_{\mathrm{marg}})
\right), \\
\mathbf{\Sigma}
&=
\mathrm{Var}\!\left(
\sum_{i=1}^n \frac{g_i}{N_i}
\sum_{j=1}^{N_i}
\mathbf B_{(*,ij)}\circ(Y_{ij}\mathbf{1}-\boldsymbol{\tau})
\right),
\end{align*}
where $\mathbf B_{(*,ij)}$ and $(\mathbf B_{\mathrm{marg}})_{(*,ij)}$ denote the columns of $\mathbf B$ and $\mathbf B_{\mathrm{marg}}$ associated with unit $(i,j)$, respectively. Then:
\begin{enumerate}
    \item[(i)] If either Assumption~\ref{ass:independent_rand}(i) or Assumption~\ref{ass:complete_rand} (i) holds, then
    \begin{align*}
    \max\{\|\boldsymbol\Sigma_{\mathrm{marg}}\|,\|\boldsymbol\Sigma\|\}
    = O\!\left(n^{-1}\right).
    \end{align*}
    
    \item[(ii)] If Assumption~\ref{ass:weight_cluster_agnostic} holds, Assumption~\ref{ass:CLT_regularity} (i) holds, and either Assumption~\ref{ass:independent_rand}(ii) or Assumption~\ref{ass:complete_rand}(ii) holds, then
    \begin{align*}
    \max\{\|\boldsymbol\Sigma_{\mathrm{marg}}\|,\|\boldsymbol\Sigma\|\}
    = O\!\left(N^{-1}\right).
    \end{align*}
    Furthermore, we could write the two asymptotic covariance matrices as
      $\mathbf{\Sigma}_{\mathrm{marg}}
      =
      \mathbb{E}[\mathbf{\Sigma}_{\mathrm{marg}}(\mathbf{C})]$ and $
      \mathbf{\Sigma}=\mathbb{E}[\mathbf{\Sigma}(\mathbf{C})]$,
      where 
      \[
        \mathbf{\Sigma}_{\mathrm{marg}}(\cdot)=\mathrm{Var}\!\left(
          \sum_{i=1}^n \frac{g_i}{N_i}
          \sum_{j=1}^{N_i}
          (\mathbf B_{\mathrm{marg}})_{(*,ij)}\circ(Y_{ij}\mathbf{1}-\boldsymbol{\tau}_{\mathrm{marg}})
          \mid \mathbf{C}=\cdot\right)
      \]
      and 
      \[
        \mathbf{\Sigma}(\cdot)=\mathrm{Var}\!\left(
          \sum_{i=1}^n \frac{g_i}{N_i}
          \sum_{j=1}^{N_i}
          \mathbf B_{(*,ij)}\circ(Y_{ij}\mathbf{1}-\boldsymbol{\tau})
          \mid \mathbf{C}=\cdot\right)
      \]
      are the conditional covariance matrices with respect to the value of cluster-level treatments $\mathbf{C}=(C_1,\cdots,C_n)^\top$.
\end{enumerate}
\end{theorem}
\begin{remark}
    Here we claim the two matrices to be the asymptotic covariance matrices of $\hat{\boldsymbol{\tau}}_{\mathrm{marg}}$ and $\hat{\boldsymbol{\tau}}$ without proof, which we defer to the proof of the central limit theorem below. In this theorem we only analyze the properties of the two matrices.
\end{remark}
\begin{remark}
Part (i) provides variance bounds for general estimators, while part (ii) establishes bounds specifically for cluster-agnostic estimators. When $N/n \to \infty$ (i.e., as cluster sizes diverge), the convergence rate of cluster-agnostic estimators is faster than that of general estimators. The intuition is that cluster-level treatments induce dependence among all units within a cluster unless the weights are cluster-agnostic, and such dependence will generate a dominant (lower-order) term in the variance. 
\end{remark}
\begin{remark}
For general estimators, since cluster-level treatments already induce dependence within clusters, the variance bounds allow arbitrary unit-level treatment rules and only require assumptions on the cluster-level randomization scheme. In contrast, for cluster-agnostic estimators, cluster-level treatments do not directly contribute a term to the variance, and the variance bounds only require assumptions on unit-level randomization scheme.
\end{remark}
Finally, we establish asymptotic normality for $\hat{\boldsymbol{\tau}}_{\mathrm{marg}}$ and $\hat{\boldsymbol{\tau}}$.
\begin{theorem}[Central limit theorem]
\label{thm:clt}
Suppose Assumptions~\ref{ass:cluster_assignment}--\ref{ass:interference}, \ref{ass:weight}, 
and \ref{ass:CLT_regularity}(ii)-(vi) hold. Let $\lambda_{\min}(\cdot)$ denote the smallest eigenvalue of a symmetric matrix.
\begin{enumerate}
    \item[(i)] Suppose that either Assumption~\ref{ass:independent_rand}(i) or Assumption~\ref{ass:complete_rand}(i) holds. Moreover, assume that
    \[
    \min\{\lambda_{\min}(\mathbf{\Sigma}_{\mathrm{marg}}),
    \lambda_{\min}(\mathbf{\Sigma})\}
    = \omega\!\left(n^{-4/3}\right).
    \]
    Then
    \[
    \mathbf{\Sigma}_{\mathrm{marg}}^{-1/2}
    \bigl(\hat{\boldsymbol{\tau}}_{\mathrm{marg}}-\boldsymbol{\tau}_{\mathrm{marg}}\bigr)
    \;\Rightarrow\;
    \mathcal{N}(\mathbf{0}, I),
    \]
    and
    \[
    \mathbf{\Sigma}^{-1/2}
    \bigl(\hat{\boldsymbol{\tau}}-\boldsymbol{\tau}\bigr)
    \;\Rightarrow\;
    \mathcal{N}(\mathbf{0}, I).
    \]
    
    \item[(ii)] If Assumptions~\ref{ass:weight_cluster_agnostic}, \ref{ass:CLT_regularity}(i), and~\ref{ass:independent_rand}(ii) hold, and
    \[
    \inf_{\mathbf{c}\in\mathrm{supp}(\mathbf{C})}
    \min\{\lambda_{\min}(\mathbf{\Sigma}_{\mathrm{marg}}(\mathbf{c})),
    \lambda_{\min}(\mathbf{\Sigma}(\mathbf{c}))\}
    = \omega\!\left(N^{-4/3}\right),
    \]
    where $\mathrm{supp}(\mathbf{C})$ denotes the support of the cluster-level treatments $\mathbf{C}$, then
    \[
    \mathbf{\Sigma}_{\mathrm{marg}}(\mathbf{C})^{-1/2}
    \bigl(\hat{\boldsymbol{\tau}}_{\mathrm{marg}}-\boldsymbol{\tau}_{\mathrm{marg}}\bigr)
    \;\Rightarrow\;
    \mathcal{N}(\mathbf{0}, I),
    \]
    and
    \[
    \mathbf{\Sigma}(\mathbf{C})^{-1/2}
    \bigl(\hat{\boldsymbol{\tau}}-\boldsymbol{\tau}\bigr)
    \;\Rightarrow\;
    \mathcal{N}(\mathbf{0}, I).
    \]
\end{enumerate}
\end{theorem}
The proof of this theorem is a key contribution of the paper and requires novel and careful use of Stein's technique to establish central limit theorem which accommodates complete randomization, a challenging setting that involves general statistics that can be written as sums of local components, where each component depends on the treatment variables within the interference neighborhood of a given unit. 
A detailed description of the required analysis is relegated to Section \ref{sec:technical} so as not interrupt the flow of the exposition as it is substantially more technical than the rest of the paper and is not required in order to apply our estimators; this will allow readers focused on methodology to skip it without loss of continuity.
\begin{remark}
The lower bounds on the variance in (i) and (ii) require the corresponding upper bounds for the two types of estimators in Theorem~\ref{thm:variance_order} to be relatively tight. This requirement is more easily satisfied for cluster-agnostic estimators in (ii), whereas for general estimators in (i) it may be violated when the estimator is too close to a cluster-agnostic estimator and its variance decreases approximately at the rate of $O(N^{-1})$.
\end{remark}
\begin{remark}
  Beyond the difference in convergence rates, parts (i) and (ii) also differ qualitatively in the nature of the limiting distribution. In part (i), the normalization is by the deterministic covariance matrix $\mathbf{\Sigma}$ (or $\mathbf{\Sigma}_{\mathrm{marg}}$), so the estimator is asymptotically normal in the usual unconditional sense, with the traditional unimodal normal distribution as its limit. In part (ii), by contrast, asymptotic normality is established conditionally on the cluster-level treatments, with normalization by the random conditional covariance matrix $\mathbf{\Sigma}(\mathbf{C})$ (or $\mathbf{\Sigma}_{\mathrm{marg}}(\mathbf{C})$). Unconditionally, the asymptotic distribution of a cluster-agnostic estimator is therefore not a normal distribution but a mixture of normal distributions, with one component for each realization of the cluster-level assignment $\mathbf{C}$; such a normal mixture is in general non-normal and can exhibit multiple modes. This phenomenon is not specific to our framework: as noted by \citet{Liu2014}, mixed-normal (and hence possibly multimodal) limiting distributions commonly arise in the two-stage randomization experiment framework. In contrast, estimators based on weights that are not cluster-agnostic fall under part (i) and retain the traditional unimodal normal limit.
\end{remark}
\begin{remark}
When complete randomization is implemented at the unit level, we are unable to establish asymptotic normality for the cluster-agnostic estimators. This may reflect limitations of the proof technique discussed in Section~\ref{sec:technical}, as well as more intricate asymptotic behavior in this setting. We do not pursue this case further, as cluster-agnostic estimators often exhibit inferior accuracy in practice compared to estimators in the broader class, for which central limit theorems hold under condition (i) of the theorem. Moreover, even in settings where asymptotic normality can be established, the finite-sample distribution of cluster-agnostic estimators may not be well approximated by a normal distribution, as indicated by our simulation results. This discrepancy appears to be driven by the heavy-tailed nature of the cluster-agnostic weights.
\end{remark}

\section{Inference}

In this section, we construct conservative variance estimators for the three classes of estimators covered in Theorem~\ref{thm:clt}: (i) general estimators under cluster-level independent Bernoulli randomization, (ii) cluster-agnostic estimators under unit-level independent Bernoulli randomization, and (iii) general estimators under cluster-level complete randomization. We treat the two Bernoulli designs jointly in the first subsection and the complete-randomization design in the second subsection.

We begin by introducing some matrix notation used throughout this section. For matrix $\mathbf A$, the \emph{Frobenius norm} $\|\mathbf A\|_F$ is defined as 
\(
\|\mathbf A\|_F
:=
\left(\sum_{l=1}^d\sum_{l'=1}^d A_{ll'}^2\right)^{1/2}.
\)
For symmetric matrices $\mathbf A$ and $\mathbf B$ of the same dimension, write $\mathbf A\preceq\mathbf B$ if $\mathbf B-\mathbf A$ is positive semidefinite, and $\mathbf A\succeq\mathbf B$ if $\mathbf A-\mathbf B$ is positive semidefinite. Equivalently, $\mathbf A\preceq\mathbf B$ if and only if $\mathbf B\succeq\mathbf A$. For any symmetric matrix $\mathbf M$ with eigendecomposition
\[
\mathbf M=\mathbf Q\,\mathrm{diag}(d_1,\dots,d_p)\,\mathbf Q^\top,
\]
where $\mathbf Q$ is orthogonal and $d_1,\dots,d_p$ are the real eigenvalues of $\mathbf M$, define its \emph{positive semidefinite part} and \emph{negative semidefinite part} as
\[
(\mathbf M)_{+}
:=
\mathbf Q\,\mathrm{diag}\bigl(\max\{d_1,0\},\dots,\max\{d_p,0\}\bigr)\,\mathbf Q^\top,
\]
and
\[
(\mathbf M)_{-}
:=
\mathbf Q\,\mathrm{diag}\bigl(\max\{-d_1,0\},\dots,\max\{-d_p,0\}\bigr)\,\mathbf Q^\top.
\]
For two symmetric matrices $\mathbf A$ and $\mathbf B$ of the same dimension, their \emph{L\"owner maximum} is defined as
\[
\mathbf A\vee\mathbf B
:=
\mathbf A+(\mathbf B-\mathbf A)_+.
\]

We next introduce a unified kernel-based notation for the network HAC variance estimators used below. Given a symmetric kernel matrix $\mathbf K$ on $\mathcal P\times\mathcal P$, the corresponding variance estimator for $\hat{\boldsymbol\tau}$ is defined as
\[
\hat{\mathbf\Sigma}^{\mathbf K}
:=
\sum_{(i,j),(i',j')}\mathbf K_{ij,i'j'}
\left(\frac{g_i}{N_i}\mathbf B_{(*,ij)}\circ(Y_{ij}\mathbf 1-\hat{\boldsymbol\tau})\right)
\left(\frac{g_{i'}}{N_{i'}}\mathbf B_{(*,i'j')}\circ(Y_{i'j'}\mathbf 1-\hat{\boldsymbol\tau})\right)^T.
\]
Analogously, for the marginal estimator $\hat{\boldsymbol\tau}_{\mathrm{marg}}$, we define
\[
\hat{\mathbf\Sigma}^{\mathbf K}_{\mathrm{marg}}
:=
\sum_{(i,j),(i',j')}\mathbf K_{ij,i'j'}
\left(\frac{g_i}{N_i}(\mathbf B_{\mathrm{marg}})_{(*,ij)}\circ(Y_{ij}\mathbf 1-\hat{\boldsymbol\tau}_{\mathrm{marg}})\right)
\left(\frac{g_{i'}}{N_{i'}}(\mathbf B_{\mathrm{marg}})_{(*,i'j')}\circ(Y_{i'j'}\mathbf 1-\hat{\boldsymbol\tau}_{\mathrm{marg}})\right)^T.
\]
Both estimators are weighted covariance-type quadratic forms: the kernel matrix $\mathbf K$ determines which pairs of observations contribute to the estimator and with what weights. In the following subsections, different choices of $\mathbf K$ are used to reflect the dependence structure induced by each randomization scheme.

\subsection{Variance estimators under independent Bernoulli randomization}

For general estimators under cluster-level independent Bernoulli randomization, we use two kernels $\mathbf K_1$ and $\mathbf K_2$ to construct the corresponding HAC variance estimator. The within-cluster kernel $\mathbf K_1$ is defined by
\[
(\mathbf{K}_1)_{ij,i'j'} = 1\{i=i'\},
\]
which equals one whenever $(i,j)$ and $(i',j')$ belong to the same cluster. $\mathbf K_1$ is the dependency matrix that only accounts for within-cluster dependency induced by cluster-level treatment assignments or within-cluster interference. It is block-diagonal with all-one blocks, hence positive semidefinite. The cross-cluster interference kernel $\mathbf{K}_2$ is defined by
\[
(\mathbf{K}_2)_{ij,i'j'} = 1\{(i',j') \in \Lambda_{ij}\},
\qquad
\Lambda_{ij} := \left\{(i',j')\in\mathcal P : \mathcal N_{ij}^{(cl)}\cap \mathcal N_{i'j'}^{(cl)}\neq\varnothing\right\},
\]
which equals one whenever the cluster-level interference neighborhoods of $(i,j)$ and $(i',j')$ overlap. $\mathbf K_2$ is the dependency matrix that accounts for both the within-cluster dependency induced by cluster-level treatment assignments, and the cross-cluster dependency induced by interference. In this sense, $\mathbf K_2$ accounts for a broader class of potentially dependent pairs than $\mathbf K_1$. Although $\mathbf K_2$ is not guaranteed to be positive semidefinite, its proximity to the positive semidefinite kernel $\mathbf K_1$ suggests that it may be approximately positive semidefinite in practice.

We combine the HAC estimators with the two different kernels by taking their L\"owner maximum. Specifically, following \cite{Leung2025}, we define the variance estimators for $\hat{\boldsymbol\tau}$ and $\hat{\boldsymbol\tau}_{\mathrm{marg}}$ as
\[
\hat{\mathbf\Sigma}_1
=
\hat{\mathbf\Sigma}^{\mathbf K_1}\vee \hat{\mathbf\Sigma}^{\mathbf K_2},
\qquad
\hat{\mathbf\Sigma}_{1,\mathrm{marg}}
=
\hat{\mathbf\Sigma}^{\mathbf K_1}_{\mathrm{marg}}\vee
\hat{\mathbf\Sigma}^{\mathbf K_2}_{\mathrm{marg}}.
\]
By definition of the L\"owner maximum,
\[
\hat{\mathbf\Sigma}^{\mathbf K_1}\preceq \hat{\mathbf\Sigma}_1,
\qquad
\hat{\mathbf\Sigma}^{\mathbf K_2}\preceq \hat{\mathbf\Sigma}_1,
\]
and similarly,
\[
\hat{\mathbf\Sigma}^{\mathbf K_1}_{\mathrm{marg}}
\preceq
\hat{\mathbf\Sigma}_{1,\mathrm{marg}},
\qquad
\hat{\mathbf\Sigma}^{\mathbf K_2}_{\mathrm{marg}}
\preceq
\hat{\mathbf\Sigma}_{1,\mathrm{marg}}.
\]
Therefore, $\hat{\mathbf\Sigma}_1$ is conservative whenever either
$\hat{\mathbf\Sigma}^{\mathbf K_1}$ or $\hat{\mathbf\Sigma}^{\mathbf K_2}$ is conservative, and the same argument applies to
$\hat{\mathbf\Sigma}_{1,\mathrm{marg}}$. The validity of this construction relies on the following assumption on the sparsity of cross-cluster interference.

\begin{assumption}[Cross-cluster interference sparsity]\label{ass:sparsity}
The number of pairs between units with cross-cluster interference is asymptotically negligible relative to the number of pairs between units within the same cluster. 
Formally, with the notation of Frobenius norm, 
\[
\|\mathbf K_2-\mathbf K_1\|_{F}^2
=o\left(
\|\mathbf K_1\|_{F}^2
\right).
\]
\end{assumption}
\begin{remark}
Under Assumption \ref{ass:CLT_regularity}(iii)-(vi),
\(
  \|\mathbf K_1\|_{F}^2=\sum_{i}N_i^2=\Theta\left(\frac{N^2}{n}\right),
\)
and this assumption equivalently requires
\(
  \|\mathbf K_2-\mathbf K_1\|_{F}^2=o\left(\frac{N^2}{n}\right).
\)
\end{remark}
\begin{remark}
Although this assumption is relatively strong, it is standard in network HAC variance estimation with clustered interference structures. For example, Assumption~4 in \cite{Leung2023} provides a sufficient condition for Assumption~\ref{ass:sparsity}. Similar sparsity conditions are also implied by Assumption~2(ii) and Assumption~5 in \cite{Leung2025}. Since we do not want our inference to rely too heavily on this assumption, we use the L\"owner maximum of the two HAC estimators to obtain a more conservative variance estimate.
\end{remark}
With the sparsity assumption, the theoretical guarantee for these variance estimators is provided in the following theorem.
\begin{theorem}[Variance conservativeness for general estimators under independent Bernoulli randomization]
\label{thm:var_1}
Suppose Assumptions \ref{ass:cluster_assignment}--\ref{ass:interference}, \ref{ass:weight}, \ref{ass:CLT_regularity}(ii)-(vi), \ref{ass:independent_rand}(i), and \ref{ass:sparsity} hold. For each $t\in\{1,2\}$, the component estimator $\hat{\mathbf\Sigma}^{\mathbf K_t}$ satisfies
\[
\hat{\mathbf\Sigma}^{\mathbf K_t}
=
\mathbf\Sigma + \mathbf R_1+o_P(n^{-1}),
\qquad
\hat{\mathbf\Sigma}^{\mathbf K_t}_{\mathrm{marg}}
=
\mathbf\Sigma_{\mathrm{marg}} + \mathbf R_{1,\mathrm{marg}}+o_P(n^{-1}),
\]
where the remainders $\mathbf R_1$ and $\mathbf R_{1,\mathrm{marg}}$ are positive semidefinite and of order $O(n^{-1})$, the order of the asymptotic variance in Theorem~\ref{thm:variance_order}(i). Hence $\hat{\mathbf\Sigma}_1$ and $\hat{\mathbf\Sigma}_{1,\mathrm{marg}}$ are asymptotically conservative.
\end{theorem}
\begin{remark}
In finite samples, $\hat{\mathbf\Sigma}_1$ and $\hat{\mathbf\Sigma}_{\mathrm{marg}}$ may perform well even when the sparsity assumption is not fully credible. For example, in our simulation study in the next section, the sparsity assumption is deliberately violated, yet these variance estimators still provide accurate approximations to the sampling variance under cluster-level independent Bernoulli randomization.
\end{remark}
\begin{remark}
We focus on conservative rather than consistent variance estimators because our analysis is conducted within a design-based causal inference framework. In this framework, population-level average potential outcomes can be precisely identified from the randomization distribution, but the discrepancy between population-level and unit-level average potential outcomes cannot be precisely recovered without additional modeling assumptions. As a result, the full variance associated with unit-level average potential outcomes cannot, in general, be consistently estimated.
\end{remark}

For cluster-agnostic estimators under unit-level independent Bernoulli randomization, the construction is analogous to the above, but the kernel matrix is based on $\mathbf K_3$ defined by
\[
(\mathbf K_3)_{ij,i'j'}
=
1\{\mathcal N_{ij}\cap\mathcal N_{i'j'}\neq\varnothing\}.
\]
$\mathbf K_3$ is the dependency matrix corresponding exclusively to interference-induced dependence. Since cluster-level treatment assignment does not create nonzero covariance among the summands of cluster-agnostic estimators, all dependence in $\boldsymbol\Sigma$ is captured by $\mathbf K_3$ in this case.
Unlike $\mathbf K_1$ and $\mathbf K_2$, $\mathbf K_3$ itself cannot be approximated by a positive semidefinite matrix, so we use its positive-semidefinite part $\mathbf K_3^{+}$ as the kernel to build the HAC variance estimators
\[
\hat{\mathbf \Sigma}_{2}
=
\hat{\mathbf\Sigma}^{\mathbf K_3^+},
\]
and
\[
\hat{\mathbf \Sigma}_{2,\mathrm{marg}}
=
\hat{\mathbf\Sigma}_{\mathrm{marg}}^{\mathbf K_3^+}.
\]
In contrast with the variance estimators under cluster-level Bernoulli randomization in Theorem~\ref{thm:var_1}, where the within-cluster kernel $\mathbf K_1$ provides $O(n^{-1})$ asymptotic conservativeness and the relevant sparsity condition (Assumption~\ref{ass:sparsity}) is used only to ensure that the cross-cluster correction in $\mathbf K_2$ does not blow up the bias of $\hat{\mathbf\Sigma}_1$, the cluster-agnostic kernel $\mathbf K_3^+$ targets an asymptotic variance of order $O(N^{-1})$. To guarantee asymptotic \emph{consistency} of $\hat{\mathbf\Sigma}_2$ and $\hat{\mathbf\Sigma}_{2,\mathrm{marg}}$ at this faster rate, we need the random fluctuation of the kernel-based quadratic form to be of strictly smaller order than the squared target variance. This requires an additional condition that controls the cluster-level dependency complexity of $\mathbf K_3^+$, together with a regime condition on the number of clusters, formalized below.
\begin{assumption}[Network density for cluster-agnostic variance estimation]\label{ass:var_cluster_agnostic}
The kernel $\mathbf K_3^+$ inherits the cluster-level dependency sparsity of $\mathbf K_3$ in the sense that
    \[
    \sum_{\substack{((i,j),(i',j'),(k,l),(k',l'))\in\mathcal P^4:\\(\mathcal N_{ij}^{(cl)}\cup\mathcal N_{i'j'}^{(cl)})\cap(\mathcal N_{kl}^{(cl)}\cup\mathcal N_{k'l'}^{(cl)})\neq\varnothing}}\!\!\bigl|(\mathbf K_3^+)_{ij,i'j'}\bigr|\,\bigl|(\mathbf K_3^+)_{kl,k'l'}\bigr|
    =
    o\!\left(N^2\right).
    \]
\end{assumption}
\begin{remark}
The summation on the left-hand side ranges over all four-unit tuples $((i,j),(i',j'),(k,l),(k',l'))\in\mathcal P^4$ for which the joint cluster-level dependency neighborhoods $\mathcal N_{ij}^{(cl)}\cup\mathcal N_{i'j'}^{(cl)}$ and $\mathcal N_{kl}^{(cl)}\cup\mathcal N_{k'l'}^{(cl)}$ overlap. This is precisely the set of four-unit tuples for which the pair of variance summands at $((i,j),(i',j'))$ is not independent of the pair at $((k,l),(k',l'))$, i.e., the tuples whose covariance contributes to the full variance.

To see what this assumption requires, consider the analogous quantity with $\mathbf K_3^+$ replaced by $\mathbf K_3$ itself. Under Assumption~\ref{ass:CLT_regularity}(ii)--(vi), each pair $((i,j),(i',j'))$ in the support of $\mathbf K_3$ has a joint cluster-level dependency neighborhood of size
\[
\Big|\bigcup_{(k'',l''):\,k''\in\mathcal N_{ij}^{(cl)}\cup\mathcal N_{i'j'}^{(cl)}}\mathcal N_{k''l''}\Big|
=O\!\left(\frac{N}{n}\right),
\]
which bounds the number of partner pairs $((k,l),(k',l'))$ in the support of $\mathbf K_3$ sharing at least one cluster-level neighbor with $((i,j),(i',j'))$. Multiplying by the $O(N)$ choices for the first pair, the count of such four-unit tuples is $O(N^2/n)$, which becomes $o(N^2)$ precisely when $n\to\infty$. Assumption~\ref{ass:var_cluster_agnostic} therefore combines two requirements into a single condition: (a) that $\mathbf K_3^+$ inherits this cluster-level dependency sparsity from $\mathbf K_3$, and (b) that the number of clusters $n$ grows.

The reason we require (b) is that, although the cluster-agnostic point estimator itself attains an $O(N^{-1/2})$ convergence rate that does not depend on $n$ (Theorem~\ref{thm:clt}),
its variance estimator $\hat{\mathbf\Sigma}_2$ does not inherit the cluster-agnostic property and therefore requires $n\to\infty$ for consistency.
Intuitively, the summands in the variance estimator, being products of two conditional mean-zero terms, are themselves no longer conditional mean-zero.
\end{remark}
Their theoretical guarantees are summarized in the following theorem.

\begin{theorem}[Variance conservativeness for cluster-agnostic estimators under independent Bernoulli randomization]
\label{thm:var_2}
Suppose Assumptions \ref{ass:cluster_assignment}--\ref{ass:interference},
\ref{ass:weight}, \ref{ass:weight_cluster_agnostic},
\ref{ass:CLT_regularity}(i), (iii)--(vii), \ref{ass:independent_rand}, and
\ref{ass:var_cluster_agnostic} hold. Then the variance estimators
$\hat{\mathbf \Sigma}_{2}$ and $\hat{\mathbf \Sigma}_{2,\mathrm{marg}}$ satisfy
\[
\hat{\mathbf \Sigma}_{2}
=
\mathbf \Sigma(\mathbf C)
+
\mathbf R_{2}(\mathbf{C})
+
o_P(N^{-1}),
\qquad
\hat{\mathbf \Sigma}_{2,\mathrm{marg}}
=
\mathbf \Sigma_{\mathrm{marg}}(\mathbf{C})
+
\mathbf R_{2,\mathrm{marg}}(\mathbf{C})
+
o_P(N^{-1}),
\]
where the data-dependent remainders $\mathbf R_{2}(\mathbf{C})$ and $\mathbf R_{2,\mathrm{marg}}(\mathbf{C})$ are positive semidefinite for every realization of $\mathbf{C}$ and, uniformly over its support, of the order of the asymptotic variance in Theorem~\ref{thm:variance_order}(ii),
\[
\sup_{\mathbf{c}\in\mathrm{supp}(\mathbf{C})}\|\mathbf R_{2}(\mathbf{c})\|
=O(N^{-1}),
\qquad
\sup_{\mathbf{c}\in\mathrm{supp}(\mathbf{C})}\|\mathbf R_{2,\mathrm{marg}}(\mathbf{c})\|
=O(N^{-1}).
\]
Hence $\hat{\mathbf \Sigma}_{2}$ and $\hat{\mathbf \Sigma}_{2,\mathrm{marg}}$ are asymptotically conservative.
\end{theorem}
\begin{remark}
Because both the target $\mathbf\Sigma(\mathbf C)$ and the remainder $\mathbf R_2(\mathbf C)$ are defined conditionally on $\mathbf C$, the dependency between the summands of the variance estimator is controlled through the conditional independence induced by the unit-level randomization in Assumption~\ref{ass:independent_rand}(ii).
\end{remark}
\subsection{Variance estimators under complete randomization}\label{subsec:bias_correction}
Variance estimation under cluster-level stratified complete randomization is
complicated by two distinct forms of dependence in the summands of
\(\hat{\boldsymbol\tau}\) and \(\hat{\boldsymbol\tau}_{\mathrm{marg}}\).
The first is sparse but potentially strong dependence among units with
overlapping cluster-level interference neighborhoods; this component is
accounted for by the variance estimators \(\hat{\mathbf\Sigma}_{1}\) and
\(\hat{\mathbf\Sigma}_{1,\mathrm{marg}}\) from the previous subsection. The
second is weak but dense dependence induced by complete randomization. Within
each stratum \(\mathcal I_k\), cluster-level assignments are jointly
correlated, so all units in clusters belonging to the same stratum are linked
through the assignment mechanism. Since only \(O(1)\) strata cover the entire
population, this dependence is too pervasive to be represented by a sparse
kernel matrix in the standard HAC framework.


However, this dense component does not invalidate the HAC approach. 
Instead, for the HAC variance estimators $\hat{\mathbf\Sigma}_{1}$ and $\hat{\mathbf\Sigma}_{1,\mathrm{marg}}$, the bias induced by this dense component is a positive-semidefinite matrix that can be expressed in closed form. 
This characterization allows us to establish conservative properties for two types of variance estimators: the standard HAC estimators in Theorem~\ref{thm:var_1}, and bias-corrected estimators that estimate and subtract the additional term.
Concretely, under complete randomization, the biases of $\hat{\mathbf\Sigma}_{1}$ and $\hat{\mathbf\Sigma}_{1,\mathrm{marg}}$ asymptotically decompose as
\[
\mathbf R_3=\mathbf R_1+\mathbf M,
\qquad
\mathbf R_{3,\mathrm{marg}}=\mathbf R_{1,\mathrm{marg}}+\mathbf M_{\mathrm{marg}},
\]
where $\mathbf R_1,\mathbf R_{1,\mathrm{marg}}\succeq\mathbf 0$ are the intrinsic bias matrices already present under independent Bernoulli randomization, as in Theorem~\ref{thm:var_1}, and $\mathbf M,\mathbf M_{\mathrm{marg}}\succeq\mathbf 0$ are the additional terms generated by the dense dependence that complete randomization induces across disjoint cluster-level neighborhoods.
The positive-semidefiniteness of these additional terms follows from the \emph{negative association} of the treatment variables generated by complete randomization, established by \citet{JoagDev1983}.

To express the additional terms, write
\[
\mathbf V_{ij}=\frac{g_i}{N_i}\mathbf B_{(*,ij)}\circ(Y_{ij}\mathbf 1-\boldsymbol\tau),
\qquad
\mathbf V_{ij}^{\mathrm{marg}}=\frac{g_i}{N_i}\mathbf B_{\mathrm{marg},(*,ij)}\circ(Y_{ij}\mathbf 1-\boldsymbol\tau_{\mathrm{marg}})
\]
for the per-unit summands underlying $\hat{\mathbf\Sigma}_{1}$ and
$\hat{\mathbf\Sigma}_{1,\mathrm{marg}}$, respectively. For each stratum $k$, let
\[
T_{ij,k}=\sum_{i''\in\mathcal N_{ij}^{(cl)}\cap\mathcal I_k}C_{i''}
\]
be the number of treated clusters in stratum $k$ within the cluster-level
neighborhood of unit $(i,j)$, and define
\[
\mathbf s_k=\sum_{(i,j)\in\mathcal P}\mathrm{Cov}\bigl(\mathbf V_{ij},\,T_{ij,k}\bigr),
\qquad
\mathbf s_k^{\mathrm{marg}}=\sum_{(i,j)\in\mathcal P}\mathrm{Cov}\bigl(\mathbf V_{ij}^{\mathrm{marg}},\,T_{ij,k}\bigr).
\]
Then, as derived in the proof of Theorem~\ref{thm:var_3}, the additional terms
admit the closed-form expressions
\[
\mathbf M=\sum_{k:\,p_k\in[p,1-p]}
\frac{\mathbf s_k\mathbf s_k^\top}{|\mathcal I_k|\,p_k(1-p_k)},
\qquad
\mathbf M_{\mathrm{marg}}=\sum_{k:\,p_k\in[p,1-p]}
\frac{\mathbf s_k^{\mathrm{marg}}(\mathbf s_k^{\mathrm{marg}})^\top}
{|\mathcal I_k|\,p_k(1-p_k)},
\]
where $p_k=I_k/|\mathcal I_k|$ is the known treated fraction in stratum $k$.

We next construct the bias-corrected estimators. Since
\[
\mathbb E\,T_{ij,k}=m_{ij,k}p_k,
\qquad
m_{ij,k}:=|\mathcal N_{ij}^{(cl)}\cap\mathcal I_k|,
\]
we have
\[
\mathrm{Cov}(\mathbf V_{ij},T_{ij,k})
=
\mathbb E\!\left[\mathbf V_{ij}(T_{ij,k}-m_{ij,k}p_k)\right],
\]
and analogously for \(\mathbf V_{ij}^{\mathrm{marg}}\). Replacing the infeasible
summands by their observable counterparts,
\[
\hat{\mathbf V}_{ij}
=
\frac{g_i}{N_i}\mathbf B_{(*,ij)}\circ(Y_{ij}\mathbf 1-\hat{\boldsymbol\tau}),
\qquad
\hat{\mathbf V}_{ij}^{\mathrm{marg}}
=
\frac{g_i}{N_i}\mathbf B_{\mathrm{marg},(*,ij)}
\circ(Y_{ij}\mathbf 1-\hat{\boldsymbol\tau}_{\mathrm{marg}}),
\]
we estimate
\[
\hat{\mathbf s}_k
:=
\sum_{(i,j)\in\mathcal P}
\bigl(T_{ij,k}-m_{ij,k}p_k\bigr)\hat{\mathbf V}_{ij},
\qquad
\hat{\mathbf s}_k^{\mathrm{marg}}
:=
\sum_{(i,j)\in\mathcal P}
\bigl(T_{ij,k}-m_{ij,k}p_k\bigr)\hat{\mathbf V}_{ij}^{\mathrm{marg}}.
\]
Substituting these plug-in estimates into the closed-form expressions for
\(\mathbf M\) and \(\mathbf M_{\mathrm{marg}}\) gives
\[
\hat{\mathbf M}
:=
\sum_{k:\,p_k\in[p,1-p]}
\frac{\hat{\mathbf s}_k\hat{\mathbf s}_k^\top}
{|\mathcal I_k|\,p_k(1-p_k)},
\qquad
\hat{\mathbf M}_{\mathrm{marg}}
:=
\sum_{k:\,p_k\in[p,1-p]}
\frac{\hat{\mathbf s}_k^{\mathrm{marg}}
(\hat{\mathbf s}_k^{\mathrm{marg}})^\top}
{|\mathcal I_k|\,p_k(1-p_k)}.
\]
The resulting bias-corrected variance estimators are
\[
\hat{\boldsymbol\Sigma}_{3}
:=
\hat{\mathbf\Sigma}_1-\hat{\mathbf M},
\qquad
\hat{\boldsymbol\Sigma}_{3,\mathrm{marg}}
:=
\hat{\mathbf\Sigma}_{1,\mathrm{marg}}-\hat{\mathbf M}_{\mathrm{marg}}.
\]
Because \(\hat{\mathbf M},\hat{\mathbf M}_{\mathrm{marg}}\succeq \mathbf 0\),
the correction removes the additional variance inflation due to complete
randomization while retaining the intrinsic conservative bias
\(\mathbf R_1\) and \(\mathbf R_{1,\mathrm{marg}}\).

The following theorem formalizes these conservativeness properties for both the
HAC estimator and the bias-corrected estimator.
\begin{theorem}[Variance conservativeness for general estimators under complete randomization]
\label{thm:var_3}
Suppose Assumptions \ref{ass:cluster_assignment}--\ref{ass:interference},
\ref{ass:weight}, \ref{ass:CLT_regularity}(ii)-(vi),
\ref{ass:complete_rand}(i), and \ref{ass:sparsity} hold.
\begin{enumerate}
\item[(i)] \emph{(HAC estimator.)} For each $t\in\{1,2\}$, the component estimator $\hat{\mathbf\Sigma}^{\mathbf K_t}$ satisfies
\[
\hat{\mathbf\Sigma}^{\mathbf K_t}
=
\mathbf\Sigma + \mathbf R_3+o_P(n^{-1}),
\qquad
\hat{\mathbf\Sigma}^{\mathbf K_t}_{\mathrm{marg}}
=
\mathbf\Sigma_{\mathrm{marg}} + \mathbf R_{3,\mathrm{marg}}+o_P(n^{-1}),
\]
where the bias matrices decompose as
\[
\mathbf R_3=\mathbf R_1+\mathbf M,
\qquad
\mathbf R_{3,\mathrm{marg}}=\mathbf R_{1,\mathrm{marg}}+\mathbf M_{\mathrm{marg}},
\]
with $\mathbf R_1,\mathbf M\succeq\mathbf 0$ (and likewise $\mathbf R_{1,\mathrm{marg}},\mathbf M_{\mathrm{marg}}\succeq\mathbf 0$) both of order $O(n^{-1})$, the order of the asymptotic variance in Theorem~\ref{thm:variance_order}(i). Here $\mathbf R_1$ is the intrinsic bias matrix of Theorem~\ref{thm:var_1} and $\mathbf M$ is the additional dense-dependence term. Consequently $\hat{\mathbf\Sigma}_1$ and $\hat{\mathbf\Sigma}_{1,\mathrm{marg}}$ are asymptotically conservative.
\item[(ii)] \emph{(Bias-corrected estimator.)} The bias-corrected estimators $\hat{\boldsymbol\Sigma}_3=\hat{\mathbf\Sigma}_1-\hat{\mathbf M}$ and $\hat{\boldsymbol\Sigma}_{3,\mathrm{marg}}=\hat{\mathbf\Sigma}_{1,\mathrm{marg}}-\hat{\mathbf M}_{\mathrm{marg}}$ satisfy
\[
\hat{\boldsymbol\Sigma}_3
=
\mathbf\Sigma + \mathbf R_1+o_P(n^{-1}),
\qquad
\hat{\boldsymbol\Sigma}_{3,\mathrm{marg}}
=
\mathbf\Sigma_{\mathrm{marg}} + \mathbf R_{1,\mathrm{marg}}+o_P(n^{-1}).
\]
That is, $\hat{\boldsymbol\Sigma}_3$ and $\hat{\boldsymbol\Sigma}_{3,\mathrm{marg}}$ remove the dense-dependence terms $\mathbf M$ and $\mathbf M_{\mathrm{marg}}$ and remain asymptotically conservative, with the same intrinsic bias matrices $\mathbf R_1$ and $\mathbf R_{1,\mathrm{marg}}$ as under independent Bernoulli randomization (Theorem~\ref{thm:var_1}).
\end{enumerate}
\end{theorem}
\begin{remark}
  The Neyman--Jackknife technique based on Gibbs rerandomization recently proposed by \cite{Park2026} may offer an alternative route to constructing general variance estimators under complete randomization. We do not pursue this approach here and leave its formal development to future work.
  \end{remark}

\section{Simulation study}

We study the finite-sample performance of the proposed estimators through Monte Carlo simulations based on the spatial network design of \citet{Leung2025}, extended to accommodate two cluster-level randomization schemes and varying levels of cross-cluster interference.

\subsection{Data-generating process}

\paragraph{Cluster structure and interference network.}
We independently draw $N$ unit locations uniformly from the square
$[-\sqrt{N\alpha},\sqrt{N\alpha}]^2$, with $\alpha=0.6$, and partition the units into
$n=\lfloor N^{2/3}\rfloor$ spatially contiguous clusters using the $k$-medoids algorithm of
\cite{Leung2025} based on the Euclidean distance matrix. Cluster-level quantities are aggregated using weights $g_i=N_i/N$, so that each individual unit receives equal weight.

We construct the interference network as the union $G=G_{\mathrm{geo}}\cup G_{\mathrm{ER}}$ of two independently generated components. The first component is a \emph{geometric graph} $G_{\mathrm{geo}}$, in which two units are connected if and only if their Euclidean distance is at most $r=1.5$. The second component is an \emph{Erd\H{o}s--R\'enyi graph} $G_{\mathrm{ER}}$, in which each pair of units is connected independently with probability
$p_{\mathrm{extra}}=\frac{\bar d}{N-1}\rho$, where $\bar d$ is the average degree of $G_{\mathrm{geo}}$ and $\rho\in\{0,0.5,1.0\}$ controls the relative density of the Erd\H{o}s--R\'enyi component. For each unit, the interference neighborhood $\mathcal{N}_{ij}$ is defined as the set of its neighbors in the union graph $G$, and $A$ denotes the adjacency matrix of $G$. By construction, all of these network designs exhibit cross-cluster interference, with larger values of \(\rho\) generating denser cross-cluster connections and hence stronger cross-cluster interference.

For each simulation setup, the cluster assignments and the union interference network $G$ are generated once and then held fixed across simulation replications. When $\rho=0.5$ or $\rho=1.0$, the Erd\H{o}s--R\'enyi component introduces substantial cross-cluster interference, thereby violating Assumption~\ref{ass:sparsity}. The other regularity conditions in Assumption~\ref{ass:CLT_regularity}, however, continue to hold for the resulting networks and cluster structures.

\paragraph{Cluster-level treatment assignment.}
We consider two different randomization schemes of cluster-level treatment assignment.
\begin{enumerate}
  \item[(i)] \textbf{Bernoulli cluster randomization.} Each cluster is independently assigned $C_i\overset{\mathrm{iid}}{\sim}\mathrm{Bernoulli}(0.7)$, so that the expected proportion of treated clusters is $70\%$.

  \item[(ii)] \textbf{Complete cluster randomization.} Exactly $\lfloor 0.7n\rfloor$ clusters are sampled for treatment uniformly at random without replacement. Unlike (i), the fixed treatment count induces negative correlation among cluster-level assignments within each experiment, a dependence structure absent under Bernoulli randomization.
\end{enumerate}

\paragraph{Individual-level treatment assignment.}
In both schemes, the individual treatment regime is
\[
  \phi_1:\quad W_{ij}\mid C_i=1 \overset{\mathrm{iid}}{\sim} \mathrm{Bernoulli}(0.5),
  \qquad
  \phi_0:\quad W_{ij}\mid C_i=0 \equiv 0,
\]
that is, units in treated clusters each receive treatment independently with probability $1/2$, while units in control clusters receive no treatment.

\paragraph{Outcome model.}
Unit-level treatment effects are governed by two sets of random coefficients,
$\beta_{ij}\overset{\mathrm{iid}}{\sim}\mathcal{N}(2,1)$ and
$\gamma_{ij}\overset{\mathrm{iid}}{\sim}\mathcal{N}(1,1)$, which are drawn once and then held fixed across simulation replications. These coefficients allow both the spillover effects induced by treated neighbors and the interactions between own treatment and neighborhood exposure to vary across units. We also generate the outcome error vector as
$\boldsymbol\epsilon=(\mathbf I+\mathbf D^{-1}\mathbf A)\boldsymbol\eta$, where
$\boldsymbol\eta\overset{\mathrm{iid}}{\sim}\mathcal{N}(\mathbf 0,\mathbf I)$ is redrawn in each replication and
$\mathbf D=\mathrm{diag}(|\mathcal N_{ij}|)$ is the degree-normalizing matrix. This construction induces network autocorrelation in $\boldsymbol\epsilon$. Observed outcomes are given by
\[
  Y_{ij} = -1
    + \sum_{(i',j')\in\mathcal{N}_{ij}} \beta_{i'j'}\,W_{i'j'}
    + W_{ij}\!\sum_{(i',j')\in\mathcal{N}_{ij}} \gamma_{i'j'}\,W_{i'j'}
    + \epsilon_{ij},
\]
so the outcome model incorporates heterogeneous spillover effects and treatment--exposure interactions. Although the random error term goes beyond the standard design-based causal inference framework, we include it to assess the performance of the proposed procedures in a more general data-generating environment. The counterfactual marginal average potential outcomes under $\phi_c$ are then computed in closed form as
\[
Y_{ij}(\phi_c)=\mathbb{E}_{\mathbf W\sim P_{\phi_c}}\left[-1
    + \sum_{(i',j')\in\mathcal{N}_{ij}} \beta_{i'j'}\,W_{i'j'}
    + W_{ij}\!\sum_{(i',j')\in\mathcal{N}_{ij}} \gamma_{i'j'}\,W_{i'j'}\right].
\]

\paragraph{Estimand.}
The target is the \emph{population-level overall causal effect}
\[
  \mathrm{CE}^{O}(\phi_1,\phi_0) = \bar{Y}(\phi_1) - \bar{Y}(\phi_0),
\]
aggregated with cluster weights $g_i = N_i/N$ so that each individual receives equal weight.

\paragraph{Experimental design.}
Each scheme is run for $N\in\{500,\,1000,\,2000\}$ and $\rho\in\{0,\,0.5,\,1.0\}$, giving $3\times 3=9$ parameter combinations per scheme, with $2{,}000$ Monte Carlo replications each (random seed 123).

\subsection{Estimators}

We compare four estimators of $\mathrm{CE}^O(\phi_1,\phi_0)$. Each is constructed by first estimating the vector of marginal population average potential outcomes with H\'ajek LW estimator
\[
\hat{\boldsymbol\tau}_{\mathrm{marg}}
=
\bigl(\hat\tau(\phi_1),\,\hat\tau(\phi_0)\bigr)^\top
=
\left(\sum_{(i,j)}\frac{g_i}{N_i}(\mathbf B_{\mathrm{marg}})_{(*,ij)}\right)^{\!\circ -1}
\circ
\sum_{(i,j)}\frac{g_i}{N_i}Y_{ij}(\mathbf B_{\mathrm{marg}})_{(*,ij)},
\]
and then forming
\[
\widehat{\mathrm{CE}}^O(\phi_1,\phi_0)
=
\mathbf c^\top\hat{\boldsymbol\tau}_{\mathrm{marg}},
\quad
\mathbf c=(1,-1)^\top.
\]
Inference is based on the scalar variance estimate $\mathbf c^\top\hat{\mathbf\Sigma}_{\mathrm{marg}}\,\mathbf c$, where $\hat{\mathbf\Sigma}_{\mathrm{marg}}$ is a $2\times 2$ conservative estimator of $\mathrm{Var}(\hat{\boldsymbol\tau}_{\mathrm{marg}})$. The four estimators differ solely in their choice of identification weight matrix $\mathbf B_{\mathrm{marg}}$; the corresponding choice of $\hat{\mathbf\Sigma}_{\mathrm{marg}}$ is described for each estimator below.

\paragraph{Difference-in-means estimator.}
This serves as a reference baseline that does not account for cross-cluster interference. The identification weights depend only on the cluster-level assignment of unit $(i,j)$'s own cluster:
\[
\beta_{ij}(\phi_1)=\frac{\mathbf 1(C_i=1)}{P(C_i=1)},
\qquad
\beta_{ij}(\phi_0)=\frac{\mathbf 1(C_i=0)}{P(C_i=0)}.
\]
Under the cluster weights \(g_i=N_i/N\), the resulting H\'ajek estimator reduces to the difference in mean outcomes between treated and control clusters. It is biased under cross-cluster interference. For variance estimation, we use $\hat{\mathbf\Sigma}^{\mathbf K_1}_{\mathrm{marg}}$, the HAC estimator with the within-cluster interference kernel $\mathbf K_1$; this is equivalent to the standard cluster-robust HAC variance estimator.

\paragraph{Inverse-probability-of-treatment weighted (IPTW) estimator (Proposition~\ref{prop:examples_of_weights}(iv)).}
Proposed by \citet{Leung2025}, this is the H\'ajek LW estimator with identification weights
\[
\beta_{ij}(\phi_c)
=
\frac{\mathbf 1(\mathbf C_{\mathcal N_{ij}}=c\mathbf{1})}{P(\mathbf C_{\mathcal N_{ij}}=c\mathbf{1})}.
\]
Under both cluster-level independent Bernoulli randomization and complete randomization,
Equation~\eqref{eq:transportability} holds, so the identification weights yield an
asymptotically unbiased estimator. Under complete randomization, this condition is
satisfied because the maximum number of neighboring clusters is smaller than both the
number of treated clusters and the number of control clusters. The same argument also
establishes asymptotic unbiasedness for the next two estimators. The estimator is
asymptotically normal at the \(O(n^{-1/2})\) rate. For variance estimation, we use
\(
\hat{\mathbf\Sigma}_{1,\mathrm{marg}}
=
\hat{\mathbf\Sigma}^{\mathbf K_1}_{\mathrm{marg}}
\vee
\hat{\mathbf\Sigma}^{\mathbf K_2}_{\mathrm{marg}},
\)
the L\"owner maximum of the within-cluster and cross-cluster HAC estimators
(Theorem~\ref{thm:var_1}), consistent with the variance estimator used in
\citet{Leung2025}.

\paragraph{Marginal Radon--Nikodym derivative (MRN) estimator (Proposition~\ref{prop:examples_of_weights}(i)).}
This is the H\'ajek LW estimator with identification weights
\[
\beta_{ij}(\phi_c) = \alpha_{ij}(\phi_c),
\]
where $\alpha_{ij}(\phi_c)$ is the marginal Radon--Nikodym derivative defined in Theorem~\ref{thm:all_weights}. This estimator is asymptotically
unbiased and attains asymptotic normality at the \(O(n^{-1/2})\) rate as well. Moreover, Theorem~\ref{thm:all_weights} shows that $\alpha_{ij}(\phi_c)$ has smaller variance than the IPT weight for each unit $(i,j)$. Therefore, using these marginal weights may yield an LW estimator with lower variance. Under Bernoulli cluster randomization, the variance estimator is $\hat{\mathbf\Sigma}_{1,\mathrm{marg}}$ (Theorem~\ref{thm:var_1}). 
Under complete cluster randomization, we consider both variance estimators of Section~\ref{subsec:bias_correction}: the uncorrected HAC estimator $\hat{\mathbf\Sigma}_{1,\mathrm{marg}}$ and the bias-corrected estimator $\hat{\boldsymbol\Sigma}_{3,\mathrm{marg}}=\hat{\mathbf\Sigma}_{1,\mathrm{marg}}-\hat{\mathbf M}_{\mathrm{marg}}$. 
In Table~\ref{tab:sim_hypergeom}, we distinguish the average standard errors and CI coverages computed from these two variance estimators using the labels HAC and BC, respectively.

\paragraph{Complete Radon--Nikodym derivative (CRN) estimator (Proposition~\ref{prop:examples_of_weights}(ii)).}
This is the H\'ajek LW estimator with identification weights
\[
\beta_{ij}(\phi_c) = \alpha_{ij}^{(\mathrm{comp})}(\phi_c).
\]
As a cluster-agnostic estimator, it attains the faster $O(N^{-1/2})$ rate by Theorem~\ref{thm:clt}. However, since the weights depend jointly on unit-level treatments $\mathbf{W}_{\mathcal{N}_{ij}}$ and cluster-level assignments $\mathbf{C}_{\mathcal{N}_{ij}}$, they have higher variance, particularly when cross-cluster interference is strong. The variance estimator is $\hat{\mathbf\Sigma}_{2,\mathrm{marg}}$ (Theorem~\ref{thm:var_2}), applied under both Bernoulli and complete cluster randomization.

\subsection{Results}

Tables~\ref{tab:sim_binom} and~\ref{tab:sim_hypergeom} report results for $\widehat{\mathrm{CE}}^O(\phi_1,\phi_0)$ under Bernoulli and complete cluster randomization, respectively (2{,}000 Monte Carlo replications). 
Within each estimator block, RMSE is the root mean squared error; SE is the average estimated standard error from the variance estimator; SE$^*$ is the oracle standard error of the estimator, obtained as the standard deviation across simulation draws; CI~SE and CI~SE$^*$ are empirical coverage rates of the nominal 95\% confidence interval using SE and SE$^*$ as the standard error, respectively. 
For the MRN estimator under complete randomization (Table~\ref{tab:sim_hypergeom}), the SE and CI rows are further split into the HAC and BC variants introduced in Section~\ref{subsec:bias_correction}.

\begin{table}[H]
\centering
\caption{Simulation results under \textbf{Bernoulli cluster randomization} ($n=\lfloor N^{2/3}\rfloor$, $C_i\overset{\mathrm{iid}}{\sim}\mathrm{Bernoulli}(0.7)$). $\phi_1=\mathrm{Bernoulli}(0.5)$ vs.\ $\phi_0\equiv 0$. Overall effect $\mathrm{CE}^O(\phi_1,\phi_0)$.}
\label{tab:sim_binom}
\small
\begin{tabular}{l rrr rrr rrr}
\toprule
&
  \multicolumn{3}{c}{$N=500$ ($n=63$)} &
  \multicolumn{3}{c}{$N=1000$ ($n=100$)} &
  \multicolumn{3}{c}{$N=2000$ ($n=159$)} \\
\cmidrule(lr){2-4}\cmidrule(lr){5-7}\cmidrule(lr){8-10}
$\rho$ & 0.0 & 0.5 & 1.0 & 0.0 & 0.5 & 1.0 & 0.0 & 0.5 & 1.0 \\
\midrule
Difference-in-means estimator\\
RMSE      & 0.612 & 1.924 & 3.551 & 0.430 & 1.888 & 3.356 & 0.380 & 1.895 & 3.498 \\
CI SE     & 0.880 & 0.034 & 0.000 & 0.872 & 0.000 & 0.000 & 0.774 & 0.000 & 0.000 \\
CI SE$^*$ & 0.858 & 0.033 & 0.000 & 0.836 & 0.000 & 0.000 & 0.715 & 0.000 & 0.000 \\
SE        & 0.482 & 0.486 & 0.525 & 0.329 & 0.344 & 0.367 & 0.243 & 0.255 & 0.269 \\
SE$^*$    & 0.454 & 0.483 & 0.536 & 0.307 & 0.345 & 0.386 & 0.219 & 0.246 & 0.272 \\
\midrule
IPTW estimator \citep{Leung2025}\\
RMSE      & 0.512 & 0.736 & 1.115 & 0.328 & 0.547 & 0.836 & 0.235 & 0.412 & 0.678 \\
CI SE     & 0.970 & 0.948 & 0.898 & 0.964 & 0.951 & 0.920 & 0.972 & 0.957 & 0.924 \\
CI SE$^*$ & 0.955 & 0.956 & 0.942 & 0.947 & 0.950 & 0.958 & 0.948 & 0.951 & 0.953 \\
SE        & 0.558 & 0.712 & 0.934 & 0.355 & 0.548 & 0.729 & 0.261 & 0.425 & 0.620 \\
SE$^*$    & 0.512 & 0.735 & 1.109 & 0.328 & 0.547 & 0.833 & 0.235 & 0.411 & 0.674 \\
\midrule
MRN estimator\\
RMSE      & 0.456 & 0.613 & 0.848 & 0.306 & 0.465 & 0.625 & 0.228 & 0.352 & 0.488 \\
CI SE     & 0.979 & 0.965 & 0.939 & 0.962 & 0.965 & 0.944 & 0.969 & 0.969 & 0.968 \\
CI SE$^*$ & 0.958 & 0.957 & 0.950 & 0.945 & 0.954 & 0.944 & 0.951 & 0.952 & 0.954 \\
SE        & 0.521 & 0.636 & 0.805 & 0.333 & 0.491 & 0.620 & 0.249 & 0.385 & 0.532 \\
SE$^*$    & 0.456 & 0.613 & 0.846 & 0.306 & 0.464 & 0.623 & 0.228 & 0.352 & 0.487 \\
\midrule
CRN estimator\\
RMSE      & 0.621 & 1.088 & 1.697 & 0.469 & 0.975 & 1.483 & 0.419 & 0.953 & 1.537 \\
CI SE     & 0.919 & 0.842 & 0.750 & 0.906 & 0.808 & 0.658 & 0.883 & 0.617 & 0.406 \\
CI SE$^*$ & 0.852 & 0.629 & 0.436 & 0.827 & 0.464 & 0.278 & 0.732 & 0.255 & 0.082 \\
SE        & 0.544 & 0.769 & 1.028 & 0.396 & 0.638 & 0.824 & 0.314 & 0.509 & 0.695 \\
SE$^*$    & 0.463 & 0.573 & 0.729 & 0.328 & 0.429 & 0.542 & 0.250 & 0.338 & 0.442 \\
\bottomrule
\end{tabular}
\end{table}

\begin{table}[H]
\centering
\caption{Simulation results under \textbf{complete cluster randomization} ($n=\lfloor N^{2/3}\rfloor$, exactly $\lfloor 0.7n\rfloor$ treated clusters). $\phi_1=\mathrm{Bernoulli}(0.5)$ vs.\ $\phi_0\equiv 0$. Overall effect $\mathrm{CE}^O(\phi_1,\phi_0)$.}
\label{tab:sim_hypergeom}
\small
\begin{tabular}{l rrr rrr rrr}
\toprule
&
  \multicolumn{3}{c}{$N=500$ ($n=63$)} &
  \multicolumn{3}{c}{$N=1000$ ($n=100$)} &
  \multicolumn{3}{c}{$N=2000$ ($n=159$)} \\
\cmidrule(lr){2-4}\cmidrule(lr){5-7}\cmidrule(lr){8-10}
$\rho$ & 0.0 & 0.5 & 1.0 & 0.0 & 0.5 & 1.0 & 0.0 & 0.5 & 1.0 \\
\midrule
Difference-in-means estimator\\
RMSE      & 0.607 & 1.922 & 3.547 & 0.441 & 1.898 & 3.370 & 0.370 & 1.890 & 3.488 \\
CI SE     & 0.881 & 0.028 & 0.000 & 0.850 & 0.000 & 0.000 & 0.790 & 0.000 & 0.000 \\
CI SE$^*$ & 0.842 & 0.026 & 0.000 & 0.796 & 0.000 & 0.000 & 0.706 & 0.000 & 0.000 \\
SE        & 0.483 & 0.488 & 0.526 & 0.328 & 0.343 & 0.366 & 0.243 & 0.255 & 0.269 \\
SE$^*$    & 0.444 & 0.473 & 0.535 & 0.298 & 0.341 & 0.380 & 0.214 & 0.246 & 0.269 \\
\midrule
IPTW estimator \citep{Leung2025}\\
RMSE      & 0.501 & 0.742 & 1.148 & 0.321 & 0.559 & 0.825 & 0.231 & 0.402 & 0.638 \\
CI SE     & 0.979 & 0.941 & 0.880 & 0.972 & 0.947 & 0.914 & 0.970 & 0.958 & 0.932 \\
CI SE$^*$ & 0.950 & 0.948 & 0.946 & 0.952 & 0.949 & 0.944 & 0.950 & 0.951 & 0.947 \\
SE        & 0.562 & 0.711 & 0.923 & 0.353 & 0.546 & 0.720 & 0.260 & 0.423 & 0.604 \\
SE$^*$    & 0.501 & 0.742 & 1.148 & 0.320 & 0.559 & 0.825 & 0.231 & 0.402 & 0.638 \\
\midrule
MRN estimator\\
RMSE       & 0.425 & 0.541 & 0.720 & 0.272 & 0.398 & 0.511 & 0.197 & 0.285 & 0.371 \\
CI (HAC)   & 0.988 & 0.979 & 0.975 & 0.985 & 0.987 & 0.986 & 0.988 & 0.992 & 0.997 \\
CI (BC)    & 0.977 & 0.959 & 0.934 & 0.976 & 0.973 & 0.954 & 0.978 & 0.980 & 0.979 \\
CI SE$^*$  & 0.946 & 0.947 & 0.956 & 0.955 & 0.954 & 0.948 & 0.952 & 0.950 & 0.953 \\
SE (HAC)   & 0.525 & 0.645 & 0.812 & 0.332 & 0.492 & 0.624 & 0.249 & 0.386 & 0.537 \\
SE (BC)    & 0.497 & 0.565 & 0.659 & 0.311 & 0.431 & 0.517 & 0.232 & 0.336 & 0.451 \\
SE$^*$     & 0.425 & 0.541 & 0.720 & 0.272 & 0.398 & 0.510 & 0.197 & 0.286 & 0.371 \\
\midrule
CRN estimator\\
RMSE      & 0.622 & 1.067 & 1.676 & 0.475 & 0.977 & 1.485 & 0.403 & 0.938 & 1.515 \\
CI SE     & 0.927 & 0.861 & 0.756 & 0.907 & 0.817 & 0.662 & 0.885 & 0.594 & 0.341 \\
CI SE$^*$ & 0.847 & 0.597 & 0.354 & 0.797 & 0.428 & 0.203 & 0.711 & 0.129 & 0.018 \\
SE        & 0.548 & 0.775 & 1.023 & 0.396 & 0.634 & 0.814 & 0.310 & 0.491 & 0.679 \\
SE$^*$    & 0.454 & 0.536 & 0.663 & 0.318 & 0.411 & 0.501 & 0.231 & 0.290 & 0.366 \\
\bottomrule
\end{tabular}
\end{table}

In all settings, the MRN estimator attains a lower RMSE than both
the difference-in-means estimator and the IPTW estimator of \citet{Leung2025}. This
advantage is especially pronounced under strong cross-cluster interference and complete
randomization. For example, under independent Bernoulli randomization with
\(\rho=1.0\) and \(N=2000\), the RMSE of the MRN estimator is
\(0.488\), which is 28\% lower than that of the IPTW estimator (\(0.678\)) and 86\%
lower than that of the difference-in-means estimator (\(3.498\)).
Moreover, the MRN estimator appears to exploit the cluster-level
complete randomization design to further reduce RMSE, whereas the other estimators show
little improvement from this design. Under complete randomization with \(\rho=1.0\) and
\(N=2000\), its RMSE decreases to \(0.371\), which is 42\% lower than that of the IPTW
estimator (\(0.638\)) and 89\% lower than that of the difference-in-means estimator
(\(3.488\)).

Turning to inference, the HAC variance estimator $\hat{\mathbf\Sigma}_{1,\mathrm{marg}}$ for the MRN estimator is conservative across all settings, with coverage at or above the nominal \(0.95\).
Under complete randomization (Table~\ref{tab:sim_hypergeom}), it is somewhat over-conservative relative to Bernoulli randomization, as predicted by Theorem~\ref{thm:var_3}(i): the HAC bias matrix $\mathbf R_{3,\mathrm{marg}}=\mathbf R_{1,\mathrm{marg}}+\mathbf M_{\mathrm{marg}}$ carries the extra dense-dependence term $\mathbf M_{\mathrm{marg}}\succeq\mathbf 0$, which inflates SE~(HAC).
The bias-corrected estimator $\hat{\boldsymbol\Sigma}_{3,\mathrm{marg}}$ of Theorem~\ref{thm:var_3}(ii) subtracts an estimate of this term, so SE~(BC) is uniformly smaller than SE~(HAC) and closer to the oracle SE$^*$, removing much of the excess conservativeness while keeping CI~(BC) valid.
In the most demanding setting (\(\rho=1.0\), \(N=2000\)), for instance, the correction shrinks the estimated standard error from about \(45\%\) above the oracle to about \(22\%\) above it, bringing CI coverage down from \(0.997\) to \(0.979\).

The CRN estimator exhibits substantially larger RMSE in both tables,
despite its faster asymptotic convergence rate. This implies that, in finite samples, per-unit weight
variance can matter more for estimation accuracy than the convergence rate governed by the degree of
cross-unit dependence. In the H\'ajek estimator, the high variance of the weights also causes the normalizing
denominator to deviate substantially from its expectation, introducing large bias.
Such bias not only inflates RMSE but also leads to undercoverage of confidence intervals.

\section{Central limit theorem under complete randomization}\label{sec:technical}

This section provides details on the technical developments made in this paper to establish central limit theorem under complete randomization. The technique applies to general statistics that can be written as sums of local components, where each component depends on the treatment variables within the interference neighborhood of a given unit. 
Establishing asymptotic normality for this class of statistics is technically challenging because the summands exhibit both dense dependence induced by complete randomization, which has a clean permutation structure, and local dependence induced by interference, which lacks such a structure; the coexistence of these two forms of dependence makes existing central limit theorem techniques tailored to either type alone inapplicable. 

Under our setup, for each \((i,j)\in\mathcal P\), let \(X_{ij}=X(\mathbf{W}_{\mathcal{N}_{ij}},\mathbf{C}_{\mathcal{N}_{ij}})\) denote such a local component, which is a known function of \(\mathbf{W}_{\mathcal{N}_{ij}}\) and \(\mathbf{C}_{\mathcal{N}_{ij}}\). We study the asymptotic behavior of the standardized sum
\[
S:=\sigma^{-1}\sum_{(i,j)\in\mathcal P} X_{ij},
\]
where
\[
\sigma^2:=\operatorname{Var}\left(\sum_{(i,j)\in\mathcal P} X_{ij}\right),
\]
so that \(S\) has unit variance. 

Our technique is based on Stein's method. To show that $S$
is approximately standard normal, the method controls a discrepancy of the form
\(
\mathbb{E}\bigl[f'(S)-S f(S)\bigr],
\)
and the crux of the argument is to bound a covariance term between each individual summand \(X_{ij}\) and a function of the full sum \(S\). When the dependence across summands is local, one can apply the standard decomposition of \cite{Ross2011}. Here, one approximates \(S\) by a partial sum \(S_{ij}\) that is independent of \(X_{ij}\), so that the covariance of \(X_{ij}\) with the corresponding function of \(S_{ij}\) vanishes, leaving only its covariance with the small approximation error.
Under complete randomization, however, the dependence between $X_{ij}$ and $S$ is not confined to a small subset of summands, so no partial sum $S_{ij}$ can be simultaneously independent of $X_{ij}$ and a good approximation to $S$, both of which this decomposition requires.

The key innovation of our approach is to construct, for \emph{each} index $(u,v)\in\mathcal P$, a separate coupling array $(\tilde X_{ij}^{(uv)})_{(i,j)\in\mathcal P}$, guided by three heuristic criteria: it should have the same distribution as the original array $(X_{ij})_{(i,j)\in\mathcal P}$, be only weakly dependent on the single term $X_{uv}$, and differ from the original array at only a small fraction of units. 
Then, we can bound the covariance term between $X_{uv}$ and the function of $S$ with a similar approximation to $S$ as above, now with the coupling sum
\[
S^{(uv)}:=\sigma^{-1}\sum_{(i,j)\in\mathcal P}\tilde X_{ij}^{(uv)}
\]
playing the role of the auxiliary partial sum $S_{uv}$. In this way, although no partial sum $S_{uv}$ can be independent of $X_{uv}$ under complete randomization, the coupling sum $S^{(uv)}$ can still be made only weakly dependent on $X_{uv}$ while remaining a good approximation to $S$. We will show later that this dependence is often sufficiently weak to control the covariance term and yield asymptotic normality.
At a high level, the new technique controls the gap between the distributions of the sum statistic with and without conditioning on selected components, in the spirit of \cite{Merlevde2009}; the new decomposition of the Stein discrepancy is also related to the coupling framework of \cite{Chen2010} and \cite{Fang2015}, except that, to accommodate the dependence induced by complete randomization, we construct a distinct coupling for each index $(u,v)$.

The coupling array $(\tilde X_{ij}^{(uv)})_{(i,j)\in\mathcal P}$ admits different constructions depending on the design; in what follows, we develop one for the complete randomization design in Assumption~\ref{ass:complete_rand}(i), and use it to prove the central limit theorem under Assumption~\ref{ass:complete_rand}(i). The construction proceeds at the level of the underlying treatment vectors: we construct the couplings $\tilde{\mathbf{C}}^{(uv)}$ and $\tilde{\mathbf{W}}^{(uv)}$ of the treatment vectors $\mathbf{C}$ and $\mathbf W$, and then immediately obtain the coupling of $X_{ij}$ as
\[
\tilde{X}_{ij}^{(uv)}
:=
X\bigl(\tilde{\mathbf{W}}_{\mathcal{N}_{ij}}^{(uv)},\tilde{\mathbf{C}}_{\mathcal{N}_{ij}}^{(uv)}\bigr).
\]

We start with the cluster-level coupling. Within the cluster-level interference neighborhood $\mathcal N_{uv}^{(cl)}$, draw
\[
\tilde{\mathbf C}_{\mathcal N_{uv}}^{(uv)}
\sim
\mathbf C_{\mathcal N_{uv}},
\]
where the draw is from the marginal distribution of $\mathbf C_{\mathcal N_{uv}}$, independently of the original array $(\mathbf W,\mathbf C)$ and independently across coupling indices $(u,v)$.
For clusters outside $\mathcal N_{uv}^{(cl)}$, the assignments are adjusted to preserve the complete randomization constraints. Specifically, for each stratum $k$, define
\[
\Delta_{uv,k}^{(1)}
=
\sum_{i\in\mathcal N_{uv}^{(cl)}\cap\mathcal I_k}
\bigl(\tilde C_i^{(uv)}-C_i\bigr).
\]
Then, for $i\in\mathcal I_k\setminus\mathcal N_{uv}^{(cl)}$, set
\[
\tilde C_i^{(uv)}
=
C_i
-
\operatorname{sgn}\!\bigl(\Delta_{uv,k}^{(1)}\bigr)
\epsilon_i^{(uv)},
\]
where $\epsilon_i^{(uv)}$ is a random indicator variable.
For each stratum $k$, let
\[
\boldsymbol\epsilon_k^{(uv)}
=
\bigl(\epsilon_i^{(uv)}\bigr)_{i\in\mathcal I_k}.
\]
Conditional on the original realization $(\mathbf W,\mathbf C)$ and on the redraw 
$\tilde{\mathbf C}_{\mathcal N_{uv}}^{(uv)}$, the vectors 
$\boldsymbol\epsilon_k^{(uv)}$ are independent across strata $k$, independent across coupling indices $(u,v)$, and uniformly distributed over the set of indicator vectors satisfying
\[
\sum_{\substack{
i\in\mathcal I_k\setminus\mathcal N_{uv}^{(cl)}\\
2C_i-1=\operatorname{sgn}(\Delta_{uv,k}^{(1)})
}}
\epsilon_i^{(uv)}
=
\bigl|\Delta_{uv,k}^{(1)}\bigr|,
\]
and
\[
\epsilon_i^{(uv)}=0
\quad\text{for all } 
i\notin
\left\{
i'\in\mathcal I_k\setminus\mathcal N_{uv}^{(cl)}
:
2C_{i'}-1=\operatorname{sgn}(\Delta_{uv,k}^{(1)})
\right\}.
\]
Thus, within each stratum, the adjustment changes exactly 
$\bigl|\Delta_{uv,k}^{(1)}\bigr|$ assignments outside the neighborhood in the direction needed to offset the imbalance generated inside the neighborhood.

For the unit-level coupling, let
\[
\mathcal M_{uv}^{(cl)}
=
\mathcal N_{uv}^{(cl)}
\cup
\{i:\tilde C_i^{(uv)}\neq C_i\}
\]
denote the set of clusters affected by the cluster-level coupling.
Conditional on $\tilde{\mathbf C}^{(uv)}$, we define the unit-level coupling as follows. For clusters in $\mathcal M_{uv}^{(cl)}$, draw
\[
\tilde{\mathbf W}_i^{(uv)}\mid \tilde C_i^{(uv)}
\sim
\mathbf W_i\mid C_i=\tilde C_i^{(uv)},
\qquad
i\in\mathcal M_{uv}^{(cl)}.
\]
These draws are independent across clusters $i$, independent across coupling indices $(u,v)$, and independent of the original realization $(\mathbf C,\mathbf W)$, conditional on $\tilde{\mathbf C}^{(uv)}$. For clusters outside $\mathcal M_{uv}^{(cl)}$, set
\[
\tilde{\mathbf W}_i^{(uv)}
=
\mathbf W_i.
\]
That is, for affected clusters, $\tilde{\mathbf W}_i^{(uv)}$ is drawn from the same conditional distribution as in the original assignment mechanism given the coupled cluster-level assignment $\tilde C_i^{(uv)}$, whereas unit-level treatments in unaffected clusters remain unchanged.

The coupling construction above for each unit $(u,v)$ meets the heuristic criteria we mentioned: the coupled array has the same distribution as the original, is weakly dependent on $X_{uv}$, and alters only a small fraction of the treatments. A rigorous proof, however, requires more precise conditions on the coupling, which we state formally in the next proposition.
Throughout, $\|R\|_{\infty}:=\inf\{M\ge 0:\mathbb P(|R|\le M)=1\}$ denotes the
essential supremum norm of a random variable $R$, and
$\|R\|_{1}:=\mathbb E|R|$ its $\mathcal L^{1}$ norm.
\begin{proposition}[Stein bound for summands under complete randomization]
  \label{prop:scalar_clt}
  Let \(\mathcal P\) be an index set for $(i,j)$ with \(|\mathcal P|=N\), and let
  the array of scalar-type summands \((X_{ij})_{(i,j)\in\mathcal P}\) and the
  coupling arrays
  \(\bigl((\tilde X_{ij}^{(uv)})_{(i,j)\in\mathcal P}\bigr)_{(u,v)\in\mathcal P}\)
  be jointly defined on a common probability space with finitely many states.
  Assume that
  \begin{enumerate}
      \item[(i)] \(\mathbb{E}[X_{ij}]=0\) for all \((i,j)\in\mathcal P\);
      \item[(ii)] \(|X_{ij}|\le C/N\) almost surely for some finite constant \(C\);
      \item[(iii)] for every \((u,v)\in\mathcal P\),
      \[
      (\tilde X_{ij}^{(uv)})_{(i,j)\in\mathcal P}
      \stackrel{d}{=}
      (X_{ij})_{(i,j)\in\mathcal P},
      \qquad
      \mathbb{E}\bigl[X_{uv}\bigm|(\tilde X_{ij}^{(uv)})_{(i,j)\in\mathcal P}\bigr]=0 .
      \]
  \end{enumerate}
  Write
  \[
  \sigma^2 := \mathrm{Var}\!\Bigl(\sum_{(i,j)\in\mathcal P} X_{ij}\Bigr),
  \qquad
  S := \sigma^{-1}\sum_{(i,j)\in\mathcal P} X_{ij},
  \qquad
  Z \sim \mathcal N(0,1).
  \]
  and abbreviate the coupling discrepancy of unit \((i,j)\) under the coupling
  indexed by \((u,v)\) as
  \[
  D_{ij}^{(uv)}:=X_{ij}-\tilde X_{ij}^{(uv)} .
  \]
  For each quadruple $\bigl((u,v),(u',v'),(i,j),(i',j')\bigr)\in\mathcal P^4$,
  define the $\psi$-mixing coefficient
  \[
  \psi_{ij,i'j'}^{(uv),(u'v')}
  :=
  \sup\;\left|\frac{P(A\cap B)}{P(A)P(B)}-1\right| ,
  \]
  which compares the joint probability with the product of the
  marginal probabilities. Finally, define the two error functionals
  \[
  \mathcal E_1
  :=
  \frac{1}{\sigma^{3}N}
  \sum_{(u,v)\in\mathcal P}
  \Bigl\|\sum_{(i,j)\in\mathcal P}D_{ij}^{(uv)}\Bigr\|_{\infty}^{2},
  \qquad
  \mathcal E_2
  :=
  \frac{1}{\sigma^{4}N^{2}}
  \sum_{\substack{\bigl((u,v),(u',v')\bigr)\in\mathcal P^2\\ \bigl((i,j),(i',j')\bigr)\in\mathcal P^2}}
  T_{ij,i'j'}^{(uv),(u'v')},
  \]
  in which each summand of \(\mathcal E_2\) is the smaller of a mixing bound and a
  joint-support bound,
  \[
  T_{ij,i'j'}^{(uv),(u'v')}
  :=
  \min\bigl\{
  T^{\mathrm{mix}}_{ij,i'j'},\;
  T^{\mathrm{sup}}_{ij,i'j'}
  \bigr\},
  \]
  whose two entries are
  \[
  T^{\mathrm{mix}}_{ij,i'j'}
  :=
  \psi_{ij,i'j'}^{(uv),(u'v')}\,
  \bigl\|D_{ij}^{(uv)}\bigr\|_{1}\,
  \bigl\|D_{i'j'}^{(u'v')}\bigr\|_{1},
  \]
  \[
  T^{\mathrm{sup}}_{ij,i'j'}
  :=
  \bigl\|D_{ij}^{(uv)}\bigr\|_{\infty}
  \bigl\|D_{i'j'}^{(u'v')}\bigr\|_{\infty}
  \Bigl(
  P\bigl(D_{ij}^{(uv)}D_{i'j'}^{(u'v')}\neq 0\bigr)
  +
  P\bigl(D_{ij}^{(uv)}\neq 0\bigr)P\bigl(D_{i'j'}^{(u'v')}\neq 0\bigr)
  \Bigr),
  \]
  the dependence of both on \((u,v)\) and \((u',v')\) being suppressed in the
  notation. Then the Wasserstein distance between \(S\) and \(Z\) satisfies
  \[
  d_W(S,Z)
  =
  O\Bigl(\mathcal E_1+\sqrt{\mathcal E_2}\Bigr).
  \]
  \end{proposition}
\begin{remark}
  In other words, Proposition~\ref{prop:scalar_clt} delivers a central limit theorem under the following three conditions on the coupling arrays: (a) each coupling array has the same distribution as the original array and is independent of the relevant components of the original array; (b) the discrepancy between each coupling array and the original array is asymptotically negligible; and (c) the dependence between the discrepancy terms associated with different unit pairs $((u,v),(i,j))$ and $((u',v'),(i',j'))$ is sufficiently weak for most choices of these pairs. While the first two conditions are covered in the heuristics we mentioned above, the last one is sutble and not covered yet.
\end{remark}
The following proposition shows that the coupled treatment vectors constructed above can be applied to Proposition \ref{prop:scalar_clt} to achieve asymptotic normality. In fact, it establishes results somewhat stronger than those required by Proposition~\ref{prop:scalar_clt}, but in a more interpretable form, as it characterizes the couplings directly at the level of the treatment vectors rather than at the level of the summand statistics. This is useful because quantities such as mixing coefficients become abstract when defined with respect to the $\sigma$-fields generated by summand discrepancies.

Its four parts are organized around the assumptions and the two error
functionals of Proposition~\ref{prop:scalar_clt}. Part~(i) checks assumptions
(i) and (iii) of Proposition~\ref{prop:scalar_clt}. Parts~(ii) and~(iii)
together establish that the coupling discrepancy is negligible, the former
controlling its aggregate effect and the latter its unit-level activation
probabilities. The aggregate effect, which enters through $\mathcal E_1$, is
controlled in Part~(ii) by the sup-norm of the accumulated discrepancy
$\sum_{(i,j)}D_{ij}^{(uv)}$. The unit-level effect, which enters through
$\mathcal E_2$, is controlled in Part~(iii) by bounds on the marginal
probabilities $\mathbb P(D_{ij}^{(uv)}\neq 0)$ and the joint probabilities
$\mathbb P(D_{ij}^{(uv)}D_{i'j'}^{(u'v')}\neq 0)$. Both parts rely on the same
structural property: each coupling changes the treatments of only a vanishing
fraction of units. Part~(iv) then bounds the mixing coefficient. Assumption~(ii)
of Proposition~\ref{prop:scalar_clt}, the almost sure bound $|X_{ij}|\le C/N$,
is verified directly from Assumption~\ref{ass:CLT_regularity}(iii)--(vi) where
the proposition is applied.

\begin{proposition}[Coupling properties]\label{prop:coupling}
Suppose Assumptions~\ref{ass:cluster_assignment}--\ref{ass:interference},
\ref{ass:weight}, \ref{ass:CLT_regularity}(ii)--(vi), and
\ref{ass:complete_rand}(i) hold.
Consider the coupling arrays 
$(\tilde{\mathbf W}^{(uv)},\tilde{\mathbf C}^{(uv)})$ constructed above for all $(u,v)\in\mathcal P$. 
Then the following properties hold.

\begin{enumerate}

\item[(i)] 
For every $(u,v)\in\mathcal P$,
\[
  (\tilde{\mathbf W}^{(uv)},\tilde{\mathbf C}^{(uv)})
  \overset{d}{=}
  (\mathbf W,\mathbf C),
\]
and the coupling array is independent of the original assignments in the interference neighborhood of $(u,v)$:
\[
  (\tilde{\mathbf W}^{(uv)},\tilde{\mathbf C}^{(uv)})
  \indep
  (\mathbf W_{\mathcal N_{uv}},\mathbf C_{\mathcal N_{uv}}).
\]

\item[(ii)]
The discrepancy between each coupling array and the original assignment array is
uniformly negligible in the aggregate: the number of units whose neighborhood
assignments are disturbed by a single coupling is almost surely of smaller order
than $N$,
\[
\max_{(u,v)\in\mathcal P}
\left\{
\left|\left\{(i,j)\in \mathcal P:(\mathbf W_{\mathcal{N}_{ij}},\mathbf C_{\mathcal{N}_{ij}})\neq (\tilde{\mathbf W}^{(uv)}_{\mathcal{N}_{ij}},\tilde{\mathbf C}^{(uv)}_{\mathcal{N}_{ij}})\right\}
\right|
\right\}
=
O\!\left(\frac{N}{n}\right).
\]

\item[(iii)]
The discrepancy is also negligible unit by unit: its activation probability is of
order $1/n$, and two activations occur simultaneously only with probability of
order $1/n^{2}$, except in explicitly identified configurations of the indices.
Write
\[
\mathcal B_{ij}^{(uv)}
:=
\Bigl\{
(\mathbf W_{\mathcal{N}_{ij}},\mathbf C_{\mathcal{N}_{ij}})
\neq
(\tilde{\mathbf W}^{(uv)}_{\mathcal{N}_{ij}},\tilde{\mathbf C}^{(uv)}_{\mathcal{N}_{ij}})
\Bigr\}
\]
for the event that the coupling indexed by $(u,v)$ disturbs the neighborhood of
unit $(i,j)$. There are constants $C_{P}^{(0)},C_{P}^{(1)}>0$ such that a single
activation probability can exceed the order $1/n$ only when the two
cluster-level neighborhoods involved intersect,
\[
P\bigl(\mathcal B_{ij}^{(uv)}\bigr)\geq\frac{C_{P}^{(0)}}{n}
\quad\text{only if}\quad
\mathcal N_{ij}^{(cl)}\cap\mathcal N_{uv}^{(cl)}\neq\varnothing ,
\]
and a joint activation probability can exceed the order $1/n^{2}$ only when one
of the two sides already has intersecting neighborhoods, or the two couplings
coincide and the two disturbed neighborhoods intersect,
\[
P\bigl(\mathcal B_{ij}^{(uv)}\cap\mathcal B_{i'j'}^{(u'v')}\bigr)
\geq\frac{C_{P}^{(1)}}{n^{2}}
\quad\text{only if}\quad
\begin{cases}
\mathcal N_{ij}^{(cl)}\cap\mathcal N_{uv}^{(cl)}\neq\varnothing,
\quad\text{or}\\[2pt]
\mathcal N_{i'j'}^{(cl)}\cap\mathcal N_{u'v'}^{(cl)}\neq\varnothing,
\quad\text{or}\\[2pt]
(u,v)=(u',v')\ \text{and}\
\mathcal N_{ij}^{(cl)}\cap\mathcal N_{i'j'}^{(cl)}\neq\varnothing .
\end{cases}
\]
\item[(iv)]
For each pair $\bigl((u,v),(i,j)\bigr)\in\mathcal P^2$, define the localized
$\sigma$-field
\[
\mathcal G_{uv,ij}
:=
\sigma\Bigl(
\mathbf W_{\mathcal N_{uv}},\,\mathbf C_{\mathcal N_{uv}},\;
\mathbf W_{\mathcal N_{ij}},\,\mathbf C_{\mathcal N_{ij}},\;
\tilde{\mathbf W}^{(uv)}_{\mathcal N_{ij}},\,\tilde{\mathbf C}^{(uv)}_{\mathcal N_{ij}}
\Bigr),
\]
which reveals the original treatments on the two interference neighborhoods
$\mathcal N_{uv}$ and $\mathcal N_{ij}$ together with the coupled treatments on
the latter. 
Furthermore, define
\[
\psi_{ij,i'j'}^{(0),(uv),(u'v')}
:=
\sup_{\substack{
A\in \mathcal G_{uv,ij},\;B\in \mathcal G_{u'v',i'j'}\\
P(A)P(B)>0
}}
\left|\frac{P(A\cap B)}{P(A)P(B)}-1\right| .
\]
Then there exists a constant $C_{\psi}^{(0)}>0$ such that the coefficient can be
large only when the cluster-level interference neighborhoods associated with the two pairs of units intersect:
\[
\psi_{ij,i'j'}^{(0),(uv),(u'v')}
\geq
\frac{C_{\psi}^{(0)}}{n}
\quad\text{only if}\quad
\bigl(\mathcal N_{ij}^{(cl)}\cup\mathcal N_{uv}^{(cl)}\bigr)
\cap
\bigl(\mathcal N_{i'j'}^{(cl)}\cup\mathcal N_{u'v'}^{(cl)}\bigr)
\neq\varnothing .
\]

\end{enumerate}
\end{proposition}
Proposition~\ref{prop:coupling} describes the coupling at the level of the
treatment vectors, whereas Proposition~\ref{prop:scalar_clt} is stated for the
scalar summands that are actually summed. The next corollary performs the
translation, and it does so part by part: the distributional statement~(i)
transfers because each summand is a fixed measurable function of the treatments
in its own neighborhood; the aggregate and unit-level discrepancy bounds~(ii)
and~(iii) transfer because a summand can change only when the treatments in
that neighborhood change, so the aggregate counts and the disturbance
probabilities carry over after multiplication (respectively, comparison) by the
uniform summand scale $O(1/N)$; and the coefficient bound~(iv) transfers
because the $\sigma$-field generated by a summand and its coupling discrepancy
is contained in the localized field $\mathcal G_{uv,ij}$, so that monotonicity
of the mixing coefficient applies. The result is exactly the input required by
Proposition~\ref{prop:scalar_clt}: conditions (a) and (b) of its
hypothesis~(iii), the aggregated and unit-level discrepancies entering
$\mathcal E_1$ and $\mathcal E_2$, and the coefficients
$\psi_{ij,i'j'}^{(uv),(u'v')}$ themselves.

To state the result, define
\[
\mathbf X_{ij}^{(\mathrm{marg})}
=
\frac{g_i}{N_i}
(\mathbf B_{\mathrm{marg}})_{(*,ij)}
\circ
\bigl(Y_{ij}\mathbf 1-\boldsymbol\tau_{ij,\mathrm{marg}}\bigr)
+
\frac{g_i}{N_i}\bigl((\mathbf B_{\mathrm{marg}})_{(*,ij)}-\mathbf{1}\bigr)
\circ
\bigl(\boldsymbol\tau_{ij,\mathrm{marg}}-\boldsymbol\tau_{\mathrm{marg}}\bigr),
\]
and 
\[
\mathbf X_{ij}
=
\frac{g_i}{N_i}
\mathbf B_{(*,ij)}
\circ
\bigl(Y_{ij}\mathbf 1-\boldsymbol\tau_{ij}\bigr)+
\frac{g_i}{N_i}\bigl(\mathbf B_{(*,ij)}-\mathbf{1}\bigr)
\circ
\bigl(\boldsymbol\tau_{ij}-\boldsymbol\tau\bigr).
\]
For each coupling index $(u,v)\in\mathcal P$, define the coupled counterparts by
\[
(\tilde{\mathbf X}_{ij}^{(\mathrm{marg})})^{(uv)}
=
\frac{g_i}{N_i}
(\tilde{\mathbf B}_{\mathrm{marg}}^{(uv)})_{(*,ij)}
\circ
\bigl(
\tilde Y_{ij}^{(uv)}\mathbf 1-\boldsymbol\tau_{ij,\mathrm{marg}}
\bigr)
+
\frac{g_i}{N_i}
\bigl((\tilde{\mathbf B}_{\mathrm{marg}}^{(uv)})_{(*,ij)}-\mathbf 1\bigr)
\circ
\bigl(\boldsymbol\tau_{ij,\mathrm{marg}}-\boldsymbol\tau_{\mathrm{marg}}\bigr),
\]
and
\[
\tilde{\mathbf X}_{ij}^{(uv)}
=
\frac{g_i}{N_i}
\tilde{\mathbf B}_{(*,ij)}^{(uv)}
\circ
\bigl(
\tilde Y_{ij}^{(uv)}\mathbf 1-\boldsymbol\tau_{ij}
\bigr)
+
\frac{g_i}{N_i}
\bigl(\tilde{\mathbf B}_{(*,ij)}^{(uv)}-\mathbf 1\bigr)
\circ
\bigl(\boldsymbol\tau_{ij}-\boldsymbol\tau\bigr),
\]
where $\tilde Y_{ij}^{(uv)}$, $\tilde{\mathbf B}_{\mathrm{marg}}^{(uv)}$, and 
$\tilde{\mathbf B}^{(uv)}$ denote the realized outcome, the marginal weight matrix, 
and the non-marginal weight matrix evaluated at 
$(\tilde{\mathbf W}^{(uv)},\tilde{\mathbf C}^{(uv)})$ rather than 
$(\mathbf W,\mathbf C)$. With these definitions, the properties in Proposition~\ref{prop:coupling} imply
the following coupling conditions for the scalar projections
$\mathbf a^\top\mathbf X_{ij}^{(\mathrm{marg})}$ and
$\mathbf b^\top\mathbf X_{ij}$, which are the summands to which
Proposition~\ref{prop:scalar_clt} is applied.

\begin{corollary}[Coupling properties for summand statistics]\label{cor:coupling}
Suppose the assumptions of Proposition~\ref{prop:coupling} hold. 
Let $\mathbf a$ and $\mathbf b$ be vectors conformable with 
$\mathbf X_{ij}^{(\mathrm{marg})}$ and $\mathbf X_{ij}$, respectively. 
Then the scalar projections 
$\mathbf a^\top\mathbf X_{ij}^{(\mathrm{marg})}$ and 
$\mathbf b^\top\mathbf X_{ij}$ satisfy the following coupling properties.

\begin{enumerate}

\item[(i)] 
For every $(u,v)\in\mathcal P$,
\[
\left(
\mathbf a^\top
(\tilde{\mathbf X}_{ij}^{(\mathrm{marg})})^{(uv)}
\right)_{(i,j)\in\mathcal P}
\overset{d}{=}
\left(
\mathbf a^\top
\mathbf X_{ij}^{(\mathrm{marg})}
\right)_{(i,j)\in\mathcal P},
\]
and
\[
\left(
\mathbf b^\top
\tilde{\mathbf X}_{ij}^{(uv)}
\right)_{(i,j)\in\mathcal P}
\overset{d}{=}
\left(
\mathbf b^\top
\mathbf X_{ij}
\right)_{(i,j)\in\mathcal P}.
\]
Moreover,
\[
\mathbb{E}\!\left[
\mathbf a^\top
\mathbf X_{ij}^{(\mathrm{marg})}
\,\middle|\,
\left(
\mathbf a^\top
(\tilde{\mathbf X}_{ij}^{(\mathrm{marg})})^{(uv)}
\right)_{(i,j)\in\mathcal P}
\right]
=
0,
\]
and
\[
\mathbb{E}\!\left[
\mathbf b^\top
\mathbf X_{ij}
\,\middle|\,
\left(
\mathbf b^\top
\tilde{\mathbf X}_{ij}^{(uv)}
\right)_{(i,j)\in\mathcal P}
\right]
=
0.
\]

\item[(ii)]
The discrepancies between the original summand statistics and their coupled
counterparts are uniformly negligible after aggregation over units,
\[
\max\Biggl\{
\sup_{(u,v)\in\mathcal P}
\left\|
\sum_{(i,j)\in\mathcal P}
\left[
\mathbf a^\top\mathbf X_{ij}^{(\mathrm{marg})}
-
\mathbf a^\top
(\tilde{\mathbf X}_{ij}^{(\mathrm{marg})})^{(uv)}
\right]
\right\|_{\infty},
\]
\[
\sup_{(u,v)\in\mathcal P}
\left\|
\sum_{(i,j)\in\mathcal P}
\left[
\mathbf b^\top\mathbf X_{ij}
-
\mathbf b^\top
\tilde{\mathbf X}_{ij}^{(uv)}
\right]
\right\|_{\infty}
\Biggr\}
=
O\!\left(\frac{1}{n}\right).
\]

\item[(iii)]
The discrepancies are also negligible unit by unit, at the level of the
individual summands. Writing
\[
D_{ij}^{(uv)}:=\mathbf a^\top\mathbf X_{ij}^{(\mathrm{marg})}
-\mathbf a^\top(\tilde{\mathbf X}_{ij}^{(\mathrm{marg})})^{(uv)}
\quad\text{(and likewise with $\mathbf b^\top\mathbf X_{ij}$),}
\]
there is a global constant $C'$ such that, with $\mathcal B_{ij}^{(uv)}$ the
disturbance event of Proposition~\ref{prop:coupling}(iii),
\[
\bigl\{D_{ij}^{(uv)}\neq 0\bigr\}\subseteq\mathcal B_{ij}^{(uv)},
\qquad
\bigl\|D_{ij}^{(uv)}\bigr\|_{\infty}\leq\frac{C'}{N},
\qquad
\bigl\|D_{ij}^{(uv)}\bigr\|_{1}\leq\frac{C'}{N}\,P\bigl(\mathcal B_{ij}^{(uv)}\bigr).
\]

\item[(iv)]
Matching the coefficient required by Proposition~\ref{prop:scalar_clt}, the
mixing coefficients are indexed by a coupling pair together with a pair of
units. For each quadruple $\bigl((u,v),(u',v'),(i,j),(i',j')\bigr)\in\mathcal P^4$, define
\[
\psi_{ij,i'j'}^{(uv),(u'v'),(\mathrm{marg})}
:=
\sup_{\substack{
A\in \sigma\!\left(
\mathbf a^\top\mathbf X_{uv}^{(\mathrm{marg})},\;
\mathbf a^\top\mathbf X_{ij}^{(\mathrm{marg})}
-
\mathbf a^\top
(\tilde{\mathbf X}_{ij}^{(\mathrm{marg})})^{(uv)}
\right)\\
B\in \sigma\!\left(
\mathbf a^\top\mathbf X_{u'v'}^{(\mathrm{marg})},\;
\mathbf a^\top\mathbf X_{i'j'}^{(\mathrm{marg})}
-
\mathbf a^\top
(\tilde{\mathbf X}_{i'j'}^{(\mathrm{marg})})^{(u'v')}
\right)\\
P(A)P(B)>0
}}
\left|\frac{P(A\cap B)}{P(A)P(B)}-1\right| ,
\]
and
\[
\psi_{ij,i'j'}^{(uv),(u'v')}
:=
\sup_{\substack{
A\in \sigma\!\left(
\mathbf b^\top\mathbf X_{uv},\;
\mathbf b^\top\mathbf X_{ij}
-
\mathbf b^\top
\tilde{\mathbf X}_{ij}^{(uv)}
\right)\\
B\in \sigma\!\left(
\mathbf b^\top\mathbf X_{u'v'},\;
\mathbf b^\top\mathbf X_{i'j'}
-
\mathbf b^\top
\tilde{\mathbf X}_{i'j'}^{(u'v')}
\right)\\
P(A)P(B)>0
}}
\left|\frac{P(A\cap B)}{P(A)P(B)}-1\right| .
\]
Then there exists a constant $C_{\psi}^{(1)}>0$ such that either coefficient can
be large only when the two localized cluster sets meet:
\[
\max\Bigl\{
\psi_{ij,i'j'}^{(uv),(u'v'),(\mathrm{marg})},\;
\psi_{ij,i'j'}^{(uv),(u'v')}
\Bigr\}
\geq
\frac{C_{\psi}^{(1)}}{n}
\quad\text{only if}\quad
\bigl(\mathcal N_{ij}^{(cl)}\cup\mathcal N_{uv}^{(cl)}\bigr)
\cap
\bigl(\mathcal N_{i'j'}^{(cl)}\cup\mathcal N_{u'v'}^{(cl)}\bigr)
\neq\varnothing .
\]

\end{enumerate}
\end{corollary}
\begin{remark}
Using the rates established in parts (ii)--(iv) of Corollary~\ref{cor:coupling}, we can bound the Wasserstein distance between the Gaussian limit $Z$ and the statistic $S$ constructed in Proposition~\ref{prop:scalar_clt} based on either $X_{ij}=\mathbf{a}^\top \mathbf{X}_{ij}^{(\mathrm{marg})}$ or $X_{ij}=\mathbf{b}^\top \mathbf{X}_{ij}$. Under the coupling construction above,
the aggregated discrepancy bound of part~(ii) makes the first error functional
\[
\mathcal E_1
=
\frac{1}{\sigma^3 N}\cdot N\cdot O\!\left(\frac{1}{n^2}\right)
=
O\!\left(\frac{1}{\sigma^3 n^2}\right),
\]
while the unit-level counts of part~(iii) and the structural characterization of
part~(iv) together give $\mathcal E_2=O(\sigma^{-4}n^{-3})$; the accounting is
carried out in the proof of Theorem~\ref{thm:clt}, where the quadruples are
split according to whether the mixing coefficient is of order $1/n$ and, when it
is not, according to the order of the two disturbance probabilities. Hence
\[
d_W(S,Z)
=
O\!\left(
\frac{1}{\sigma^3 n^2}
+
\frac{1}{\sigma^{2}n^{3/2}}
\right),
\]
which validates the variance rate conditions required in Theorem~\ref{thm:clt}.
\end{remark}
\begin{remark}
The procedures developed in this section remain valid even when the randomized trial does not involve a cluster design and instead implements complete randomization at the unit level to assign treatment $\mathbf{W}$ directly. In this case, one can interpret the design as a special case by treating each unit as its own cluster, so that the cluster-level treatment coincides with the unit-level treatment, i.e., $\mathbf{C} = \mathbf{W}$. Under this formulation, the current coupling can still be used to prove central limit theorems under interference.
\end{remark}
\begin{remark}
The main obstacle to extending this technique to the analysis of cluster-agnostic estimators lies in condition (iii), which requires control over mixing coefficients. Although carefully constructed summand statistics $X_{ij}$ can eliminate the mean contribution of cluster-level treatments and render $X_{ij}$ and $X_{i'j'}$ uncorrelated, it remains unclear how to control the mixing coefficients between the $\sigma$-fields they generate.
Interested readers may refer to the proof of Proposition~\ref{prop:scalar_clt} for insight into how the relevant terms might be controlled without relying on the classical additive mixing coefficients.
\end{remark}

\bibliographystyle{plainnat}
\bibliography{reference}

\begin{thebibliography}{23}
\providecommand{\natexlab}[1]{#1}
\providecommand{\url}[1]{\texttt{#1}}
\expandafter\ifx\csname urlstyle\endcsname\relax
  \providecommand{\doi}[1]{doi: #1}\else
  \providecommand{\doi}{doi: \begingroup \urlstyle{rm}\Url}\fi

\bibitem[Aronow and Samii(2017)]{Aronow2017}
Peter~M. Aronow and Cyrus Samii.
\newblock Estimating average causal effects under general interference, with application to a social network experiment.
\newblock \emph{Annals of Applied Statistics}, 11:\penalty0 1912--1947, 12 2017.
\newblock ISSN 19417330.
\newblock \doi{10.1214/16-AOAS1005}.

\bibitem[Chen and Röllin(2010)]{Chen2010}
Louis H.~Y. Chen and Adrian Röllin.
\newblock Stein couplings for normal approximation, 2010.
\newblock URL \url{http://arxiv.org/abs/1003.6039}.
\newblock arXiv:1003.6039.

\bibitem[Fang and Röllin(2015)]{Fang2015}
Xiao Fang and Adrian Röllin.
\newblock Rates of convergence for multivariate normal approximation with applications to dense graphs and doubly indexed permutation statistics.
\newblock \emph{Bernoulli}, 21\penalty0 (4):\penalty0 2157--2189, 2015.
\newblock \doi{10.3150/14-BEJ639}.

\bibitem[Gao et~al.(2026)Gao, Harshaw, Sävje, and Wang]{Gao2026}
Chao Gao, Christopher Harshaw, Fredrik Sävje, and Yitan Wang.
\newblock On the impossibility of specification testing of interference models based on exposure mappings, 2026.
\newblock URL \url{https://arxiv.org/abs/2605.09726}.
\newblock arXiv:2605.09726.

\bibitem[Gao and Ding(2025)]{Gao2025}
Mengsi Gao and Peng Ding.
\newblock Causal inference in network experiments: regression-based analysis and design-based properties.
\newblock \emph{Journal of Econometrics}, 252:\penalty0 106119, 2025.
\newblock \doi{10.1016/j.jeconom.2025.106119}.

\bibitem[Godambe(1955)]{Godambe1955}
V.~P. Godambe.
\newblock A unified theory of sampling from finite populations.
\newblock \emph{Journal of the Royal Statistical Society Series B: Statistical Methodology}, 17:\penalty0 269--278, 7 1955.
\newblock ISSN 1369-7412.
\newblock \doi{10.1111/j.2517-6161.1955.tb00203.x}.

\bibitem[Halloran and Struchiner(1991)]{Halloran1991}
M~Elizabeth Halloran and Claudio~J Struchiner.
\newblock Study designs for dependent happenings.
\newblock \emph{Epidemiology}, 2\penalty0 (5):\penalty0 331--338, 1991.

\bibitem[Harshaw et~al.(2025)Harshaw, Sävje, and Wang]{Harshaw2025}
Christopher Harshaw, Fredrik Sävje, and Yitan Wang.
\newblock A general design-based framework and estimator for randomized experiments, 2025.
\newblock URL \url{http://arxiv.org/abs/2210.08698}.
\newblock arXiv:2210.08698.

\bibitem[Hudgens and Halloran(2008)]{Hudgens2008}
Michael~G. Hudgens and M.~Elizabeth Halloran.
\newblock Toward causal inference with interference.
\newblock \emph{Journal of the American Statistical Association}, 103:\penalty0 832--842, 6 2008.
\newblock ISSN 01621459.
\newblock \doi{10.1198/016214508000000292}.

\bibitem[Joag-Dev and Proschan(1983)]{JoagDev1983}
Kumar Joag-Dev and Frank Proschan.
\newblock Negative association of random variables, with applications.
\newblock \emph{The Annals of Statistics}, 11\penalty0 (1):\penalty0 286--295, 1983.
\newblock \doi{10.1214/aos/1176346079}.

\bibitem[Kojevnikov et~al.(2021)Kojevnikov, Marmer, and Song]{Kojevnikov2021}
Denis Kojevnikov, Vadim Marmer, and Kyungchul Song.
\newblock Limit theorems for network dependent random variables.
\newblock \emph{Journal of Econometrics}, 222\penalty0 (2):\penalty0 882--908, 2021.
\newblock \doi{10.1016/j.jeconom.2020.05.019}.

\bibitem[Leung(2023)]{Leung2023}
Michael~P. Leung.
\newblock Network cluster‐robust inference.
\newblock \emph{Econometrica}, 91:\penalty0 641--667, 2023.
\newblock ISSN 0012-9682.
\newblock \doi{10.3982/ecta19816}.

\bibitem[Leung(2025)]{Leung2025}
Michael~P. Leung.
\newblock Cluster-randomized trials with cross-cluster interference.
\newblock \emph{Journal of the American Statistical Association}, 2025.
\newblock \doi{10.1080/01621459.2025.2566416}.
\newblock Published online; in press.

\bibitem[Li and Ding(2017)]{Li2017}
Xinran Li and Peng Ding.
\newblock General forms of finite population central limit theorems with applications to causal inference.
\newblock \emph{Journal of the American Statistical Association}, 112:\penalty0 1759--1769, 10 2017.
\newblock ISSN 1537274X.
\newblock \doi{10.1080/01621459.2017.1295865}.

\bibitem[Liu and Hudgens(2014)]{Liu2014}
Lan Liu and Michael~G. Hudgens.
\newblock Large sample randomization inference of causal effects in the presence of interference.
\newblock \emph{Journal of the American Statistical Association}, 109:\penalty0 288--301, 2014.
\newblock ISSN 1537274X.
\newblock \doi{10.1080/01621459.2013.844698}.

\bibitem[Lu et~al.(2026)Lu, Shi, and Ding]{Lu2026}
Sizhu Lu, Lei Shi, and Peng Ding.
\newblock Estimating within-cluster and between-cluster spillover effects in randomized saturation designs, 2026.
\newblock URL \url{http://arxiv.org/abs/2603.19573}.
\newblock arXiv:2603.19573.

\bibitem[Makofane et~al.(2023)Makofane, Kim, Tchetgen, Bassett, Berkman, Adeagbo, Mcgrath, Seeley, Shahmanesh, Yapa, Herbst, Tanser, and Bärnighausen]{Makofane2023}
Keletso Makofane, Hae-Young Kim, Eric~Tchetgen Tchetgen, Mary~T Bassett, Lisa Berkman, Oluwafemi Adeagbo, Nuala Mcgrath, Janet Seeley, Maryam Shahmanesh, H~Manisha Yapa, Kobus Herbst, Frank Tanser, and Till Bärnighausen.
\newblock Impact of family networks on uptake of health interventions: evidence from a community-randomized control trial aimed at increasing hiv testing in south africa.
\newblock \emph{Journal of the International AIDS Society}, 26:\penalty0 e26142, 2023.
\newblock \doi{10.1002/jia2.26142}.

\bibitem[Merlevède et~al.(2011)Merlevède, Peligrad, and Rio]{Merlevde2009}
Florence Merlevède, Magda Peligrad, and Emmanuel Rio.
\newblock A bernstein type inequality and moderate deviations for weakly dependent sequences.
\newblock \emph{Probability Theory and Related Fields}, 151\penalty0 (3--4):\penalty0 435--474, 2011.
\newblock \doi{10.1007/s00440-010-0304-9}.

\bibitem[Ogburn et~al.(2024)Ogburn, Sofrygin, Díaz, and van~der Laan]{Ogburn2024}
Elizabeth~L. Ogburn, Oleg Sofrygin, Iván Díaz, and Mark~J. van~der Laan.
\newblock Causal inference for social network data.
\newblock \emph{Journal of the American Statistical Association}, 119:\penalty0 597--611, 2024.
\newblock ISSN 1537274X.
\newblock \doi{10.1080/01621459.2022.2131557}.

\bibitem[Park and Wager(2026)]{Park2026}
Bryan Park and Stefan Wager.
\newblock Neyman jackknife: Design-based variance estimation for causal inference under interference, 2026.
\newblock URL \url{http://arxiv.org/abs/2604.24017}.
\newblock arXiv:2604.24017.

\bibitem[Ross(2011)]{Ross2011}
Nathan Ross.
\newblock Fundamentals of {S}tein's method.
\newblock \emph{Probability Surveys}, 8:\penalty0 210--293, 2011.
\newblock \doi{10.1214/11-PS182}.

\bibitem[Shi and Ding(2026)]{Shi2025}
Lei Shi and Peng Ding.
\newblock Berry--esseen bounds for design-based causal inference with possibly diverging treatment levels and varying group sizes.
\newblock \emph{The Annals of Statistics}, 54\penalty0 (1), 2026.
\newblock \doi{10.1214/25-AOS2569}.

\bibitem[Sävje(2024)]{Svje2023}
Fredrik Sävje.
\newblock Causal inference with misspecified exposure mappings: separating definitions and assumptions.
\newblock \emph{Biometrika}, 111\penalty0 (1):\penalty0 1--15, 2024.
\newblock \doi{10.1093/biomet/asad019}.

\end{thebibliography}

\section{Proof}
This section collects the proofs of the main results.


\subsection{Proof of identification theory}
In this subsection, we prove theory results about the identification weights.
\subsubsection{Proof of Theorem \ref{thm:all_weights}}
\begin{proof}[Proof of Theorem \ref{thm:all_weights}]
We prove the ``if'' direction via a change-of-measure argument and the ``only if'' direction via the Riesz representation theorem.

\medskip
\noindent\textbf{``If'' direction.}
By Assumption~\ref{ass:measure}(i), $P_{\phi}(\mathbf{W}_{\mathcal{N}_{ij}})\ll P(\mathbf{W}_{\mathcal{N}_{ij}})$,
so the marginal Radon--Nikodym derivative
$\alpha_{ij}(\phi)=dP_{\phi}(\mathbf{W}_{\mathcal{N}_{ij}})/dP(\mathbf{W}_{\mathcal{N}_{ij}})$
exists, is $\sigma(\mathbf{W}_{\mathcal{N}_{ij}})$-measurable, and satisfies
\begin{equation}\label{eq:change_measure_local}
\mathbb{E}\!\left[f(\mathbf{W}_{\mathcal{N}_{ij}})\,\alpha_{ij}(\phi)\right]
=
\mathbb{E}_{\phi}\!\left[f(\mathbf{W}_{\mathcal{N}_{ij}})\right]
\end{equation}
for every $f\in\mathcal{L}(\mathbf{W}_{\mathcal{N}_{ij}})$. Since \(\mathcal{L}(\mathbf{W}_{\mathcal{N}_{ij}})\) contains all well-defined functions of
\(\mathbf{W}_{\mathcal{N}_{ij}}\), and
\(Y_{ij}=Y_{ij}(\mathbf{W}_{\mathcal{N}_{ij}})\) by
Assumption \ref{ass:potential_outcome}-\ref{ass:interference}, it follows that, for any
\(h\in\mathcal{L}(\mathbf{W}_{\mathcal{N}_{ij}})\),
\[
Y_{ij}(\cdot)\,h(\cdot)\in\mathcal{L}(\mathbf{W}_{\mathcal{N}_{ij}}).
\]
Applying~\eqref{eq:change_measure_local} with $f=Y_{ij}(\cdot)\,h(\cdot)$ yields
\[
\mathbb{E}\!\left[Y_{ij}\,\alpha_{ij}(\phi)\,h(\mathbf{W}_{\mathcal{N}_{ij}})\right]
=
\mathbb{E}_{\phi}\!\left[Y_{ij}\,h(\mathbf{W}_{\mathcal{N}_{ij}})\right],
\]
establishing the ``if'' direction. Parts~(a) and~(b) follow by taking
$h\equiv 1$ and $h(\mathbf{w}_{\mathcal{N}_{ij}})=1(w_{ij}=w)/P_{\phi}(W_{ij}=w)$, respectively.

\medskip
\noindent\textbf{``Only if'' direction.}
Since \(\mathcal{L}(\mathbf{W}_{\mathcal{N}_{ij}})\) contains all well-defined functions of
\(\mathbf{W}_{\mathcal{N}_{ij}}\), we view it as a Hilbert space equipped with the
inner product induced by the product expectation. Fix
\(h\in\mathcal{L}(\mathbf{W}_{\mathcal{N}_{ij}})\). We first show that the map
\[
Y_{ij}(\cdot)
\mapsto
\mathbb{E}_{\phi}\!\left[
Y_{ij}h(\mathbf{W}_{\mathcal{N}_{ij}})
\right]
\]
is a continuous linear functional on
\(\mathcal{L}(\mathbf{W}_{\mathcal{N}_{ij}})\). Linearity is immediate. For continuity, the finite support of \(\mathbf{W}_{\mathcal{N}_{ij}}\) implies
\[
\begin{aligned}
\left|
\mathbb{E}_{\phi}\!\left[
Y_{ij}h(\mathbf{W}_{\mathcal{N}_{ij}})
\right]
\right|
&\leq
\max_{\mathbf{w}_{\mathcal{N}_{ij}}}
\left\{|h(\mathbf{w}_{\mathcal{N}_{ij}})|\right\}
\mathbb{E}_{\phi}\!\left[|Y_{ij}|\right] \\
&\leq
\max_{\mathbf{w}_{\mathcal{N}_{ij}}}
\left\{|h(\mathbf{w}_{\mathcal{N}_{ij}})|\right\}
\max_{\mathbf{w}_{\mathcal{N}_{ij}}:\,
P(\mathbf{W}_{\mathcal{N}_{ij}}=\mathbf{w}_{\mathcal{N}_{ij}})>0}
\left\{
\frac{
P_{\phi}(\mathbf{W}_{\mathcal{N}_{ij}}=\mathbf{w}_{\mathcal{N}_{ij}})
}{
P(\mathbf{W}_{\mathcal{N}_{ij}}=\mathbf{w}_{\mathcal{N}_{ij}})
}
\right\}
\left(\mathbb{E}\!\left[|Y_{ij}|^2\right]\right)^{1/2}.
\end{aligned}
\]
The multiplicative constant on the right-hand side is finite because
\(\mathbf{W}_{\mathcal{N}_{ij}}\) has finite support. Hence the functional is bounded,
and therefore continuous.

By the Riesz representation theorem, there exists a unique
\(\gamma_{ij}\in\mathcal{L}(\mathbf{W}_{\mathcal{N}_{ij}})\) such that
\[
\mathbb{E}\!\left[Y_{ij}\gamma_{ij}\right]
=
\mathbb{E}_{\phi}\!\left[
Y_{ij}h(\mathbf{W}_{\mathcal{N}_{ij}})
\right]
\]
for all \(Y_{ij}\in\mathcal{L}(\mathbf{W}_{\mathcal{N}_{ij}})\). By the ``if'' direction, this
representer is
\[
\gamma_{ij}
=
\alpha_{ij}(\phi)h.
\]
Notice that for any
\(\beta_{ij}\in
\mathcal{L}(\mathbf{W}_{\mathcal{N}_{ij}},\mathbf{C}_{\mathcal{N}_{ij}})\),
\[
\begin{aligned}
\mathbb{E}\!\left[Y_{ij}\beta_{ij}\right]
&=
\mathbb{E}\!\left[
Y_{ij}\,
\mathbb{E}\!\left[\beta_{ij}\mid \mathbf{W}_{\mathcal{N}_{ij}}\right]
\right].
\end{aligned}
\]
It follows that
\[
\mathbb{E}\!\left[Y_{ij}\beta_{ij}\right]
=
\mathbb{E}_{\phi}\!\left[
Y_{ij}h(\mathbf{W}_{\mathcal{N}_{ij}})
\right]
\quad
\text{for all }
Y_{ij}\in\mathcal{L}(\mathbf{W}_{\mathcal{N}_{ij}})
\]
if and only if
\[
\mathbb{E}\!\left[\beta_{ij}\mid \mathbf{W}_{\mathcal{N}_{ij}}\right]
=
\alpha_{ij}(\phi)h.
\]
This completes the proof.
\end{proof}

\subsubsection{Proof of Theorem~\ref{thm:cluster_agnostic_weights}}

\begin{proof}[Proof of Theorem~\ref{thm:cluster_agnostic_weights}]
Fix \(\phi\in\{\phi_0,\phi_1\}\).

\medskip
\noindent\textbf{Part (i).}
We prove the ``if'' and ``only if'' directions in parallel with the proof of
Theorem~\ref{thm:all_weights}. By Assumption~\ref{ass:measure}(ii),
\[
P_{\phi}(\mathbf{W}_{\mathcal{N}_{ij}},\mathbf{C}_{\mathcal{N}_{ij}})
\ll
P(\mathbf{W}_{\mathcal{N}_{ij}},\mathbf{C}_{\mathcal{N}_{ij}}),
\]
so the complete Radon--Nikodym derivative
\[
\alpha_{ij}^{(\mathrm{comp})}(\phi)
=
\frac{
dP_{\phi}(\mathbf{W}_{\mathcal{N}_{ij}},\mathbf{C}_{\mathcal{N}_{ij}})
}{
dP(\mathbf{W}_{\mathcal{N}_{ij}},\mathbf{C}_{\mathcal{N}_{ij}})
}
\]
exists. By construction of the counterfactual law and Assumption \ref{ass:cluster_assignment},
\(\mathbf{W}_{\mathcal{N}_{ij}}\) is independent of
\(\mathbf{C}_{\mathcal{N}_{ij}}\) under \(P_{\phi}\).

\smallskip
\noindent\textbf{``If'' direction.}
For any
\(u\in\mathcal{L}(\mathbf{W}_{\mathcal{N}_{ij}})\) and
\(v\in\mathcal{L}(\mathbf{C}_{\mathcal{N}_{ij}})\),
\begin{align*}
& \mathbb{E}\!\left[
v(\mathbf{C}_{\mathcal{N}_{ij}})
u(\mathbf{W}_{\mathcal{N}_{ij}})
\alpha_{ij}^{(\mathrm{comp})}(\phi)
\right] \\
&\quad =
\mathbb{E}_{\phi}\!\left[
v(\mathbf{C}_{\mathcal{N}_{ij}})
u(\mathbf{W}_{\mathcal{N}_{ij}})
\right] \\
&\quad =
\mathbb{E}_{\phi}\!\left[
v(\mathbf{C}_{\mathcal{N}_{ij}})
\right]
\mathbb{E}_{\phi}\!\left[
u(\mathbf{W}_{\mathcal{N}_{ij}})
\right] \\
&\quad =
\mathbb{E}\!\left[
v(\mathbf{C}_{\mathcal{N}_{ij}})
\right]
\mathbb{E}_{\phi}\!\left[
u(\mathbf{W}_{\mathcal{N}_{ij}})
\right],
\end{align*}
where the second equality uses the \(P_{\phi}\)-independence of
\(\mathbf{W}_{\mathcal{N}_{ij}}\) and \(\mathbf{C}_{\mathcal{N}_{ij}}\), and the third
uses \(P_{\phi}(\mathbf{C})=P(\mathbf{C})\). Therefore,
\[
\mathbb{E}\!\left[
u(\mathbf{W}_{\mathcal{N}_{ij}})
\alpha_{ij}^{(\mathrm{comp})}(\phi)
\mid \mathbf{C}
\right]
=
\mathbb{E}\!\left[
u(\mathbf{W}_{\mathcal{N}_{ij}})
\alpha_{ij}^{(\mathrm{comp})}(\phi)
\mid \mathbf{C}_{\mathcal{N}_{ij}}
\right]
=
\mathbb{E}_{\phi}\!\left[
u(\mathbf{W}_{\mathcal{N}_{ij}})
\right].
\]
Taking \(u=Y_{ij}(\cdot)h(\cdot)\) establishes the ``if'' direction. Parts~(i)(a)
and~(i)(b) follow by setting \(h\equiv 1\) and
\(
h(\mathbf{w}_{\mathcal{N}_{ij}})
=
1(w_{ij}=w)/
P_{\phi}(W_{ij}=w)
\)
respectively.

\smallskip
\noindent\textbf{``Only if'' direction.}
Fix \(h\in\mathcal{L}(\mathbf{W}_{\mathcal{N}_{ij}})\) and
\(\beta_{ij}\in
\mathcal{L}(\mathbf{W}_{\mathcal{N}_{ij}},\mathbf{C}_{\mathcal{N}_{ij}})\).
Suppose that
\[
\mathbb{E}\!\left[
Y_{ij}\beta_{ij}
\mid \mathbf{C}
\right]
=
\mathbb{E}_{\phi}\!\left[
Y_{ij}h(\mathbf{W}_{\mathcal{N}_{ij}})
\right]
\]
for all
\(Y_{ij}\in\mathcal{L}(\mathbf{W}_{\mathcal{N}_{ij}})\). Then, for any
\(v\in\mathcal{L}(\mathbf{C}_{\mathcal{N}_{ij}})\),
\begin{align*}
\mathbb{E}\!\left[
Y_{ij}\beta_{ij}v(\mathbf{C}_{\mathcal{N}_{ij}})
\right]
&=
\mathbb{E}_{\phi}\!\left[
Y_{ij}h(\mathbf{W}_{\mathcal{N}_{ij}})
\right]
\mathbb{E}\!\left[
v(\mathbf{C}_{\mathcal{N}_{ij}})
\right] \\
&=
\mathbb{E}_{\phi}\!\left[
Y_{ij}
v(\mathbf{C}_{\mathcal{N}_{ij}})
h(\mathbf{W}_{\mathcal{N}_{ij}})
\right],
\end{align*}
where the equalities again use the \(P_{\phi}\)-independence of
\(\mathbf{W}_{\mathcal{N}_{ij}}\) and \(\mathbf{C}_{\mathcal{N}_{ij}}\) and
\(P_{\phi}(\mathbf{C})=P(\mathbf{C})\).

Since \(\mathbf{W}_{\mathcal{N}_{ij}}\) and
\(\mathbf{C}_{\mathcal{N}_{ij}}\) both have finite support, intersections of events in
\(\sigma(\mathbf{W}_{\mathcal{N}_{ij}})\) and events of
\(\sigma(\mathbf{C}_{\mathcal{N}_{ij}})\) span 
\(\sigma(\mathbf{W}_{\mathcal{N}_{ij}},\mathbf{C}_{\mathcal{N}_{ij}})\). Therefore, products of functions in
\(\mathcal{L}(\mathbf{W}_{\mathcal{N}_{ij}})\) and functions in
\(\mathcal{L}(\mathbf{C}_{\mathcal{N}_{ij}})\) span 
\(\mathcal{L}(\mathbf{W}_{\mathcal{N}_{ij}},\mathbf{C}_{\mathcal{N}_{ij}})\).
Hence,
\[
\mathbb{E}\!\left[
f(\mathbf{W}_{\mathcal{N}_{ij}},\mathbf{C}_{\mathcal{N}_{ij}})
\beta_{ij}
\right]
=
\mathbb{E}_{\phi}\!\left[
f(\mathbf{W}_{\mathcal{N}_{ij}},\mathbf{C}_{\mathcal{N}_{ij}})
h(\mathbf{W}_{\mathcal{N}_{ij}})
\right]
\]
holds for all
\(f\in
\mathcal{L}(\mathbf{W}_{\mathcal{N}_{ij}},\mathbf{C}_{\mathcal{N}_{ij}})\).

As in the proof of Theorem~\ref{thm:all_weights}, we view
\(\mathcal{L}(\mathbf{W}_{\mathcal{N}_{ij}},\mathbf{C}_{\mathcal{N}_{ij}})\)
as a Hilbert space equipped with the inner product induced by the product
expectation. The map
\[
f(\cdot)
\mapsto
\mathbb{E}_{\phi}\!\left[
f(\mathbf{W}_{\mathcal{N}_{ij}},\mathbf{C}_{\mathcal{N}_{ij}})
h(\mathbf{W}_{\mathcal{N}_{ij}})
\right]
\]
is a continuous linear functional on this Hilbert space. By the Riesz representation
theorem, its representer is unique. Therefore, the identification weight
\(\beta_{ij}\) can only be $\alpha_{ij}^{(comp)}(\phi)$.

\medskip
\noindent\textbf{Part (ii).}
We prove part~(ii)(a). Part~(ii)(b) follows as a corollary: if valid weights
\(\beta_{ij}(w,\phi)\) existed for all \(w\) with \(P_{\phi}(W_{ij}=w)>0\), then
\[
\beta_{ij}(\phi)
=
\sum_{w:\,P_{\phi}(W_{ij}=w)>0}
P_{\phi}(W_{ij}=w)\beta_{ij}(w,\phi)
\]
would be a valid weight for part~(ii)(a).

We prove part~(ii)(a) by contradiction. Suppose that there exists
\[
\beta_{ij}
\in
\mathcal{L}(\mathbf{W}_{\mathcal{N}_{ij}},\mathbf{C}_{\mathcal{N}_{ij}})
\]
such that
\[
\mathbb{E}\!\left[
Y_{ij}\beta_{ij}
\mid \mathbf{C}
\right]
=
\mathbb{E}_{\phi}\!\left[Y_{ij}\right]
\]
for all \(Y_{ij}\in\mathcal{L}(\mathbf{W}_{\mathcal{N}_{ij}})\). Repeating the argument
from part~(i), we obtain
\[
\mathbb{E}\!\left[
f(\mathbf{W}_{\mathcal{N}_{ij}},\mathbf{C}_{\mathcal{N}_{ij}})
\beta_{ij}
\right]
=
\mathbb{E}_{\phi}\!\left[
f(\mathbf{W}_{\mathcal{N}_{ij}},\mathbf{C}_{\mathcal{N}_{ij}})
\right]
\]
for all
\(f\in
\mathcal{L}(\mathbf{W}_{\mathcal{N}_{ij}},\mathbf{C}_{\mathcal{N}_{ij}})\).
Thus \(\beta_{ij}\) must be the Radon--Nikodym derivative of
\(P_{\phi}(\mathbf{W}_{\mathcal{N}_{ij}},\mathbf{C}_{\mathcal{N}_{ij}})\) with
respect to
\(P(\mathbf{W}_{\mathcal{N}_{ij}},\mathbf{C}_{\mathcal{N}_{ij}})\). However, when
Assumption~\ref{ass:measure}(ii) fails, this Radon--Nikodym derivative does not
exist. This contradiction proves part~(ii)(a). 
\end{proof}

\subsubsection{Proof of Proposition \ref{prop:examples_of_weights}}
\begin{proof}[Proof of Proposition \ref{prop:examples_of_weights}]
Fix $(i,j)$ and $\phi_c\in\{\phi_0,\phi_1\}$.

\medskip
\noindent\emph{(i) Marginal Radon--Nikodym derivative weights.}
Since $\alpha_{ij}(\phi_c)$ is $\sigma(\mathbf{W}_{\mathcal{N}_{ij}})$-measurable,
$\mathbb{E}[\alpha_{ij}(\phi_c)\mid\mathbf{W}_{\mathcal{N}_{ij}}]=\alpha_{ij}(\phi_c)$.
Hence the weights in~(i) satisfy the ``if'' condition of Theorem~\ref{thm:all_weights}
(with $h\equiv 1$ and $h=1(W_{ij}=w)/P_{\phi_c}(W_{ij}=w)$, respectively), and therefore achieve unbiased identification.

\medskip
\noindent\emph{(ii) Complete Radon--Nikodym derivative / cluster-agnostic weights.}
The weights in~(ii) are $\alpha_{ij}^{(\mathrm{comp})}(\phi_c)\cdot h$ with $h\equiv 1$ and
$h=1(W_{ij}=w)/P_{\phi_c}(W_{ij}=w)$, respectively.
They therefore satisfy the ``if'' condition of Theorem~\ref{thm:cluster_agnostic_weights}(i)
and achieve cluster-agnostic unbiased identification.

\medskip
\noindent\emph{(iii) Exposure-mapping-level Radon--Nikodym derivative weights.}
Under Assumption~\ref{ass:measure}(i), $P_{\phi_c}(\mathbf{W}_{\mathcal{N}_{ij}})\ll P(\mathbf{W}_{\mathcal{N}_{ij}})$,
which implies $P_{\phi_c}(\mathbf{e}(\mathbf{W}_{\mathcal{N}_{ij}}))\ll P(\mathbf{e}(\mathbf{W}_{\mathcal{N}_{ij}}))$.
Hence the Radon--Nikodym derivative
$\beta_{ij}(\phi_c)=dP_{\phi_c}(\mathbf{e}(\mathbf{W}_{\mathcal{N}_{ij}}))/dP(\mathbf{e}(\mathbf{W}_{\mathcal{N}_{ij}}))$
exists, is $\sigma(\mathbf{e}(\mathbf{W}_{\mathcal{N}_{ij}}))$-measurable, and satisfies
\begin{equation}\label{eq:change_measure_exposure}
\mathbb{E}\!\left[f\!\left(\mathbf{e}(\mathbf{W}_{\mathcal{N}_{ij}})\right)\beta_{ij}(\phi_c)\right]
=\mathbb{E}_{\phi_c}\!\left[f\!\left(\mathbf{e}(\mathbf{W}_{\mathcal{N}_{ij}})\right)\right]
\quad\text{for every }f\in\mathcal{L}(\mathbf{e}(\mathbf{W}_{\mathcal{N}_{ij}})).
\end{equation}
Applying~\eqref{eq:change_measure_exposure} with $f=Y_{ij}$ gives
$\mathbb{E}[Y_{ij}\beta_{ij}(\phi_c)]=\mathbb{E}_{\phi_c}[Y_{ij}]=\bar{Y}_{ij}(\phi_c)$, and applying~\eqref{eq:change_measure_exposure}
with $f(\mathbf{e})=Y_{ij}(\mathbf{e})\cdot 1(w_{ij}=w)/P_{\phi_c}(W_{ij}=w)$
yields $\mathbb{E}[Y_{ij}\beta_{ij}(w,\phi_c)]=\mathbb{E}_{\phi_c}[Y_{ij}\mid W_{ij}=w]=\bar{Y}_{ij}(w,\phi_c)$.

\medskip
\noindent\emph{(iv) IPT weights.}
Under $\phi_c$, $\mathbf W_{\mathcal N_{ij}}$ is drawn from the conditional
unit-level assignment distribution given that all clusters in
$\mathcal N_{ij}$ receive cluster-level treatment $c$. Hence, for any
feasible $\mathbf w_{\mathcal{N}_{ij}}$,
\[
P_{\phi_c}\!\left(\mathbf W_{\mathcal N_{ij}}=\mathbf w_{\mathcal{N}_{ij}}\right)
=
P\!\left(
\mathbf W_{\mathcal N_{ij}}=\mathbf w_{\mathcal{N}_{ij}}
\mid
\mathbf C_{\mathcal N_{ij}}=c\mathbf 1
\right).
\]
Then, for any $Y_{ij}\in\mathcal{L}(\mathbf{W}_{\mathcal{N}_{ij}})$, and any $h\in \mathcal{L}(\mathbf{W}_{\mathcal{N}_{ij}})$,
\begin{align*}
\mathbb{E}\!\left[Y_{ij}h(\mathbf{W}_{\mathcal{N}_{ij}})\,\frac{1(\mathbf{C}_{\mathcal{N}_{ij}}=c\mathbf{1})}{P(\mathbf{C}_{\mathcal{N}_{ij}}=c\mathbf{1})}\right]
&=\sum_{\mathbf{w}_{\mathcal{N}_{ij}}}P(\mathbf{W}_{\mathcal{N}_{ij}}=\mathbf{w}_{\mathcal{N}_{ij}},\mathbf{C}_{\mathcal{N}_{ij}}=c\mathbf{1})\frac{1}{P(\mathbf{C}_{\mathcal{N}_{ij}}=c\mathbf{1})}Y_{ij}(\mathbf{w}_{\mathcal{N}_{ij}})h(\mathbf{w}_{\mathcal{N}_{ij}})\\
&=\sum_{\mathbf{w}_{\mathcal{N}_{ij}}}P(\mathbf{W}_{\mathcal{N}_{ij}}=\mathbf{w}_{\mathcal{N}_{ij}}|\mathbf{C}_{\mathcal{N}_{ij}}=c\mathbf{1})Y_{ij}(\mathbf{w}_{\mathcal{N}_{ij}})h(\mathbf{w}_{\mathcal{N}_{ij}})\\
&=\sum_{\mathbf{w}_{\mathcal{N}_{ij}}}P_{\phi_c}(\mathbf{W}_{\mathcal{N}_{ij}}=\mathbf{w}_{\mathcal{N}_{ij}})Y_{ij}(\mathbf{w}_{\mathcal{N}_{ij}})h(\mathbf{w}_{\mathcal{N}_{ij}})\\
&=\mathbb{E}_{\phi_c}\!\left[Y_{ij}h(\mathbf{W}_{\mathcal{N}_{ij}})\right].
\end{align*}
By plugging in $h\equiv 1$ and
$h=1(W_{ij}=w)/P_{\phi_c}(W_{ij}=w)$ respectively, we obtain the unbiasedness of $\beta_{ij}(\phi_c)$ and of the product weight $\frac{1(\mathbf{C}_{\mathcal{N}_{ij}}=c\mathbf{1})}{P(\mathbf{C}_{\mathcal{N}_{ij}}=c\mathbf{1})}\cdot\frac{1(W_{ij}=w)}{P_{\phi_c}(W_{ij}=w)}$. Marginalizing the identity above to the single coordinate $W_{ij}$ gives $P_{\phi_c}(W_{ij}=w)=P(W_{ij}=w\mid\mathbf{C}_{\mathcal{N}_{ij}}=c\mathbf{1})$, hence $P(\mathbf{C}_{\mathcal{N}_{ij}}=c\mathbf{1},\,W_{ij}=w)=P(\mathbf{C}_{\mathcal{N}_{ij}}=c\mathbf{1})\,P_{\phi_c}(W_{ij}=w)$, so this product weight coincides with $\beta_{ij}(w,\phi_c)$ in the proposition statement, completing the proof.
\end{proof}


\subsection{Proof of asymptotic theory}
This subsection establishes the convergence rate and asymptotic normality of our LW estimators. In some proofs, we first state auxiliary lemmas that simplify the argument; these are either drawn directly from the literature and cited by their source, or are new results whose proofs are collected in Section~\ref{subsec:aux_lemmas} below.

\subsubsection{Proof of Theorem \ref{thm:variance_order}}
The proof relies on the following three auxiliary lemmas.

\begin{lemma}[Cross-cluster dependency set size]\label{lem:complete_neighborhood_size}
Under Assumptions~\ref{ass:interference} and~\ref{ass:CLT_regularity}(ii)--(iv), the dependency set
\[
\Lambda_{ij} := \bigl\{(i',j')\in\mathcal{P} : \mathcal{N}_{ij}^{(cl)}\cap\mathcal{N}_{i'j'}^{(cl)}\neq\varnothing\bigr\},
\]
which comprises all units $(i',j')$ whose cluster-level interference neighborhood overlaps with that of $(i,j)$, satisfies
\[
\sup_{(i,j)\in\mathcal{P}}|\Lambda_{ij}| = O\!\left(\frac{N}{n}\right).
\]
\end{lemma}

\begin{lemma}[Covariance bounds via the mixing coefficient and via the joint support]
\label{lem:mixing_coef}
Let $\mathcal{F}_1$ and $\mathcal{F}_2$ be two $\sigma$-fields on a probability
space with finitely many states, and measure their dependence by the
$\psi$-mixing coefficient
\[
\psi(\mathcal{F}_1,\mathcal{F}_2)
:=
\sup_{\substack{A\in\mathcal{F}_1,\;B\in\mathcal{F}_2\\ P(A)P(B)>0}}
\left|\frac{P(A\cap B)}{P(A)P(B)}-1\right| ,
\]
which compares the joint probability with the product of the marginal
probabilities multiplicatively rather than additively. Then, for any
$\mathcal{F}_1$-measurable $f$ and any $\mathcal{F}_2$-measurable $g$,
\begin{enumerate}
\item[(i)] if $f$ and $g$ are integrable,
\[
\bigl|\mathrm{Cov}(f,g)\bigr|
\leq
\psi(\mathcal{F}_1,\mathcal{F}_2)\,\mathbb{E}|f|\,\mathbb{E}|g| ;
\]
\item[(ii)] if $f$ and $g$ are bounded,
\[
\bigl|\mathrm{Cov}(f,g)\bigr|
\leq
\|f\|_{\infty}\|g\|_{\infty}
\Bigl(P(fg\neq 0)+P(f\neq 0)P(g\neq 0)\Bigr).
\]
\end{enumerate}
\end{lemma}

\begin{lemma}[$\psi$-mixing coefficient for complete randomization]
\label{lem:mixing_coef_complete_randomization}
Let $\{W_i\}_{i=1}^{N}$ be treatment assignments generated by a stratified complete
randomization scheme: the index set $\{1,\ldots,N\}$ is partitioned into $K$ disjoint
strata $\mathcal{I}_1,\ldots,\mathcal{I}_K$ (where $K$ may grow with $N$), and within each stratum $k$
the subvector $\mathbf{W}_{\mathcal{I}_k}$ is drawn uniformly from all binary vectors
summing to $m_k$, independently across strata.
Suppose there exist constants $p>0$ and $C_S<\infty$, independent of $N$, and a
sequence $M=M(N)\to\infty$, such that
\[
  \min_{k}|\mathcal{I}_k|=\Omega(M),
  \qquad
  \frac{m_k}{|\mathcal{I}_k|}\in\{0,1\}\cup[p,\,1-p]
  \quad\text{for all }k.
\]
Let $S_1,S_2\subseteq\{1,\ldots,N\}$ be two disjoint index sets with
$|S_1|\vee|S_2|\leq C_S$, and let
$\mathcal{F}_t:=\sigma\bigl(\{W_i:i\in S_t\}\bigr)$ for $t=1,2$. Then, with
$\psi$ the $\psi$-mixing coefficient of Lemma~\ref{lem:mixing_coef},
\[
  \psi(\mathcal{F}_1,\mathcal{F}_2)
  \leq
  \frac{C}{M},
\]
where $C<\infty$ depends only on $C_S$, on $p$, and on the implied constant in
$\min_k|\mathcal I_k|=\Omega(M)$. In particular $C$ does not depend on $N$, on
the strata, or on the choice of $S_1$ and $S_2$, so both bounds hold uniformly
over all such pairs of index sets.
\end{lemma}

\begin{proof}[Proof of Theorem \ref{thm:variance_order}]
Since the arguments for $\boldsymbol\Sigma_{\mathrm{marg}}$ and $\boldsymbol\Sigma$ are identical
(replacing $\mathbf{B}_{\mathrm{marg}}$ and $\boldsymbol\tau_{ij,\mathrm{marg}}$ by $\mathbf{B}$ and $\boldsymbol\tau_{ij}$),
we prove only the bound for $\boldsymbol\Sigma$.

Define the vector-valued summands
$\mathbf{V}_{ij}:=\frac{g_i}{N_i}\mathbf{B}_{(*,ij)}\circ(Y_{ij}\mathbf{1}-\boldsymbol\tau)$,
so $\boldsymbol\Sigma=\mathrm{Var}\!\bigl(\sum_{(i,j)}\mathbf{V}_{ij}\bigr)=\sum_{(i,j)}\sum_{(i',j')}\mathrm{Cov}(\mathbf{V}_{ij},\mathbf{V}_{i'j'})$.
Throughout, $\|\mathbf{V}_{ij}\|_\infty=\max_k\operatorname{ess\,sup}|(\mathbf{V}_{ij})_k|$ for random vectors and $\|\mathbf{M}\|_\infty=\max_{k,l}|M_{kl}|$ for deterministic matrices; the operator norm satisfies $\|\mathbf{M}\|\le d\|\mathbf{M}\|_\infty$ for $d\times d$ matrices.
Assumptions~\ref{ass:CLT_regularity}(iii)--(iv) give $\|\mathbf{V}_{ij}\|_\infty\le C\,g_i/N_i$ a.s.,
and together with $N_i\ge N/(C_N n)$ from Assumption~\ref{ass:CLT_regularity}(v) and $\max_i g_i=O(1/n)$ from Assumption~\ref{ass:CLT_regularity}(vi),
\[
\sup_{(i,j)}\|\mathbf{V}_{ij}\|_\infty=O\!\left(\frac{1}{N}\right),
\qquad
\sum_{(i,j)}\|\mathbf{V}_{ij}\|_\infty\le C\sum_i g_i=O(1).
\]

\medskip
\noindent\textbf{Part (i) under Assumption~\ref{ass:independent_rand}(i).}
Since $\mathbf{V}_{ij}$ depends on $(\mathbf{C}_{\mathcal{N}_{ij}^{(cl)}},\mathbf{W}_{\mathcal{N}_{ij}})$, the law of total covariance gives
\[
\mathrm{Cov}(\mathbf{V}_{ij},\mathbf{V}_{i'j'})
=\mathrm{Cov}\!\bigl(\mathbb{E}[\mathbf{V}_{ij}\mid\mathbf{C}],\,\mathbb{E}[\mathbf{V}_{i'j'}\mid\mathbf{C}]\bigr)
+\mathbb{E}\!\bigl[\mathrm{Cov}(\mathbf{V}_{ij},\mathbf{V}_{i'j'}\mid\mathbf{C})\bigr].
\]
For $(i',j')\notin\Lambda_{ij}$, i.e., $\mathcal{N}_{ij}^{(cl)}\cap\mathcal{N}_{i'j'}^{(cl)}=\varnothing$, both terms vanish.
The second term is zero because, given $\mathbf{C}$, Assumption~\ref{ass:cluster_assignment}(ii) makes $\mathbf{W}_i\indep\mathbf{W}_{i'}$ for all $i\neq i'$; since $\mathcal{N}_{ij}$ and $\mathcal{N}_{i'j'}$ lie in disjoint clusters, $\mathbf{W}_{\mathcal{N}_{ij}}\indep\mathbf{W}_{\mathcal{N}_{i'j'}}\mid\mathbf{C}$, so $\mathrm{Cov}(\mathbf{V}_{ij},\mathbf{V}_{i'j'}\mid\mathbf{C})=\mathbf{0}$ a.s.
The first term is zero because $\mathbb{E}[\mathbf{V}_{ij}\mid\mathbf{C}]$ is a function of $\mathbf{C}_{\mathcal{N}_{ij}^{(cl)}}$ alone, and by Assumption~\ref{ass:independent_rand}(i) the $C_i$ are mutually independent; disjointness of the two cluster index sets therefore implies $\mathbf{C}_{\mathcal{N}_{ij}^{(cl)}}\indep\mathbf{C}_{\mathcal{N}_{i'j'}^{(cl)}}$, hence $\mathrm{Cov}(\mathbb{E}[\mathbf{V}_{ij}\mid\mathbf{C}],\mathbb{E}[\mathbf{V}_{i'j'}\mid\mathbf{C}])=\mathbf{0}$.
Only pairs $(i',j')\in\Lambda_{ij}$ contribute to $\boldsymbol\Sigma$. Bounding each entry by $\|\mathbf{V}_{ij}\|_\infty\|\mathbf{V}_{i'j'}\|_\infty$ and applying Lemma~\ref{lem:complete_neighborhood_size},
\[
\|\boldsymbol\Sigma\|_\infty
\le\sup_{(i,j)}|\Lambda_{ij}|\cdot\sup_{(i',j')}\|\mathbf{V}_{i'j'}\|_\infty\cdot\sum_{(i,j)}\|\mathbf{V}_{ij}\|_\infty
=O\!\left(\frac{N}{n}\right)\cdot O\!\left(\frac{1}{N}\right)\cdot O(1)=O\!\left(\frac{1}{n}\right).
\]
Since $\boldsymbol\Sigma$ has fixed dimension $d$, $\|\boldsymbol\Sigma\|\le d\|\boldsymbol\Sigma\|_\infty=O(n^{-1})$.

\medskip
\noindent\textbf{Part (i) under Assumption~\ref{ass:complete_rand}(i).}
The cluster-level assignments $\{C_i\}_{i=1}^n$ satisfy stratified complete randomization with strata $\{\mathcal{I}_k\}$ meeting the conditions of Lemma~\ref{lem:mixing_coef_complete_randomization} with population size~$n$.
The same law-of-total-covariance decomposition applies:
\[
\mathrm{Cov}(\mathbf{V}_{ij},\mathbf{V}_{i'j'})
=\mathrm{Cov}\!\bigl(\mathbb{E}[\mathbf{V}_{ij}\mid\mathbf{C}],\,\mathbb{E}[\mathbf{V}_{i'j'}\mid\mathbf{C}]\bigr)
+\mathbb{E}\!\bigl[\mathrm{Cov}(\mathbf{V}_{ij},\mathbf{V}_{i'j'}\mid\mathbf{C})\bigr].
\]
For $(i',j')\notin\Lambda_{ij}$, the second term vanishes by the same reasoning as in the Bernoulli case: given $\mathbf{C}$, Assumption~\ref{ass:cluster_assignment}(ii) implies $\mathbf{W}_{\mathcal{N}_{ij}}\indep\mathbf{W}_{\mathcal{N}_{i'j'}}\mid\mathbf{C}$, so $\mathrm{Cov}(\mathbf{V}_{ij},\mathbf{V}_{i'j'}\mid\mathbf{C})=\mathbf{0}$ a.s.
The first term, however, need not vanish for pairs $(i',j')\notin\Lambda_{ij}$. This is because, under complete randomization, cluster-level assignments within the same stratum are not independent but only weakly dependent. Therefore, $\mathbb{E}[\mathbf{V}_{ij}\mid\mathbf{C}]$ and $\mathbb{E}[\mathbf{V}_{i'j'}\mid\mathbf{C}]$ may remain weakly dependent for pairs in the same stratum.
For such pairs, Lemma~\ref{lem:mixing_coef_complete_randomization} applies with $M\asymp n$ and $C_S=2C_{\mathcal{N}^{(cl)}}$, giving a $\psi$-mixing coefficient $\psi=O(1/n)$ that is uniform over the pairs, and then Lemma~\ref{lem:mixing_coef}(i) yields
\[
\bigl\|\mathrm{Cov}\!\bigl(\mathbb{E}[\mathbf{V}_{ij}\mid\mathbf{C}],\,\mathbb{E}[\mathbf{V}_{i'j'}\mid\mathbf{C}]\bigr)\bigr\|_\infty
\le \bigl\|\mathbb{E}[\mathbf{V}_{ij}\mid\mathbf{C}]\bigr\|_\infty\,
     \bigl\|\mathbb{E}[\mathbf{V}_{i'j'}\mid\mathbf{C}]\bigr\|_\infty\cdot O\!\left(\frac{1}{n}\right)
\le \|\mathbf{V}_{ij}\|_\infty\,\|\mathbf{V}_{i'j'}\|_\infty\cdot O\!\left(\frac{1}{n}\right),
\]
using $\mathbb{E}|\cdot|\le\|\cdot\|_\infty$ entrywise and the fact that conditional expectation is a contraction for the sup-norm.
Summing over all non-$\Lambda_{ij}$ pairs,
\[
\Bigl\|\sum_{(i,j)}\sum_{(i',j')\notin\Lambda_{ij}}\mathrm{Cov}(\mathbf{V}_{ij},\mathbf{V}_{i'j'})\Bigr\|_\infty
\le O\!\left(\frac{1}{n}\right)\cdot\Bigl(\sum_{(i,j)}\|\mathbf{V}_{ij}\|_\infty\Bigr)^2=O\!\left(\frac{1}{n}\right).
\]
For pairs $(i',j')\in\Lambda_{ij}$, Lemma~\ref{lem:complete_neighborhood_size} gives $|\Lambda_{ij}|=O(N/n)$, and bounding both terms of the variance decomposition directly by $\|\mathbf{V}_{ij}\|_\infty\|\mathbf{V}_{i'j'}\|_\infty$ yields $O(1/n)$ by the same calculation as in the Bernoulli case.
Combining both contributions and using fixed dimension, $\|\boldsymbol\Sigma\|=O(n^{-1})$.

\medskip
\noindent\textbf{Part (ii) under Assumption~\ref{ass:independent_rand}(ii).}
Under Assumption~\ref{ass:weight_cluster_agnostic} and Assumption~\ref{ass:CLT_regularity}(i), each $\mathbf{V}_{ij}$ is $\sigma(\mathbf{W}_{\mathcal{N}_{ij}})$-measurable with $|\mathcal{N}_{ij}|\le C_{\mathcal{N}}$.
Theorem~\ref{thm:cluster_agnostic_weights} with $h\equiv 1$ gives 
\[
\mathbb{E}[\mathbf{V}_{ij}\mid\mathbf{C}]
=
\frac{g_i}{N_i}\mathbb{E}[\mathbf{B}_{(*,ij)}\circ(Y_{ij}\mathbf{1})\mid\mathbf{C}]-\frac{g_i}{N_i}\mathbb{E}[\mathbf{B}_{(*,ij)}\mid\mathbf{C}]\circ \boldsymbol\tau
=
\frac{g_i}{N_i}(\boldsymbol\tau_{ij}-\boldsymbol\tau),
\]
which is a constant. Consequently, $\mathbb{E}\bigl[\sum_{(i,j)}\mathbf{V}_{ij}\mid\mathbf{C}\bigr]=\sum_{(i,j)}\frac{g_i}{N_i}(\boldsymbol\tau_{ij}-\boldsymbol\tau)$ is a deterministic vector, so its variance vanishes, and the law of total variance gives
\[
\boldsymbol\Sigma
=\mathbb{E}\Bigl[\mathrm{Var}\Bigl(\sum_{(i,j)}\mathbf{V}_{ij}\Bigm|\mathbf{C}\Bigr)\Bigr]
+\mathrm{Var}\Bigl(\mathbb{E}\Bigl[\sum_{(i,j)}\mathbf{V}_{ij}\Bigm|\mathbf{C}\Bigr]\Bigr)
=\mathbb{E}[\boldsymbol\Sigma(\mathbf{C})],
\]
where $\boldsymbol\Sigma(\cdot)$ is the conditional covariance matrix defined in the statement of the theorem. Since Theorem~\ref{thm:cluster_agnostic_weights} applies to $\mathbf{B}_{\mathrm{marg}}$ and $\boldsymbol\tau_{\mathrm{marg}}$ in exactly the same way, the identical argument yields $\boldsymbol\Sigma_{\mathrm{marg}}=\mathbb{E}[\boldsymbol\Sigma_{\mathrm{marg}}(\mathbf{C})]$. This proves the decomposition claimed in part (ii). Moreover, expanding the conditional variance entrywise, the decomposition reduces $\boldsymbol\Sigma$ to
\[
\boldsymbol\Sigma=\sum_{(i,j)}\sum_{(i',j')}\mathbb{E}\bigl[\mathrm{Cov}(\mathbf{V}_{ij},\mathbf{V}_{i'j'}\mid\mathbf{C})\bigr].
\]
By Assumptions~\ref{ass:independent_rand}(ii) and~\ref{ass:cluster_assignment}(ii), given $\mathbf{C}$ the assignments $\{W_{ij}\}$ are mutually independent, so $\mathrm{Cov}(\mathbf{V}_{ij},\mathbf{V}_{i'j'}\mid\mathbf{C})=\mathbf{0}$ a.s.\ whenever $\mathcal{N}_{ij}\cap\mathcal{N}_{i'j'}=\varnothing$.
Since $|\mathcal{N}_{ij}|\le C_{\mathcal{N}}$, at most $C_{\mathcal{N}}^2$ partners $(i',j')$ per $(i,j)$ can satisfy $\mathcal{N}_{ij}\cap\mathcal{N}_{i'j'}\neq\varnothing$. Bounding each nonzero entry by $\|\mathbf{V}_{ij}\|_\infty\|\mathbf{V}_{i'j'}\|_\infty$,
\[
\|\boldsymbol\Sigma\|_\infty
\le C_{\mathcal{N}}^2\cdot\sup_{(i',j')}\|\mathbf{V}_{i'j'}\|_\infty\cdot\sum_{(i,j)}\|\mathbf{V}_{ij}\|_\infty
=C_{\mathcal{N}}^2\cdot O\!\left(\frac{1}{N}\right)\cdot O(1)=O\!\left(\frac{1}{N}\right).
\]
Hence $\|\boldsymbol\Sigma\|=O(N^{-1})$.

\medskip
\noindent\textbf{Part (ii) under Assumption~\ref{ass:complete_rand}(ii).}
By definition of the cluster-agnostic weights, the computation above again gives the constant $\mathbb{E}[\mathbf{V}_{ij}\mid\mathbf{C}]=\frac{g_i}{N_i}(\boldsymbol\tau_{ij}-\boldsymbol\tau)$. Hence the decompositions $\boldsymbol\Sigma=\mathbb{E}[\boldsymbol\Sigma(\mathbf{C})]$ and $\boldsymbol\Sigma_{\mathrm{marg}}=\mathbb{E}[\boldsymbol\Sigma_{\mathrm{marg}}(\mathbf{C})]$ continue to hold in this case, and the same entrywise reduction applies: $\boldsymbol\Sigma=\sum_{(i,j),(i',j')}\mathbb{E}[\mathrm{Cov}(\mathbf{V}_{ij},\mathbf{V}_{i'j'}\mid\mathbf{C})]$.
By Assumption~\ref{ass:cluster_assignment}(ii), within-cluster assignments are independent across clusters conditional on $\mathbf{C}$, so $\mathrm{Cov}(\mathbf{V}_{ij},\mathbf{V}_{i'j'}\mid\mathbf{C})=\mathbf{0}$ a.s.\ whenever $\mathcal{N}_{ij}^{(cl)}\cap\mathcal{N}_{i'j'}^{(cl)}=\varnothing$.
For the remaining pairs $(i',j')\in\Lambda_{ij}$, we distinguish two sub-cases.

If $\mathcal{N}_{ij}\cap\mathcal{N}_{i'j'}\neq\varnothing$: there are at most $C_{\mathcal{N}}^2$ such pairs per $(i,j)$, and bounding each entry directly by $\|\mathbf{V}_{ij}\|_\infty\|\mathbf{V}_{i'j'}\|_\infty$ gives a contribution of $O(1/N)$, as in the Bernoulli case.

If $\mathcal{N}_{ij}\cap\mathcal{N}_{i'j'}=\varnothing$ but $(i',j')\in\Lambda_{ij}$:
conditional on $\mathbf{C}$, the full unit-level treatment vector $\mathbf{W}$ follows a stratified complete randomization on all $N$ units, with strata $\{\mathcal{J}_{ik}:1\le i\le n,\,1\le k\le K_i\}$.
These strata are independent of one another (both across groups within each cluster by Assumption~\ref{ass:complete_rand}(ii) and across clusters by Assumption~\ref{ass:cluster_assignment}(ii)), and every stratum satisfies $|\mathcal{J}_{ik}|=\Omega(N_i)=\Omega(N/n)$ with treatment fraction in $\{0,1\}\cup[p,1-p]$.
Applying Lemma~\ref{lem:mixing_coef_complete_randomization} to this full conditional scheme with $S_1=\mathcal{N}_{ij}$, $S_2=\mathcal{N}_{i'j'}$ (disjoint, each of size at most $C_{\mathcal{N}}$) and $M=\Theta(N/n)$ gives a $\psi$-mixing coefficient $\psi=O(M^{-1})=O(n/N)$, uniformly over these pairs.
Applying Lemma~\ref{lem:mixing_coef}(i) conditionally on $\mathbf{C}$ and taking expectations,
\[
\bigl\|\mathbb{E}\bigl[\mathrm{Cov}(\mathbf{V}_{ij},\mathbf{V}_{i'j'}\mid\mathbf{C})\bigr]\bigr\|_\infty
\le \|\mathbf{V}_{ij}\|_\infty\|\mathbf{V}_{i'j'}\|_\infty\cdot O\!\left(\frac{n}{N}\right)=O\!\left(\frac{n}{N^3}\right)
\]
uniformly for all units.
Using $|\Lambda_{ij}|=O(N/n)$ from Lemma~\ref{lem:complete_neighborhood_size} and summing over all such pairs,
\[
\Bigl\|\sum_{(i,j)}\!\!\sum_{\substack{(i',j')\in\Lambda_{ij}\\\mathcal{N}_{ij}\cap\mathcal{N}_{i'j'}=\varnothing}}\!\!
\mathbb{E}\bigl[\mathrm{Cov}(\mathbf{V}_{ij},\mathbf{V}_{i'j'}\mid\mathbf{C})\bigr]\Bigr\|_\infty
\le N\cdot O\!\left(\frac{N}{n}\right)\cdot O\!\left(\frac{n}{N^3}\right)
=O\!\left(\frac{1}{N}\right).
\]
Combining all contributions, $\|\boldsymbol\Sigma\|=O(N^{-1})$.
\end{proof}

\subsubsection{Proof of Proposition \ref{prop:scalar_clt}}

\begin{proof}[Proof of Proposition \ref{prop:scalar_clt}]

By Theorem~3.1 of \cite{Ross2011}, it suffices to bound
\[
\bigl|\mathbb{E}[f'(S)-Sf(S)]\bigr|
\]
for functions \(f\) satisfying
\[
\|f\|_{\infty},\ \|f''\|_{\infty}\le 2,
\qquad
\|f'\|_{\infty}\le \sqrt{2/\pi}.
\]
Now
\[
Sf(S)=\frac1\sigma\sum_{(u,v)}X_{uv}f(S),
\]
and the independence condition gives us
\[
\begin{aligned}
\bigl|\mathbb{E}[f'(S)-Sf(S)]\bigr|
&=
\left|
\mathbb{E}\left[
f'(S)-\frac1\sigma\sum_{(u,v)}X_{uv}f(S)
\right]
\right| \\
&\le
\left|
\mathbb{E}\left[
\frac1\sigma\sum_{(u,v)}
X_{uv}\Bigl(f(S)-f(\tilde{S}^{(uv)})-(S-\tilde{S}^{(uv)})f'(S)\Bigr)
\right]
\right| \\
&\quad+
\left|
\mathbb{E}\left[
f'(S)
\left(\frac1\sigma\sum_{(u,v)}X_{uv}(S-\tilde{S}^{(uv)})-1\right)
\right]
\right|,
\end{aligned}
\]
where 
\[\tilde{S}^{(uv)}:=\sigma^{-1} \sum_{(i,j)} \tilde{X}^{(uv)}_{ij}.\]
Thus
\[
\bigl|\mathbb{E}[f'(S)-Sf(S)]\bigr|
\le T_1+T_2,
\]
where
\[
T_1:=
\left|
\mathbb{E}\left[
\frac1\sigma\sum_{(u,v)}
X_{uv}\Bigl(f(S)-f(\tilde{S}^{(uv)})-(S-\tilde{S}^{(uv)})f'(S)\Bigr)
\right]
\right|
\]
and
\[
T_2:=
\left|
\mathbb{E}\left[
f'(S)
\left(\frac1\sigma\sum_{(u,v)}X_{uv}(S-\tilde{S}^{(uv)})-1\right)
\right]
\right|.
\]
We first bound \(T_1\). By Taylor's theorem,
\[
\bigl|f(S)-f(\tilde{S}^{(uv)})-(S-\tilde{S}^{(uv)})f'(S)\bigr|
\le
\frac12\|f''\|_{\infty}|S-\tilde{S}^{(uv)}|^2
\le
|S-\tilde{S}^{(uv)}|^2,
\]
since \(\|f''\|_{\infty}\le 2\). Therefore,
\begin{align*}
T_1
&=
O\left(\frac1\sigma\sum_{(u,v)}\mathbb{E}\Bigl[|X_{uv}|\,|S-\tilde{S}^{(uv)}|^2\Bigr]\right)
=
O\left(\frac{1}{\sigma N}\sum_{(u,v)}\mathbb{E}\Bigl[|S-\tilde{S}^{(uv)}|^2\Bigr]\right)\\
&=
O\left(\frac{1}{\sigma^3 N}\sum_{(u,v)}\left\|\sum_{(i,j)}(X_{ij}-\tilde{X}_{ij}^{(uv)})\right\|_{\infty}^2\right).
\end{align*}
Next we consider \(T_2\). Since \(\|f'\|_{\infty}\le \sqrt{2/\pi}\) and $\mathbb{E}\left[\frac1\sigma\sum_{(u,v)}X_{uv}(S-\tilde{S}_{uv})\right]=1$, we can apply Cauchy's inequality to get
\[
T_2
=
O\left(
\sqrt{
\mathrm{Var}\!\left(
\frac1\sigma\sum_{(u,v)}X_{uv}(S-\tilde{S}^{(uv)})
\right)}\right).
\]
Rather than treating each discrepancy \(S-\tilde S^{(uv)}\) as a single random
variable, we resolve it into its individual summands. By the definition of
\(\tilde S^{(uv)}\),
\[
\frac1\sigma\sum_{(u,v)}X_{uv}\bigl(S-\tilde{S}^{(uv)}\bigr)
=
\frac{1}{\sigma^{2}}\sum_{(u,v)\in\mathcal P}\sum_{(i,j)\in\mathcal P}
X_{uv}\bigl(X_{ij}-\tilde X_{ij}^{(uv)}\bigr),
\]
so that the variance decomposes over quadruples of indices,
\[
\mathrm{Var}\!\left(
\frac1\sigma\sum_{(u,v)}X_{uv}(S-\tilde{S}^{(uv)})
\right)
=
\frac{1}{\sigma^{4}}
\sum_{\substack{\bigl((u,v),(u',v')\bigr)\in\mathcal P^2\\ \bigl((i,j),(i',j')\bigr)\in\mathcal P^2}}
\mathrm{Cov}\Bigl(
X_{uv}\bigl(X_{ij}-\tilde X_{ij}^{(uv)}\bigr),\;
X_{u'v'}\bigl(X_{i'j'}-\tilde X_{i'j'}^{(u'v')}\bigr)
\Bigr).
\]
This finer decomposition is what the coefficients
\(\psi_{ij,i'j'}^{(uv),(u'v')}\) are built for: the first argument of each
covariance is measurable with respect to
\(\sigma\bigl(X_{uv},D_{ij}^{(uv)}\bigr)\) and the second with respect to
\(\sigma\bigl(X_{u'v'},D_{i'j'}^{(u'v')}\bigr)\), which are exactly the two
\(\sigma\)-fields entering the definition of
\(\psi_{ij,i'j'}^{(uv),(u'v')}\). Both parts of Lemma~\ref{lem:mixing_coef}
therefore apply to each term, with
\(f=X_{uv}D_{ij}^{(uv)}\) and \(g=X_{u'v'}D_{i'j'}^{(u'v')}\). Part~(i) gives
\[
\bigl|\mathrm{Cov}(f,g)\bigr|
\leq
\psi_{ij,i'j'}^{(uv),(u'v')}\,
\mathbb{E}\bigl|X_{uv}D_{ij}^{(uv)}\bigr|\,
\mathbb{E}\bigl|X_{u'v'}D_{i'j'}^{(u'v')}\bigr| ,
\]
and part~(ii), noting that \(fg\neq 0\) forces
\(D_{ij}^{(uv)}D_{i'j'}^{(u'v')}\neq 0\) and that \(f\neq 0\) forces
\(D_{ij}^{(uv)}\neq 0\), gives
\[
\bigl|\mathrm{Cov}(f,g)\bigr|
\leq
\bigl\|X_{uv}D_{ij}^{(uv)}\bigr\|_{\infty}
\bigl\|X_{u'v'}D_{i'j'}^{(u'v')}\bigr\|_{\infty}
\Bigl(
P\bigl(D_{ij}^{(uv)}D_{i'j'}^{(u'v')}\neq 0\bigr)
+
P\bigl(D_{ij}^{(uv)}\neq 0\bigr)P\bigl(D_{i'j'}^{(u'v')}\neq 0\bigr)
\Bigr).
\]
Condition~(ii) of the proposition gives \(|X_{uv}|\le C/N\) almost surely, so
the factor \(X_{uv}\) may be pulled out of each norm at the cost of \(C/N\), and
likewise for \(X_{u'v'}\). Both displays therefore reduce to
\(C^{2}N^{-2}\) times \(T^{\mathrm{mix}}_{ij,i'j'}\) and
\(T^{\mathrm{sup}}_{ij,i'j'}\) respectively, and taking whichever is smaller for
each quadruple,
\[
\mathrm{Var}\!\left(
\frac1\sigma\sum_{(u,v)}X_{uv}(S-\tilde{S}^{(uv)})
\right)
=
O\!\left(
\frac{1}{\sigma^{4}N^{2}}
\sum_{\substack{\bigl((u,v),(u',v')\bigr)\in\mathcal P^2\\ \bigl((i,j),(i',j')\bigr)\in\mathcal P^2}}
T_{ij,i'j'}^{(uv),(u'v')}
\right)
=
O\bigl(\mathcal E_2\bigr),
\]
whence \(T_2=O\bigl(\sqrt{\mathcal E_2}\bigr)\).

Combining the bounds for \(T_1\) and \(T_2\) and recalling the error functionals
\(\mathcal E_1\) and \(\mathcal E_2\) of the statement, we conclude that
\[
\bigl|\mathbb{E}[f'(S)-Sf(S)]\bigr|
=
O\Bigl(\mathcal E_1+\sqrt{\mathcal E_2}\Bigr).
\]
Applying Theorem~3.1 of \cite{Ross2011} finishes the proof.
\end{proof}

\subsubsection{Proof of Proposition \ref{prop:coupling}}

\begin{proof}[Proof of Proposition~\ref{prop:coupling}]
We prove the four parts separately. By construction (the coupling of Section~\ref{sec:technical}), the coupling $(\tilde{\mathbf W}^{(uv)},\tilde{\mathbf C}^{(uv)})$ agrees with $(\mathbf W,\mathbf C)$ outside $\mathcal M_{uv}^{(cl)}$.

\vspace{0.5em}
\noindent\textbf{Part (i).}
Fix $(u,v)\in\mathcal P$. For each stratum $k$, write
\[
A_k := \mathcal N_{uv}^{(cl)}\cap\mathcal I_k,\quad
B_k := \mathcal I_k\setminus\mathcal N_{uv}^{(cl)},\quad
a_k:=|A_k|,\quad b_k:=|B_k|.
\]
We begin with a calculation that will be used to prove both claims. Fix any
$\mathbf a\in\{0,1\}^{A_k}$ with $s:=\sum_{i}a_i$,
$\mathbf b\in\{0,1\}^{B_k}$ with $t:=\sum_{i}b_i=I_k-s$, and
$\mathbf a_0\in\{0,1\}^{A_k}$ with $s_0:=\sum_{i}a_{0,i}$.
Set $\Delta:=s-s_0$ and $t_0:=t+\Delta$. Without loss of generality, we assume $\Delta\geq 0$ in the subsequent analysis; the case $\Delta<0$ is symmetric. In each coupling, we first generate $\tilde{\mathbf C}_{A_k}^{(uv)}$ from
the marginal distribution of $\mathbf C_{A_k}$ independently of
$(\mathbf C,\mathbf W)$, and then generate $\tilde{\mathbf C}_{B_k}^{(uv)}$
by zeroing out exactly $\Delta\ge 0$ uniformly chosen treated entries of
$\mathbf C_{B_k}$. Consider all possible indicator vectors $\mathbf d$ such
that $\mathbf C_{B_k}=\mathbf d$ is compatible with
$\tilde{\mathbf C}_{B_k}^{(uv)}=\mathbf b$,
$\tilde{\mathbf C}_{A_k}^{(uv)}=\mathbf a$, and
$\mathbf C_{A_k}=\mathbf a_0$. Such a vector must take the form
$\mathbf d=\mathbf b+\boldsymbol\rho$, where
$\boldsymbol\rho$ is an indicator vector satisfying
$\operatorname{supp}(\boldsymbol\rho)\subseteq\{i:b_i=0\}$ and
$\sum_{i}\rho_i=\Delta$. Hence there are
$\binom{b_k-t}{\Delta}$ possible values of $\mathbf d$ in total. For each
such $\mathbf d$, the conditional probability of obtaining
$\tilde{\mathbf C}_{B_k}^{(uv)}=\mathbf b$ from
$\mathbf C_{B_k}=\mathbf d$ is $1/\binom{t_0}{\Delta}$. Therefore,
\begin{align}\label{eq:key_coupling}
&P\!\bigl(\tilde{\mathbf C}_{B_k}^{(uv)}=\mathbf b
\mid\tilde{\mathbf C}_{A_k}^{(uv)}=\mathbf a,\mathbf C_{A_k}=\mathbf a_0\bigr)
\notag\\
=&
\sum_{\mathbf d}
P\!\bigl(\mathbf C_{B_k}=\mathbf d
\mid\mathbf C_{A_k}=\mathbf a_0\bigr)
P\!\bigl(\tilde{\mathbf C}_{B_k}^{(uv)}=\mathbf b
\mid\mathbf C_{B_k}=\mathbf d,\tilde{\mathbf C}_{A_k}^{(uv)}=\mathbf a,
\mathbf C_{A_k}=\mathbf a_0\bigr)
\notag\\
=&
\binom{b_k-t}{\Delta}\cdot\frac{1}{\binom{b_k}{t_0}}\cdot
\frac{1}{\binom{t_0}{\Delta}}
\notag\\
=&
\frac{(b_k-t)!}{\Delta!\,(b_k-t-\Delta)!}\cdot
\frac{t_0!\,(b_k-t_0)!}{b_k!}\cdot
\frac{\Delta!\,t!}{t_0!}
=
\frac{1}{\binom{b_k}{t}},
\end{align}
where the last equality uses $t_0-\Delta=t$ and
$b_k-t_0=b_k-t-\Delta$.

We first establish distributional equality. For the cluster-level
assignments,
\[
P(\tilde{\mathbf C}_{A_k}^{(uv)}=\mathbf a)
=
P(\mathbf C_{A_k}=\mathbf a)
=
\frac{\binom{b_k}{t}}{\binom{|\mathcal I_k|}{I_k}}.
\]
Together with~\eqref{eq:key_coupling}, this implies
\[
P\!\bigl(\tilde{\mathbf C}_{\mathcal I_k}^{(uv)}=(\mathbf a,\mathbf b)\bigr)
=
\frac{\binom{b_k}{t}}{\binom{|\mathcal I_k|}{I_k}}\cdot
\frac{1}{\binom{b_k}{t}}
=
\frac{1}{\binom{|\mathcal I_k|}{I_k}}
=
P\!\bigl(\mathbf C_{\mathcal I_k}=(\mathbf a,\mathbf b)\bigr).
\]
By cross-strata independence in Assumption~\ref{ass:complete_rand}(i), we
obtain
\[
\tilde{\mathbf C}^{(uv)}\overset{d}{=}\mathbf C.
\]
It remains to verify that $\tilde{\mathbf W}^{(uv)}\mid\tilde{\mathbf C}^{(uv)}$ and $\mathbf W\mid\mathbf C$ have the same conditional law. Fix realizations $\mathbf c_0$ and $\mathbf c$ of $\mathbf C$ and $\tilde{\mathbf C}^{(uv)}$, respectively, so that $\mathcal M_{uv}^{(cl)}$ is determined. For any $\mathbf w=(\mathbf w_1,\ldots,\mathbf w_m)$,
\begin{align*}
&P\!\left(\tilde{\mathbf W}^{(uv)}=\mathbf w
\mid \tilde{\mathbf C}^{(uv)}=\mathbf c,\mathbf C=\mathbf c_0\right)\\
=&
P\!\left(\tilde{\mathbf W}_{i}^{(uv)}=\mathbf w_i,\ \forall i\in \mathcal M_{uv}^{(cl)}
\mid \tilde{\mathbf C}^{(uv)}=\mathbf c,\mathbf C=\mathbf c_0\right)
P\!\left(\tilde{\mathbf W}_{i}^{(uv)}=\mathbf w_i,\ \forall i\notin \mathcal M_{uv}^{(cl)}
\mid \tilde{\mathbf C}^{(uv)}=\mathbf c,\mathbf C=\mathbf c_0\right)\\
=&
\prod_{i\in \mathcal M_{uv}^{(cl)}}
P\!\left(\mathbf W_i=\mathbf w_i\mid C_i=c_i\right)
\prod_{i\notin \mathcal M_{uv}^{(cl)}}
P\!\left(\mathbf W_i=\mathbf w_i\mid C_i=c_{0,i}\right)\\
=&
\prod_{i=1}^m
P\!\left(\mathbf W_i=\mathbf w_i\mid C_i=c_i\right)\\
=&
P\!\left(\mathbf W=\mathbf w\mid \mathbf C=\mathbf c\right).
\end{align*}
where the third equality uses $c_{0,i}=c_i$ for all $i\notin\mathcal M_{uv}^{(cl)}$ (clusters outside $\mathcal M_{uv}^{(cl)}$ are unchanged by the coupling). Since the right-hand side does not depend on $\mathbf c_0$, marginalizing over $\mathbf C=\mathbf c_0$ gives, for every $(\mathbf w,\mathbf c)$,
\[
P\!\left(\tilde{\mathbf W}^{(uv)}=\mathbf w
\mid \tilde{\mathbf C}^{(uv)}=\mathbf c\right)
=
P\!\left(\mathbf W=\mathbf w\mid \mathbf C=\mathbf c\right).
\]
Together with $\tilde{\mathbf C}^{(uv)}\overset{d}{=}\mathbf C$, this
completes the proof of distributional equality.

We next establish independence. Since~\eqref{eq:key_coupling} does not
depend on $\mathbf a_0$, we have
\[
\tilde{\mathbf C}_{B_k}^{(uv)}
\indep
\mathbf C_{A_k}
\mid
\tilde{\mathbf C}_{A_k}^{(uv)}.
\]
Together with
$\tilde{\mathbf C}_{A_k}^{(uv)}\indep(\mathbf C,\mathbf W)$, this yields
\[
(\tilde{\mathbf C}_{A_k}^{(uv)},\tilde{\mathbf C}_{B_k}^{(uv)})
\indep
\mathbf C_{A_k}
\]
for each $k$. By cross-strata independence in
Assumption~\ref{ass:complete_rand}(i) and the coupling construction in Section~\ref{sec:technical}, we therefore obtain
\[
\tilde{\mathbf C}^{(uv)}
\indep
\mathbf C_{\mathcal N_{uv}}.
\]
Since $\tilde{\mathbf C}^{(uv)}$ depends on $\mathbf W_{\mathcal N_{uv}}$ only through $\mathbf C_{\mathcal N_{uv}}$, the conditional independence $\tilde{\mathbf C}^{(uv)}\indep\mathbf W_{\mathcal N_{uv}}\mid\mathbf C_{\mathcal N_{uv}}$ holds by construction. Together with $\tilde{\mathbf C}^{(uv)}\indep\mathbf C_{\mathcal N_{uv}}$, this gives
\[
\tilde{\mathbf C}^{(uv)}
\indep
(\mathbf W_{\mathcal N_{uv}},\mathbf C_{\mathcal N_{uv}}).
\]
For the unit-level component, the same argument used for distributional equality gives, for any realizations $\mathbf c$, $\mathbf c_0$, and $\mathbf w$,
\begin{align*}
&P\!\left(\tilde{\mathbf W}^{(uv)}=\mathbf w
\mid \tilde{\mathbf C}^{(uv)}=\mathbf c,\mathbf C=\mathbf c_0,
\mathbf W_{i_0}=\mathbf w_{0,i_0},\ \forall i_0\in\mathcal N_{uv}^{(cl)}
\right)\\
=&
\prod_{i\in \mathcal M_{uv}^{(cl)}}
P\!\left(\mathbf W_i=\mathbf w_i\mid C_i=c_i\right)
\prod_{i\notin \mathcal M_{uv}^{(cl)}}
P\!\left(\mathbf W_i=\mathbf w_i
\mid C_i=c_{0,i},
\mathbf W_{i_0}=\mathbf w_{0,i_0},\ \forall i_0\in\mathcal N_{uv}^{(cl)}
\right)\\
=&
\prod_{i=1}^m
P\!\left(\mathbf W_i=\mathbf w_i\mid C_i=c_i\right)\\
=&
P\!\left(\mathbf W=\mathbf w\mid \mathbf C=\mathbf c\right).
\end{align*}
The first product drops the neighborhood conditioning because the coupling construction in Section~\ref{sec:technical} draws $\tilde{\mathbf W}_i^{(uv)}$ independently of $(\mathbf C,\mathbf W)$ for $i\in\mathcal M_{uv}^{(cl)}$; the second does so by the cross-cluster independence in Assumption~\ref{ass:cluster_assignment}(ii). Hence,
\[
\tilde{\mathbf W}^{(uv)}
\indep
\bigl(\mathbf C,(\mathbf W_i)_{i\in\mathcal N_{uv}^{(cl)}}\bigr)
\mid
\tilde{\mathbf C}^{(uv)}.
\]
It follows that
\[
\tilde{\mathbf W}^{(uv)}
\indep
(\mathbf W_{\mathcal N_{uv}},\mathbf C_{\mathcal N_{uv}})
\mid
\tilde{\mathbf C}^{(uv)}.
\]
Combining this conditional independence with
\[
\tilde{\mathbf C}^{(uv)}
\indep
(\mathbf W_{\mathcal N_{uv}},\mathbf C_{\mathcal N_{uv}}),
\]
we conclude that
\[
(\tilde{\mathbf W}^{(uv)},\tilde{\mathbf C}^{(uv)})
\indep
(\mathbf W_{\mathcal N_{uv}},\mathbf C_{\mathcal N_{uv}}).
\]
\vspace{0.5em}
\noindent\textbf{Part (ii).}
Fix $(u,v)\in\mathcal P$. By construction (the coupling of Section~\ref{sec:technical}), the coupling $(\tilde{\mathbf W}^{(uv)},\tilde{\mathbf C}^{(uv)})$ differs from $(\mathbf W,\mathbf C)$ only on clusters in $\mathcal M_{uv}^{(cl)}$. Hence $(\mathbf W_{\mathcal N_{ij}},\mathbf C_{\mathcal N_{ij}})\neq(\tilde{\mathbf W}^{(uv)}_{\mathcal N_{ij}},\tilde{\mathbf C}^{(uv)}_{\mathcal N_{ij}})$ only if $\mathcal N_{ij}^{(cl)}\cap\mathcal M_{uv}^{(cl)}\neq\varnothing$, so it suffices to bound
\[
\bigl|\bigl\{(i,j)\in\mathcal P:\mathcal N_{ij}^{(cl)}\cap\mathcal M_{uv}^{(cl)}\neq\varnothing\bigr\}\bigr|.
\]

We first bound $|\mathcal M_{uv}^{(cl)}|$. By Assumption~\ref{ass:CLT_regularity}(ii), $|\mathcal N_{uv}^{(cl)}|\leq C_{\mathcal N^{(cl)}}$. Within each stratum $k$, the balance adjustment flips at most $|\Delta_{uv,k}^{(1)}|\leq|\mathcal N_{uv}^{(cl)}\cap\mathcal I_k|$ clusters outside the neighborhood, so the total number of adjusted clusters satisfies $\sum_k|\Delta_{uv,k}^{(1)}|\leq|\mathcal N_{uv}^{(cl)}|\leq C_{\mathcal N^{(cl)}}$. Therefore
\[
|\mathcal M_{uv}^{(cl)}|\leq 2|\mathcal N_{uv}^{(cl)}|\leq 2C_{\mathcal N^{(cl)}}=O(1),
\]
uniformly over $(u,v)\in\mathcal P$.

Next, for each cluster $m\in\mathcal M_{uv}^{(cl)}$, the interference network is undirected, so $(i,j)\in\mathcal N_{ml}$ if and only if $(m,l)\in\mathcal N_{ij}$. Hence $m\in\mathcal N_{ij}^{(cl)}$ if and only if $(i,j)\in\bigcup_{l=1}^{N_m}\mathcal N_{ml}$, and
\[
\bigl|\bigl\{(i,j)\in\mathcal P:m\in\mathcal N_{ij}^{(cl)}\bigr\}\bigr|
=
\Bigl|\bigcup_{l=1}^{N_m}\mathcal N_{ml}\Bigr|
\leq
C_{\mathcal N^{(cl)}}N_m,
\]
where the last inequality is Assumption~\ref{ass:CLT_regularity}(ii) applied to cluster $m$. By the balanced cluster-size condition (Assumption~\ref{ass:CLT_regularity}(v)), $N_m\leq C_N(N/n)$, so the right-hand side is $O(N/n)$.

Applying a union bound over the $O(1)$ clusters in $\mathcal M_{uv}^{(cl)}$:
\[
\bigl|\bigl\{(i,j)\in\mathcal P:\mathcal N_{ij}^{(cl)}\cap\mathcal M_{uv}^{(cl)}\neq\varnothing\bigr\}\bigr|
\leq
|\mathcal M_{uv}^{(cl)}|\cdot C_{\mathcal N^{(cl)}}C_N\frac{N}{n}
=
O\!\left(\frac{N}{n}\right).
\]
Since all constants are uniform over $(u,v)\in\mathcal P$, taking the maximum over $(u,v)$ gives the stated bound.

\vspace{0.5em}
\noindent\textbf{Part (iii).}
Write $\mathcal J_{uv}$ for the (random) set of clusters whose treatments are flipped by the balance
adjustment of the coupling $(u,v)$, so that
$\mathcal M_{uv}^{(cl)}=\mathcal N_{uv}^{(cl)}\cup\mathcal J_{uv}$ and, as shown
above, $\mathcal B_{ij}^{(uv)}$ can occur only if $\mathcal N_{ij}^{(cl)}$ meets
$\mathcal N_{uv}^{(cl)}\cup\mathcal J_{uv}$. Both assertions follow from a single
description of the law of $\mathcal J_{uv}$, which we record first.

By the construction of Section~\ref{sec:technical}, conditionally on
$(\mathbf W,\mathbf C)$ and on the independently drawn $\tilde{\mathbf C}_{\mathcal{N}_{uv}}^{(uv)}$, the set
$\mathcal J_{uv}\cap\mathcal I_k$ is a uniformly chosen subset of size
$|\Delta_{uv,k}^{(1)}|$ of the eligible clusters of
$\mathcal I_k\setminus\mathcal N_{uv}^{(cl)}$, namely those whose current
assignment can absorb the imbalance. Two features of this description are all we
use. The number of flipped clusters is uniformly bounded,
$\sum_k|\Delta_{uv,k}^{(1)}|\leq|\mathcal N_{uv}^{(cl)}|=O(1)$, while the pool
they are drawn from is large: since the number $I_k$ of treated clusters per
stratum is deterministic under Assumption~\ref{ass:complete_rand}(i), the
eligible pool in stratum $k$ is $I_k$ or $|\mathcal I_k|-I_k$ minus at most
$|\mathcal N_{uv}^{(cl)}|=O(1)$ clusters, hence of size $\Omega(n)$ for large $n$
because $I_k/|\mathcal I_k|\in[p,1-p]$ and $\min_k|\mathcal I_k|=\Omega(n)$. As
the flipped clusters are exchangeable within each pool, it follows that
\begin{equation}\label{eq:hit_prob}
P\bigl(m\in\mathcal J_{uv}\bigm|\mathbf C,\tilde{\mathbf C}_{\mathcal{N}_{uv}}^{(uv)}\bigr)=O\!\left(\frac1n\right),
\qquad
P\bigl(\{m,m'\}\subseteq\mathcal J_{uv}\bigm|\mathbf C,\tilde{\mathbf C}_{\mathcal{N}_{uv}}^{(uv)}\bigr)=O\!\left(\frac1{n^{2}}\right),
\end{equation}
uniformly over clusters $m\neq m'$ and over the conditioning variables.

The first assertion is now a dichotomy. If
$\mathcal N_{ij}^{(cl)}\cap\mathcal N_{uv}^{(cl)}=\varnothing$, then
$\mathcal B_{ij}^{(uv)}$ forces
$\mathcal N_{ij}^{(cl)}\cap\mathcal J_{uv}\neq\varnothing$, so
a union bound over the $O(1)$ clusters of $\mathcal N_{ij}^{(cl)}$, which is
$O(1)$ by Assumption~\ref{ass:CLT_regularity}(ii), combined with the first bound
in \eqref{eq:hit_prob}, gives $P(\mathcal B_{ij}^{(uv)})=O(1/n)$ after averaging
over $(\mathbf C,\mathbf C_{\mathcal{N}_{uv}}^{(uv)})$. Choosing $C_{P}^{(0)}$ larger than the constant
implied here therefore makes
$P(\mathcal B_{ij}^{(uv)})\geq C_{P}^{(0)}/n$ possible only when
$\mathcal N_{ij}^{(cl)}\cap\mathcal N_{uv}^{(cl)}\neq\varnothing$, which is the
first assertion.

For the second assertion, suppose that none of the three listed configurations
holds. Then $\mathcal N_{ij}^{(cl)}\cap\mathcal N_{uv}^{(cl)}=\varnothing$ and
$\mathcal N_{i'j'}^{(cl)}\cap\mathcal N_{u'v'}^{(cl)}=\varnothing$, so
$\mathcal B_{ij}^{(uv)}$ requires
$\mathcal J_{uv}\cap\mathcal N_{ij}^{(cl)}\neq\varnothing$ and
$\mathcal B_{i'j'}^{(u'v')}$ requires
$\mathcal J_{u'v'}\cap\mathcal N_{i'j'}^{(cl)}\neq\varnothing$. If
$(u,v)\neq(u',v')$, leveraging some conditional independence in the construction of couplings in Section \ref{sec:technical} yields
\begin{align*}
P(\mathcal B_{ij}^{(uv)}\cap\mathcal B_{i'j'}^{(u'v')})=&E\left[P(\mathcal B_{ij}^{(uv)}|\mathbf W,\mathbf C,\tilde{\mathbf C}_{\mathcal{N}_{uv}}^{(uv)},\tilde{\mathbf C}_{\mathcal{N}_{u'v'}}^{(u'v')})P(\mathcal B_{i'j'}^{(u'v')}|\mathbf W,\mathbf C,\tilde{\mathbf C}_{\mathcal{N}_{uv}}^{(uv)},\tilde{\mathbf C}_{\mathcal{N}_{u'v'}}^{(u'v')})\right]\\
=&E\left[P(\mathcal B_{ij}^{(uv)}|\mathbf W,\mathbf C,\tilde{\mathbf C}_{\mathcal{N}_{uv}}^{(uv)})P(\mathcal B_{i'j'}^{(u'v')}|\mathbf W,\mathbf C,\tilde{\mathbf C}_{\mathcal{N}_{u'v'}}^{(u'v')})\right]\\
\leq&E\left[P(\cup_{m\in \mathcal{N}_{ij}}m\in\mathcal{J}_{uv}|\mathbf W,\mathbf C,\tilde{\mathbf C}_{\mathcal{N}_{uv}}^{(uv)})P(\cup_{m'\in \mathcal{N}_{i'j'}}m'\in\mathcal{J}_{u'v'}|\mathbf W,\mathbf C,\tilde{\mathbf C}_{\mathcal{N}_{u'v'}}^{(u'v')})\right]\\
=&E\left[P(\cup_{m\in \mathcal{N}_{ij}}m\in\mathcal{J}_{uv}|\mathbf C,\tilde{\mathbf C}_{\mathcal{N}_{uv}}^{(uv)})P(\cup_{m'\in \mathcal{N}_{i'j'}}m'\in\mathcal{J}_{u'v'}|\mathbf C,\tilde{\mathbf C}_{\mathcal{N}_{u'v'}}^{(u'v')})\right].
\end{align*}
Since each conditional probability in the final formula is uniformly $O(1/n)$ by the first bound in
\eqref{eq:hit_prob}, multiplying and averaging over the possible values of $\bigl(\mathbf C,\mathbf C_{\mathcal{N}_{uv}}^{(uv)},\mathbf C_{\mathcal{N}_{u'v'}}^{(u'v')}\bigr)$ gives
$P(\mathcal B_{ij}^{(uv)}\cap\mathcal B_{i'j'}^{(u'v')})=O(1/n^{2})$. If
$(u,v)=(u',v')$, then
$\mathcal N_{ij}^{(cl)}$ and $\mathcal N_{i'j'}^{(cl)}$ are disjoint, so the two
required hits must be realized by two distinct clusters of the same set
$\mathcal J_{uv}$, and the second bound in \eqref{eq:hit_prob} together with a
union bound over the $O(1)$ pairs of clusters within $\mathcal{N}_{ij}\times \mathcal{N}_{i'j'}$ again gives $O(1/n^{2})$. Choosing $C_{P}^{(1)}$ larger than the constant implied in these
two cases yields the second assertion.

\vspace{0.5em}
\noindent\textbf{Part (iv).}
Write $\mathcal S_{uv,ij}:=\mathcal N_{uv}^{(cl)}\cup\mathcal N_{ij}^{(cl)}$ for
the clusters visible to a pair, so that
$|\mathcal S_{uv,ij}|\leq 2C_{\mathcal N^{(cl)}}=O(1)$ by
Assumption~\ref{ass:CLT_regularity}(ii), and fix a quadruple
$\bigl((u,v),(u',v'),(i,j),(i',j')\bigr)\in\mathcal P^4$ with
$\mathcal S_{uv,ij}\cap\mathcal S_{u'v',i'j'}=\varnothing$; we show
$\psi_{ij,i'j'}^{(0),(uv),(u'v')}\leq C_{\psi}^{(0)}/n$ for such quadruples,
which is the assertion.

Collect the generators of $\mathcal G_{uv,ij}$ into the
local configuration
\[
Z:=\bigl(
\underbrace{\mathbf C_{\mathcal N_{uv}},\,\mathbf C_{\mathcal N_{ij}},\,
\tilde{\mathbf C}^{(uv)}_{\mathcal N_{ij}}}_{=:\;\zeta},\;
\underbrace{\mathbf W_{\mathcal N_{uv}},\,\mathbf W_{\mathcal N_{ij}},\,
\tilde{\mathbf W}^{(uv)}_{\mathcal N_{ij}}}_{=:\;\xi}
\bigr),
\]
and let $Z'$, with cluster-level part $\zeta'$ and unit-level part $\xi'$, be its counterpart for
$\mathcal G_{u'v',i'j'}$. Since $\mathcal G_{uv,ij}=\sigma(Z)$ and
$\mathcal G_{u'v',i'j'}=\sigma(Z')$, every $A\in\mathcal G_{uv,ij}$ is of the
form $A=\{Z\in\Gamma\}$ and every $B\in\mathcal G_{u'v',i'j'}$ of the form
$B=\{Z'\in\Gamma'\}$, for sets $\Gamma,\Gamma'$ of configurations; the atoms are
the events $\{Z=z\}$ and $\{Z'=z'\}$. For each realization $z$ of $Z$, let $\zeta(z)$ and $\xi(z)$ be the realization values of $\zeta$ and $\xi$, and similarly define $\zeta(z')$ and $\xi(z')$ for $z'$. We also write
$\mathcal G^{C}_{uv,ij}:=\sigma(\zeta)$ and
$\mathcal G^{C}_{u'v',i'j'}:=\sigma(\zeta')$ as the $\sigma$-fields generated by the cluster-level parts.

The proof simplifies the $\sigma$-fields compared in the mixing coefficients in two steps, first passing from
$\mathcal G_{uv,ij}$ to its cluster-level part $\mathcal G^{C}_{uv,ij}$ and then
from that to $\sigma(\mathbf C_{\mathcal S_{uv,ij}})$, and similarly on the primed
side. Each step is justified in the same way: the generators being discarded have
a conditional law which, given what remains, is a fixed kernel of the retained
variables on their own side, and the two sides' discarded variables are
conditionally independent. Such generators cancel from the ratio defining $\psi$,
so the coefficient is unchanged.

In the first step we discarded unit-level treatment variables to simplify the $\sigma$-fields. The original treatments
$\mathbf W_{\mathcal N_{uv}},\mathbf W_{\mathcal N_{ij}}$ are determined by the
blocks $(\mathbf W_{m})_{m\in\mathcal S_{uv,ij}}$, whose conditional law given
$\mathbf C$ depends on $\mathbf C$ only through
$\mathbf C_{\mathcal S_{uv,ij}}$ by Assumption~\ref{ass:cluster_assignment}(i)--(ii);
and by the construction of Section~\ref{sec:technical} the coupled treatments
$\tilde{\mathbf W}^{(uv)}_{\mathcal N_{ij}}$ are, conditionally on
$\tilde{\mathbf C}^{(uv)}$, drawn cluster by cluster from those same
within-cluster laws, independently across clusters and independently of
$(\mathbf C,\mathbf W)$. Hence, with the two cluster
sets $\mathcal S_{uv,ij}$ and $\mathcal S_{u'v',i'j'}$ being disjoint, we have the conditional independence
\[
\xi\indep \bigl(\mathbf{C},\mathbf{C}^{(uv)},\mathbf{C}^{(u'v')}\bigr)\;\Bigm|\;\zeta, \qquad 
\xi'\indep \bigl(\mathbf{C},\mathbf{C}^{(uv)},\mathbf{C}^{(u'v')}\bigr)\;\Bigm|\;\zeta'
\]
 Writing $h(z)$ and $h'(z')$ for
these conditional probabilities $P(\xi=\xi(z)|\zeta=\zeta(z))$ and $P(\xi'=\xi(z')|\zeta'=\zeta(z'))$, 
\begin{equation}\label{eq:atom_split}
P(Z=z)=P\bigl(\zeta=\zeta(z)\bigr)h(z),
\quad
P(Z'=z')=P\bigl(\zeta'=\zeta'(z')\bigr)h'(z'),
\end{equation}
\[
P(Z=z,\,Z'=z')
=
P\bigl(\zeta=\zeta(z),\,\zeta'=\zeta'(z')\bigr)h(z)h'(z') .
\]
Consequently, for any $A=\{Z=z\}$ and $B=\{Z'=z'\}$ with
$P(A)P(B)>0$,
\[
\frac{P(A\cap B)}{P(A)P(B)}
=
\frac{
P\bigl(\zeta=\zeta(z),\,\zeta'=\zeta'(z')\bigr)h(z)h'(z')}
{
P\bigl(\zeta=\zeta(z)\bigr)P\bigl(\zeta'=\zeta'(z')\bigr)h(z)h'(z')}
=
\frac{P\bigl(\zeta=\zeta(z),\,\zeta'=\zeta'(z')\bigr)}{P\bigl(\zeta=\zeta(z)\bigr)P\bigl(\zeta'=\zeta'(z')\bigr)},
\]
It's easy to see that for any event $A=\{Z\in\Gamma\}$ and $B=\{Z\in\Gamma'\}$, their corresponding $\frac{P(A\cap B)}{P(A)P(B)}$ could be written as the weighted mean of the atom ratio above, and its maximum distance to 1 is exactly the maximum distance between the atom ratios and 1. Therefore,
\begin{equation}\label{eq:strip_W}
\psi_{ij,i'j'}^{(0),(uv),(u'v')}
=
\psi\bigl(\mathcal G_{uv,ij},\mathcal G_{u'v',i'j'}\bigr)
=
\psi\bigl(\mathcal G^{C}_{uv,ij},\mathcal G^{C}_{u'v',i'j'}\bigr).
\end{equation}

In the second step the discarded generator is the coupled cluster-level block
$\tilde{\mathbf C}^{(uv)}_{\mathcal N_{ij}}$, and the argument is the same, so we
do not repeat the computation. All that is needed is the analogue of the two
properties used above, namely
\[
\tilde{\mathbf C}^{(uv)}_{\mathcal N_{ij}}
\indep
\mathbf C
\;\Bigm|\;
\mathbf C_{\mathcal S_{uv,ij}},
\qquad
\tilde{\mathbf C}^{(u'v')}_{\mathcal N_{i'j'}}
\indep
\mathbf C
\;\Bigm|\;
\mathbf C_{\mathcal S_{u'v',i'j'}}.
\]
For the first, on $\mathcal N_{ij}^{(cl)}\cap\mathcal N_{uv}^{(cl)}$ the coupled
assignment is the redraw of Section~\ref{sec:technical}, drawn from the marginal
law of $\mathbf C_{\mathcal N_{uv}}$ independently of $\mathbf C$,
while on the remaining clusters of $\mathcal N_{ij}^{(cl)}$ it equals $\mathbf C$
except at the clusters selected by the balance adjustment, whose conditional law
given $\mathbf C$ and the redraw is a uniform draw from an eligible pool whose
size and whose intersection with $\mathcal N_{ij}^{(cl)}$ are determined by
$\mathbf C_{\mathcal S_{uv,ij}}$ alone, because the number $I_k$ of treated
clusters per stratum is deterministic under
Assumption~\ref{ass:complete_rand}(i). The second holds similarly. Replacing
$\zeta$ by $\mathbf C_{\mathcal S_{uv,ij}}$ and $h$ by a similar conditional probability then yield
\[
\psi_{ij,i'j'}^{(0),(uv),(u'v')}
=
\psi\Bigl(\sigma\bigl(\mathbf C_{\mathcal S_{uv,ij}}\bigr),\,
\sigma\bigl(\mathbf C_{\mathcal S_{u'v',i'j'}}\bigr)\Bigr),
\]
a mixing coefficient between the original cluster-level treatment assignments on two
disjoint blocks alone.

These two blocks have size $O(1)$, so under
Assumption~\ref{ass:complete_rand}(i) they satisfy the hypotheses of
Lemma~\ref{lem:mixing_coef_complete_randomization} with population size $n$,
$M\asymp n$ and $C_S=2C_{\mathcal N^{(cl)}}$. That lemma bounds the right-hand
side by $C/n$ for a constant $C$ depending only on $p$, on
$C_{\mathcal N^{(cl)}}$ and on the implied constant in
$\min_k|\mathcal I_k|=\Omega(n)$. Taking $C_{\psi}^{(0)}:=C$ completes the proof
of part~(iv).

\end{proof}

\subsubsection{Proof of Corollary \ref{cor:coupling}}
\begin{proof}[Proof of Corollary \ref{cor:coupling}]
We prove the four parts separately. Throughout, recall that
$\mathbf X_{ij}^{(\mathrm{marg})}$ and $\mathbf X_{ij}$ depend on
$(\mathbf W,\mathbf C)$ only through $(\mathbf W_{\mathcal N_{ij}},\mathbf C_{\mathcal N_{ij}})$,
and the coupled statistics $(\tilde{\mathbf X}_{ij}^{(\mathrm{marg})})^{(uv)}$ and
$\tilde{\mathbf X}_{ij}^{(uv)}$ are obtained by evaluating the same functions at
$(\tilde{\mathbf W}^{(uv)},\tilde{\mathbf C}^{(uv)})$ in place of $(\mathbf W,\mathbf C)$.

\vspace{0.5em}
\noindent\textbf{Part (i).}
Fix $(u,v)\in\mathcal P$. We address the marginal and non-marginal statistics in parallel, as the arguments are identical.

We first establish distributional equality. Each $\mathbf X_{ij}^{(\mathrm{marg})}$ is a measurable function of
$(\mathbf W_{\mathcal N_{ij}},\mathbf C_{\mathcal N_{ij}})
\subseteq \sigma(\mathbf W,\mathbf C)$,
and $(\tilde{\mathbf X}_{ij}^{(\mathrm{marg})})^{(uv)}$ is the same function evaluated at
$(\tilde{\mathbf W}^{(uv)},\tilde{\mathbf C}^{(uv)})$.
Since Proposition~\ref{prop:coupling}(i) gives
$(\tilde{\mathbf W}^{(uv)},\tilde{\mathbf C}^{(uv)})\overset{d}{=}(\mathbf W,\mathbf C)$,
the entire array of coupled statistics has the same joint distribution as the original array:
\[
\bigl(\mathbf a^\top(\tilde{\mathbf X}_{ij}^{(\mathrm{marg})})^{(uv)}\bigr)_{(i,j)\in\mathcal P}
\overset{d}{=}
\bigl(\mathbf a^\top\mathbf X_{ij}^{(\mathrm{marg})}\bigr)_{(i,j)\in\mathcal P},
\]
and the same holds for $\mathbf b^\top\mathbf X_{ij}$.

We next turn to the zero conditional mean. The statistic $\mathbf X_{uv}^{(\mathrm{marg})}$ is
$\sigma(\mathbf W_{\mathcal N_{uv}},\mathbf C_{\mathcal N_{uv}})$-measurable,
since both $Y_{uv}(\mathbf w)$ and $\beta_{uv}(\phi)$ depend on the realized
treatment vector only through $(\mathbf W_{\mathcal N_{uv}},\mathbf C_{\mathcal N_{uv}})$
by Assumptions~\ref{ass:interference} and~\ref{ass:weight}.
On the other hand, the entire coupled array
$(\mathbf a^\top(\tilde{\mathbf X}_{ij}^{(\mathrm{marg})})^{(uv)})_{(i,j)\in\mathcal P}$
is $\sigma(\tilde{\mathbf W}^{(uv)},\tilde{\mathbf C}^{(uv)})$-measurable.
By Proposition~\ref{prop:coupling}(i),
\[
(\tilde{\mathbf W}^{(uv)},\tilde{\mathbf C}^{(uv)})
\indep
(\mathbf W_{\mathcal N_{uv}},\mathbf C_{\mathcal N_{uv}}),
\]
so $\mathbf X_{uv}^{(\mathrm{marg})}$ is independent of the coupled array.
Hence, conditioning on the coupled array does not change the conditional mean of
$\mathbf a^\top\mathbf X_{uv}^{(\mathrm{marg})}$:
\[
\mathbb{E}\!\left[
\mathbf a^\top\mathbf X_{uv}^{(\mathrm{marg})}
\,\middle|\,
\bigl(\mathbf a^\top(\tilde{\mathbf X}_{ij}^{(\mathrm{marg})})^{(uv)}\bigr)_{(i,j)\in\mathcal P}
\right]
=
\mathbb{E}\!\left[\mathbf a^\top\mathbf X_{uv}^{(\mathrm{marg})}\right].
\]
Fix any $\phi\in\{\phi_0,\phi_1\}$. By assumption \ref{ass:weight}(i), 
$\mathbb{E}[\beta_{uv}(\phi)Y_{uv}]=\bar Y_{uv}(\phi)$ and $\mathbb{E}[\beta_{uv}(\phi)]=1$ hold. Consequently, the
$\phi$-component of $\mathbb{E}[\mathbf X_{uv}^{(\mathrm{marg})}]$ is
\[
\frac{g_u}{N_u}
\bigl(\mathbb{E}[\beta_{uv}(\phi)Y_{uv}]-\bar Y_{uv}(\phi)\,\mathbb{E}[\beta_{uv}(\phi)]\bigr)+\frac{g_u}{N_u}
\bigl(\mathbb{E}[\beta_{uv}(\phi)]-1\bigr)\bigl(\bar Y_{uv}(\phi)-\bar Y(\phi)\bigr)
=
\frac{g_u}{N_u}\bigl(\bar Y_{uv}(\phi)-\bar Y_{uv}(\phi)\bigr)
=0.
\]
Therefore, 
\[\mathbb{E}\!\left[
\mathbf a^\top\mathbf X_{uv}^{(\mathrm{marg})}
\,\middle|\,
\bigl(\mathbf a^\top(\tilde{\mathbf X}_{ij}^{(\mathrm{marg})})^{(uv)}\bigr)_{(i,j)\in\mathcal P}
\right]=0.
\]
The identical reasoning applies to $\mathbf b^\top\mathbf X_{uv}$ using
Assumption~\ref{ass:weight}(ii) and $\boldsymbol\tau_{ij}$.

\vspace{0.5em}
\noindent\textbf{Part (ii).}
Fix $(u,v)\in\mathcal P$.
By construction, $(\tilde{\mathbf W}^{(uv)},\tilde{\mathbf C}^{(uv)})$ agrees with
$(\mathbf W,\mathbf C)$ outside the affected cluster set $\mathcal M_{uv}^{(cl)}$,
so
\[
(\mathbf W_{\mathcal N_{ij}},\mathbf C_{\mathcal N_{ij}})
=
(\tilde{\mathbf W}^{(uv)}_{\mathcal N_{ij}},\tilde{\mathbf C}^{(uv)}_{\mathcal N_{ij}})
\quad\text{whenever}\quad
\mathcal N_{ij}^{(cl)}\cap\mathcal M_{uv}^{(cl)}=\varnothing.
\]
For such $(i,j)$, the coupled statistic equals the original:
$\mathbf a^\top(\tilde{\mathbf X}_{ij}^{(\mathrm{marg})})^{(uv)}
=\mathbf a^\top\mathbf X_{ij}^{(\mathrm{marg})}$, and similarly for
$\mathbf b^\top\mathbf X_{ij}$.
Hence only units $(i,j)$ with $\mathcal N_{ij}^{(cl)}\cap\mathcal M_{uv}^{(cl)}\neq\varnothing$
contribute to the discrepancy sums.
By Proposition~\ref{prop:coupling}(ii), the number of such units is $O(N/n)$, uniformly in $(u,v)$.
For each contributing unit, Assumptions~\ref{ass:CLT_regularity}(iii)--(iv) bound
the individual discrepancy:
\begin{align*}
\bigl|\mathbf a^\top\mathbf X_{ij}^{(\mathrm{marg})}
-\mathbf a^\top(\tilde{\mathbf X}_{ij}^{(\mathrm{marg})})^{(uv)}\bigr|
&\le
\|\mathbf a\|
\bigl\|\mathbf X_{ij}^{(\mathrm{marg})}\bigr\|_\infty
+\|\mathbf a\|
\bigl\|(\tilde{\mathbf X}_{ij}^{(\mathrm{marg})})^{(uv)}\bigr\|_\infty
\le
\frac{4C_\beta C_Y g_i}{N_i},
\end{align*}
where we used $\|\mathbf a\|\le 1$,
$|Y_{ij}(\mathbf w)|\le C_Y$ and $|\bar Y_{ij}(\phi)|\le C_Y$ so that
$|Y_{ij}-\bar Y_{ij}(\phi)|\le 2C_Y$, and
$|\beta_{ij}(\phi)|\le C_\beta$ almost surely.
Combining with the cluster-weight balance
$g_i/N_i\le C_N g_i\cdot(n/N)$ and
$\max_i g_i=O(1/n)$ from
Assumption~\ref{ass:CLT_regularity}(v)--(vi), each individual discrepancy is $O(1/N)$.
Summing over the $O(N/n)$ discrepant units and taking the supremum over $(u,v)$:
\[
\sup_{(u,v)\in\mathcal P}
\left\|
\sum_{(i,j)\in\mathcal P}
\Bigl[
\mathbf a^\top\mathbf X_{ij}^{(\mathrm{marg})}
-
\mathbf a^\top(\tilde{\mathbf X}_{ij}^{(\mathrm{marg})})^{(uv)}
\Bigr]
\right\|_{\infty}
=
O\!\left(\frac{N}{n}\right)\cdot O\!\left(\frac{1}{N}\right)
=
O\!\left(\frac{1}{n}\right),
\]
and the same bound holds for the $\mathbf b^\top\mathbf X_{ij}$ discrepancies.

\vspace{0.5em}
\noindent\textbf{Part (iii).}
For the unit-level claims, the individual discrepancy vanishes whenever the
coupled neighborhood assignments agree, since both statistics are then the same
function of the same arguments; this is the inclusion
$\{D_{ij}^{(uv)}\neq0\}\subseteq\mathcal B_{ij}^{(uv)}$. The display above bounds
$|D_{ij}^{(uv)}|$ by $C'/N$ almost surely, which is the sup-norm bound, and
therefore
\[
\bigl\|D_{ij}^{(uv)}\bigr\|_{1}
=
\mathbb E\bigl[|D_{ij}^{(uv)}|\,\mathbf 1_{\mathcal B_{ij}^{(uv)}}\bigr]
\leq
\frac{C'}{N}\,P\bigl(\mathcal B_{ij}^{(uv)}\bigr).
\]
The inclusion for products follows by intersecting the two one-index inclusions.
The same argument applies to $\mathbf b^\top\mathbf X_{ij}$.

\vspace{0.5em}
\noindent\textbf{Part (iv).}
The claim follows from a $\sigma$-field inclusion into the localized field
$\mathcal G_{uv,ij}$ of Proposition~\ref{prop:coupling}(iv). Fix a pair
$\bigl((u,v),(i,j)\bigr)\in\mathcal P^2$; we show that
\[
\sigma\!\left(
\mathbf a^\top\mathbf X_{uv}^{(\mathrm{marg})},\;
\mathbf a^\top\mathbf X_{ij}^{(\mathrm{marg})}
-\mathbf a^\top(\tilde{\mathbf X}_{ij}^{(\mathrm{marg})})^{(uv)}
\right)
\subseteq
\mathcal G_{uv,ij},
\]
and likewise for the non-marginal statistic.

Each of the three ingredients is a generator of $\mathcal G_{uv,ij}$, or a
function of one. The statistic $\mathbf a^\top\mathbf X_{uv}^{(\mathrm{marg})}$
is a function of $(\mathbf W_{\mathcal N_{uv}},\mathbf C_{\mathcal N_{uv}})$ by
Assumptions~\ref{ass:interference} and~\ref{ass:weight}; the statistic
$\mathbf a^\top\mathbf X_{ij}^{(\mathrm{marg})}$ is the corresponding function of
$(\mathbf W_{\mathcal N_{ij}},\mathbf C_{\mathcal N_{ij}})$; and
$\mathbf a^\top(\tilde{\mathbf X}_{ij}^{(\mathrm{marg})})^{(uv)}$ is that same
function evaluated at
$(\tilde{\mathbf W}^{(uv)}_{\mathcal N_{ij}},\tilde{\mathbf C}^{(uv)}_{\mathcal N_{ij}})$.
All of $\mathbf W_{\mathcal N_{uv}},\mathbf C_{\mathcal N_{uv}},
\mathbf W_{\mathcal N_{ij}},\mathbf C_{\mathcal N_{ij}},
\tilde{\mathbf W}^{(uv)}_{\mathcal N_{ij}},\tilde{\mathbf C}^{(uv)}_{\mathcal N_{ij}}$
generate $\mathcal G_{uv,ij}$, which proves the inclusion. Because the
coefficient of Proposition~\ref{prop:coupling}(iv) is attached to a single unit
rather than to the whole discrepancy sum, no enlargement beyond the two
interference neighborhoods is needed here. The identical argument applies to
$\mathbf b^\top\mathbf X_{ij}$.

The mixing coefficient is monotone with respect to $\sigma$-field inclusion,
since shrinking either $\sigma$-field shrinks the class over which its defining
supremum is taken. Applying this to both sides of a quadruple,
\[
\psi_{ij,i'j'}^{(uv),(u'v'),(\mathrm{marg})}
\leq
\psi_{ij,i'j'}^{(0),(uv),(u'v')},
\qquad
\psi_{ij,i'j'}^{(uv),(u'v')}
\leq
\psi_{ij,i'j'}^{(0),(uv),(u'v')},
\]
for every $\bigl((u,v),(u',v'),(i,j),(i',j')\bigr)\in\mathcal P^4$. Consequently, with $C_{\psi}^{(1)}:=C_{\psi}^{(0)}$, either summand coefficient
can reach $C_{\psi}^{(1)}/n$ only where
$\psi_{ij,i'j'}^{(0),(uv),(u'v')}$ does, and
Proposition~\ref{prop:coupling}(iv) confines the latter to the quadruples whose
localized cluster sets meet. Since the domination holds quadruple by quadruple,
the structural characterization transfers verbatim.
\end{proof}

\subsubsection{Proof of Theorem \ref{thm:clt}}
The proof under complete randomization is based on Proposition \ref{prop:scalar_clt} and Corollary~\ref{cor:coupling}, while the proof under independent Bernoulli randomization uses the following Stein-type bound from~\cite{Ross2011}.

\begin{lemma}[Stein bound for summands under independent Bernoulli randomization (Theorem 3.5 in \cite{Ross2011})]
\label{lem:stein_bernoulli}
Let $\mathcal{P}$ be an index set for $(i,j)$ with $|\mathcal{P}|=N$.
Assume that $\{X_{ij}\}_{(i,j)\in\mathcal{P}}$ satisfies:
\begin{enumerate}
    \item[(i)] $\mathbb{E}[X_{ij}]=0$ for all $(i,j)\in\mathcal{P}$;
    \item[(ii)] $\mathbb{E}[X_{ij}^4]<\infty$ for all $(i,j)\in\mathcal{P}$.
\end{enumerate}
Suppose further that $\{X_{ij}\}_{(i,j)\in\mathcal{P}}$ admits dependency neighborhoods $\{\bar{\mathcal{N}}_{ij}\}_{(i,j)\in\mathcal{P}}$, in the sense that
\[
X_{ij}\indep\bigl(X_{i'j'}\bigr)_{(i',j')\notin\bar{\mathcal{N}}_{ij}}
\quad\text{for all }(i,j)\in\mathcal{P},
\]
and let $D:=\sup_{(i,j)\in\mathcal{P}}|\bar{\mathcal{N}}_{ij}|$.
Let
\[
\sigma^2 := \mathrm{Var}\!\Bigl(\sum_{(i,j)\in\mathcal{P}} X_{ij}\Bigr),
\qquad
S := \sigma^{-1}\sum_{(i,j)\in\mathcal{P}} X_{ij},
\qquad
Z\sim\mathcal{N}(0,1).
\]
Then the Wasserstein distance between $S$ and $Z$ satisfies
\[
d_W(S,Z)
\;\le\;
\frac{D^2}{\sigma^3}\sum_{(i,j)\in\mathcal{P}} \mathbb{E}|X_{ij}|^3
\;+\;
\frac{\sqrt{26}\,D^{3/2}}{\sqrt{\pi}\,\sigma^2}
\sqrt{\sum_{(i,j)\in\mathcal{P}} \mathbb{E}[X_{ij}^4]}.
\]
\end{lemma}

\begin{proof}[Proof of Theorem \ref{thm:clt}]
Here, we prove only the non-marginal case. The marginal case follows by the
same argument, with $\mathbf B_{\mathrm{marg}}$,
$\boldsymbol\tau_{\mathrm{marg}}$, and
$\boldsymbol\tau_{ij,\mathrm{marg}}$ replacing $\mathbf B$,
$\boldsymbol\tau$, and $\boldsymbol\tau_{ij}$, respectively.
The proof covers three settings: general estimators under cluster-level
independent Bernoulli randomization, general estimators under
cluster-level complete randomization, and cluster-agnostic estimators under
unit-level independent Bernoulli randomization. The first two settings are covered by part (i) of the theorem, and the last by part (ii). We first establish several preliminary
arguments that apply to all three settings. We then prove part (i), analyzing its two randomization schemes in turn, and finally prove part (ii), whose random normalization matrix $\boldsymbol\Sigma(\mathbf C)$ requires both a different reduction from the estimator to the sum $\sum_{(i,j)}\mathbf X_{ij}$ and a conditional application of Stein's method.

Recall the individual vector-valued summands from Corollary~\ref{cor:coupling}:
\[
\mathbf X_{ij}
:=
\frac{g_i}{N_i}\mathbf B_{(*,ij)}\circ\bigl(Y_{ij}\mathbf 1-\boldsymbol\tau_{ij}\bigr)+\frac{g_i}{N_i}\bigl(\mathbf B_{(*,ij)}-\mathbf{1}\bigr)\circ\bigl(\boldsymbol\tau_{ij}-\boldsymbol\tau\bigr).
\]
Under Assumption~\ref{ass:weight}(ii), the same argument as in
Corollary~\ref{cor:coupling} implies that
$\mathbb{E}[\mathbf X_{ij}]=\mathbf 0$ for each $(i,j)$. Under Assumption~\ref{ass:weight_cluster_agnostic}(ii), we furthermore have $\mathbb{E}[\mathbf X_{ij}\mid \mathbf{C}]=\mathbf 0$ for each $(i,j)$.
Assumptions~\ref{ass:CLT_regularity}(iii)--(vi) give
$\|\mathbf X_{ij}\|_\infty\le C/N$ a.s.\ for a global constant $C$.

Using the vectorized representation of the Hájek estimator, together with the fact that 
\(\sum_{(i,j)}\frac{g_i}{N_i}\bigl(\boldsymbol\tau_{ij}-\boldsymbol\tau\bigr)=0,\)
we have
\begin{align*}
\hat{\boldsymbol\tau}-\boldsymbol\tau
=&
\left(\sum_{(i,j)}\frac{g_i}{N_i}\mathbf B_{(*,ij)}\right)^{\!\circ-1}
\!\circ\!
\left(
\sum_{(i,j)}\frac{g_i}{N_i}\mathbf B_{(*,ij)}\circ\bigl(Y_{ij}\mathbf 1-\boldsymbol\tau\bigr)
\right)\\
=&\left(\sum_{(i,j)}\frac{g_i}{N_i}\mathbf B_{(*,ij)}\right)^{\!\circ-1}
\!\circ\!
\left(
\sum_{(i,j)}\frac{g_i}{N_i}\mathbf B_{(*,ij)}\circ\bigl(Y_{ij}\mathbf 1-\boldsymbol\tau\bigr)-\frac{g_i}{N_i}\bigl(\boldsymbol\tau_{ij}-\boldsymbol\tau\bigr)
\right)\\
=&\left(\sum_{(i,j)}\frac{g_i}{N_i}\mathbf B_{(*,ij)}\right)^{\!\circ-1}
\!\circ\!
\left(
\sum_{(i,j)}\mathbf X_{ij}
\right).
\end{align*}
By the unbiased identification property (Assumption~\ref{ass:weight}),
$\mathbb{E}\!\bigl[\sum_{(i,j)}\frac{g_i}{N_i}B_{(*,ij)}\bigr]=\mathbf 1$, and a variance-order argument analogous to Theorem~\ref{thm:variance_order} (substituting $Y_{ij}\equiv 1$) gives
\[
\Bigl\|\mathrm{Var}\!\Bigl(\sum_{(i,j)}\tfrac{g_i}{N_i}\mathbf B_{(*,ij)}\Bigr)\Bigr\|
=
\begin{cases}
O(n^{-1}) & \text{in the two settings of part (i)},\\
O(N^{-1}) & \text{in the setting of part (ii)}.
\end{cases}
\]
Writing
\[
\boldsymbol\delta:=\Bigl(\sum_{(i,j)}\tfrac{g_i}{N_i}\mathbf B_{(*,ij)}\Bigr)^{\circ-1}-\mathbf 1,
\]
Chebyshev's inequality therefore gives $\boldsymbol\delta=O_p(n^{-1/2})$ in the part-(i) settings and $\boldsymbol\delta=O_p(N^{-1/2})$ in the part-(ii) setting, and in every case the H\'ajek identity above reads $\hat{\boldsymbol\tau}-\boldsymbol\tau=(\mathbf 1+\boldsymbol\delta)\circ\sum_{(i,j)}\mathbf X_{ij}$.

The passage from the normalized estimator to the normalized sum $\sum_{(i,j)}\mathbf X_{ij}$ takes a single unified form in every setting. Let $\mathbf A$ denote the symmetric positive-definite matrix used to standardize the estimator: $\mathbf A=\boldsymbol\Sigma$, which is deterministic, in part (i), and $\mathbf A=\boldsymbol\Sigma(\mathbf C)$, which is $\sigma(\mathbf C)$-measurable and hence random, in part (ii). The H\'ajek identity gives the exact decomposition
\[
\mathbf A^{-1/2}(\hat{\boldsymbol\tau}-\boldsymbol\tau)
=
\mathbf A^{-1/2}\sum_{(i,j)}\mathbf X_{ij}
+
\mathbf R,
\qquad
\mathbf R:=\mathbf A^{-1/2}\Bigl(\boldsymbol\delta\circ\sum_{(i,j)}\mathbf X_{ij}\Bigr),
\]
and, using $\|\boldsymbol\delta\circ\mathbf v\|\le\|\boldsymbol\delta\|_\infty\|\mathbf v\|$ and $\|\sum_{(i,j)}\mathbf X_{ij}\|\le\|\mathbf A^{1/2}\|\,\|\mathbf A^{-1/2}\sum_{(i,j)}\mathbf X_{ij}\|$, the remainder obeys the bound
\begin{equation}
\|\mathbf R\|
\le
\underbrace{\|\mathbf A^{1/2}\|\,\|\mathbf A^{-1/2}\|}_{\text{square-root condition number of }\mathbf A}
\cdot\|\boldsymbol\delta\|_\infty\cdot
\Bigl\|\mathbf A^{-1/2}\sum_{(i,j)}\mathbf X_{ij}\Bigr\|.
\tag{$\star$}
\end{equation}
Hence, in either part, it suffices to (a) prove the central limit theorem $\mathbf A^{-1/2}\sum_{(i,j)}\mathbf X_{ij}\Rightarrow\mathcal N(\mathbf 0,I)$ for the normalized sum and (b) show that the right-hand side of $(\star)$ is $o_p(1)$, so that $\mathbf R\xrightarrow{p}\mathbf 0$. The two parts differ precisely in how (a) and (b) are carried out, and this difference stems entirely from whether $\mathbf A$ is deterministic or random, as we now explain.

\medskip
\noindent\textbf{Proof of part (i).}
Here $\mathbf A=\boldsymbol\Sigma$ is deterministic, which simplifies both (a) and (b).

For (b), since $\mathbb E[\sum_{(i,j)}\mathbf X_{ij}]=\mathbf 0$ and $\boldsymbol\Sigma$ is the variance matrix of $\sum_{(i,j)}\mathbf X_{ij}$,
\[
\mathbb E\Bigl\|\boldsymbol\Sigma^{-1/2}\sum_{(i,j)}\mathbf X_{ij}\Bigr\|^2
=\mathrm{tr}\bigl(\boldsymbol\Sigma^{-1/2}\boldsymbol\Sigma\,\boldsymbol\Sigma^{-1/2}\bigr)=d,
\]
so $\|\boldsymbol\Sigma^{-1/2}\sum_{(i,j)}\mathbf X_{ij}\|=O_p(1)$. Theorem~\ref{thm:variance_order}(i) gives $\|\boldsymbol\Sigma\|=O(n^{-1})$, which with the hypothesis $\lambda_{\min}(\boldsymbol\Sigma)=\omega(n^{-4/3})$ makes the square-root condition number in $(\star)$ equal to $(\|\boldsymbol\Sigma\|/\lambda_{\min}(\boldsymbol\Sigma))^{1/2}=o(n^{1/6})$. With $\boldsymbol\delta=O_p(n^{-1/2})$, the bound $(\star)$ gives $\|\mathbf R\|\le o(n^{1/6})\cdot O_p(n^{-1/2})\cdot O_p(1)=o_p(1)$.

For (a), because the normalizer $\boldsymbol\Sigma^{-1/2}$ is a fixed matrix, $\boldsymbol\Sigma^{-1/2}\sum_{(i,j)}\mathbf X_{ij}$ is a fixed linear image of $\sum_{(i,j)}\mathbf X_{ij}$, and the Cram\'er--Wold theorem reduces its asymptotic normality to a scalar central limit theorem with a \emph{deterministic} normalizing variance. Fix any deterministic unit vector $\mathbf a$ with $\|\mathbf a\|=1$, let $X_{ij}:=\mathbf a^\top\mathbf X_{ij}$, and set $\sigma_n^2:=\mathrm{Var}(\sum_{(i,j)}X_{ij})=\mathbf a^\top\boldsymbol\Sigma\mathbf a\ge\lambda_{\min}(\boldsymbol\Sigma)$; it then suffices to show
\[
S_n:=\sigma_n^{-1}\sum_{(i,j)}X_{ij}\;\xrightarrow{d}\;\mathcal N(0,1).
\]
We prove this scalar central limit theorem separately for the two randomization schemes covered by part (i), using Stein's method in a scheme-specific way.

\medskip
\noindent\textbf{Cluster-level independent Bernoulli randomization
(Assumption~\ref{ass:independent_rand}(i)).}
Each $X_{ij}$ is a function of the cluster blocks $(\mathbf W_{i'},C_{i'})_{i'\in\mathcal N_{ij}^{(cl)}}$, and by Assumptions~\ref{ass:cluster_assignment}(ii) and~\ref{ass:independent_rand}(i) the blocks $\{(\mathbf W_i,C_i)\}_{i=1}^n$ are mutually independent. Fix $(i,j)$: every summand $X_{i'j'}$ with $(i',j')\notin\Lambda_{ij}$ has $\mathcal N_{i'j'}^{(cl)}\cap\mathcal N_{ij}^{(cl)}=\varnothing$ and is therefore measurable with respect to the blocks on clusters outside $\mathcal N_{ij}^{(cl)}$, so the entire family $\{X_{i'j'}:(i',j')\notin\Lambda_{ij}\}$ is \emph{jointly} independent of $X_{ij}$. Hence $X_{ij}$ admits, in the sense of Lemma~\ref{lem:stein_bernoulli}, the dependency neighborhood
\[
\bar{\mathcal N}_{ij}
:=
\Lambda_{ij}
=
\bigl\{(i',j'):\mathcal N_{ij}^{(cl)}\cap\mathcal N_{i'j'}^{(cl)}\neq\varnothing\bigr\},
\]
and by Lemma~\ref{lem:complete_neighborhood_size},
$D:=\sup_{(i,j)\in\mathcal P}|\bar{\mathcal N}_{ij}|=O(N/n)$.
By Lemma~\ref{lem:stein_bernoulli}, together with the fact that $\|\mathbf X_{ij}\|_\infty\le C/N$, we have
\begin{align*}
d_W(S_n,Z)
&\le
\frac{D^2}{\sigma_n^3}\sum_{(i,j)}\mathbb{E}|X_{ij}|^3
+\frac{\sqrt{26}\,D^{3/2}}{\sqrt\pi\,\sigma_n^2}
\sqrt{\sum_{(i,j)}\mathbb{E}[X_{ij}^4]}\\
&=O\!\left(\frac{N^2}{\sigma_n^3 n^2}\cdot \frac{N}{N^3}\right)+O\!\left(\frac{N^{3/2}}{\sigma_n^2 n^{3/2}}\cdot \sqrt{\frac{N}{N^4}}\right)\\
&=
O\!\left(\frac{1}{\sigma_n^3 n^2}\right)
+
O\!\left(\frac{1}{\sigma_n^2 n^{3/2}}\right).
\end{align*}
Under $\sigma_n^2\ge\lambda_{\min}(\boldsymbol\Sigma)=\omega(n^{-4/3})$ we have
$\sigma_n^3 n^2=\omega(1)$ and $\sigma_n^2 n^{3/2}=\omega(n^{1/6})\to\infty$,
so $d_W(S_n,Z)\to 0$.

\medskip
\noindent\textbf{Cluster-level complete randomization
(Assumption~\ref{ass:complete_rand}(i)).}
We apply Corollary~\ref{cor:coupling} to the scalar summands
$X_{ij}=\mathbf a^\top\mathbf X_{ij}$,
with coupled counterparts
$\tilde X_{ij}^{(uv)}:=\mathbf a^\top\tilde{\mathbf X}_{ij}^{(uv)}$
constructed via the coupling construction in Section~\ref{sec:technical} and Corollary~\ref{cor:coupling}.
Condition~(ii) of Proposition~\ref{prop:scalar_clt} holds for these summands
because Assumptions~\ref{ass:CLT_regularity}(iii)--(vi) give
$|X_{ij}|\le\|\mathbf a\|\,\|\mathbf X_{ij}\|\le C/N$ almost surely, as recorded
in the Bernoulli case above, and condition~(i) is
$\mathbb E[\mathbf X_{ij}]=\mathbf 0$.
Corollary~\ref{cor:coupling}(i) verifies the remaining conditions of
Proposition~\ref{prop:scalar_clt}: the coupled array
$(\tilde X_{ij}^{(uv)})_{(i,j)\in\mathcal P}$ is distributionally equal to
$(X_{ij})_{(i,j)\in\mathcal P}$, and
$\mathbb{E}\bigl[X_{uv}\mid(\tilde X_{ij}^{(uv)})_{(i,j)\in\mathcal P}\bigr]=0$.
Corollary~\ref{cor:coupling}(ii) gives the discrepancy bound
\[
\sup_{(u,v)\in\mathcal P}
\left\|\sum_{(i,j)\in\mathcal P}
\bigl(X_{ij}-\tilde X_{ij}^{(uv)}\bigr)\right\|_\infty
=O\!\left(\frac{1}{n}\right),
\]
which makes the first error term of Proposition~\ref{prop:scalar_clt} equal to
$\sigma_n^{-3}N^{-1}\cdot N\cdot O(n^{-2})=O(\sigma_n^{-3}n^{-2})$.

It remains to bound $\mathcal E_2$. Abbreviate
$D_{ij}^{(uv)}:=X_{ij}-\tilde X_{ij}^{(uv)}$,
\[
q_{ij}^{(uv)}:=P\bigl(\mathcal B_{ij}^{(uv)}\bigr),
\qquad
Q_{ij,i'j'}^{(uv),(u'v')}:=P\bigl(\mathcal B_{ij}^{(uv)}\cap\mathcal B_{i'j'}^{(u'v')}\bigr),
\]
with $\mathcal B_{ij}^{(uv)}$ the disturbance event of
Proposition~\ref{prop:coupling}(iii), and recall from
Corollary~\ref{cor:coupling}(iii) that $\|D_{ij}^{(uv)}\|_\infty\le C'/N$,
$\|D_{ij}^{(uv)}\|_1\le (C'/N)q_{ij}^{(uv)}$ and
$P\bigl(D_{ij}^{(uv)}D_{i'j'}^{(u'v')}\neq0\bigr)\le Q_{ij,i'j'}^{(uv),(u'v')}$. Hence, up to a global constant,
\[
T^{\mathrm{mix}}_{ij,i'j'}
=
O\!\left(\frac{1}{N^2}\,\psi_{ij,i'j'}^{(uv),(u'v')}\,q_{ij}^{(uv)}q_{i'j'}^{(u'v')}\right),
\qquad
T^{\mathrm{sup}}_{ij,i'j'}
=
O\!\left(\frac{1}{N^2}\Bigl(Q_{ij,i'j'}^{(uv),(u'v')}+q_{ij}^{(uv)}q_{i'j'}^{(u'v')}\Bigr)\right),
\]
so that $\mathcal E_2=O\bigl(\sigma_n^{-4}N^{-4}\Sigma\bigr)$ with
$\Sigma:=\sum_{\mathcal P^4}\min\{\,\psi qq',\;Q+qq'\,\}$, the sum running over
quadruples. We show $\Sigma=O(N^4/n^3)$.

All the counting rests on two index sets, both controlled by
Lemma~\ref{lem:complete_neighborhood_size}. Call an index pair
$\bigl((u,v),(i,j)\bigr)\in\mathcal P^2$ \emph{resonant} if
$\mathcal N_{ij}^{(cl)}\cap\mathcal N_{uv}^{(cl)}\neq\varnothing$, and write
$R$ for the set of resonant pairs. Since $|\mathcal N_{uv}^{(cl)}|=O(1)$ and each
cluster lies in the cluster-level neighborhood of only $O(N/n)$ units, every
slice of $R$ in either index has size $O(N/n)$, whence
\begin{equation}\label{eq:R_count}
\sup_{(u,v)}\bigl|\{(i,j):((u,v),(i,j))\in R\}\bigr|
=
\sup_{(i,j)}\bigl|\{(u,v):((u,v),(i,j))\in R\}\bigr|
=O\!\left(\frac Nn\right),
\qquad
|R|=O\!\left(\frac{N^2}{n}\right).
\end{equation}
Call a quadruple \emph{exceptional} if
$\mathcal S_{uv,ij}\cap\mathcal S_{u'v',i'j'}\neq\varnothing$ with
$\mathcal S_{uv,ij}:=\mathcal N_{uv}^{(cl)}\cup\mathcal N_{ij}^{(cl)}$; by
Corollary~\ref{cor:coupling}(iv) this is the only way
$\psi_{ij,i'j'}^{(uv),(u'v')}$ can reach $C_{\psi}^{(1)}/n$. For a fixed
unprimed pair, an exceptional partner must have $\mathcal N_{u'v'}^{(cl)}$ or
$\mathcal N_{i'j'}^{(cl)}$ meeting the $O(1)$ set $\mathcal S_{uv,ij}$, which
confines one primed index to $O(N/n)$ values and leaves the other free; hence
\begin{equation}\label{eq:exc_count}
\#\{\text{exceptional partners}\}=O\!\left(\frac{N^2}{n}\right),
\qquad
\#\{\text{exceptional partners lying in }R\}=O\!\left(\frac{N^2}{n^{2}}\right),
\end{equation}
the second count because once one primed index is confined to $O(N/n)$ values,
membership in $R$ confines the other to $O(N/n)$ values as well by
\eqref{eq:R_count}. Finally, Proposition~\ref{prop:coupling}(iii) says that
$q_{ij}^{(uv)}\geq C_P^{(0)}/n$ forces the pair into $R$, so
\begin{equation}\label{eq:sum_q}
\sup_{(u,v)}\sum_{(i,j)}q_{ij}^{(uv)}=O\!\left(\frac Nn\right),
\qquad
\sup_{(i,j)}\sum_{(u,v)}q_{ij}^{(uv)}=O\!\left(\frac Nn\right),
\qquad
\sum_{(u,v),(i,j)}q_{ij}^{(uv)}=O\!\left(\frac{N^2}{n}\right),
\end{equation}
since in each sum the $O(N/n)$ resonant terms contribute at most $1$ each and the
remaining terms less than $C_P^{(0)}/n$ each.

For the regular quadruples, where $\psi_{ij,i'j'}^{(uv),(u'v')}<C_{\psi}^{(1)}/n$, we use
$T^{\mathrm{mix}}$. Summing over all of $\mathcal P^4$ and factoring,
\[
\sum_{\text{regular}}\psi\,q\,q'
\leq
\frac{C_{\psi}^{(1)}}{n}
\Bigl(\sum_{(u,v),(i,j)}q_{ij}^{(uv)}\Bigr)
\Bigl(\sum_{(u',v'),(i',j')}q_{i'j'}^{(u'v')}\Bigr)
=
\frac{1}{n}\cdot O\!\left(\frac{N^2}{n}\right)^{2}
=
O\!\left(\frac{N^4}{n^3}\right),
\]
by \eqref{eq:sum_q}. 

For the exceptional quadruples we use $T^{\mathrm{sup}}$ and treat its two pieces
separately. For the $qq'$ piece, the analysis before \eqref{eq:exc_count} confines an exceptional
partner of a fixed unprimed pair to $O(N/n)$ values of one primed index with the
other free, so summing $q'$ over those partners and using \eqref{eq:sum_q} in the
matching direction gives $O(N/n)\cdot O(N/n)=O(N^2/n^2)$; multiplying by
$\sum_{(u,v),(i,j)}q=O(N^2/n)$ yields
$\sum_{\text{exceptional}}qq'=O(N^4/n^3)$. Both directions of \eqref{eq:sum_q}
are used here, since the conflict may originate in $(i',j')$ rather than in
$(u',v')$.

For the $Q$ piece we use Proposition~\ref{prop:coupling}(iii): $Q$ can reach
$C_P^{(1)}/n^2$ only if the unprimed pair lies in $R$, or the primed pair lies in
$R$, or the two couplings coincide and $\mathcal N_{ij}^{(cl)}$ meets
$\mathcal N_{i'j'}^{(cl)}$. Splitting the exceptional quadruples accordingly, and
bounding $Q$ by $C_P^{(1)}/n^2$ when none of the three holds and by
$\min\{q,q'\}$ otherwise:
\[
\renewcommand{\arraystretch}{2.1}
\begin{array}{l|c|c|c}
\text{case} & \#\{\text{quadruples}\} & Q\leq & \text{contribution to }\Sigma\\\hline
\text{none of the three}
& O\!\left(\frac{N^4}{n}\right)
& \frac{C_P^{(1)}}{n^2}
& O\!\left(\frac{N^4}{n^3}\right)\\
\text{unprimed}\in R,\ \text{primed}\notin R
& O\!\left(\frac{N^2}{n}\right)O\!\left(\frac{N^2}{n}\right)
& \frac{C_P^{(0)}}{n}
& O\!\left(\frac{N^4}{n^3}\right)\\
\text{primed}\in R,\ \text{unprimed}\notin R
& N^2\cdot O\!\left(\frac{N^2}{n^{2}}\right)
& \frac{C_P^{(0)}}{n}
& O\!\left(\frac{N^4}{n^3}\right)\\
\text{both}\in R
& O\!\left(\frac{N^2}{n}\right)O\!\left(\frac{N^2}{n^{2}}\right)
& 1
& O\!\left(\frac{N^4}{n^3}\right)\\
(u,v)=(u',v'),\ \text{neither}\in R
& N^2\cdot O\!\left(\frac Nn\right)
& \frac{C_P^{(0)}}{n}
& O\!\left(\frac{N^3}{n^{2}}\right)
\end{array}
\]
Each count follows from \eqref{eq:R_count} and \eqref{eq:exc_count}. In the first
row only exceptionality is used, giving $N^2$ unprimed pairs times
$O(N^2/n)$ partners, and the case hypothesis bounds $Q$ directly. In the second
row the unprimed pair ranges over $R$, of size $O(N^2/n)$, and the partner over
the $O(N^2/n)$ exceptional ones in \eqref{eq:exc_count}; since the primed pair is not resonant,
$Q\leq q'<C_P^{(0)}/n$. In the third row the unprimed pair is free and the
partner must be both exceptional and resonant, so there would uniformly be $O(N^2/n^2)$ partners by
\eqref{eq:exc_count}; since the unprimed pair is not resonant,
$Q\leq q<C_P^{(0)}/n$. In the fourth row both restrictions apply at once, and the
trivial bound $Q\le1$ already suffices. In the last row the coupling indices
coincide, so $(u,v)$ and $(i,j)$ are free while $(i',j')$ must meet
$\mathcal N_{ij}^{(cl)}$, which by
Lemma~\ref{lem:complete_neighborhood_size} leaves $O(N/n)$ choices; the unprimed
pair is not resonant, so again $Q\leq q<C_P^{(0)}/n$. The resulting
$O(N^3/n^2)$ is $O(N^4/n^3)$ because $n\leq N$.

Adding the regular contribution and the six exceptional ones,
$\Sigma=O(N^4/n^3)$, whence
\[
\mathcal E_2=O\!\left(\frac{1}{\sigma_n^{4}n^{3}}\right),
\qquad
\sqrt{\mathcal E_2}=O\!\left(\frac{1}{\sigma_n^{2}n^{3/2}}\right).
\]
Combining with $\mathcal E_1=O(\sigma_n^{-3}n^{-2})$,
Proposition~\ref{prop:scalar_clt} gives
\[
d_W(S_n,Z)
=
O\!\left(
\frac{1}{\sigma_n^3 n^2}
+
\frac{1}{\sigma_n^{2}n^{3/2}}
\right).
\]
Under $\sigma_n^2\ge\lambda_{\min}(\boldsymbol\Sigma)=\omega(n^{-4/3})$ we have
$\sigma_n^3n^2=\omega(1)$ and $\sigma_n^{2}n^{3/2}=\omega(n^{1/6})\to\infty$, so
both terms are $o(1)$ and $d_W(S_n,Z)\to 0$.

\medskip
In both randomization schemes $d_W(S_n,Z)\to 0$, hence $S_n\xrightarrow{d}\mathcal N(0,1)$, which completes the proof of part (i).

\medskip
\noindent\textbf{Proof of part (ii).}
Here $\mathbf A=\boldsymbol\Sigma(\mathbf C)$ is random. Write
\[
\lambda_N:=\inf_{\mathbf c\in\mathrm{supp}(\mathbf C)}
\min\{\lambda_{\min}(\boldsymbol\Sigma_{\mathrm{marg}}(\mathbf c)),
\lambda_{\min}(\boldsymbol\Sigma(\mathbf c))\},
\]
so that the eigenvalue condition of part (ii) reads $\lambda_N=\omega(N^{-4/3})$. Because $\mathbf X_{ij}$ and $\frac{g_i}{N_i}\mathbf B_{(*,ij)}\circ(Y_{ij}\mathbf 1-\boldsymbol\tau)$ differ by the deterministic vector $\frac{g_i}{N_i}(\boldsymbol\tau_{ij}-\boldsymbol\tau)$, the matrix $\boldsymbol\Sigma(\mathbf C)$ from Theorem~\ref{thm:variance_order}(ii) is the conditional covariance matrix of $\sum_{(i,j)}\mathbf X_{ij}$ given $\mathbf C$.

For (b), since $\mathbb E[\sum_{(i,j)}\mathbf X_{ij}\mid\mathbf C]=\mathbf 0$ and $\boldsymbol\Sigma(\mathbf C)$ is the conditional variance matrix of $\sum_{(i,j)}\mathbf X_{ij}$,
\[
\mathbb E\biggl[\Bigl\|\boldsymbol\Sigma(\mathbf C)^{-1/2}\sum_{(i,j)}\mathbf X_{ij}\Bigr\|^2\biggr]
=\mathbb E\biggl[\mathrm{tr}\Bigl(\boldsymbol\Sigma(\mathbf C)^{-1/2}\,
\mathrm{Var}\Bigl(\sum_{(i,j)}\mathbf X_{ij}\Bigm|\mathbf C\Bigr)\,
\boldsymbol\Sigma(\mathbf C)^{-1/2}\Bigr)\biggr]
=\mathrm{tr}(I)=d,
\]
so $\|\boldsymbol\Sigma(\mathbf C)^{-1/2}\sum_{(i,j)}\mathbf X_{ij}\|=O_p(1)$ by Markov's inequality. The entrywise bounds in the proof of Theorem~\ref{thm:variance_order}(ii) hold uniformly for every realization of $\mathbf C$, so $\|\boldsymbol\Sigma(\mathbf c)\|=O(N^{-1})$ uniformly over $\mathbf c\in\mathrm{supp}(\mathbf C)$, and the square-root condition number in $(\star)$ satisfies $\|\boldsymbol\Sigma(\mathbf C)^{1/2}\|\,\|\boldsymbol\Sigma(\mathbf C)^{-1/2}\|\le(O(N^{-1})/\lambda_N)^{1/2}=o(N^{1/6})$. With $\boldsymbol\delta=O_p(N^{-1/2})$, the bound $(\star)$ gives $\|\mathbf R\|=o(N^{1/6})\cdot O_p(N^{-1/2})\cdot O_p(1)=o_p(1)$.

For (a), the randomness of $\boldsymbol\Sigma(\mathbf C)$ changes the argument. Unlike part (i), the normalizer $\boldsymbol\Sigma(\mathbf C)^{-1/2}$ is itself random and does not converge to a constant matrix, so $\boldsymbol\Sigma(\mathbf C)^{-1/2}\sum_{(i,j)}\mathbf X_{ij}$ is not a fixed linear image of $\sum_{(i,j)}\mathbf X_{ij}$ and cannot be standardized by a single deterministic variance; correspondingly, its unconditional law is a mixture of Gaussians indexed by the realization of $\mathbf C$ rather than a single Gaussian, and the Cram\'er--Wold reduction of part (i) does not apply directly. We therefore condition on $\mathbf C=\mathbf c$, under which the normalizer becomes deterministic, establish the central limit theorem for each realization with a Berry--Esseen bound that is uniform over $\mathbf c$, and finally integrate over $\mathbf C$.

Fix a realization $\mathbf c\in\mathrm{supp}(\mathbf C)$ and work under the conditional law $\mathbb P(\,\cdot\mid\mathbf C=\mathbf c)$. Each $\mathbf X_{ij}$ is a function of $(\mathbf W_{\mathcal N_{ij}},\mathbf C_{\mathcal N_{ij}})$, hence of $\mathbf W_{\mathcal N_{ij}}$ alone once $\mathbf C=\mathbf c$ is fixed. By Assumptions~\ref{ass:independent_rand}(ii) and~\ref{ass:cluster_assignment}(ii), the unit-level assignments $\{W_{ij}\}$ are mutually independent given $\mathbf C$. Fix $(i,j)$: every summand $\mathbf X_{i'j'}$ with $\mathcal N_{i'j'}\cap\mathcal N_{ij}=\varnothing$ is a function of the assignments on $\mathcal N_{i'j'}$, an index set disjoint from $\mathcal N_{ij}$, so under $\mathbb P(\,\cdot\mid\mathbf C=\mathbf c)$ the whole family $\{\mathbf X_{i'j'}:\mathcal N_{i'j'}\cap\mathcal N_{ij}=\varnothing\}$ is \emph{jointly} independent of $\mathbf X_{ij}$. Hence the array $\{\mathbf X_{ij}\}$ admits the dependency neighborhoods $\bar{\mathcal N}_{ij}:=\{(i',j'):\mathcal N_{ij}\cap\mathcal N_{i'j'}\neq\varnothing\}$ in the sense of Lemma~\ref{lem:stein_bernoulli}. Since every unit has at most $C_{\mathcal N}$ interference neighbors by Assumption~\ref{ass:CLT_regularity}(i), $D:=\sup_{(i,j)\in\mathcal P}|\bar{\mathcal N}_{ij}|\le C_{\mathcal N}^2$ for every realization $\mathbf c$.

By the Cram\'er--Wold theorem, (a) holds once we show $\mathbf a^\top\boldsymbol\Sigma(\mathbf C)^{-1/2}\sum_{(i,j)}\mathbf X_{ij}\xrightarrow{d}\mathcal N(0,1)$ for every deterministic unit vector $\mathbf a$. Conditionally on $\mathbf C=\mathbf c$, this statistic equals
\[
S_n(\mathbf c):=\sum_{(i,j)}\tilde X_{ij}(\mathbf c),
\qquad
\tilde X_{ij}(\mathbf c):=\mathbf b_{\mathbf c}^\top\mathbf X_{ij},
\qquad
\mathbf b_{\mathbf c}:=\boldsymbol\Sigma(\mathbf c)^{-1/2}\mathbf a,
\]
where $\mathbf b_{\mathbf c}$ is nonrandom once $\mathbf c$ is fixed. Using $\mathbb E[\mathbf X_{ij}\mid\mathbf C=\mathbf c]=\mathbf 0$ from Assumption~\ref{ass:weight_cluster_agnostic}(ii) and the fact that $\boldsymbol\Sigma(\mathbf c)$ is the conditional covariance matrix of $\sum_{(i,j)}\mathbf X_{ij}$ given $\mathbf C=\mathbf c$,
\[
\mathbb E[S_n(\mathbf c)\mid\mathbf C=\mathbf c]=0,
\qquad
\mathrm{Var}\!\left(S_n(\mathbf c)\mid\mathbf C=\mathbf c\right)
=\mathbf b_{\mathbf c}^\top\boldsymbol\Sigma(\mathbf c)\,\mathbf b_{\mathbf c}
=\mathbf a^\top\mathbf a=1.
\]
Applying Lemma~\ref{lem:stein_bernoulli} under $\mathbb P(\,\cdot\mid\mathbf C=\mathbf c)$ with $\sigma^2=1$,
\begin{align*}
d_W\bigl(\mathcal L(S_n(\mathbf c)\mid\mathbf C=\mathbf c),\,\mathcal N(0,1)\bigr)
\le{}&
D^2\sum_{(i,j)}\mathbb E\bigl[|\tilde X_{ij}(\mathbf c)|^3\mid\mathbf C=\mathbf c\bigr]\\
&+\frac{\sqrt{26}\,D^{3/2}}{\sqrt\pi}
\sqrt{\sum_{(i,j)}\mathbb E\bigl[\tilde X_{ij}(\mathbf c)^4\mid\mathbf C=\mathbf c\bigr]}.
\end{align*}
Assumptions~\ref{ass:CLT_regularity}(iii)--(vi) give $\|\mathbf X_{ij}\|\le\sqrt d\,\|\mathbf X_{ij}\|_\infty\le C/N$ almost surely for a global constant $C$ (using $g_i/N_i\le C_N g_i\cdot n/N=O(1/N)$), and $\|\mathbf b_{\mathbf c}\|\le\|\boldsymbol\Sigma(\mathbf c)^{-1/2}\|=\lambda_{\min}(\boldsymbol\Sigma(\mathbf c))^{-1/2}\le\lambda_N^{-1/2}$, whence $|\tilde X_{ij}(\mathbf c)|\le C\lambda_N^{-1/2}/N$ almost surely. Therefore, uniformly over $\mathbf c\in\mathrm{supp}(\mathbf C)$,
\begin{align*}
\sum_{(i,j)}\mathbb E\bigl[|\tilde X_{ij}(\mathbf c)|^3\mid\mathbf C=\mathbf c\bigr]
&=O\!\left(\frac{\lambda_N^{-3/2}}{N^{2}}\right),\\
\sqrt{\sum_{(i,j)}\mathbb E\bigl[\tilde X_{ij}(\mathbf c)^4\mid\mathbf C=\mathbf c\bigr]}
&=O\!\left(\frac{\lambda_N^{-1}}{N^{3/2}}\right),
\end{align*}
and since $D\le C_{\mathcal N}^2=O(1)$ and $\lambda_N=\omega(N^{-4/3})$,
\begin{align*}
\sup_{\mathbf c\in\mathrm{supp}(\mathbf C)}
d_W\bigl(\mathcal L(S_n(\mathbf c)\mid\mathbf C=\mathbf c),\,\mathcal N(0,1)\bigr)
&=O\!\left(\frac{\lambda_N^{-3/2}}{N^{2}}\right)+O\!\left(\frac{\lambda_N^{-1}}{N^{3/2}}\right)\\
&=o(1)+o\!\left(N^{-1/6}\right)
=o(1).
\end{align*}
Thus the conditional central limit theorem holds for every realization of $\mathbf C$, with the Wasserstein distance to $\mathcal N(0,1)$ vanishing uniformly over the support of $\mathbf C$.

Integrating over $\mathbf C$ then yields the unconditional statement. For any $1$-Lipschitz test function $f$, writing $S_n:=\mathbf a^\top\boldsymbol\Sigma(\mathbf C)^{-1/2}\sum_{(i,j)}\mathbf X_{ij}$,
\begin{align*}
\bigl|\mathbb E[f(S_n)]-\mathbb E[f(Z)]\bigr|
&\le
\mathbb E\Bigl[\bigl|\mathbb E[f(S_n)\mid\mathbf C]-\mathbb E[f(Z)]\bigr|\Bigr]\\
&\le
\sup_{\mathbf c\in\mathrm{supp}(\mathbf C)}
d_W\bigl(\mathcal L(S_n(\mathbf c)\mid\mathbf C=\mathbf c),\,\mathcal N(0,1)\bigr)
\to 0,
\end{align*}
so $S_n\xrightarrow{d}\mathcal N(0,1)$ for every unit vector $\mathbf a$, and the Cram\'er--Wold theorem establishes (a). Combined with the reduction (b) through $(\star)$, this gives $\boldsymbol\Sigma(\mathbf C)^{-1/2}(\hat{\boldsymbol\tau}-\boldsymbol\tau)\Rightarrow\mathcal N(\mathbf 0,I)$; the marginal counterpart $\boldsymbol\Sigma_{\mathrm{marg}}(\mathbf C)^{-1/2}(\hat{\boldsymbol\tau}_{\mathrm{marg}}-\boldsymbol\tau_{\mathrm{marg}})\Rightarrow\mathcal N(\mathbf 0,I)$ follows by the identical argument with $\boldsymbol\Sigma_{\mathrm{marg}}(\mathbf C)$ in place of $\boldsymbol\Sigma(\mathbf C)$, completing the proof of part (ii).
\end{proof}

\subsection{Proof of variance conservativeness}

Throughout this subsection, $\|\mathbf{V}_{ij}\|_\infty=\max_k\operatorname{ess\,sup}|(\mathbf{V}_{ij})_k|$ for random vectors and $\|\mathbf{M}\|_\infty=\max_{k,l}|M_{kl}|$ for deterministic matrices; the operator norm satisfies $\|\mathbf{M}\|\le d\|\mathbf{M}\|_\infty$ for $d\times d$ matrices.

\subsubsection{Proof of Theorem \ref{thm:var_1}}
\begin{proof}[Proof of Theorem \ref{thm:var_1}]
  We prove the result for the HAC estimators
  $\hat{\mathbf\Sigma}^{\mathbf K_1}$ and $\hat{\mathbf\Sigma}^{\mathbf K_2}$.
  The marginal counterparts
  $\hat{\mathbf\Sigma}_{\mathrm{marg}}^{\mathbf K_1}$ and
  $\hat{\mathbf\Sigma}_{\mathrm{marg}}^{\mathbf K_2}$ follow from an identical
  argument upon replacing $\mathbf B$ with
  $\mathbf B_{\mathrm{marg}}$ and $\boldsymbol\tau_{ij}$ with
  $\boldsymbol\tau_{ij,\mathrm{marg}}$.

We begin by introducing the notation used throughout the proof. Define
\[
\mathbf V_{ij}
:=
\frac{g_i}{N_i}\mathbf B_{(*,ij)}\circ(Y_{ij}\mathbf 1-\boldsymbol\tau),
\qquad
\hat{\mathbf V}_{ij}
:=
\frac{g_i}{N_i}\mathbf B_{(*,ij)}\circ(Y_{ij}\mathbf 1-\hat{\boldsymbol\tau}),
\]
together with
\[
\mathbf D_{ij}:=\frac{g_i}{N_i}\mathbf B_{(*,ij)},
\qquad
\boldsymbol\delta:=\hat{\boldsymbol\tau}-\boldsymbol\tau,
\]
so that $\hat{\mathbf V}_{ij}=\mathbf V_{ij}-\mathbf D_{ij}\circ\boldsymbol\delta$. In preparation for the bias analysis, we further decompose $\mathbf V_{ij}=\mathbf U_{ij,1}+\mathbf U_{ij,2}$, where
\[
\mathbf U_{ij,1}
:=
\frac{g_i}{N_i}\mathbf B_{(*,ij)}\circ(Y_{ij}\mathbf 1-\boldsymbol\tau_{ij})+\frac{g_i}{N_i}\left(\mathbf B_{(*,ij)}-\mathbf{1}\right)\circ(\boldsymbol\tau_{ij}-\boldsymbol\tau),
\]
and
\[
\mathbf U_{ij,2}
:=
\frac{g_i}{N_i}\circ(\boldsymbol\tau_{ij}-\boldsymbol\tau).
\]
By the same argument as in the proof of Theorem~\ref{thm:clt}, we have
\(
\mathbb E[\mathbf U_{ij,1}]=\mathbf 0.
\)
The bounds in Assumptions~\ref{ass:CLT_regularity}(iii)--(vi) give
\[
\sup_{(i,j)}\|\mathbf V_{ij}\|_\infty=O\!\left(\frac{1}{N}\right),
\qquad
\sum_{(i,j)}\|\mathbf V_{ij}\|_\infty=O(1),
\]
and the analogous bounds hold for $\mathbf U_{ij,1}$, $\mathbf U_{ij,2}$, and $\mathbf D_{ij}$.
From Theorem~\ref{thm:variance_order}(i), we also have $\|\boldsymbol\delta\|=O_P(n^{-1/2})$. Moreover, it's worth noting that $\mathbf{V}_{ij}$ is defined here in the same way as in Theorem~\ref{thm:variance_order}.

Substituting $\hat{\mathbf V}_{ij}=\mathbf V_{ij}-\mathbf D_{ij}\circ\boldsymbol\delta$ into the definition of $\hat{\mathbf\Sigma}^{\mathbf K_1}$ and expanding the outer product, we obtain
\begin{align*}
\hat{\mathbf\Sigma}^{\mathbf K_1}
&=
\sum_{(i,j),(i',j')}(\mathbf K_1)_{ij,i'j'}\mathbf V_{ij}\mathbf V_{i'j'}^\top
-\sum_{(i,j),(i',j')}(\mathbf K_1)_{ij,i'j'}\mathbf V_{ij}(\mathbf D_{i'j'}\circ\boldsymbol\delta)^\top \\
&\quad
-\sum_{(i,j),(i',j')}(\mathbf K_1)_{ij,i'j'}(\mathbf D_{ij}\circ\boldsymbol\delta)\mathbf V_{i'j'}^\top
+\sum_{(i,j),(i',j')}(\mathbf K_1)_{ij,i'j'}(\mathbf D_{ij}\circ\boldsymbol\delta)(\mathbf D_{i'j'}\circ\boldsymbol\delta)^\top \\
&=: \mathbf T_{1n}+\mathbf T_{2n}+\mathbf T_{2n}^\top+\mathbf T_{3n}.
\end{align*}

Using the uniform bounds on the summands established above, we obtain
\[
  \|\mathbf T_{2n}\|\leq \|\mathbf{K}_1\|_{F}^2\sup_{(i,j)}\|\mathbf V_{ij}\mathbf D_{i'j'}^\top\|_{\infty}\|\boldsymbol\delta\|
  = O\left(\frac{N^2}{n}\right)O\left(N^{-2}\right)O_P\left(n^{-1/2}\right)=o_P(n^{-1})
\]
and
\[
  \|\mathbf T_{3n}\|\leq \|\mathbf{K}_1\|_{F}^2\sup_{(i,j)}\|\mathbf D_{i'j'}\mathbf D_{i'j'}^\top\|_{\infty}\|\boldsymbol\delta\|^2
  = O\left(\frac{N^2}{n}\right)O\left(N^{-2}\right)O_P\left(n^{-1}\right)=o(n^{-1})
\]

We now analyze the leading term $\mathbf T_{1n}$, beginning with its expectation. Since $\mathbf V_{ij}=\mathbf U_{ij,1}+\mathbf U_{ij,2}$, we have
\begin{align*}
  \mathbb{E}[\mathbf T_{1n}]
  &=
  \sum_{(i,j),(i',j')}(\mathbf K_1)_{ij,i'j'}\mathbb{E}[\mathbf U_{ij,1}\mathbf U_{i'j',1}^\top]
  -\sum_{(i,j),(i',j')}(\mathbf K_1)_{ij,i'j'}\mathbb{E}[\mathbf U_{ij,1}\mathbf U_{i'j',2}^\top] \\
  &\quad
  -\sum_{(i,j),(i',j')}(\mathbf K_1)_{ij,i'j'}\mathbb{E}[\mathbf U_{ij,2}\mathbf U_{i'j',1}^\top]
  +\sum_{(i,j),(i',j')}(\mathbf K_1)_{ij,i'j'}\mathbb{E}[\mathbf U_{ij,2}\mathbf U_{i'j',2}^\top] \\
  &=\sum_{(i,j),(i',j')}(\mathbf K_1)_{ij,i'j'}\mathbb{E}[\mathbf U_{ij,1}\mathbf U_{i'j',1}^\top]
  +\sum_{(i,j),(i',j')}(\mathbf K_1)_{ij,i'j'}\mathbf U_{ij,2}\mathbf U_{i'j',2}^\top,
  \end{align*}
where the last equality uses $\mathbb{E}[\mathbf U_{i'j',1}]=\mathbf 0$ together with the fact that $\mathbf U_{i'j',2}$ is nonrandom.
Recall from the proof of Theorem~\ref{thm:variance_order} that
\[
\boldsymbol\Sigma
  =
  \sum_{(i,j),(i',j')}(\mathbf K_2)_{ij,i'j'}\mathrm{Cov}(\mathbf V_{ij},\mathbf V_{i'j'})
  =
  \sum_{(i,j),(i',j')}(\mathbf K_2)_{ij,i'j'}\mathbb{E}[\mathbf U_{ij,1}\mathbf U_{i'j',1}^\top],
  \]
so
\[
  \mathbb{E}[\mathbf T_{1n}]-\boldsymbol\Sigma=-\sum_{(i,j),(i',j')}(\mathbf K_2-\mathbf K_1)_{ij,i'j'}\mathbb{E}[\mathbf U_{ij,1}\mathbf U_{i'j',1}^\top]
  +\sum_{(i,j),(i',j')}(\mathbf K_1)_{ij,i'j'}\mathbf U_{ij,2}\mathbf U_{i'j',2}^\top.
\]
By Assumption~\ref{ass:sparsity}, the difference $\mathbf K_2-\mathbf K_1$ has at most $o(|\{((i,j),(i',j'))\in\mathcal{P}\times\mathcal{P}:\mathbf (K_1)_{ij,i'j'\neq 0}\}|)=o(\frac{N^2}{n})$ nonzero entries.
Combining this with $\sup_{(i,j)}\|\mathbf U_{ij,1}\|_{\infty}=O(N^{-1})$ yields
\[
\sum_{(i,j),(i',j')}(\mathbf K_2-\mathbf K_1)_{ij,i'j'}\mathbb{E}[\mathbf U_{ij,1}\mathbf U_{i'j',1}^\top]=o\left(\frac{N^2}{n}\right)\cdot O(N^{-2})=o_P(n^{-1}).
\]
Therefore,
\[
  \mathbb{E}[\mathbf T_{1n}]
  =
  \boldsymbol\Sigma+\mathbf R_1+o_P(n^{-1}),
  \]
where
\[
  \mathbf R_1:=\sum_{(i,j),(i',j')}(\mathbf K_1)_{ij,i'j'}\mathbf U_{ij,2}\mathbf U_{i'j',2}^\top
\]
is positive semidefinite. To see this, take any vector $\mathbf a$ with dimension compatible with $\mathbf U_{ij,2}$, and define
\[
  \mathbf u
  =
  \big(\mathbf a^\top \mathbf U_{ij,2}\big)_{(i,j)\in\mathcal P}.
\]
Then
\[
  \mathbf a^\top \mathbf R_1 \mathbf a
  =
  \mathbf u^\top \mathbf K_1 \mathbf u
  \geq 0,
\]
where the inequality follows from the positive semidefiniteness of $\mathbf K_1$. Moreover, the bias matrix is of the same order as the target variance: since $\mathbf K_1$ has $\|\mathbf K_1\|_F^2=\sum_i N_i^2=O(N^2/n)$ nonzero entries and $\sup_{(i,j)}\|\mathbf U_{ij,2}\|_\infty=O(N^{-1})$,
\[
  \|\mathbf R_1\|\leq\|\mathbf K_1\|_F^2\sup_{(i,j)}\|\mathbf U_{ij,2}\mathbf U_{i'j',2}^\top\|_\infty=O(N^2/n)\,O(N^{-2})=O(n^{-1}).
\]

We can also show that
\[
\mathbf T_{1n}-\mathbb E[\mathbf T_{1n}]=o_P(n^{-1})
\]
via a variance bound. For any fixed coordinate pair $(l,l')$, the
$(l,l')$-entry of $\mathbf T_{1n}$ is
\[
  (\mathbf T_{1n})_{l,l'}
  =
  \sum_{(i,j),(i',j')}
  (\mathbf K_1)_{ij,i'j'} V_{ij,l} V_{i'j',l'}.
\]
We claim that
\[
  \mathrm{Var}\left((\mathbf T_{1n})_{l,l'}\right)=o(n^{-2}).
\]
Expanding the variance gives
\begin{align*}
  \mathrm{Var}\left((\mathbf T_{1n})_{l,l'}\right)&=\mathrm{Var}\left(\sum_{(i,j),(i',j')}
  (\mathbf K_1)_{ij,i'j'} V_{ij,l} V_{i'j',l'}\right)\\
  &=\sum_{1\leq i,i'\leq n,1\leq j_1,j_2\leq N_i, 1\leq j'_1,j'_2\leq N_{i'}}\mathrm{Cov}\left(V_{ij_1,l} V_{ij_2,l'},V_{i'j'_1,l} V_{i'j'_2,l'}\right).
\end{align*}
By Assumptions \ref{ass:cluster_assignment} and \ref{ass:independent_rand}(i), 
\[
(V_{ij_1,l},V_{ij_2,l'})\indep (V_{i'j'_1,l},V_{i'j'_2,l'}) \quad\text{if}\quad 
(\mathcal{N}_{ij_1}^{(cl)}\cup \mathcal{N}_{ij_2}^{(cl)})\cap (\mathcal{N}_{i'j'_1}^{(cl)}\cup \mathcal{N}_{i'j'_2}^{(cl)})=\varnothing.
\]
When this independence condition fails, at least one of $(i',j'_1)$ and $(i',j'_2)$ lies in the joint dependency neighborhood of $(i,j_1)$ and $(i,j_2)$, defined as
\[
\cup_{(i'',j''):i''\in\mathcal{N}_{ij_1}^{(cl)}\cup \mathcal{N}_{ij_2}^{(cl)}}\mathcal{N}_{i''j''}
\]
whose cardinality satisfies
\[
|\cup_{(i'',j''):i''\in\mathcal{N}_{ij_1}^{(cl)}\cup \mathcal{N}_{ij_2}^{(cl)}}\mathcal{N}_{i''j''}|\leq |\mathcal{N}_{ij_1}^{(cl)}\cup \mathcal{N}_{ij_2}^{(cl)}|\max_{i''}\{\cup_{j''=1}^{N_{i''}}\mathcal{N}_{i''j''}\}\leq 2C_{\mathcal{N}^{(cl)}}^2\max_{i''}N_{i''}=O\left(\frac{N}{n}\right).
\]
For a fixed triplet $(i,j_1,j_2)$, the number of triplets $(i',j'_1,j'_2)$ violating the independence condition is at most
$O\left(\frac{N}{n}\right)\cdot O\left(\frac{N}{n}\right)+O\left(\frac{N}{n}\right)\cdot O\left(\frac{N}{n}\right)=O\left(\frac{N^2}{n^2}\right)$.
This count is obtained by enumerating all units $(i',j'_1)$ in the joint dependency neighborhood together with all $(i',j'_2)$ in the same cluster, and vice versa.
Hence the total number of sextuples $(i,j_1,j_2,i',j'_1,j'_2)$ violating the independence condition is $N\cdot O\left(\frac{N}{n}\right)\cdot O\left(\frac{N^2}{n^2}\right)=O\left(\frac{N^4}{n^3}\right)$.
Combining this with the almost-sure bound
\[
\sup_{(i,j,l)}|V_{ij,l}|\le C/N\quad\text{a.s.}
\]
for a constant $C$, we obtain
\begin{align*}
  \mathrm{Var}\left((\mathbf T_{1n})_{l,l'}\right)=O\left(\frac{N^4}{n^3}\right)\cdot O\left(\frac{1}{N^4}\right)=O\left(\frac{1}{n^3}\right)=o\left(n^{-2}\right).
\end{align*}
By Chebyshev's inequality,
\[
\mathbb{E}[(\mathbf T_{1n})_{l,l'}]-(\mathbf T_{1n})_{l,l'}=O_P\left(\mathrm{Var}\left((\mathbf T_{1n})_{l,l'}\right)^{1/2}\right)=o_P(n^{-1})
\]
for every coordinate pair $(l,l')$. It follows that
\[
\mathbf T_{1n}-\mathbb E[\mathbf T_{1n}]=o_P(n^{-1}).
\]

Combining the bounds on $\mathbf T_{2n}$, $\mathbf T_{3n}$, the bias of $\mathbf T_{1n}$, and the concentration of $\mathbf T_{1n}$ around its mean, we conclude that
\[
  \hat{\mathbf\Sigma}^{\mathbf K_1}=\boldsymbol\Sigma+\mathbf R_1+o_P(n^{-1})
\]
with $\mathbf R_1\succeq\mathbf 0$.

Turning to $\hat{\mathbf\Sigma}^{\mathbf K_2}$, we write
\begin{align*}
  \hat{\mathbf\Sigma}^{\mathbf K_2}-\hat{\mathbf\Sigma}^{\mathbf K_1}
  &=
  \sum_{(i,j),(i',j')}(\mathbf K_2-\mathbf K_1)_{ij,i'j'}\mathbf V_{ij}\mathbf V_{i'j'}^\top
  -\sum_{(i,j),(i',j')}(\mathbf K_2-\mathbf K_1)_{ij,i'j'}\mathbf V_{ij}(\mathbf D_{i'j'}\circ\boldsymbol\delta)^\top \\
  &\quad
  -\sum_{(i,j),(i',j')}(\mathbf K_2-\mathbf K_1)_{ij,i'j'}(\mathbf D_{ij}\circ\boldsymbol\delta)\mathbf V_{i'j'}^\top
  +\sum_{(i,j),(i',j')}(\mathbf K_2-\mathbf K_1)_{ij,i'j'}(\mathbf D_{ij}\circ\boldsymbol\delta)(\mathbf D_{i'j'}\circ\boldsymbol\delta)^\top \\
  &=: \tilde{\mathbf T}_{1n}+\tilde{\mathbf T}_{2n}+\tilde{\mathbf T}_{2n}^\top+\tilde{\mathbf T}_{3n}.
  \end{align*}
By the same reasoning as above, $\|\tilde{\mathbf T}_{2n}\|,\|\tilde{\mathbf T}_{3n}\|=o_P(n^{-1})$.
For $\tilde{\mathbf T}_{1n}$, the sparsity bound gives
\[
  \|\tilde{\mathbf T}_{1n}\|\leq\|\mathbf K_2-\mathbf K_1\|_{F}^2\sup_{(i,j)}\|\mathbf V_{ij}\mathbf V_{i'j'}^\top\|_{\infty}=o\left(\frac{N^2}{n}\right)O\left(\frac{1}{N^2}\right)=o(n^{-1}).
\]
Therefore $\hat{\mathbf\Sigma}^{\mathbf K_2}-\hat{\mathbf\Sigma}^{\mathbf K_1}=o_P(n^{-1})$, and hence
\[
  \hat{\mathbf\Sigma}^{\mathbf K_2}=\boldsymbol\Sigma+\mathbf R_1+o_P(n^{-1}).
\]
This completes the proof.
\end{proof}

\subsubsection{Proof of Theorem \ref{thm:var_2}}
\begin{proof}[Proof of Theorem \ref{thm:var_2}]
We prove the result for the HAC estimator $\hat{\mathbf\Sigma}_2=\hat{\mathbf\Sigma}^{\mathbf K_3^+}$.
The marginal counterpart $\hat{\mathbf\Sigma}_{2,\mathrm{marg}}=\hat{\mathbf\Sigma}_{\mathrm{marg}}^{\mathbf K_3^+}$ follows from an identical argument upon replacing $\mathbf B$ with $\mathbf B_{\mathrm{marg}}$ and $\boldsymbol\tau_{ij}$ with $\boldsymbol\tau_{ij,\mathrm{marg}}$. The structure of the proof parallels that of Theorem~\ref{thm:var_1}.

We begin by introducing the notation used throughout the proof, following the conventions in the proof of Theorem~\ref{thm:var_1}. Define
\[
\mathbf V_{ij}
:=
\frac{g_i}{N_i}\mathbf B_{(*,ij)}\circ(Y_{ij}\mathbf 1-\boldsymbol\tau),
\qquad
\hat{\mathbf V}_{ij}
:=
\frac{g_i}{N_i}\mathbf B_{(*,ij)}\circ(Y_{ij}\mathbf 1-\hat{\boldsymbol\tau}),
\]
together with
\[
\mathbf D_{ij}:=\frac{g_i}{N_i}\mathbf B_{(*,ij)},
\qquad
\boldsymbol\delta:=\hat{\boldsymbol\tau}-\boldsymbol\tau,
\]
so that $\hat{\mathbf V}_{ij}=\mathbf V_{ij}-\mathbf D_{ij}\circ\boldsymbol\delta$. In preparation for the bias analysis, we further decompose $\mathbf V_{ij}=\mathbf U_{ij,1}+\mathbf U_{ij,2}$, where
\[
\mathbf U_{ij,1}
:=
\frac{g_i}{N_i}\mathbf B_{(*,ij)}\circ(Y_{ij}\mathbf 1-\boldsymbol\tau_{ij})+\frac{g_i}{N_i}\left(\mathbf B_{(*,ij)}-\mathbf 1\right)\circ(\boldsymbol\tau_{ij}-\boldsymbol\tau),
\]
and
\[
\mathbf U_{ij,2}
:=
\frac{g_i}{N_i}(\boldsymbol\tau_{ij}-\boldsymbol\tau).
\]
The split is identical to the one used in the proof of Theorem~\ref{thm:var_1}: $\mathbf U_{ij,2}$ is a deterministic constant, 
and the same argument as in the proof of Theorem~\ref{thm:clt} (applied here under Assumption~\ref{ass:weight_cluster_agnostic}, which gives $\mathbb E[\mathbf B_{(*,ij)}\circ(Y_{ij}\mathbf 1)\mid\mathbf C]=\boldsymbol\tau_{ij}$ and $\mathbb E[\mathbf B_{(*,ij)}\mid\mathbf C]=\mathbf 1$) 
yields $\mathbb E[\mathbf U_{ij,1}\mid \mathbf{C}]=\mathbf 0$.
The bounds in Assumptions~\ref{ass:CLT_regularity}(iii)--(vi) give
\[
\sup_{(i,j)}\|\mathbf V_{ij}\|_\infty=O\!\left(\frac{1}{N}\right),
\qquad
\sum_{(i,j)}\|\mathbf V_{ij}\|_\infty=O(1),
\]
and the analogous bounds hold for $\mathbf U_{ij,1}$, $\mathbf U_{ij,2}$, and $\mathbf D_{ij}$.
From Theorem~\ref{thm:variance_order}(ii), we also have $\|\boldsymbol\delta\|=O_P(N^{-1/2})$. Moreover, it's worth noting that $\mathbf{V}_{ij}$ is defined here in the same way as in Theorem~\ref{thm:variance_order}.

Substituting $\hat{\mathbf V}_{ij}=\mathbf V_{ij}-\mathbf D_{ij}\circ\boldsymbol\delta$ into the definition of $\hat{\mathbf\Sigma}_2$ and expanding the outer product, we obtain
\begin{align*}
\hat{\mathbf\Sigma}_2
&=
\sum_{(i,j),(i',j')}(\mathbf K_3^+)_{ij,i'j'}\mathbf V_{ij}\mathbf V_{i'j'}^\top
-\sum_{(i,j),(i',j')}(\mathbf K_3^+)_{ij,i'j'}\mathbf V_{ij}(\mathbf D_{i'j'}\circ\boldsymbol\delta)^\top \\
&\quad
-\sum_{(i,j),(i',j')}(\mathbf K_3^+)_{ij,i'j'}(\mathbf D_{ij}\circ\boldsymbol\delta)\mathbf V_{i'j'}^\top
+\sum_{(i,j),(i',j')}(\mathbf K_3^+)_{ij,i'j'}(\mathbf D_{ij}\circ\boldsymbol\delta)(\mathbf D_{i'j'}\circ\boldsymbol\delta)^\top \\
&=: \mathbf T_{1n}+\mathbf T_{2n}+\mathbf T_{2n}^\top+\mathbf T_{3n}.
\end{align*}

We first dispatch the plug-in error terms $\mathbf T_{2n}$ and $\mathbf T_{3n}$. 
The kernel $\mathbf K_3^+$ is not $\{0,1\}$-valued, so we bound the operator norm $\|\mathbf K_3^+\|$ rather than the Frobenius norm $\|\mathbf K_3^+\|_{F}$ to bound the quadratic forms. Since the row sums of $\mathbf K_3$ are uniformly bounded, 
\[
  \|\mathbf K_3^+\|\leq\|\mathbf K_3\|=O(1).
\]
So, for any fixed coordinate pair $(l,l')$, 
\begin{align*}
  \sum_{(i,j),(i',j')}(\mathbf K_3^+)_{ij,i'j'}\mathbf V_{ij,l}\mathbf D_{i'j',l'}\leq&
   \|\mathbf K_3^+\|N^{1/2}\sup_{(i,j,l)}\|\mathbf V_{ij,l}\|_{\infty}N^{1/2}\sup_{(i,j,l)}\|\mathbf D_{ij,l}\|_{\infty}\\
   =&O(1)N\sup_{(i,j,l)}\|\mathbf V_{ij}\|_{\infty}\sup_{(i,j,l)}\|\mathbf D_{ij}\|_{\infty}\\
   =&O(N)\cdot O\bigl(N^{-2}\bigr)=O(N^{-1}),
\end{align*}
and similarly,
\begin{align*}
  \sum_{(i,j),(i',j')}(\mathbf K_3^+)_{ij,i'j'}\mathbf D_{ij,l}\mathbf D_{i'j',l'}=O(N^{-1}).
\end{align*}
It follows that
\[
  \|\mathbf T_{2n}\|
  \le
  \|\sum_{(i,j),(i',j')}(\mathbf K_3^+)_{ij,i'j'}\mathbf V_{ij}(\mathbf D_{i'j'})^\top\| \|\boldsymbol\delta\| 
  =O(N^{-1})O_P(N^{-1/2})
  =o(N^{-1})
\]
and
\[
  \|\mathbf T_{3n}\|
  \le
  \|\sum_{(i,j),(i',j')}(\mathbf K_3^+)_{ij,i'j'}\mathbf D_{ij}(\mathbf D_{i'j'})^\top\| \|\boldsymbol\delta\|^2 
  =O(N^{-1})O_P(N^{-1})
  =o(N^{-1}).
\]

Having reduced $\hat{\mathbf\Sigma}_2$ to the leading term $\mathbf T_{1n}$ up to the $o_P(N^{-1})$ plug-in errors, we split $\mathbf T_{1n}$ into its conditional mean given $\mathbf C$ and the fluctuation around it,
\[
\mathbf T_{1n}
=
\mathbb E[\mathbf T_{1n}\mid\mathbf C]
+
\bigl(\mathbf T_{1n}-\mathbb E[\mathbf T_{1n}\mid\mathbf C]\bigr),
\]
and analyze the two pieces separately.

Consider first the conditional mean. Since $\mathbf V_{ij}=\mathbf U_{ij,1}+\mathbf U_{ij,2}$ with $\mathbb E[\mathbf U_{ij,1}\mid\mathbf C]=\mathbf 0$ and $\mathbf U_{i'j',2}$ nonrandom, the cross terms involving exactly one factor of $\mathbf U_{\cdot,1}$ vanish conditionally, and
\begin{align*}
  \mathbb E[\mathbf T_{1n}\mid\mathbf C]
  &=
  \sum_{(i,j),(i',j')}(\mathbf K_3^+)_{ij,i'j'}\mathbb E[\mathbf U_{ij,1}\mathbf U_{i'j',1}^\top\mid\mathbf C]
  +\sum_{(i,j),(i',j')}(\mathbf K_3^+)_{ij,i'j'}\mathbf U_{ij,2}\mathbf U_{i'j',2}^\top.
\end{align*}
Decomposing $\mathbf K_3^+=\mathbf K_3+\mathbf K_3^-$, where $\mathbf K_3^-:=(-\mathbf K_3)_+\succeq\mathbf 0$ is the lifting that makes $\mathbf K_3^+$ positive semidefinite, we further split the first sum as
\begin{align*}
\sum_{(i,j),(i',j')}(\mathbf K_3^+)_{ij,i'j'}\mathbb E[\mathbf U_{ij,1}\mathbf U_{i'j',1}^\top\mid\mathbf C]
={}&
\sum_{(i,j),(i',j')}(\mathbf K_3)_{ij,i'j'}\mathbb E[\mathbf U_{ij,1}\mathbf U_{i'j',1}^\top\mid\mathbf C]\\
&+
\sum_{(i,j),(i',j')}(\mathbf K_3^-)_{ij,i'j'}\mathbb E[\mathbf U_{ij,1}\mathbf U_{i'j',1}^\top\mid\mathbf C].
\end{align*}
Under Assumption~\ref{ass:weight_cluster_agnostic} we have $\mathbb E[\mathbf V_{ij}\mid\mathbf C]=\mathbf U_{ij,2}$, so $\mathbf U_{ij,1}=\mathbf V_{ij}-\mathbb E[\mathbf V_{ij}\mid\mathbf C]$ and hence $\mathbb E[\mathbf U_{ij,1}\mathbf U_{i'j',1}^\top\mid\mathbf C]=\mathrm{Cov}(\mathbf V_{ij},\mathbf V_{i'j'}\mid\mathbf C)$ for every pair. Given $\mathbf C$, Assumptions~\ref{ass:cluster_assignment} and~\ref{ass:independent_rand}(ii) make the unit-level assignments independent, so this conditional covariance vanishes whenever $\mathcal N_{ij}\cap\mathcal N_{i'j'}=\varnothing$. Since the support of $\mathbf K_3$ is exactly $\{(ij,i'j'):\mathcal N_{ij}\cap\mathcal N_{i'j'}\neq\varnothing\}$,
\[
\sum_{(i,j),(i',j')}(\mathbf K_3)_{ij,i'j'}\mathbb E[\mathbf U_{ij,1}\mathbf U_{i'j',1}^\top\mid\mathbf C]
=
\sum_{(i,j),(i',j')}\mathrm{Cov}(\mathbf V_{ij},\mathbf V_{i'j'}\mid\mathbf C)
=
\boldsymbol\Sigma(\mathbf C),
\]
the last equality being the definition of $\boldsymbol\Sigma(\mathbf C)$ in Theorem~\ref{thm:variance_order}(ii).
Therefore, similarly to the conservativeness decomposition in \cite{Gao2025},
\[
\mathbb E[\mathbf T_{1n}\mid\mathbf C]
=
\boldsymbol\Sigma(\mathbf C)+\mathbf R_2(\mathbf C),
\qquad
\mathbf R_2(\mathbf C)
=
\mathbf R_2^{(\boldsymbol\tau)}+\mathbf R_2^{(\mathrm{lift})}(\mathbf C),
\]
where
\begin{align*}
\mathbf R_2^{(\boldsymbol\tau)}
&:=
\sum_{(i,j),(i',j')}(\mathbf K_3^+)_{ij,i'j'}\mathbf U_{ij,2}\mathbf U_{i'j',2}^\top,\\
\mathbf R_2^{(\mathrm{lift})}(\mathbf C)
&:=
\sum_{(i,j),(i',j')}(\mathbf K_3^-)_{ij,i'j'}\mathbb E[\mathbf U_{ij,1}\mathbf U_{i'j',1}^\top\mid\mathbf C].
\end{align*}
The first remainder $\mathbf R_2^{(\boldsymbol\tau)}$ is deterministic and is exactly the analog of $\mathbf R_1$ in the proof of Theorem~\ref{thm:var_1}, positive semidefinite because $\mathbf K_3^+\succeq\mathbf 0$. The second remainder $\mathbf R_2^{(\mathrm{lift})}(\mathbf C)$ is positive semidefinite for every realization of $\mathbf C$: for any unit vector $\mathbf a$, writing $\mathbf w:=(\mathbf a^\top\mathbf U_{ij,1})_{(i,j)}$, we have $\mathbf a^\top\mathbf R_2^{(\mathrm{lift})}(\mathbf C)\mathbf a=\mathbb E[\mathbf w^\top\mathbf K_3^-\mathbf w\mid\mathbf C]\ge 0$ because $\mathbf K_3^-\succeq\mathbf 0$. Hence $\mathbf R_2(\mathbf C)\succeq\mathbf 0$ for every realization of $\mathbf C$, as claimed.

It remains to verify that $\mathbf R_2(\mathbf C)$ is, uniformly over the support of $\mathbf C$, of the same order $O(N^{-1})$ as the conditional target variance. Fix a realization $\mathbf c\in\mathrm{supp}(\mathbf C)$ and an arbitrary unit vector $\mathbf a$, and write $\mathbf u:=(\mathbf a^\top\mathbf U_{ij,2})_{(i,j)}$ and $\mathbf w:=(\mathbf a^\top\mathbf U_{ij,1})_{(i,j)}$, so that
\[
\mathbf a^\top\mathbf R_2^{(\boldsymbol\tau)}\mathbf a=\mathbf u^\top\mathbf K_3^+\mathbf u,
\qquad
\mathbf a^\top\mathbf R_2^{(\mathrm{lift})}(\mathbf c)\mathbf a=\mathbb E\bigl[\mathbf w^\top\mathbf K_3^-\mathbf w\mid\mathbf C=\mathbf c\bigr].
\]
Both kernels have operator norm controlled by that of $\mathbf K_3$: we already noted $\|\mathbf K_3^+\|\le\|\mathbf K_3\|=O(1)$, and since $\mathbf K_3^-=(-\mathbf K_3)_+$ collects the negative eigenvalues of $\mathbf K_3$, also $\|\mathbf K_3^-\|\le\|\mathbf K_3\|=O(1)$. The per-unit and summed sup-norm bounds on $\mathbf U_{ij,2}$ give $\|\mathbf u\|^2\le d\sup_{(i,j)}\|\mathbf U_{ij,2}\|_\infty\sum_{(i,j)}\|\mathbf U_{ij,2}\|_\infty=O(N^{-1})$, where $d$ is the matrix dimension; a similar bound holds for $\|\mathbf w\|^2$ almost surely, so $\mathbb E[\|\mathbf w\|^2\mid\mathbf C=\mathbf c]=O(N^{-1})$ hold uniformly for all $\mathbf c\in \text{supp}(\mathbf{C})$. Since $\mathbf K_3^+$ and $\mathbf K_3^-$ are positive semidefinite,
\[
0\le\mathbf a^\top\mathbf R_2(\mathbf c)\mathbf a
\le\|\mathbf K_3^+\|\,\|\mathbf u\|^2+\|\mathbf K_3^-\|\,\mathbb E\bigl[\|\mathbf w\|^2\mid\mathbf C=\mathbf c\bigr]
=O(N^{-1}),
\]
uniformly over unit vectors $\mathbf a$ and over $\mathbf c\in\mathrm{supp}(\mathbf C)$. Consequently $\sup_{\mathbf c\in\mathrm{supp}(\mathbf C)}\|\mathbf R_2(\mathbf c)\|=O(N^{-1})$.

Consider next the fluctuation $\mathbf T_{1n}-\mathbb E[\mathbf T_{1n}\mid\mathbf C]$, which we show is $o_P(N^{-1})$ through a conditional variance bound. For any fixed coordinate pair $(l,l')$, the $(l,l')$-entry of $\mathbf T_{1n}$ is
\[
(\mathbf T_{1n})_{l,l'}
=
\sum_{(i_1,j_1),(i_2,j_2)}(\mathbf K_3^+)_{i_1j_1,i_2j_2}V_{i_1j_1,l}V_{i_2j_2,l'},
\]
so that
\begin{multline*}
\mathrm{Var}\bigl((\mathbf T_{1n})_{l,l'}\mid\mathbf C\bigr)
=
\sum_{(i_1,j_1),(i_2,j_2),(i_3,j_3),(i_4,j_4)}
(\mathbf K_3^+)_{i_1j_1,i_2j_2}(\mathbf K_3^+)_{i_3j_3,i_4j_4}\\
\cdot\,\mathrm{Cov}\bigl(V_{i_1j_1,l}V_{i_2j_2,l'},\,V_{i_3j_3,l}V_{i_4j_4,l'}\mid\mathbf C\bigr).
\end{multline*}
Given $\mathbf C$, each $\mathbf V_{ij}$ is $\sigma(\mathbf W_{\mathcal N_{ij}})$-measurable, and by Assumptions~\ref{ass:cluster_assignment} and~\ref{ass:independent_rand}(ii) the unit-level assignments are conditionally independent, so
\[
(\mathbf V_{i_1j_1},\mathbf V_{i_2j_2})\indep (\mathbf V_{i_3j_3},\mathbf V_{i_4j_4})\mid\mathbf C
\quad\text{if}\quad
(\mathcal N_{i_1j_1}\cup\mathcal N_{i_2j_2})\cap(\mathcal N_{i_3j_3}\cup\mathcal N_{i_4j_4})=\varnothing.
\]
With the almost-sure bound $\sup_{(i,j,l)}|V_{ij,l}|\le C/N$ from Assumption~\ref{ass:CLT_regularity}(iii)--(vi), and noting that any overlap of unit-level neighborhoods forces an overlap of the corresponding cluster-level neighborhoods,
\begin{multline*}
\mathrm{Var}\bigl((\mathbf T_{1n})_{l,l'}\mid\mathbf C\bigr)
\le
\frac{4C^4}{N^4}\!\!\!\!\sum_{\substack{(i_1,j_1),(i_2,j_2),(i_3,j_3),(i_4,j_4):\\(\mathcal N_{i_1j_1}^{(cl)}\cup\mathcal N_{i_2j_2}^{(cl)})\cap(\mathcal N_{i_3j_3}^{(cl)}\cup\mathcal N_{i_4j_4}^{(cl)})\neq\varnothing}}\!\!\!\!
\bigl|(\mathbf K_3^+)_{i_1j_1,i_2j_2}\bigr|\,\bigl|(\mathbf K_3^+)_{i_3j_3,i_4j_4}\bigr|\\
=
O(N^{-4})\cdot o\!\left(N^2\right)
=
o\!\left(N^{-2}\right),
\end{multline*}
where the next-to-last equality invokes Assumption~\ref{ass:var_cluster_agnostic}. The bound is deterministic, hence uniform over $\mathbf C$; taking expectations gives $\mathbb E[\mathrm{Var}((\mathbf T_{1n})_{l,l'}\mid\mathbf C)]=o(N^{-2})$, and an application of Markov’s inequality and Chebyshev’s inequality yields
\[
(\mathbf T_{1n})_{l,l'}-\mathbb E[(\mathbf T_{1n})_{l,l'}\mid\mathbf C]=o_P(N^{-1})
\]
for every coordinate pair $(l,l')$. Since the matrix dimension is fixed, this gives
\[
\mathbf T_{1n}-\mathbb E[\mathbf T_{1n}\mid\mathbf C]=o_P(N^{-1}).
\]

Combining the bounds on $\mathbf T_{2n}$ and $\mathbf T_{3n}$, the conditional mean identity $\mathbb E[\mathbf T_{1n}\mid\mathbf C]=\boldsymbol\Sigma(\mathbf C)+\mathbf R_2(\mathbf C)$, and the concentration of $\mathbf T_{1n}$ around its conditional mean, we conclude that
\[
\hat{\mathbf\Sigma}_2
=
\boldsymbol\Sigma(\mathbf C)+\mathbf R_2(\mathbf C)+o_P(N^{-1}),
\]
with $\mathbf R_2(\mathbf C)\succeq\mathbf 0$ and $\sup_{\mathbf c\in\mathrm{supp}(\mathbf C)}\|\mathbf R_2(\mathbf c)\|=O(N^{-1})$.

The proof of the marginal statement is identical, with $\mathbf B$ replaced by $\mathbf B_{\mathrm{marg}}$ and $\boldsymbol\tau_{ij}$ by $\boldsymbol\tau_{ij,\mathrm{marg}}$.
\end{proof}

\subsubsection{Proof of Theorem \ref{thm:var_3}}
The proof follows the structure of the proof of Theorem~\ref{thm:var_1}, and relies additionally on the following property of stratified complete randomization, concerning the covariance between functionals of disjoint blocks of the cluster-level assignment vector. While Lemma~\ref{lem:mixing_coef_complete_randomization} bounds the \emph{order} of the covariance induced by complete randomization, the next lemma derives the exact expression for the leading term in this covariance.

\begin{lemma}[Cross-block covariance under stratified complete randomization]
\label{lem:scr_covariance}
Suppose the cluster-level assignments $\mathbf C=(C_1,\ldots,C_n)$ satisfy Assumption~\ref{ass:complete_rand}(i), with strata $\{\mathcal I_k\}_{k=1}^K$, treated counts $I_k=\sum_{i\in\mathcal I_k}C_i$, and treatment fractions $p_k:=I_k/|\mathcal I_k|$. Let $S_1,S_2\subseteq\{1,\ldots,n\}$ be disjoint index sets with $|S_1|\vee|S_2|\leq L$ for a fixed constant $L$, and let $f\in\mathcal{L}(\mathbf C_{S_1})$ and $g\in\mathcal{L}(\mathbf C_{S_2})$ be functions of the corresponding subvectors. For each stratum $k$ define the within-stratum treated counts
\[
T_{1k}=\sum_{i\in S_1\cap\mathcal I_k}C_i,
\qquad
T_{2k}=\sum_{j\in S_2\cap\mathcal I_k}C_j .
\]
All covariances below are taken under the assignment law of Assumption~\ref{ass:complete_rand}(i). Assume that $f$ and $g$ are bounded by $1$ in absolute value, and write $\mathcal K^{\ast}=\{k:p_k\in[p,1-p]\}$ for the set of nondegenerate strata. Then
\[
\mathrm{Cov}(f,g)
=-\sum_{k\in\mathcal K^{\ast}}
\frac{1}{|\mathcal I_k|\,p_k(1-p_k)}\,
\mathrm{Cov}(f,T_{1k})\,\mathrm{Cov}(g,T_{2k})
+o(n^{-1}),
\]
where the $o(n^{-1})$ remainder is uniform over all $f,g$ bounded by $1$ and all disjoint $S_1,S_2$ with $|S_1|\vee|S_2|\leq L$.
\end{lemma}

\begin{proof}[Proof of Theorem \ref{thm:var_3}]
We prove the claims for the HAC estimator $\hat{\mathbf\Sigma}_1$ in part~(i) and for the bias-corrected estimator $\hat{\boldsymbol\Sigma}_3$ in part~(ii); the marginal claims for $\hat{\mathbf\Sigma}_{1,\mathrm{marg}}$ and $\hat{\boldsymbol\Sigma}_{3,\mathrm{marg}}$ follow from identical arguments upon replacing $\mathbf B$ with $\mathbf B_{\mathrm{marg}}$, $\boldsymbol\tau_{ij}$ with $\boldsymbol\tau_{ij,\mathrm{marg}}$, and $\mathbf V_{ij}$ with $\mathbf V_{ij}^{\mathrm{marg}}$.

\medskip
\noindent\textbf{Part (i).}
We adopt the notation of the proof of Theorem~\ref{thm:var_1}: $\mathbf V_{ij}$, $\hat{\mathbf V}_{ij}$, $\mathbf D_{ij}$, $\boldsymbol\delta=\hat{\boldsymbol\tau}-\boldsymbol\tau$, and the decomposition $\mathbf V_{ij}=\mathbf U_{ij,1}+\mathbf U_{ij,2}$, where $\mathbf U_{ij,2}$ is deterministic. Since the unbiasedness in Assumption~\ref{ass:weight} is defined with respect to the actual randomization distribution, the same argument as in the proof of Theorem~\ref{thm:clt} again gives $\mathbb E[\mathbf U_{ij,1}]=\mathbf 0$ under Assumption~\ref{ass:complete_rand}(i). 
The summands $\mathbf V_{ij}$ obey the two order bounds
\[
\sup_{(i,j)}\|\mathbf V_{ij}\|_\infty=O\!\left(\frac{1}{N}\right),
\qquad
\sum_{(i,j)}\|\mathbf V_{ij}\|_\infty=O(1),
\]
namely a uniform per-unit bound and a summable aggregate bound, and the same two bounds hold verbatim with $\mathbf V_{ij}$ replaced by any of $\mathbf U_{ij,1}$, $\mathbf U_{ij,2}$, or $\mathbf D_{ij}$. They follow from Assumptions~\ref{ass:CLT_regularity}(iii)--(vi) exactly as in the proof of Theorem~\ref{thm:var_1}, as they do not involve the randomization scheme; here $\|\cdot\|_\infty$ denotes the essential supremum, following the conventions in the proof of Theorem~\ref{thm:variance_order}. For brevity, we refer to these bounds---for $\mathbf V_{ij}$ and its analogues $\mathbf U_{ij,1}$, $\mathbf U_{ij,2}$, and $\mathbf D_{ij}$---collectively as the \emph{envelope bounds}.
Moreover, Theorem~\ref{thm:variance_order}(i) and Chebyshev's inequality gives $\|\boldsymbol\delta\|=O_P(n^{-1/2})$.

Working directly at the matrix level, the argument below departs from the proof of Theorem~\ref{thm:var_1} in exactly two places. First, the expectation of the leading term no longer matches the target variance up to the sparsity error alone: complete randomization generates additional covariance terms across disjoint cluster-level neighborhoods, which both kernels omit, and which we show form a negative-semidefinite contribution at leading order through the treated-count covariance identity of Lemma~\ref{lem:scr_covariance}. Second, the variance of the leading term involves additional weakly dependent tuples that are linked only through complete randomization, which we control through Lemmas~\ref{lem:mixing_coef} and~\ref{lem:mixing_coef_complete_randomization}.

Substituting $\hat{\mathbf V}_{ij}=\mathbf V_{ij}-\mathbf D_{ij}\circ\boldsymbol\delta$ into the definition of $\hat{\mathbf\Sigma}^{\mathbf K_1}$ and expanding the outer product, the symmetry of $\mathbf K_1$ gives
\[
\hat{\mathbf\Sigma}^{\mathbf K_1}
=
\mathbf T_{1n}-\mathbf T_{2n}-\mathbf T_{2n}^\top+\mathbf T_{3n},
\]
where
\[
\mathbf T_{1n}:=\sum_{(i,j),(i',j')}(\mathbf K_1)_{ij,i'j'}\mathbf V_{ij}\mathbf V_{i'j'}^\top,
\quad
\mathbf T_{2n}:=\sum_{(i,j),(i',j')}(\mathbf K_1)_{ij,i'j'}\mathbf V_{ij}(\mathbf D_{i'j'}\circ\boldsymbol\delta)^\top,
\]
\[
\mathbf T_{3n}:=\sum_{(i,j),(i',j')}(\mathbf K_1)_{ij,i'j'}(\mathbf D_{ij}\circ\boldsymbol\delta)(\mathbf D_{i'j'}\circ\boldsymbol\delta)^\top.
\]
Since $\mathbf K_1$ has $\|\mathbf K_1\|_F^2=\sum_iN_i^2=O(N^2/n)$ nonzero entries,
\[
\|\mathbf T_{2n}\|
\leq
\|\mathbf K_1\|_F^2\,\sup_{(i,j)}\|\mathbf V_{ij}\mathbf D_{i'j'}^\top\|_\infty\,\|\boldsymbol\delta\|
=O\!\left(\frac{N^2}{n}\right)O\!\left(N^{-2}\right)O_P\!\left(n^{-1/2}\right)=o_P(n^{-1}),
\]
and similarly $\|\mathbf T_{3n}\|=O(N^2/n)\,O(N^{-2})\,O_P(n^{-1})=o_P(n^{-1})$. These bounds use only the envelope bounds and $\|\boldsymbol\delta\|=O_P(n^{-1/2})$, exactly as in the proof of Theorem~\ref{thm:var_1}, and do not involve the randomization scheme.

We now analyze the leading term $\mathbf T_{1n}$, beginning with its expectation. Since $\mathbf V_{ij}=\mathbf U_{ij,1}+\mathbf U_{ij,2}$ with $\mathbb E[\mathbf U_{ij,1}]=\mathbf 0$ and $\mathbf U_{ij,2}$ deterministic, the cross terms vanish in expectation and
\[
\mathbb E[\mathbf T_{1n}]
=\sum_{(i,j),(i',j')}(\mathbf K_1)_{ij,i'j'}\mathbb E[\mathbf U_{ij,1}\mathbf U_{i'j',1}^\top]
+\mathbf R_1,
\qquad
\mathbf R_1:=\sum_{(i,j),(i',j')}(\mathbf K_1)_{ij,i'j'}\mathbf U_{ij,2}\mathbf U_{i'j',2}^\top,
\]
where $\mathbf R_1$ is positive semidefinite, being the same matrix as in the proof of Theorem~\ref{thm:var_1}. For the same reason, $\mathrm{Cov}(\mathbf V_{ij},\mathbf V_{i'j'})=\mathbb E[\mathbf U_{ij,1}\mathbf U_{i'j',1}^\top]$ for every pair of units, so the target variance decomposes as
\[
\boldsymbol\Sigma
=\sum_{(i,j),(i',j')}\mathrm{Cov}(\mathbf V_{ij},\mathbf V_{i'j'})
=\sum_{(i,j),(i',j')}(\mathbf K_2)_{ij,i'j'}\mathbb E[\mathbf U_{ij,1}\mathbf U_{i'j',1}^\top]
+\boldsymbol{\mathcal E},
\]
where
\[
\boldsymbol{\mathcal E}
:=\sum_{(i,j)}\sum_{(i',j')\notin\Lambda_{ij}}\mathrm{Cov}\bigl(\mathbf V_{ij},\mathbf V_{i'j'}\bigr)
\]
collects the cross-covariance matrices of the pairs whose cluster-level neighborhoods are disjoint. This is the first point of departure from the proof of Theorem~\ref{thm:var_1}: under Assumption~\ref{ass:independent_rand}(i) these covariances vanish, whereas under complete randomization they are nonzero in general, and they are omitted by both kernels. Combining the two displays,
\[
\mathbb E[\mathbf T_{1n}]-\boldsymbol\Sigma
=-\sum_{(i,j),(i',j')}(\mathbf K_2-\mathbf K_1)_{ij,i'j'}\mathbb E[\mathbf U_{ij,1}\mathbf U_{i'j',1}^\top]
+\mathbf R_1
-\boldsymbol{\mathcal E}.
\]
By Assumption~\ref{ass:sparsity}, $\mathbf K_2-\mathbf K_1$ has at most $o(N^2/n)$ nonzero entries, and $\sup_{(i,j)}\|\mathbf U_{ij,1}\|_\infty=O(N^{-1})$, so the first sum is $o(N^2/n)\cdot O(N^{-2})=o(n^{-1})$ entrywise, as in the proof of Theorem~\ref{thm:var_1}. The matrix $\mathbf R_1$ is positive semidefinite.

We now show that the leading part of $\boldsymbol{\mathcal E}$ is negative semidefinite, so that omitting these cross-covariances can only make the estimator more conservative. Fix a pair with $\mathcal N_{ij}^{(cl)}\cap\mathcal N_{i'j'}^{(cl)}=\varnothing$. By the law of total covariance,
\[
\mathrm{Cov}(\mathbf V_{ij},\mathbf V_{i'j'})
=\mathrm{Cov}\bigl(\mathbb E[\mathbf V_{ij}\mid\mathbf C],\,\mathbb E[\mathbf V_{i'j'}\mid\mathbf C]\bigr)
+\mathbb E\bigl[\mathrm{Cov}(\mathbf V_{ij},\mathbf V_{i'j'}\mid\mathbf C)\bigr].
\]
As in the proof of Theorem~\ref{thm:variance_order}, the second term vanishes: given $\mathbf C$, Assumption~\ref{ass:cluster_assignment}(ii) makes the unit-level treatments independent across clusters, and $\mathcal N_{ij}$ and $\mathcal N_{i'j'}$ lie in disjoint sets of clusters, so $\mathrm{Cov}(\mathbf V_{ij},\mathbf V_{i'j'}\mid\mathbf C)=\mathbf 0$ a.s. By the same conditional independence, $\mathbf h_{ij}:=\mathbb E[\mathbf V_{ij}\mid\mathbf C]=\mathbb E[\mathbf V_{ij}\mid\mathbf C_{\mathcal N_{ij}}]$ is a function of $\mathbf C_{\mathcal N_{ij}^{(cl)}}$ alone, and analogously $\mathbf h_{i'j'}:=\mathbb E[\mathbf V_{i'j'}\mid\mathbf C]=\mathbb E[\mathbf V_{i'j'}\mid\mathbf C_{\mathcal N_{i'j'}}]$ is a function of $\mathbf C_{\mathcal N_{i'j'}^{(cl)}}$ alone.

The identity of Lemma~\ref{lem:scr_covariance} now applies entrywise to the matrix $\mathrm{Cov}(\mathbf V_{ij},\mathbf V_{i'j'})=\mathrm{Cov}(\mathbf h_{ij},\mathbf h_{i'j'})$, whose two arguments are functions of the disjoint blocks $\mathbf C_{\mathcal N_{ij}^{(cl)}}$ and $\mathbf C_{\mathcal N_{i'j'}^{(cl)}}$, each meeting at most $C_{\mathcal N^{(cl)}}=O(1)$ clusters by Assumption~\ref{ass:CLT_regularity}(ii). Applying the lemma to each coordinate pair $(h_{ij,l},h_{i'j',l'})$, which have sup-norm $\sup_{(i,j)}\|\mathbf V_{ij}\|_\infty=O(N^{-1})$ so that the identity is rescaled by bilinearity, and assembling the $d\times d$ entries into a single matrix identity gives
\[
\mathrm{Cov}(\mathbf V_{ij},\mathbf V_{i'j'})
=-\sum_{k\in\mathcal K^{\ast}}\frac{1}{|\mathcal I_k|\,p_k(1-p_k)}\,
\mathbf v_{ij,k}\,\mathbf v_{i'j',k}^\top+\mathbf r_{ij,i'j'},
\qquad
\mathbf v_{ij,k}:=\mathrm{Cov}\bigl(\mathbf V_{ij},T_{1k}^{(ij)}\bigr)\in\mathbb R^d,
\]
where $T_{1k}^{(ij)}=\sum_{i''\in\mathcal N_{ij}^{(cl)}\cap\mathcal I_k}C_{i''}$ is the treated count of unit $(i,j)$'s block within stratum $k$, and the matrix remainder is $\mathbf r_{ij,i'j'}=O(N^{-2})\,o(n^{-1})$ entrywise, uniformly over all disjoint pairs. Here we also used $\mathrm{Cov}(\mathbf h_{ij},T_{1k}^{(ij)})=\mathrm{Cov}(\mathbf V_{ij},T_{1k}^{(ij)})$, which holds because $\mathbf h_{ij}=\mathbb E[\mathbf V_{ij}\mid\mathbf C]$ and $T_{1k}^{(ij)}$ is $\sigma(\mathbf C)$-measurable.

Summing this matrix identity over the disjoint pairs and extending it to all ordered pairs of units, where the all-pairs sum factorizes through $\sum_{(i,j),(i',j')}\mathbf v_{ij,k}\mathbf v_{i'j',k}^\top=\mathbf s_k\mathbf s_k^\top$ with $\mathbf s_k=\sum_{(i,j)}\mathbf v_{ij,k}$, gives
\[
\boldsymbol{\mathcal E}
=\sum_{(i,j)}\sum_{(i',j')\notin\Lambda_{ij}}\mathrm{Cov}(\mathbf V_{ij},\mathbf V_{i'j'})
=-\sum_{k\in\mathcal K^{\ast}}\frac{1}{|\mathcal I_k|\,p_k(1-p_k)}\,\mathbf s_k\mathbf s_k^\top+o(n^{-1}).
\]
The extension to all ordered pairs costs only $o(n^{-1})$: each summand $\mathbf v_{ij,k}\mathbf v_{i'j',k}^\top/[|\mathcal I_k|p_k(1-p_k)]$ has entries of order $O(N^{-1})\cdot O(N^{-1})\cdot O(n^{-1})=O(n^{-1}N^{-2})$ because $\|\mathbf v_{ij,k}\|=O(N^{-1})$, $|\mathcal I_k|=\Omega(n)$, and $K=O(1)$; by the neighborhood-counting argument in the proof of Theorem~\ref{thm:var_1} (Assumptions~\ref{ass:CLT_regularity}(ii) and (vi)--(vii)) at most $O(N^2/n)$ ordered pairs have overlapping cluster-level neighborhoods, so the added pairs contribute $O(N^2/n)\cdot O(n^{-1}N^{-2})=o(n^{-1})$ entrywise, and the accumulated remainders $\sum\mathbf r_{ij,i'j'}$ are likewise $o(N^{-2}n^{-1})O(N^{2})=o(n^{-1})$ entrywise over the at most $O(N^2)$ pairs. The leading term in $\boldsymbol{\mathcal E}$ is therefore exactly $-\mathbf M$, where $\mathbf M$ is the dense-dependence term of the theorem statement. Since $\mathbf M$ is a nonnegative combination of the rank-one matrices $\mathbf s_k\mathbf s_k^\top\succeq\mathbf 0$, it is positive semidefinite, so this leading term $-\mathbf M$ is negative semidefinite and omitting these cross-covariances can only enlarge the variance estimator. Hence
\[
\mathbb E[\mathbf T_{1n}]
=\boldsymbol\Sigma+\mathbf R_1-\boldsymbol{\mathcal E}+o(n^{-1})
=\boldsymbol\Sigma+\mathbf R_1+\mathbf M+o(n^{-1})
=\boldsymbol\Sigma+\mathbf R_3+o(n^{-1}),
\qquad
\mathbf R_3\succeq\mathbf R_1\succeq\mathbf 0 .
\]
It remains to check that both $\mathbf R_1$ and $\mathbf M$ are positive semidefinite and of order $O(n^{-1})$, as asserted in part~(i). Positive-semidefiniteness was already established above: $\mathbf R_1\succeq\mathbf 0$ as the bias matrix of Theorem~\ref{thm:var_1}, and $\mathbf M=\sum_{k\in\mathcal K^{\ast}}\mathbf s_k\mathbf s_k^\top/[|\mathcal I_k|p_k(1-p_k)]\succeq\mathbf 0$ as a nonnegative combination of rank-one matrices. For the orders, $\mathbf R_1$ satisfies
\[
\|\mathbf R_1\|\leq\|\mathbf K_1\|_F^2\sup_{(i,j)}\|\mathbf U_{ij,2}\mathbf U_{i'j',2}^\top\|_\infty=O(N^2/n)\,O(N^{-2})=O(n^{-1}),
\]
exactly as in the proof of Theorem~\ref{thm:var_1}, while $\mathbf M$ inherits its order from $\boldsymbol{\mathcal E}$: the mixing bound established in the proof of Theorem~\ref{thm:variance_order} (the case of Assumption~\ref{ass:complete_rand}(i)) gives
\[
\|\boldsymbol{\mathcal E}\|
\leq
4\cdot O\!\left(\frac1n\right)\Bigl(\sum_{(i,j)}\|\mathbf V_{ij}\|_\infty\Bigr)^2
=O\!\left(\frac1n\right),
\]
so $\|\mathbf M\|\leq\|\boldsymbol{\mathcal E}\|+o(n^{-1})=O(n^{-1})$.

We can also show that
\[
\mathbf T_{1n}-\mathbb E[\mathbf T_{1n}]=o_P(n^{-1})
\]
via a variance bound. For any fixed coordinate pair $(l,l')$, the $(l,l')$-entry of $\mathbf T_{1n}$ is $(\mathbf T_{1n})_{l,l'}=\sum_{(i,j),(i',j')}(\mathbf K_1)_{ij,i'j'}V_{ij,l}V_{i'j',l'}$, and we claim that $\mathrm{Var}((\mathbf T_{1n})_{l,l'})=o(n^{-2})$. Expanding the variance over the within-cluster pairs selected by $\mathbf K_1$,
\[
\mathrm{Var}\bigl((\mathbf T_{1n})_{l,l'}\bigr)
=\sum_{\substack{1\leq i,i'\leq n\\ 1\leq j_1,j_2\leq N_i,\ 1\leq j_1',j_2'\leq N_{i'}}}
\mathrm{Cov}\bigl(V_{ij_1,l}V_{ij_2,l'},\,V_{i'j_1',l}V_{i'j_2',l'}\bigr).
\]
We split the sextuples $(i,j_1,j_2,i',j_1',j_2')$ according to whether the joint cluster-level neighborhoods overlap, i.e., whether
\[
\bigl(\mathcal N_{ij_1}^{(cl)}\cup\mathcal N_{ij_2}^{(cl)}\bigr)\cap\bigl(\mathcal N_{i'j_1'}^{(cl)}\cup\mathcal N_{i'j_2'}^{(cl)}\bigr)\neq\varnothing.
\]

For the overlapping sextuples, we use the crude bound
\[
\bigl|\mathrm{Cov}\bigl(V_{ij_1,l}V_{ij_2,l'},\,V_{i'j_1',l}V_{i'j_2',l'}\bigr)\bigr|
\leq2\sup_{(i,j,l)}|V_{ij,l}|^4=O(N^{-4}),
\]
which requires no independence between the two products. The number of overlapping sextuples is $O(N^4/n^3)$ by exactly the counting argument in the proof of Theorem~\ref{thm:var_1}, which uses only Assumptions~\ref{ass:CLT_regularity}(ii) and (vi)--(vii) and not the randomization scheme. The total contribution of the overlapping sextuples is therefore $O(N^4/n^3)\cdot O(N^{-4})=O(n^{-3})$.

For the non-overlapping sextuples, the second point of departure from the proof of Theorem~\ref{thm:var_1} arises: under Assumption~\ref{ass:independent_rand}(i) the two products would be independent and these covariances would vanish, whereas under complete randomization the two products remain weakly dependent through the cluster-level assignment, even though their cluster-level neighborhoods are disjoint. Write
\[
\Pi:=V_{ij_1,l}V_{ij_2,l'},
\qquad
\Pi':=V_{i'j_1',l}V_{i'j_2',l'},
\qquad
G:=\mathcal N_{ij_1}^{(cl)}\cup\mathcal N_{ij_2}^{(cl)},
\qquad
G':=\mathcal N_{i'j_1'}^{(cl)}\cup\mathcal N_{i'j_2'}^{(cl)},
\]
so that $G\cap G'=\varnothing$ and $|G|\vee|G'|\leq2C_{\mathcal N^{(cl)}}$ by Assumption~\ref{ass:CLT_regularity}(ii). By the law of total covariance and the conditional independence of unit-level treatments across clusters given $\mathbf C$ (Assumption~\ref{ass:cluster_assignment}(ii)),
\[
\mathrm{Cov}(\Pi,\Pi')
=\mathrm{Cov}\bigl(\mathbb E[\Pi\mid\mathbf C],\,\mathbb E[\Pi'\mid\mathbf C]\bigr),
\]
where $\mathbb E[\Pi\mid\mathbf C]$ is a function of $\mathbf C_G$ with $\|\mathbb E[\Pi\mid\mathbf C]\|_\infty\leq\sup_{(i,j,l)}|V_{ij,l}|^2=O(N^{-2})$, and analogously for $\mathbb E[\Pi'\mid\mathbf C]$. Under Assumption~\ref{ass:complete_rand}(i), the cluster-level assignments satisfy the conditions of Lemma~\ref{lem:mixing_coef_complete_randomization} with population size $n$, $M\asymp n$, and $C_S=2C_{\mathcal N^{(cl)}}$, so
\[
\psi\bigl(\sigma(\mathbf C_G),\,\sigma(\mathbf C_{G'})\bigr)=O\!\left(\frac1n\right)
\]
uniformly over the non-overlapping sextuples. Lemma~\ref{lem:mixing_coef} then yields
\[
\bigl|\mathrm{Cov}(\Pi,\Pi')\bigr|
\leq4\,\bigl\|\mathbb E[\Pi\mid\mathbf C]\bigr\|_\infty\bigl\|\mathbb E[\Pi'\mid\mathbf C]\bigr\|_\infty\cdot O\!\left(\frac1n\right)
=O\!\left(\frac{1}{N^4n}\right).
\]
The number of non-overlapping sextuples is at most the total number of sextuples,
\(
O\left(\bigl(\sum_iN_i^2\bigr)^2\right)=O(N^4/n^2),
\)
so their total contribution is $O(N^4/n^2)\cdot O(N^{-4}n^{-1})=O(n^{-3})$. Combining the two parts,
\[
\mathrm{Var}\bigl((\mathbf T_{1n})_{l,l'}\bigr)=O(n^{-3})=o(n^{-2}),
\]
and Chebyshev's inequality gives $(\mathbf T_{1n})_{l,l'}-\mathbb E[(\mathbf T_{1n})_{l,l'}]=o_P(n^{-1})$ for every coordinate pair $(l,l')$, hence $\mathbf T_{1n}-\mathbb E[\mathbf T_{1n}]=o_P(n^{-1})$. 

Combining the bounds on $\mathbf T_{2n}$ and $\mathbf T_{3n}$, the bias of $\mathbf T_{1n}$, and the concentration of $\mathbf T_{1n}$ around its mean, we conclude that
\[
\hat{\mathbf\Sigma}^{\mathbf K_1}
=\boldsymbol\Sigma+\mathbf R_3+o_P(n^{-1}).
\]

Turning to $\hat{\mathbf\Sigma}^{\mathbf K_2}$, arguments analogous to those used in the proof of Theorem~\ref{thm:var_1} yield
\[
\hat{\mathbf\Sigma}^{\mathbf K_2}
=
\hat{\mathbf\Sigma}^{\mathbf K_1}
+o_P(n^{-1})
=
\boldsymbol\Sigma+\mathbf R_3+o_P(n^{-1}).
\]
Part~(i) follows by combining the results for $\hat{\mathbf\Sigma}^{\mathbf K_1}$ and $\hat{\mathbf\Sigma}^{\mathbf K_2}$ with the facts that $\mathbf R_1$ and $\mathbf M$ are positive semidefinite of order $O(n^{-1})$.

\medskip
\noindent\textbf{Part (ii).}
By part~(i), the HAC estimator satisfies $\hat{\mathbf\Sigma}_1=\boldsymbol\Sigma+\mathbf R_1+\mathbf M+o_P(n^{-1})$, so the bias-corrected estimator obeys
\[
\hat{\boldsymbol\Sigma}_3
=\hat{\mathbf\Sigma}_1-\hat{\mathbf M}
=\boldsymbol\Sigma+\mathbf R_1+(\mathbf M-\hat{\mathbf M})+o_P(n^{-1}).
\]
The only new ingredient relative to part (i) is therefore the behavior of the
plug-in $\hat{\mathbf M}$, and it suffices to prove that $\hat{\mathbf M}$ is
consistent for $\mathbf M$ at the relevant scale, i.e.\ $\hat{\mathbf M}=\mathbf M+o_P(n^{-1})$.
Recall the plug-in estimates
\[
\hat{\mathbf M}=\sum_{k\in\mathcal K^{\ast}}\frac{\hat{\mathbf s}_k\hat{\mathbf s}_k^\top}{|\mathcal I_k|\,p_k(1-p_k)},
\qquad
\hat{\mathbf s}_k=\sum_{(i,j)}\bigl(T_{ij,k}-m_{ij,k}p_k\bigr)\hat{\mathbf V}_{ij},
\]
together with their population counterparts
\[
\mathbf M=\sum_{k\in\mathcal K^{\ast}}\frac{\mathbf s_k\mathbf s_k^\top}{|\mathcal I_k|\,p_k(1-p_k)},
\qquad
\mathbf s_k=\sum_{(i,j)}\mathbb E\bigl[(T_{ij,k}-m_{ij,k}p_k)\mathbf V_{ij}\bigr],
\]
where we used $\mathbf s_k=\sum_{(i,j)}\mathrm{Cov}(\mathbf V_{ij},T_{ij,k})$ and the known centering $\mathbb E\,T_{ij,k}=m_{ij,k}p_k$. 
The envelope bounds together with $|T_{ij,k}-m_{ij,k}p_k|=O(1)$ give $\|\mathbf s_k\|=O(1)$. To prove $\hat{\mathbf M}=\mathbf M+o_P(n^{-1})$, it therefore suffices to establish
\[
\hat{\mathbf s}_k=\mathbf s_k+o_P(1)
\qquad\text{uniformly over }k\in\mathcal K^{\ast}.
\]
Combined with $\|\mathbf s_k\|=O(1)$ this gives $\hat{\mathbf s}_k\hat{\mathbf s}_k^\top=\mathbf s_k\mathbf s_k^\top+o_P(1)$, and since $p_k\in[p,1-p]$ and $|\mathcal I_k|=\Omega(n)$, dividing by $|\mathcal I_k|p_k(1-p_k)=\Omega(n)$ and summing over the $K=O(1)$ strata yields $\hat{\mathbf M}=\mathbf M+o_P(n^{-1})$.

Decompose $\hat{\mathbf s}_k=\tilde{\mathbf s}_k+(\hat{\mathbf s}_k-\tilde{\mathbf s}_k)$ with $\tilde{\mathbf s}_k:=\sum_{(i,j)}(T_{ij,k}-m_{ij,k}p_k)\mathbf V_{ij}$, the infeasible version that uses $\boldsymbol\tau$ in place of $\hat{\boldsymbol\tau}$. The plug-in difference is
\[
\hat{\mathbf s}_k-\tilde{\mathbf s}_k
=-\sum_{(i,j)}\bigl(T_{ij,k}-m_{ij,k}p_k\bigr)\,(\mathbf D_{ij}\circ\boldsymbol\delta),
\]
and since $|T_{ij,k}-m_{ij,k}p_k|=O(1)$ over the $O(1)$ neighborhood clusters, $\sup_{(i,j)}\|\mathbf D_{ij}\|_\infty=O(N^{-1})$, and $\|\boldsymbol\delta\|=O_P(n^{-1/2})$, the same accounting as for $\mathbf T_{2n}$ in part (i) gives $\hat{\mathbf s}_k-\tilde{\mathbf s}_k=O_P(n^{-1/2})=o_P(1)$. For the infeasible version, $\mathbb E[\tilde{\mathbf s}_k]=\mathbf s_k$ by construction, so it remains to bound its variance. By an argument analogous to the variance calculation used in the proof of Theorem~\ref{thm:variance_order} under Assumption~\ref{ass:complete_rand}(i)---splitting the pairs of units according to whether their cluster-level neighborhoods overlap, bounding the $O(N^2/n)$ overlapping pairs by the crude envelope bound and the non-overlapping pairs through the complete-randomization mixing coefficient of Lemma~\ref{lem:mixing_coef_complete_randomization} via Lemma~\ref{lem:mixing_coef}---we obtain $\mathrm{Var}(\tilde{\mathbf s}_k)=O(n^{-1})=o(1)$. Chebyshev's inequality then yields $\tilde{\mathbf s}_k=\mathbf s_k+o_P(1)$, and combining the two pieces gives $\hat{\mathbf s}_k=\mathbf s_k+o_P(1)$ uniformly over the $O(1)$ strata, as required.

Consequently $\hat{\mathbf M}=\mathbf M+o_P(n^{-1})$, and therefore
\[
\hat{\boldsymbol\Sigma}_3
=\boldsymbol\Sigma+\mathbf R_1+o_P(n^{-1}),
\]
with $\mathbf R_1\succeq\mathbf 0$ the same intrinsic bias matrix as in Theorem~\ref{thm:var_1}. By the same marginal substitution noted at the outset, $\hat{\boldsymbol\Sigma}_{3,\mathrm{marg}}=\boldsymbol\Sigma_{\mathrm{marg}}+\mathbf R_{1,\mathrm{marg}}+o_P(n^{-1})$.
\end{proof}

\subsection{Proof of auxiliary lemmas}\label{subsec:aux_lemmas}
\begin{proof}[Proof of Lemma \ref{lem:complete_neighborhood_size}]
By Assumption \ref{ass:CLT_regularity}(ii), we have
\[
|\mathcal{N}_{ij}^{(cl)}|\leq C_{\mathcal{N}^{(cl)}}
\]
for any unit $(i,j)$ and
    \[
    |\cup_{j=1}^{N_i}\mathcal{N}_{ij}|\leq C_{\mathcal{N}^{(cl)}}N_i\leq C_{\mathcal{N}^{(cl)}}C_N\frac{N}{n}
    \]
for any cluster $i$, where the last inequality uses the cluster-size balance in Assumption \ref{ass:CLT_regularity}(v).
    Then, for any unit $(i,j)$, since
    \[
    \Lambda_{ij}= \{(i'',j''):\exists (i',j')\in\mathcal{N}_{ij}\ s.t.\ i'\in\mathcal{N}_{ij}^{(cl)},\ (i'',j'')\in\mathcal{N}_{i'j'}\},
    \]
    we have
    \[
    |\Lambda_{ij}|\leq \sum_{i': i'\in\mathcal{N}_{ij}^{(cl)}}|\cup_{j'=1}^{N_{i'}}\mathcal{N}_{i'j'}|\leq |\mathcal{N}_{ij}^{(cl)}|\max_{i'}\{|\cup_{j'=1}^{N_{i'}}\mathcal{N}_{i'j'}|\}\leq C_{\mathcal{N}^{(cl)}}^2C_{N}\frac{N}{n}.
    \]
Therefore,
\[
\max_{(i,j)\in\mathcal{P}}|\Lambda_{ij}|=O\left(\frac{N}{n}\right).
\]
\end{proof}

\begin{proof}[Proof of Lemma \ref{lem:mixing_coef_complete_randomization}]
Consider $\mathbf{w}_t\in\{0,1\}^{S_t}$ with $P(\mathbf{W}_{S_t}=\mathbf{w}_t)>0$,
$t=1,2$. For $t=1,2$ and each stratum $k$, let $\mathbf{w}_{t,k}:=\mathbf{w}_t|_{S_t\cap\mathcal{I}_k}
\in\{0,1\}^{S_t\cap\mathcal{I}_k}$ denote the restriction of $\mathbf{w}_t$ to the index
set $S_t\cap\mathcal{I}_k$.
Let $R(\mathbf{w}_1,\mathbf{w}_2)$ denote the ratio between the conditional and marginal probability of $\mathbf{W}_{S_2}$, defined as
\[
R(\mathbf{w}_1,\mathbf{w}_2):=\frac{P(\mathbf{W}_{S_2}=\mathbf{w}_2\mid\mathbf{W}_{S_1}=\mathbf{w}_1)}
          {P(\mathbf{W}_{S_2}=\mathbf{w}_2)}.
\]
Since treatments across strata are
independent, $R$ admits the stratum-by-stratum factorization
\begin{equation}\label{eq:ratio_factored}
\begin{aligned}
  R&=\prod_{k=1}^{K}R_k(\mathbf{w}_{1,k},\mathbf{w}_{2,k}), \qquad R_k&:=\frac{P(\mathbf{W}_{S_2\cap\mathcal{I}_k}=\mathbf{w}_{2,k}
               \mid\mathbf{W}_{S_1\cap\mathcal{I}_k}=\mathbf{w}_{1,k})}
            {P(\mathbf{W}_{S_2\cap\mathcal{I}_k}=\mathbf{w}_{2,k})}.
\end{aligned}
\end{equation}
Three classes of strata contribute $R_k=1$ exactly and may be ignored.
First, if $S_1\cap\mathcal{I}_k=\varnothing$, the conditioning on
$\mathbf{W}_{S_1\cap\mathcal{I}_k}$ is vacuous and $R_k=1$.
Second, if $S_2\cap\mathcal{I}_k=\varnothing$, both the numerator and denominator of
$R_k$ equal $1$ and again $R_k=1$.
Third, if $m_k/|\mathcal{I}_k|\in\{0,1\}$, the within-stratum treatment is
deterministic ($\mathbf{W}_{\mathcal{I}_k}$ is a.s.\ constant), so $R_k=1$.
Since $|S_1|\leq C_S$, at most $C_S$ strata have $S_1\cap\mathcal{I}_k\neq\varnothing$,
and we restrict attention to those strata for which additionally
$m_k/|\mathcal{I}_k|\in[p,1-p]$, i.e., the treatment proportion is bounded away from
$0$ and $1$.
It therefore suffices to show $|R_k-1|=O(M^{-1})$ for each such stratum.

Fix one such stratum $k$.  Write $n_k:=|\mathcal{I}_k|$, and set
$s_t:=|S_t\cap\mathcal{I}_k|$ and $a_t:=\sum\mathbf{w}_{t,k}$ for $t=1,2$.
Under simple complete randomization over $n_k$ units with $m_k$ treated, both
conditional and marginal probabilities depend on $\mathbf{w}_{t,k}$ only through $a_t$:
\[
  P(\mathbf{W}_{S_2\cap\mathcal{I}_k}=\mathbf{w}_{2,k}
    \mid\mathbf{W}_{S_1\cap\mathcal{I}_k}=\mathbf{w}_{1,k})
  =\frac{\dbinom{n_k-s_1-s_2}{m_k-a_1-a_2}}{\dbinom{n_k-s_1}{m_k-a_1}},
  \qquad
  P(\mathbf{W}_{S_2\cap\mathcal{I}_k}=\mathbf{w}_{2,k})
  =\frac{\dbinom{n_k-s_2}{m_k-a_2}}{\dbinom{n_k}{m_k}}.
\]
To compute $R_k$, write $n^{\underline{j}}:=\prod_{l=0}^{j-1}(n-l)$ for the falling
factorial (with $n^{\underline{0}}:=1$) and use $\binom{n}{j}=n^{\underline{j}}/j!$.
Direct cancellation of the factorials gives
\[
\frac{\dbinom{n_k-s_1-s_2}{m_k-a_1-a_2}}{\dbinom{n_k-s_1}{m_k-a_1}}
  =
\frac{(m_k-a_1)^{\underline{a_2}}
       \cdot (n_k-s_1-s_2)^{\underline{m_k-a_1-a_2}}}
      {(n_k-s_1)^{\underline{m_k-a_1}}}
  =
\frac{(m_k-a_1)^{\underline{a_2}}
       \cdot (n_k-s_1-m_k+a_1)^{\underline{s_2-a_2}}}
      {(n_k-s_1)^{\underline{s_2}}}
\]
and
\[
\frac{\dbinom{n_k}{m_k}}{\dbinom{n_k-s_2}{m_k-a_2}}
  =
\frac{n_k^{\underline{m_k}}}
     {m_k^{\underline{a_2}}
      \cdot (n_k-s_2)^{\underline{m_k-a_2}}}
  =
\frac{n_k^{\underline{s_2}}}
     {m_k^{\underline{a_2}}
      \cdot (n_k-m_k)^{\underline{s_2-a_2}}}.
\]
Multiplying these two expressions and grouping the terms of each falling-factorial
type yields the three-factor decomposition
\[
  R_k
  =\underbrace{\frac{(m_k-a_1)^{\underline{a_2}}}{m_k^{\underline{a_2}}}}_{T_1}
   \cdot
   \underbrace{\frac{(n_k-s_1-m_k+a_1)^{\underline{s_2-a_2}}}{(n_k-m_k)^{\underline{s_2-a_2}}}}_{T_2}
   \cdot
   \underbrace{\frac{n_k^{\underline{s_2}}}{(n_k-s_1)^{\underline{s_2}}}}_{T_3},
\]
which in product form reads
\[
  T_1=\prod_{l=0}^{a_2-1}\frac{m_k-a_1-l}{m_k-l},
  \quad
  T_2=\prod_{l=0}^{s_2-a_2-1}\frac{n_k-s_1-m_k+a_1-l}{n_k-m_k-l},
  \quad
  T_3=\prod_{l=0}^{s_2-1}\frac{n_k-l}{n_k-s_1-l}.
\]
For $T_1$, each factor equals $1-a_1/(m_k-l)$.  Since $m_k\geq p\,n_k=\Omega(M)$
and $a_1\leq s_1\leq C_S$, we have $|a_1/(m_k-l)|\leq C_S/(p\,n_k-C_S)=O(M^{-1})$,
so $|T_1-1|=O(M^{-1})$ (the product has at most $C_S$ terms).
For $T_2$, each factor equals $1-(s_1-a_1)/(n_k-m_k-l)$; since
$n_k-m_k\geq(1-p)\,n_k=\Omega(M)$ and $|s_1-a_1|\leq C_S$, the same reasoning gives
$|T_2-1|=O(M^{-1})$.
For $T_3$, each factor equals $1/\bigl(1-s_1/(n_k-l)\bigr)$; since
$n_k-l\geq n_k-C_S=\Omega(M)$, we get $|T_3-1|=O(M^{-1})$.
Multiplying them together yields $R_k=\bigl(1+O(M^{-1})\bigr)^{O(1)}$, and there exists a constant
$C<\infty$ (depending only on $p$ and $C_S$) such that
\begin{equation}\label{eq:Rk_bound}
  \bigl|R_k(\mathbf{w}_{1,k},\mathbf{w}_{2,k})-1\bigr|\leq \frac{C}{M}
  \quad\text{uniformly over all }\mathbf{w}_{1,k},\mathbf{w}_{2,k}.
\end{equation}

Returning to the global ratio, since at most $C_S$ factors in~\eqref{eq:ratio_factored}
differ from $1$, applying~\eqref{eq:Rk_bound} and the bound $|R_k|\leq 1+C/M$ gives
\[
  \bigl|R(\mathbf{w}_1,\mathbf{w}_2)-1\bigr|
  \leq \sum_{k:\,S_1\cap\mathcal{I}_k\neq\varnothing}\bigl|R_k-1\bigr|
       \cdot\prod_{k'\neq k}|R_{k'}|
  =O(M^{-1})
\]
uniformly over $(\mathbf{w}_1,\mathbf{w}_2)$.  Hence, for every $\mathbf{w}_1$ in the
support of $\mathbf{W}_{S_1}$,
\[
  \bigl|P(\mathbf{W}_{S_2}=\mathbf{w}_2\mid\mathbf{W}_{S_1}=\mathbf{w}_1)
        -P(\mathbf{W}_{S_2}=\mathbf{w}_2)\bigr|
  \leq\frac{C'}{M}\,P(\mathbf{W}_{S_2}=\mathbf{w}_2)
\]
for some constant $C'$, where $C'$ is determined by $C_S$, by $p$, and by the
implied constant in $\min_k|\mathcal I_k|=\Omega(M)$ alone, since these are the
only quantities entering the bounds on $T_1$, $T_2$ and $T_3$ and the number of
nontrivial factors in \eqref{eq:ratio_factored}. In particular $C'$ is uniform
over all admissible pairs $(S_1,S_2)$.

We now pass from atoms to events. With a slight abuse of notation, for
$B\in\mathcal{F}_2$ write $\mathbf{w}_2:B$ for those realizations
$\mathbf{w}_2\in\{0,1\}^{S_2}$ with
$\{\mathbf{W}_{S_2}=\mathbf{w}_2\}\subseteq B$ and
$P(\mathbf{W}_{S_2}=\mathbf{w}_2)>0$, and similarly $\mathbf{w}_1:A$ for
$A\in\mathcal{F}_1$. Summing the atom-level estimate over $\mathbf{w}_2:B$ gives,
for every $\mathbf{w}_1$ in the support of $\mathbf{W}_{S_1}$,
\[
  \bigl|P(B\mid\mathbf{W}_{S_1}=\mathbf{w}_1)-P(B)\bigr|
  \leq
  \sum_{\mathbf{w}_2:\,B}\frac{C'}{M}P(\mathbf{W}_{S_2}=\mathbf{w}_2)
  =\frac{C'}{M}P(B).
\]
For a general $A\in\mathcal{F}_1$ with $P(A)>0$, conditioning on which atom of
$\mathcal{F}_1$ occurs and averaging,
\[
  \bigl|P(B\mid A)-P(B)\bigr|
  =\Bigl|\sum_{\mathbf{w}_1:\,A}P(\mathbf{W}_{S_1}=\mathbf{w}_1\mid A)
   \bigl[P(B\mid\mathbf{W}_{S_1}=\mathbf{w}_1)-P(B)\bigr]\Bigr|
  \leq\frac{C'}{M}P(B) .
\]
Dividing by $P(B)$ for any pair with $P(A)P(B)>0$ and taking the supremum over
$A$ and $B$ yields $\psi(\mathcal{F}_1,\mathcal{F}_2)\leq C'/M$, so the lemma
follows with $C:=C'$.
\end{proof}

\begin{proof}[Proof of Lemma~\ref{lem:scr_covariance}]
Because each $C_i$ is binary, every function of a finite block of coordinates has a unique multilinear expansion, so we may write $f(\mathbf C_{S_1})=\sum_{A\subseteq S_1}\Delta_A f\prod_{i\in A}C_i$ and $g(\mathbf C_{S_2})=\sum_{B\subseteq S_2}\Delta_B g\prod_{j\in B}C_j$ for real coefficients $\Delta_A f$ and $\Delta_B g$; the empty-set terms are constants and affect no covariance, so we discard them and retain only $A\neq\varnothing$ and $B\neq\varnothing$. By the bilinearity of the covariance,
\[
\mathrm{Cov}(f,g)=\sum_{A\neq\varnothing}\sum_{B\neq\varnothing}\Delta_A f\,\Delta_B g\,\mathrm{Cov}\Bigl(\prod_{i\in A}C_i,\prod_{j\in B}C_j\Bigr),
\]
so the whole computation reduces to evaluating the monomial covariance $\mathrm{Cov}\bigl(\prod_{i\in A}C_i,\prod_{j\in B}C_j\bigr)$ for a single pair of nonempty sets $A\subseteq S_1$ and $B\subseteq S_2$. The next several paragraphs therefore fix one such pair $(A,B)$ and analyze this monomial covariance; we reassemble the full sum over $(A,B)$ only afterward. Throughout we write $m^{\underline r}:=\prod_{\ell=0}^{r-1}(m-\ell)$ for the falling factorial, with $m^{\underline 0}:=1$, so that $\binom{m}{r}=m^{\underline r}/r!$. Having fixed $(A,B)$, for each stratum $k$ we abbreviate $a_k:=|A\cap\mathcal I_k|$ and $b_k:=|B\cap\mathcal I_k|$, so that $|A|=\sum_k a_k$ and $|B|=\sum_k b_k$; since $|S_1|\vee|S_2|\leq L$, all of these counts are at most $L=O(1)$, and there are only finitely many such pairs $(A,B)$.

Within a single stratum $\mathcal I_k$, the subvector $\mathbf C_{\mathcal I_k}$ is uniform over the assignments with $\sum_{i\in\mathcal I_k}C_i=I_k$, so for any $A\cap\mathcal I_k$ the probability that all of its coordinates are treated is the ratio of the number of ways to place the remaining $I_k-a_k$ treated units among the remaining $|\mathcal I_k|-a_k$ coordinates to the total number of assignments, namely $\mathbb E[\prod_{i\in A\cap\mathcal I_k}C_i]=\binom{|\mathcal I_k|-a_k}{I_k-a_k}/\binom{|\mathcal I_k|}{I_k}=I_k^{\underline{a_k}}/|\mathcal I_k|^{\underline{a_k}}$. Because the strata are assigned independently under Assumption~\ref{ass:complete_rand}(i), the monomial expectation factorizes across strata, and the same computation applied to $A$, to $B$, and to $A\cup B$ (the last using $A\cap B=\varnothing$, which holds because $S_1\cap S_2=\varnothing$) gives
\[
\mathbb E\Bigl[\prod_{i\in A}C_i\Bigr]=\prod_{k=1}^K\frac{I_k^{\underline{a_k}}}{|\mathcal I_k|^{\underline{a_k}}},
\quad
\mathbb E\Bigl[\prod_{j\in B}C_j\Bigr]=\prod_{k=1}^K\frac{I_k^{\underline{b_k}}}{|\mathcal I_k|^{\underline{b_k}}},
\quad
\mathbb E\Bigl[\prod_{i\in A}C_i\prod_{j\in B}C_j\Bigr]=\prod_{k=1}^K\frac{I_k^{\underline{a_k+b_k}}}{|\mathcal I_k|^{\underline{a_k+b_k}}} .
\]
Consequently the monomial covariance is exactly
\[
\mathrm{Cov}\Bigl(\prod_{i\in A}C_i,\prod_{j\in B}C_j\Bigr)
=\prod_{k=1}^K\frac{I_k^{\underline{a_k+b_k}}}{|\mathcal I_k|^{\underline{a_k+b_k}}}
-\Bigl(\prod_{k=1}^K\frac{I_k^{\underline{a_k}}}{|\mathcal I_k|^{\underline{a_k}}}\Bigr)
\Bigl(\prod_{k=1}^K\frac{I_k^{\underline{b_k}}}{|\mathcal I_k|^{\underline{b_k}}}\Bigr).
\]

We now expand the exact weight to leading order. For any fixed $r=O(1)$ and any stratum with $p_k\in[p,1-p]$, writing $p_k=I_k/|\mathcal I_k|$ and using $|\mathcal I_k|=\Omega(n)$, each factor satisfies $(I_k-\ell)/(|\mathcal I_k|-\ell)=p_k-\ell(1-p_k)/|\mathcal I_k|+O(|\mathcal I_k|^{-2})=p_k-\ell(1-p_k)/|\mathcal I_k|+o(n^{-1})$, and multiplying the $r$ factors out gives
\[
\frac{I_k^{\underline r}}{|\mathcal I_k|^{\underline r}}
= p_k^{\,r}\Bigl[1-\frac{r(r-1)}{2|\mathcal I_k|}\,\frac{1-p_k}{p_k}\Bigr]+o(n^{-1}) .
\]
Applying this expansion to the three falling-factorial ratios appearing in the monomial covariance and using the elementary identity $a_k(a_k-1)+b_k(b_k-1)-(a_k+b_k)(a_k+b_k-1)=-2a_kb_k$, the difference between the product of the two marginal monomial expectations and the joint expectation, contributed by stratum $k$ with all other strata carrying their leading factors, equals
\[
\Bigl(\prod_{k=1}^K\frac{I_k^{\underline{a_k}}}{|\mathcal I_k|^{\underline{a_k}}}\frac{I_k^{\underline{b_k}}}{|\mathcal I_k|^{\underline{b_k}}}\Bigr)
-\prod_{k=1}^K\frac{I_k^{\underline{a_k+b_k}}}{|\mathcal I_k|^{\underline{a_k+b_k}}}
=\Bigl[\prod_{\ell=1}^K p_\ell^{a_\ell+b_\ell}\Bigr]\sum_{k=1}^K\frac{a_kb_k}{|\mathcal I_k|}\frac{1-p_k}{p_k}+o(n^{-1}),
\]
where the $o(n^{-1})$ term is uniform over all subsets $A,B$ and distribution compatible with Assumption \ref{ass:complete_rand}(i). 
This uniformity follows because there are only finitely many choices of $A$ and $B$, 
because $K=O(1)$ by $\min_k |\mathcal I_k|=\Omega(n)$, and because every nondegenerate 
$p_k$ is bounded away from both zero and one, namely $p_k \in [p,1-p]$.
Degenerate strata with $p_k\in\{0,1\}$ contribute nothing, since there $\mathbf C_{\mathcal I_k}$ is almost surely constant, so any monomial involving a coordinate in $\mathcal I_k$ is almost surely constant and drops out of every covariance; this is why the surviving sum runs over $k\in\mathcal K^{\ast}$. The displayed difference is precisely $-\mathrm{Cov}(\prod_{i\in A}C_i,\prod_{j\in B}C_j)$, so substituting it into the covariance decomposition, exchanging the finite order of summation, and using the factorization $a_kb_k\,p_k^{a_k+b_k-1}\prod_{\ell\neq k}p_\ell^{a_\ell+b_\ell}=p_k\,[a_kp_k^{a_k-1}\prod_{\ell\neq k}p_\ell^{a_\ell}]\,[b_kp_k^{b_k-1}\prod_{\ell\neq k}p_\ell^{b_\ell}]$, we obtain
\[
\mathrm{Cov}(f,g)
=-\sum_{k\in\mathcal K^{\ast}}\frac{p_k(1-p_k)}{|\mathcal I_k|}
\Bigl[\sum_{A\neq\varnothing}a_k\,p_k^{a_k-1}\!\!\prod_{\ell\neq k}p_\ell^{a_\ell}\,\Delta_A f\Bigr]
\Bigl[\sum_{B\neq\varnothing}b_k\,p_k^{b_k-1}\!\!\prod_{\ell\neq k}p_\ell^{b_\ell}\,\Delta_B g\Bigr]
+o(n^{-1}) .
\]

It remains to identify the two bracketed sums as treated-count covariances. Introduce the auxiliary product law $\mathbb P_{\mathbf p}$ under which the coordinates $C_i$ are mutually independent with $C_i\sim\mathrm{Bernoulli}(p_k)$ for $i\in\mathcal I_k$, and let $\mathbb E_{\mathbf p}$ and $\mathrm{Cov}_{\mathbf p}$ denote expectation and covariance under it. Under this law $\mathbb E_{\mathbf p}[\prod_{i\in A}C_i]=\prod_k p_k^{a_k}$, so $\mathbb E_{\mathbf p}[f]=\sum_{A}\Delta_A f\prod_k p_k^{a_k}$ is a polynomial in $\mathbf p=(p_1,\ldots,p_K)$, and differentiating term by term, the $A=\varnothing$ term being constant, gives $\partial \mathbb E_{\mathbf p}[f]/\partial p_k=\sum_{A\neq\varnothing}a_k\,p_k^{a_k-1}\prod_{\ell\neq k}p_\ell^{a_\ell}\,\Delta_A f$, which is exactly the bracket attached to $f$ above; the analogous identity holds for $g$. The full score of $\mathbb P_{\mathbf p}$ with respect to $p_k$ is $\sum_{i\in\mathcal I_k}(C_i-p_k)/[p_k(1-p_k)]$, and this yields $\partial\mathbb E_{\mathbf p}[f]/\partial p_k=\mathbb E_{\mathbf p}[f\sum_{i\in\mathcal I_k}(C_i-p_k)]/[p_k(1-p_k)]$; since $f$ depends only on $\mathbf C_{S_1}$ and the coordinates are independent under $\mathbb P_{\mathbf p}$, the terms with $i\notin S_1$ vanish and those with $i\in S_1\cap\mathcal I_k$ sum to $T_{1k}$ up to a constant, so
\[
\frac{\partial \mathbb E_{\mathbf p}[f]}{\partial p_k}=\frac{1}{p_k(1-p_k)}\,\mathrm{Cov}_{\mathbf p}(f,T_{1k}),
\qquad
\frac{\partial \mathbb E_{\mathbf p}[g]}{\partial p_k}=\frac{1}{p_k(1-p_k)}\,\mathrm{Cov}_{\mathbf p}(g,T_{2k}) .
\]
Substituting these into the previous display, the two factors $p_k(1-p_k)$ in the denominators combine with the prefactor $p_k(1-p_k)/|\mathcal I_k|$ to leave $1/[|\mathcal I_k|p_k(1-p_k)]$, so that
\[
\mathrm{Cov}(f,g)
=-\sum_{k\in\mathcal K^{\ast}}\frac{1}{|\mathcal I_k|\,p_k(1-p_k)}\,
\mathrm{Cov}_{\mathbf p}(f,T_{1k})\,\mathrm{Cov}_{\mathbf p}(g,T_{2k})+o(n^{-1}) .
\]

Finally we replace the auxiliary covariances $\mathrm{Cov}_{\mathbf p}$ by the actual ones $\mathrm{Cov}$, and we now make the local comparison between the two laws explicit. Write $\mathbb P$ for the law of Assumption~\ref{ass:complete_rand}(i). Both $\mathbb P$ and $\mathbb P_{\mathbf p}$ make the strata mutually independent, so the restriction of either law to the coordinates in $S_1$ factorizes as a product over $k$ of the within-stratum law of $\mathbf C_{S_1\cap\mathcal I_k}$, and likewise for $S_2$; the two laws differ only in these within-stratum factors. Fix a stratum $k$ and a coordinate set $F\subseteq\mathcal I_k$ with $|F|=s\leq L=O(1)$. For a pattern $\mathbf c_F\in\{0,1\}^F$ with $m:=\sum_{i\in F}c_i$ treated coordinates, counting the assignments that realize it gives the complete-randomization marginal
\[
\mathbb P(\mathbf C_F=\mathbf c_F)
=\frac{\binom{|\mathcal I_k|-s}{I_k-m}}{\binom{|\mathcal I_k|}{I_k}}
=\frac{I_k^{\underline m}\,(|\mathcal I_k|-I_k)^{\underline{\,s-m}}}{|\mathcal I_k|^{\underline s}},
\]
whereas the auxiliary marginal is $\mathbb P_{\mathbf p}(\mathbf C_F=\mathbf c_F)=p_k^{m}(1-p_k)^{s-m}$. For a nondegenerate stratum, $p_k\in[p,1-p]$, and the falling-factorial expansion already used above gives $I_k^{\underline m}=(p_k|\mathcal I_k|)^m(1+O(|\mathcal I_k|^{-1}))$, $(|\mathcal I_k|-I_k)^{\underline{\,s-m}}=((1-p_k)|\mathcal I_k|)^{s-m}(1+O(|\mathcal I_k|^{-1}))$, and $|\mathcal I_k|^{\underline s}=|\mathcal I_k|^{s}(1+O(|\mathcal I_k|^{-1}))$, so the ratio $\mathbb P(\mathbf C_F=\mathbf c_F)/\mathbb P_{\mathbf p}(\mathbf C_F=\mathbf c_F)=1+O(|\mathcal I_k|^{-1})$ uniformly over the $2^{s}=O(1)$ patterns; hence the two within-stratum marginals differ in total-variation distance by $O(|\mathcal I_k|^{-1})$. (A degenerate stratum, $p_k\in\{0,1\}$, makes $\mathbf C_{\mathcal I_k}$ the same point mass under both laws, contributing zero.) Since total-variation distance is subadditive under products of independent coordinates,
\[
d_{\mathrm{TV}}\bigl(\mathbb P|_{\mathbf C_{S_1}},\,\mathbb P_{\mathbf p}|_{\mathbf C_{S_1}}\bigr)
\leq\sum_{k=1}^{K}d_{\mathrm{TV}}\bigl(\mathbb P|_{\mathbf C_{S_1\cap\mathcal I_k}},\,\mathbb P_{\mathbf p}|_{\mathbf C_{S_1\cap\mathcal I_k}}\bigr)
=K\cdot O\bigl(\min_k|\mathcal I_k|^{-1}\bigr)=O(n^{-1}),
\]
using $K=O(1)$ and $\min_k|\mathcal I_k|=\Omega(n)$, and the same bound holds for $\mathbf C_{S_2}$. Consequently, for any function $\Phi$ of $\mathbf C_{S_1}$ with $\|\Phi\|_\infty\leq M$, the coupling characterization of total variation gives $|\mathbb E_{\mathbf p}[\Phi]-\mathbb E[\Phi]|\leq 2M\,d_{\mathrm{TV}}=O(Mn^{-1})$. Applying this to $\Phi\in\{f,\,T_{1k},\,fT_{1k}\}$, which depend only on $\mathbf C_{S_1}$ and are bounded by $1$, $L$, and $L$ respectively, yields $\mathbb E_{\mathbf p}[fT_{1k}]=\mathbb E[fT_{1k}]+O(n^{-1})$, $\mathbb E_{\mathbf p}[f]=\mathbb E[f]+O(n^{-1})$, and $\mathbb E_{\mathbf p}[T_{1k}]=\mathbb E[T_{1k}]+O(n^{-1})$; since these moments are all $O(1)$, expanding $\mathrm{Cov}_{\mathbf p}(f,T_{1k})=\mathbb E_{\mathbf p}[fT_{1k}]-\mathbb E_{\mathbf p}[f]\mathbb E_{\mathbf p}[T_{1k}]$ gives $\mathrm{Cov}_{\mathbf p}(f,T_{1k})=\mathrm{Cov}(f,T_{1k})+O(n^{-1})$, and likewise $\mathrm{Cov}_{\mathbf p}(g,T_{2k})=\mathrm{Cov}(g,T_{2k})+O(n^{-1})$. For arbitrary $f$ and $g$ satisfying the boundedness condition, both covariance terms under the original complete-randomization distribution are bounded by $2L=O(1)$, and hence the remainder terms are uniformly $O(n^{-1})$. 
Combining this with $K=O(1)$ and
\(
\frac{1}{|\mathcal I_k|\,p_k(1-p_k)}=O(n^{-1}),
\)
we obtain
\[
\mathrm{Cov}(f,g)
=-\sum_{k\in\mathcal K^{\ast}}\frac{1}{|\mathcal I_k|\,p_k(1-p_k)}\,
\mathrm{Cov}(f,T_{1k})\,\mathrm{Cov}(g,T_{2k})+o(n^{-1}),
\]
with all covariances taken under Assumption~\ref{ass:complete_rand}(i), which is the claim of the lemma.
\end{proof}

\end{document}